\documentclass[12pt]{ociamthesis}  

\pdfoutput=1 

\usepackage{url}
\usepackage{amssymb}
\usepackage{amsmath}
\usepackage{amsthm}
\usepackage[titletoc]{appendix} 
\usepackage{xcolor} 
\definecolor{grey}{rgb}{0.35, 0.35, 0.35}
\usepackage[colorlinks=true, linkcolor=grey, citecolor=grey, urlcolor=grey]{hyperref} 
\usepackage{textcomp} 
\usepackage{enumitem} 
\usepackage{bbm} 
\usepackage[T1,OT1]{fontenc} 
\usepackage[mono=false]{libertine} 
\usepackage{sectsty} 
\usepackage{microtype} 
\usepackage{etoolbox}
\apptocmd{\sloppy}{\hbadness 10000\relax}{}{}

\allsectionsfont{\fontfamily{ppl}\bfseries\selectfont\raggedright}
\usepackage{quotchap} 
\renewcommand*{\sectfont}{\fontfamily{ppl}\mdseries\selectfont\raggedleft}
\qsetcnfont{ppl}

\UndeclareTextCommand{\l}{OT1}
\DeclareTextSymbolDefault{\l}{T1}

\makeatletter
\def\thmhead@plain#1#2#3{%
  \thmname{#1}\thmnumber{\@ifnotempty{#1}{ }\@upn{#2}}%
  \thmnote{ {\the\thm@notefont#3}}}
\let\thmhead\thmhead@plain
\makeatother

\makeatletter
\g@addto@macro{\UrlBreaks}{\UrlOrds}
\makeatother

\newtheorem{definition}{Definition}[chapter]
\newtheorem{theorem}{Theorem}[chapter]
\newtheorem{lemma}{Lemma}[chapter]

\newtheorem*{theorem*}{Theorem}

\newtheorem{proposition}{Proposition}[chapter]

\DeclareMathOperator{\Tr}{Tr}

\DeclareMathOperator{\Ln}{Ln}

\newcommand\id{\mathrm{id}}

\DeclareMathOperator{\Pa}{Pa}
\DeclareMathOperator{\Ch}{Ch}
\DeclareMathOperator{\Desc}{De}
\DeclareMathOperator{\Ndesc}{Nd}

\DeclareRobustCommand{\trace}{%
  \begingroup\normalfont
  \includegraphics[height=\fontcharht\font`\B]{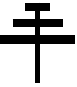}%
  \endgroup
}

\DeclareRobustCommand{\ccopy}{%
  \begingroup\normalfont
  \includegraphics[height=\fontcharht\font`\B]{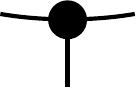}%
  \endgroup
}

\DeclareRobustCommand{\qcopy}{%
  \begingroup\normalfont
  \includegraphics[height=\fontcharht\font`\B]{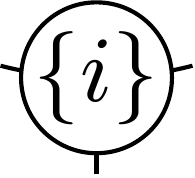}%
  \endgroup
}

\DeclareRobustCommand{\CNOT}{%
  \begingroup\normalfont
  \includegraphics[height=\fontcharht\font`\B]{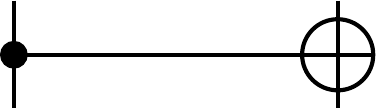}%
  \endgroup
}

\DeclareRobustCommand{\SWAP}{%
  \begingroup\normalfont
  \includegraphics[height=\fontcharht\font`\B]{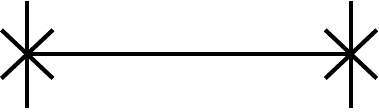}%
  \endgroup
}

\newcommand\eqdef{\mathrel{\overset{\makebox[0pt]{\mbox{\normalfont\tiny def}}}{=}}}

\newcommand\eqquest{\mathrel{\overset{\makebox[0pt]{\mbox{\normalfont\tiny ?}}}{=}}}

\newcommand{\red}[1]{\textcolor{red}{#1}}

\newcommand{\wcomment}[1]{#1}
\renewcommand{\wcomment}[1]{}

\newcommand{\transto}[1]{\mathrel{\overset{\scriptstyle #1}{\leadsto}}}

\newcommand{\docn}{\mathop{}\!\mathrm{do} \ }

\AtBeginDocument{\renewcommand{\d}{\mathop{}\!\mathrm{d}}}

\AtBeginDocument{\renewcommand{\k}{\mathchar'26\mkern-9mu k}}

\usepackage[sort&compress,numbers]{natbib}
\usepackage{doi}

\title{Reality, Causality, and Quantum Theory}   

\author{John-Mark A. Allen}             
\college{St. John's College}  

\renewcommand{\submittedtext}{A thesis submitted for the degree of}
\degree{Doctor of Philosophy}     
\degreedate{Hilary Term 2017}         

\begin{document}	

\baselineskip=18pt plus1pt

\setcounter{secnumdepth}{3}
\setcounter{tocdepth}{3}

\maketitle                  
\begin{alwayssingle}

\thispagestyle{empty}
\begin{center}
\vspace*{1.2cm}
{\LARGE }
\end{center}
\vspace{0.5cm}

\begin{changemargin}{0cm}{6cm}
\begin{flushleft}
``I daresay you haven't had much practice. When I was your age, I always did it for half-an-hour a day. Why, sometimes I've believed as many as six impossible things before breakfast.''
\end{flushleft}

\begin{flushright}
-- \textit{The White Queen} \footnote{\emph{Through the Looking Glass}, Lewis Carroll}
\end{flushright}
\end{changemargin}

\vspace*{1.5cm}

\begin{changemargin}{6cm}{0cm}
\begin{flushleft}
``When you are a Bear of Very Little Brain, and you Think of Things, you find sometimes that a Thing which seemed very Thingish inside you is quite different when it gets out into the open and has other people looking at it.''
\end{flushleft}

\begin{flushright}
-- \textit{Pooh Bear} \footnote{\emph{The House at Pooh Corner}, A. A. Milne}
\end{flushright}
\end{changemargin}

\includegraphics[scale=0.075,angle=0]{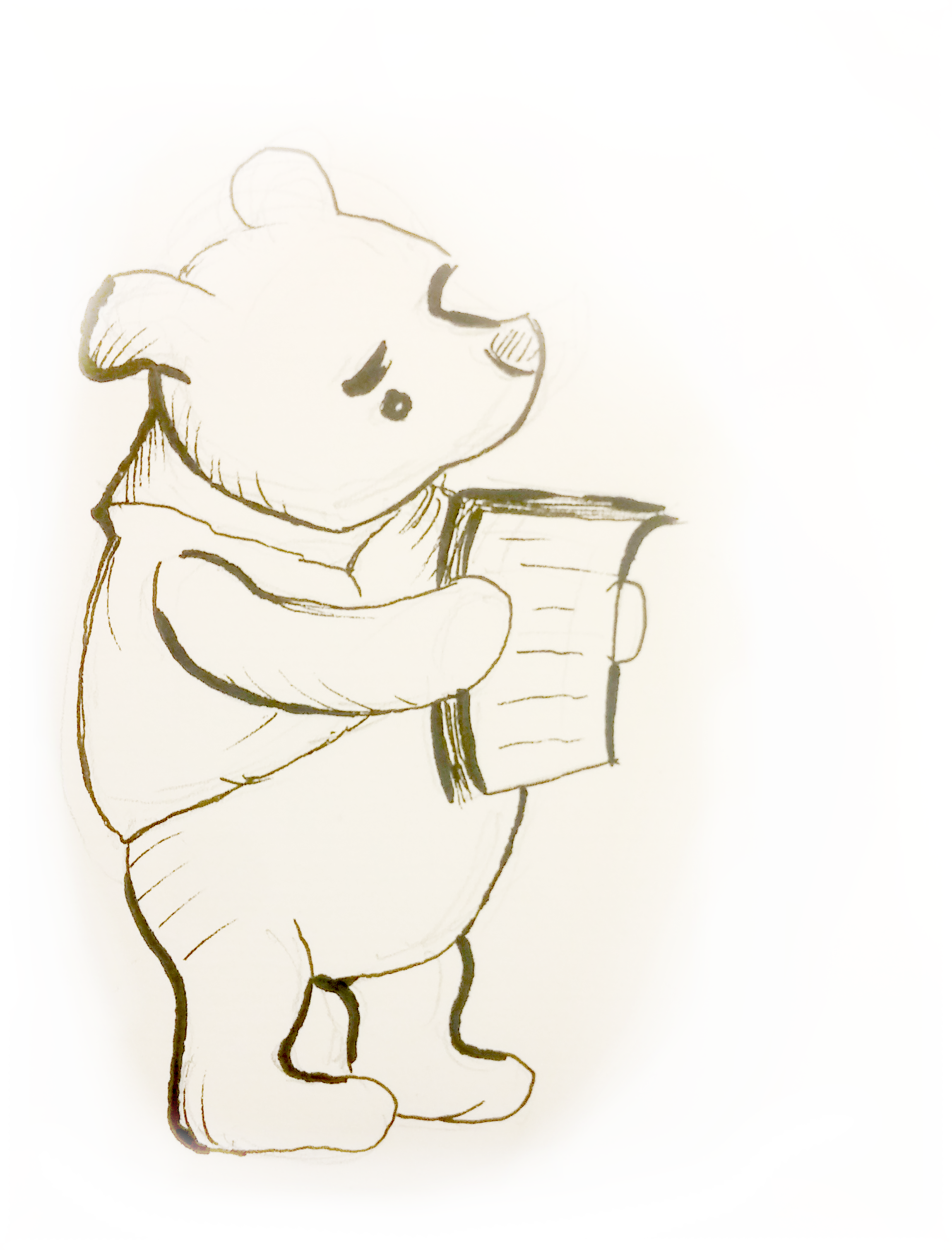} 

\end{alwayssingle}
\begin{dedication}

\noindent

\textit{To}

\bigskip

Mary

\medskip

Peter

\medskip

John

\medskip

Caroline

\medskip

Tim

\medskip

Phil	

\end{dedication}        
\begin{acknowledgements}

{\baselineskip=14pt plus1pt

Down the long wall of the main stairwell in the Scottish National Gallery of Modern Art One there are many names, neatly arranged in unassuming columns of clean Helvetica on the white gallery wall. They appear as a very standard (if impressively long) list of benefactors to the gallery. Look closer and you can recognise some famous names of modern art in the crowd. Regain some awareness of your surroundings and you notice the label identifying it as ``List of Names (Random)'' by Douglas Gordon, an ongoing piece listing all of the people the artist can remember meeting.

Thinking about writing these acknowledgements, my mind keeps coming back to that wall. Acknowledgements tend to end up looking like a long list of loosely related and largely unfamiliar names, perhaps sprinkled with some you recognise. But mostly I end up thinking of all the people Gordon must have forgotten to include. It seems horrible and cruel to forget people, yet failing to mention everyone is inevitable. So I guess I'm starting with an apology---a long-winded and somewhat pretentious apology to anybody who has helped me at all in the making of this thesis but I have failed to mention here. I'm sorry and I am certainly grateful.

Foremost thanks most definitely go to Jon Barrett. As a supervisor, collaborator, source of encouragement, fantastic explanations, and good humour he has been wonderful throughout. In particular should be noted his clarity and patience when explaining an open problem to me for the fourth time (when my memory refuses to play along) and his encouragement in finding further research opportunities and collaborations in Oxford, Waterloo, and elsewhere.

Dom Horsman also deserves more than a special mention here. He has selflessly kept me sane and on-track in many ways over many conversations. He has been both amazingly practical and delightfully irreverent, each exactly when they were needed. I have done nothing to deserve such a wonderful and effective pair to guide me through my time at Oxford and I am forever grateful.

They are far from the only people from the Quantum Group who need thanking, however. In particular Ciar\'{a}n Lee (ah, yeah, it's been grand), Matty Hoban, John Selby, and Stefano Gogioso have been exceptionally accommodating to being distracted by sometimes-interesting sometimes-banal questions. Samson Abramsky, Bob Coecke, Carlo Maria Scandolo, Nathan Walk, Niel de Beaudrap, Subhayan Roy Moulik, Robin Lorenz, Jamie Vicary, Shane Mansfield, Sean Tull, Miriam Backens, Will Zeng, Brendan Fong, Aleks Kissinger, Chris Heunen, Dan Marsden, Sina Salek, and Ognyan Oreshkov all need to be thanked for interesting and productive conversations, advice, and teaching.

I have also benefited from working closely with several collaborators, from whom I have learned much. These include Owen Maroney, Stefano Gogioso, Ciar\'{a}n Lee, Jon Barrett, and Dom Horsman in Oxford. I'm indebted to Rob Spekkens, Matt Pusey, Matt Leifer, and Elie Wolfe at Perimeter Institute and especially Lucien Hardy for hosting my visit there. Special thanks also to L\'{i}dia del Rio in Z\"{u}rich for hosting a productive and enjoyable visit.

Back in the classical world, there are many more people who have been a joy to spend time with. They have helped me to keep my head and humour throughout my DPhil process in Oxford. Alma Brodersen has been unyieldingly delightful and encouraging. Solja H\"{o}ft is absolutely brilliant and a magnificent friend. Jan Cosgrave has been determined to keep me topped up with sleep and dessert wine (dude, with Ross Haines, also deserves utmost respect for slogging all the way through this thesis in search of typos). Sally Le Page has been both inspiring and comforting. The list of wonderful people who've kept me sane and happy goes on and on with Sarah Penington, Marina Lambrakis, Jon Templeman, Katherine Fender, Laura Grima, Kate Radford, Ali Ward, Sam Forbes, Tim Bourns, Max Emmerich, Ellie Milnes-Smith, Martin Bittner, Chris Arran, and Alan Percy. In fact, many people from the St. John's MCR could happily fit on this list. One group that deserves particular note is Oriel College Choir including David Maw, Claire Lowe, Alasdair Cameron, Johanna Hockmann, and Lizzie Searle to name but a few, because there's nothing like depressed Tudor Catholic to brighten up your week.

The list of wonderful friends who have helped in many ways extends well beyond Oxford. In particular, Ed Whittle, Harriet Cook, Chris Hallam, Kenny Masters, Kieran Franklin, Katie Dooley, Emma Wilson, Emily Baker, Becca Lee, Josh Rhodes, Bonnie Barker, Darryl Hoving, Lizzie Briggs, Giulia Postal, Katharine Elliot, Sophie Reed, Elena Teh, Heather Parsons, Olivia Cleary, and Jazzy Ramsay Gray should all be thanked. Tort, Richard, Emma, and Madeleine Olver have all been hugely generous, fun, and kind.

Of course there are several institutions which have practically enabled this thesis. The Department of Computer Science, especially Julie Sheppard and Destiny Chen, has been very supportive in many ways. St. John's College has acted as a home and source of great support. The Engineering and Physical Science Research Council (and, by extension, the UK taxpayer) funded this research, for which I am forever grateful. I also received generous funding support from ETH Z\"{u}rich and the Perimeter Institute for Theoretical Physics for visits.

My family, to whom this thesis is dedicated, quite literally cannot be thanked enough. Brothers, uncles, aunts, cousins, I only become increasingly aware of how absurdly lucky I am to have you. Most of all, my parents and grandparents have always supported me with a compassion and grace that is difficult to describe. Thank you all, for everything that has led up to this thesis and everything that will follow.

Finally, Catherine. There is certainly nobody who understands the personal content of this work better than you. Thank you for your patience and your impatience. Your understanding and your refusal to understand. Your support and your drive. Your excitement and your confusion. Your celebrations and your commiserations. Your presence. Your letters. Your thoughts. Your belief. Thank you.

}

\end{acknowledgements}  
\begin{abstract}

Quantum theory describes our universe incredibly successfully. To our classic\-ally-inclined brains, however, it is a bizarre description that requires a re-imagining of what fundamental reality, or ``ontology'', could look like. This thesis examines different ontological features in light of the success of quantum theory, what it requires, and what it rules out. While these investigations are primarily foundational, they also have relevance to quantum information, quantum communication, and experiments on quantum systems.

The way that quantum theory describes the state of a system is one of its most unintuitive features. It is natural, therefore, to ask whether a similarly strange description of states is required on an ontological level. This thesis proves that almost all quantum superposition states for $d>3$ dimensions must be real---that is, present in the ontology in a well-defined sense. This is a strong requirement which prevents intuitive explanations of the many quantum phenomena which are based on superpositions. A new theorem is also presented showing that quantum theory is incompatible with macro-realist ontologies, where certain physical quantities must always have definite values. This improves on the Leggett-Garg argument, which also aims to prove incompatibility with macro-realism but contains loopholes. Variations on both of these results that are error-tolerant (and therefore amenable to experimentation) are presented, as well as numerous related theorems showing that the ontology of quantum states must be somewhat similar to the quantum states themselves in various specific ways. Extending these same methods to quantum communication, a simple proof is found showing that an exponential number of classical bits are required to communicate a linear number of qubits. That is, classical systems are exponentially bad at storing quantum data.

Causal influences are another part of ontology where quantum theory demands a revision of our classical notions. This follows from the outcomes of Bell experiments, as rigorously shown in recent analyses. Here, the task of constructing a native quantum framework for reasoning about causal influences is tackled. This is done by first analysing the simple example of a common cause, from which a quantum version of Reichenbach's principle is identified. This quantum principle relies on an identification of quantum conditional independence which can be defined in four ways, each naturally generalising a corresponding definition for classical conditional independence. Not only does this allow one to reason about common causes in a quantum experiments, but it can also be generalised to a full framework of quantum causal models (mirroring how classical causal models generalise Reichenbach's principle). This new definition of quantum causal models is illustrated by examples and strengthened by it's foundation on a robust quantum Reichenbach's principle.

An unusual, but surprisingly fruitful, setting for considering quantum ontology is found by considering time travel to the past. This provides a testbed for different ontological concepts in quantum theory and new ways to compare classical and quantum frameworks. It is especially useful for comparing computational properties. In particular, time travel introduces non-linearity to quantum theory, which brings (sometimes implicit) ontological assumptions to the fore while introducing strange new abilities. Here, a model for quantum time travel is presented which arguably has fewer objectionable features than previous attempts, while remaining similarly well-motivated. This model is discussed and compared with previous quantum models, as well as with the classical case.

Together, these threads of investigation develop a better understanding of how quantum theory affects possible ontologies and how ontological prejudices influence quantum theory.

\end{abstract}          

\begin{romanpages}          
\tableofcontents            
\chapter*{Related Publications}
\addcontentsline{toc}{chapter}{\protect\numberline{}Related Publications}

\begin{itemize}

\item[\cite{Allen14}] \textbf{JMAA}, \emph{Treating Time Travel Quantum Mechanically}. Physical Review A \textbf{90}(4), 042107 (2014). arXiv:1401.4933 [quant-ph]
	
\item[\cite{Allen16a}] \textbf{JMAA}, \emph{Quantum Superpositions Cannot be Epistemic}. Quantum Studies: Mathematics and Foundations \textbf{3}(2), 161--177 (2016). arXiv:1501.05969 [quant-ph]

\item[\cite{AllenBarrett+17}] \textbf{JMAA}, J. Barrett, D. C. Horsman, C. M. Lee, and R. W. Spekkens, \emph{Quantum Common Causes and Quantum Causal Models}. Physical Review X \textbf{In preparation} (2017). arXiv:1609.09487 [quant-ph]

\item[\cite{AllenMaroney+16}] \textbf{JMAA}, O. J. E. Maroney, and S. Gogioso, \emph{A Stronger Theorem Against Macro-realism} (2016). arXiv:1610.00022 [quant-ph]

\end{itemize}
\chapter*{Notation} \label{sec:notation}
\addcontentsline{toc}{chapter}{\protect\numberline{}Notation}

Most of the notation used in this thesis is based on standard usage in quantum foundations or physics in general. In order to streamline the text, some of the most common notation used is summarised here. If the reader comes across an unfamiliar symbol in the text, this should be their first port of call for a definition.

The symbol $\eqdef$ is used to indicate that the expression is a definition of the item on the left. Read ``is defined as being equal to''.

For logarithms, $\log$ is used for the base-two logarithm (as in information theory), while $\ln$ is used for base-$e$ natural logarithms.

\subsection*{Quantum Systems and States}

Quantum systems will normally be given capital Latin letters, such as $A$.

\begin{itemize}

\item $\mathcal{H}_A$ is the \emph{Hilbert space} (bounded inner-product complex vector space) corresponding to system $A$.

\item $d_A \eqdef \dim \mathcal{H}_A$ is the dimension of the system $A$.

\item $\mathbbm{1}_A$ is the identity operator on $\mathcal{H}_A$.

\item $\mathcal{H}_A^\ast = (\mathcal{H}_A)^\ast$ is the dual Hilbert space to $\mathcal{H}_A$, with dimension $d_{A^\ast} = d_A$ and identity $\mathbbm{1}_{A^\ast}$.

\item $\mathcal{P}(\mathcal{H}_A) \eqdef \{ |\psi\rangle\in\mathcal{H}_A : \|\psi\| = 1,\; |\psi\rangle \cong e^{\mathrm{i}\theta}|\psi\rangle, \;\forall\theta\in\mathbb{R} \}$ is the set of \emph{pure physical states} in $A$, \emph{viz.} unit vectors in $\mathcal{H}_A$ where vectors equal up to a global phase are considered equivalent.

\item $\mathcal{D}(\mathcal{H}_A)$ is the set of normalised \emph{density operators} (\emph{i.e.} mixed states) of $A$, \emph{viz.} positive Hermitian trace-one operators $\rho$ on $\mathcal{H}_A$.

\item A \emph{superposition} state with respect to some given orthonormal basis $\mathcal{B}$ of $\mathcal{H}$ is any $|\psi\rangle\in\mathcal{P}(\mathcal{H})$ such that $|\psi\rangle\not\in\mathcal{B}$.

\item $AB$, where $A$ and $B$ are each distinct quantum systems, is the bipartite contemporary system comprised of the two with Hilbert space $\mathcal{H}_A \otimes \mathcal{H}_B$. 

\item If a state for a multipartite system $AB$ is given by, \emph{e.g}, $\rho_{AB}$, then the same symbol with omitted indices indicates a partial trace of that state, such as $\rho_{A} = \Tr_B \rho_{AB}$.

\item Products of states that do not share all of their subsystems are taken to implicitly include identities on those systems, \emph{e.g.} $\rho_A \rho_B \eqdef ( \rho_A \otimes \mathbbm{1}_B )( \mathbbm{1}_A \otimes \rho_B )$ and $\rho_{AB}\rho_{BC} \eqdef ( \rho_{AB} \otimes \mathbbm{1}_C )( \mathbbm{1}_A \otimes \rho_{BC} )$.

\item A \emph{measurement} $M$ of a quantum system consists of a set of outcomes $E\in M$ which can be obtained when the measurement is performed.

\item A \emph{basis measurement} $M = \mathcal{B} = \{ |i\rangle \}_i$ is a measurement where the outcomes are pure states $|i\rangle$ which form an orthonormal basis $\mathcal{B}$ for the Hilbert space of system. The probability of obtaining outcome $|i\rangle$ when the system is in state $\rho$ is $\langle i|\rho|i\rangle$.

\item A \emph{POVM measurement} $M = \{ E_i \}_i$ is a measurement where the outcomes are Hermitian positive definite operators on the Hilbert space that form a partition of unity $\sum_i E_i = \mathbbm{1}$. The probability to obtain outcome $E_i$ when the system is in state $\rho$ is $\Tr\left(\rho E_i\right)$.

\item A \emph{von Neumann measurement} $M = \{ E_i \}_i$ is a POVM measurement where all outcomes are projectors $E_i E_j = \delta_{ij} E_i$. If all operators are one-dimensional projectors then this is equivalent to a basis measurement.

\item A \emph{quantum channel} from $A$ to $B$, $\mathcal{E}_{B|A}$, is a completely-positive trace-preserving (CPTP) map from states in $\mathcal{D}(\mathcal{H}_A)$ to states in $\mathcal{D}(\mathcal{H}_B)$.

\item A \emph{quantum instrument} from $A$ to $B$, $\{\mathcal{E}_{B|A}^k\}_k$, is a set of completely-positive trace-non-increasing maps from states in $\mathcal{D}(\mathcal{H}_A)$ to sub-normalised density operators on $\mathcal{H}_B$. Each map is labelled by a classical outcome $k$ and they sum to a quantum channel. These represent general operations on a quantum system that include transformations and measurements. The interpretation is that one such map is applied and the corresponding outcome observed, with probabilities given by the trace of the corresponding output state. The channel obtained by summing the maps outputs the proper mixed state obtained by ignoring this classical outcome.

\item $\mathcal{S}_{|\psi\rangle}$ is the set of \emph{stabiliser unitaries} for the pure state $|\psi\rangle\in\mathcal{P}(\mathcal{H})$. That is, $\mathcal{S}_{|\psi\rangle} \eqdef \{U\text{ on }\mathcal{H}\,:\,U|\psi\rangle=|\psi\rangle\}$. In other words, $\mathcal{S}_{|\psi\rangle}$ is the stabiliser subgroup of the unitary group over $\mathcal{H}$ with respect to $|\psi\rangle$.

\end{itemize}

\subsection*{Quantum Circuit Diagrams}

A standard box-and-wire notation for quantum circuits will be used, specifically with the following conventions.

\begin{itemize}

\item Circuits proceed from bottom to top.

\item Wires represent systems and may or may not be labelled with the appropriate symbol or Hilbert space.

\item Boxes represent unitary or CPTP evolutions from (the tensor product of) input systems (below) to output systems (above) and are labelled with the appropriate operation.

\item Semi-circles represent state preparations of the outgoing systems (above) and are labelled with the appropriate state.

\item The unitary CNOT gate is illustrated as \CNOT{} where the solid dot is on the control system and the $\oplus$ is on the target system.

\item The unitary SWAP gate may be illustrated as \SWAP{}.

\end{itemize}

\subsection*{Ontological Models and Overlap Measures}

Ontological models are defined in Sec.~\ref{sec:IN:ontological-models} with the notation summarised here for convenience.

\begin{itemize}

\item $\Lambda$ is the set of \emph{ontic states} $\lambda\in\Lambda$.

\item $\Sigma$ is the sigma-algebra such that $(\Lambda,\Sigma)$ is a measurable space---the \emph{ontic state space}.

\item $\Delta_P$ is the set of \emph{preparation measures} $\mu\in\Delta_P$ for a preparation $P$.

\item $\Delta_{|\psi\rangle}$ is therefore the set of preparation measures for all preparation methods that produce quantum state $|\psi\rangle$ if the system is a quantum system.

\item $\mu\in\Delta_P$ gives the probability $\mu(\Omega)$ for the resulting ontic state being in $\Omega\in\Sigma$ for some preparation method of $P$.

\item $\Gamma_T$ is the set of \emph{stochastic maps} $\gamma\in\Gamma_T$ for a transformation $T$ of the system.

\item $\Gamma_U$ is therefore the set of stochastic maps for all transformation methods corresponding to a unitary $U$ if the system is a quantum system.

\item $\gamma\in\Gamma_T$ gives the probability $\gamma(\Omega | \lambda)$ for the resulting ontic state being in $\Omega\in\Sigma$ given that some method for $T$ was applied to a system in ontic state $\lambda\in\Lambda$.

\item $\mu\transto{\gamma}\nu$ indicates that if a system is prepared according to $\mu$ and then a transformation is applied according to $\gamma$, then $\nu$ is the preparation measure that describes that composite process, Eq.~(\ref{eq:IN:ontological-model-transformation}).

\item $\Xi_M$ is the set of \emph{response functions} $\mathbb{P}_M \in\Xi_M$ for a measurement $M$. 

\item If $M = \{ |i\rangle \}_i$ is a basis measurement then $\Xi_M$ is therefore the set of response functions for all methods of performing the measurement $M$ on a quantum system.

\item $\mathbb{P}_M \in\Xi_M$ gives the probability $\mathbb{P}_M (E\,|\,\lambda)$ that the outcome $E\in M$ is obtained given that some method for measuring $M$ was performed on a system in ontic state $\lambda\in\Lambda$.

\item Similarly, $\mathbb{P}_M (E\,|\,\mu) \eqdef \int_\Lambda \d\mu(\lambda)\,\mathbb{P}_M (E\,|\,\lambda)$ is the probability of obtaining outcome $E\in M$ given that some method for measuring $M$ was performed on a system prepared according to $\mu$.

\item $\varpi(\cdot\,|\,\cdot)$ is the \emph{asymmetric overlap} as defined and explained in Sec.~\ref{sec:SO:asymmetric-overlap} \cite{AllenMaroney+16,Allen16a,Ballentine14,LeiferMaroney13,Maroney12a}.

\item $\omega(\cdot,\cdot)$ is the \emph{symmetric overlap} as defined and explained in Sec.~\ref{sec:SO:error-tolerance} \cite{AllenMaroney+16,Allen16a,BarrettCavalcanti+14,Branciard14,Leifer14a,Leifer14b,Maroney12a}.

\item $\varpi_\epsilon (\cdot\,|\,\cdot)$ for some $\epsilon\in[0,1)$ is the \emph{$\epsilon$-asymmetric overlap} as defined and explained in Sec.~\ref{sec:MR:error-tolerant-superpositions-and-mr}. $0$-asymmetric overlap is the normal asymmetric overlap.

\item $\k$ is a mapping from measurable functions $g : \Lambda \rightarrow [0,1]$ to measurable sets over $\Lambda$. $\k(g) \eqdef \ker(1 - g) \in \Sigma$ as defined in Def.~\ref{def:SO:k-bar}.

\item $\k_\epsilon$, for $\epsilon\in[0,1)$ generalises $\k$ as defined in Def.~\ref{def:MR:k-bar-e}.

\end{itemize}

\subsection*{Classical Random Variables}

Classical random variables are normally given capital Latin letters, such as $X$. Values that the variable can take are normally given by the corresponding lower-case Latin letter, such as $x$. Random variables are assumed to be discrete unless otherwise indicated.

\begin{itemize}

\item $\mathbb{P}(X=x)$ is the probability that random variable $X$ has value $x$. $\mathbb{P}(X)$ is the probability distribution for $X$. Alternatively, $\mathbb{P}(x)\eqdef\mathbb{P}(X=x)$ is used as shorthand where there is no ambiguity in doing so.

\item $\mathbb{P}(X,Y)$ is the joint probability distribution for random variables $X$ and $Y$. $\mathbb{P}(X,Y=y)$ is the distribution of probabilities that $X=x\;\wedge\;Y=y$ for all possible $x$.

\item $\mathbb{P}(X|Y)$ is the conditional probability distribution for $X$ given that $Y$ takes any given value.

\end{itemize}

\subsection*{Entropies and Information}

Classical \emph{Shannon entropies} are usually denoted with $H$. Quantum \emph{von Neumann entropies} are usually denoted with $S$. Recall that $\log$ is used for base-two logarithms. The definitions here assume that all random variables are discrete and all quantum systems are finite-dimensional.

Note that the definitions of quantum entropies can all be applied to both quantum systems and quantum states, depending on context.

\begin{itemize}

\item $H(X)$ is the classical \emph{entropy} of a classical random variable, defined $H(X) \eqdef -\sum_x \mathbb{P}(x)\log\mathbb{P}(x)$.

\item $H(X,Y)$ is the classical \emph{joint entropy}, defined $H(X,Y) \eqdef -\sum_{x,y} \mathbb{P}(x,y)\log\mathbb{P}(x,y)$.

\item $H(X|Y) \eqdef H(X,Y) - H(Y)$ is the classical \emph{conditional entropy}.

\item $I(X:Y) \eqdef H(X) + H(Y) - H(X,Y)$ is the classical \emph{mutual information}.

\item $I(X:Y|Z) \eqdef H(X,Z) + H(Y,Z) - H(Z) - H(X,Y,Z) = H(X|Z) - H(X|Y,Z)$ is the classical \emph{conditional mutual information}.

\item $S(\rho)$ is the quantum \emph{entropy} of a state $\rho$ (equivalently, of a quantum system in state $\rho$), defined $S(\rho) = -\Tr\left( \rho\log\rho \right)$.

\item $S(A)$ is the \emph{entropy} of a quantum system $A$, defined $S(\rho_A)$ where $\rho_A$ is the state of $A$.

\item $S(AB)$ is the quantum \emph{joint entropy} of bipartite system $AB$, defined $S(\rho_{AB})$ where $\rho_{AB}$ is the state of $AB$.

\item $S(A|B) \eqdef S(AB) - S(A)$ is the quantum \emph{conditional entropy}.

\item $I(A:B) \eqdef S(A) + S(B) - S(AB)$ is the quantum \emph{mutual information}.

\item $I(A:B|C) \eqdef S(AC) + S(BC) - S(C) - S(ABC) = S(A|C) - S(A|BC)$ is the quantum \emph{conditional mutual information}.

\item Quantum and classical mutual informations, despite using the same symbol, can be told apart by the types of system on which they act.

\end{itemize}
\end{romanpages}            

\chapter{Introduction} \label{ch:IN}

\begin{quote}
``... So it will be difficult. But the difficulty, really, is psychological and exists in the perpetual torment that results from your saying to yourself ``But how can it be like that?'' Which really is a reflection of an uncontrolled, but I say utterly vain, desire to see it in terms of some analogy with something familiar. I will not describe it in terms of an analogy with something familiar. I'll simply describe it.

...

So that's the way to look at the lecture---is not to try to understand---well, you have to understand the English, of course. But in any sense in terms of something else. And don't keep saying to yourself, if you can possibly avoid it, ``But how can it be like that?'' Because you'll get down the drain. You'll get down into a blind alley in which nobody has yet escaped. Nobody knows how it can be like that.'' \cite{Feynman64}
\end{quote}

These words of Richard Feynman, delivered to the audience of a public lecture on quantum theory in 1964, contain both a warning and an invitation. You must not get lost down a blind alley trying to understand quantum theory. But nobody knows how it can be like this; which, to a theorist, is an invitation if ever there was one.

More than fifty years later, it is still true that nobody knows how or why nature is successfully described by quantum theory. Some would certainly claim that they do, but unfortunately such people generally disagree with one another \cite{SaundersBarrett+10,BohmHiley95,FuchsMermin+14,DoratoEsfeld10}\footnote{Of course, this is not to say that none of these people are correct, but it certainly shows that most are likely to be incorrect.}. However, the study of quantum theory, especially in responding to Feynman's challenge, has changed significantly. Quantum foundations has bloomed into its own sub-field on the borders of physics, philosophy, mathematics, information science, and computer science. The most productive insights from this field have not come from stumbling into the alley, but by carefully peering in and trying and make out its main contours.

This thesis contributes to that effort while concentrating on \emph{quantum ontology}. In foundations of quantum theory, \emph{ontology} refers to any potential ``actual state of affairs'' or fundamental description, essentially the \emph{reality} of the physical system. What does quantum theory require or suggest about the behaviour of physics on the most fundamental level? Does the success of quantum theory require certain ways of understanding the actual state of affairs of a system, or does it rule out others? Rather than attempting to give a comprehensive account of an exact quantum ontology, such questions can be carefully drawn out in specific ways, revealing the important features of any plausible or actual ontology.

It is helpful to note that quantum theory can fruitfully be considered as a \emph{framework} rather than a \emph{theory} as such\footnote{I have heard this point being made by various people, but have been unable to find a citation for it. This is therefore presented with apologies to anybody who may lay claim to its conception.}, a framework being understood as a low-level set of rules that theories built from that framework must respect. The classical framework, for example, might be expressed as the facts that: systems occupy exactly one state from a space at any given point, joint states of multiple systems are found by taking the Cartesian product, transformations can move systems between states in such a way that probability is conserved, \emph{etc.} Out of this framework, specific theories such as Maxwellian electromagnetism, statistical mechanics, and Newtonian orbital mechanics can be built. Similarly, the quantum framework can be understood as the Hilbert space structure of system states, the unitary/CPTP (completely-positive trace-preserving) nature of transformations, the Born rule for measurement outcomes \emph{etc.} and out of this, theories such as quantum optics and the standard model can be built. This distinction between framework and theory goes a long way to explaining why quantum theories are uniquely difficult for us to understand compared to any other physical theory: they are all predicated on an unfamiliar non-classical framework. It also allows for easier comparisons between quantum and classical predictions in general by concentrating on what is allowed by these general abstract frameworks, rather than getting lost in the exact physics of specific systems. However, ``quantum theory'' will continue to be the term used throughout this thesis since ``quantum framework'' is clunky and ``quantum mechanics'' is no more specific.

The mother of the modern approach to quantum ontology is Bell's theorem \cite{Bell66}, developed in response to controversy over hidden variables arguments in quantum theory in general and the EPR argument \cite{EinsteinPodolsky+35} in particular. This approach starts by precisely defining an ontological feature of interest in a way that is independent of either the quantum or classical frameworks. In the case of Bell's theorem, this generic property has come to be known as ``Bell locality'' \cite{Shimony13}. The content of the theorem is to show, using as few assumptions as possible, that this property can be violated in quantum theories and, ultimately, in quantum experiments. The conclusion: quantum theory does not admit ontologies with this property. The key to the extraordinary success of Bell's theorem is in keeping all assumptions to a bare minimum.

Since Bell, many others have followed this broad outline when investigating quantum ontology and other foundational issues. Highlights include the Kochen-Specker theorem \cite{KochenSpecker67,Held13}, the Leggett-Garg inequality \cite{LeggettGarg85}, and the PBR theorem \cite{PuseyBarrett+12}. The most immediately impressive results are often \emph{no-go theorems}, proving that quantum theory precludes certain ontological features (as in Bell's theorem). However, there has also been good progress in finding formal examples of quantum ontologies which are compatible with quantum theory and do have other interesting ontological features, for example in Refs.~\cite{LewisJennings+12,AaronsonBouland+13,KochenSpecker67,HarriganRudolph07}. In these ways, certain major features and limitations required of a plausible quantum ontology can be carefully discerned.

In these discussions, a dichotomy is often drawn between ontic and \emph{epistemic} features or explanations. An epistemic feature is an artefact of a particular agent's description of a physical system, subject to their (often incomplete) knowledge of it. This is in contrast to ontic features, which are considered objective states of affairs. Most often, epistemic explanations make use of probability distributions (or, more generally, probability measures) representing the agent's (imperfect) knowledge. Note, however, that while this is consistent with their use in quantum foundations, philosophers may use the terms ``ontic'' and ``epistemic'' somewhat differently.

Quantum theory is unrivalled in its experimental success, having never faced an incompatible experimental result. This is certainly impressive, but much more understandable when quantum theory is understood as a framework rather than a theory---comparisons to Maxwellian electromagnetism, for instance, are unfair. This has lead to theorists often taking quantum predictions (in the broad sense of the quantum framework) to be generically correct and any incompatibility with them to be damning. However, in science experimental results should still reign supreme. Any theoretical result should be seen as preliminary until it can be subject to experimental verification. That said, it is common for a strict theoretical incompatibility between quantum theory and some ontology to be found before that result can be extended to something amenable to experiment.

This thesis further characterises necessary or plausible quantum ontologies in three ways. In the first, the subject is ontology of quantum states. This is the most common setting for results in quantum ontology, as the states of quantum systems are what most obviously differs compared to classical systems. In the second, the focus is on quantum causality. This is a much more recent subject in the study of quantum theory, but certainly of no less importance than the ontology of quantum states for understanding the reality of quantum systems. Third, a rather different approach is taken by considering quantum theory supplemented with time travel to the past, as this gives interesting alternative perspectives on many of the problems of quantum ontology.

%
%
%
%
%

\section{``Ontology'', ``Epistemic'', and ``Causality''} \label{sec:IN:philosophical-language}

Before proceeding, some comments are necessary on the use of terms like ``ontology'', ``epistemic'', and ``causality'' in this thesis. These terms have all been inherited from philosophy, in particular metaphysics. However, this thesis is in the tradition of quantum foundations which, as a field, tends to use these terms somewhat differently (and rather more loosely) than its more philosophical cousins. 

For the avoidance of doubt, therefore, it will be useful to quickly outline what the intended meanings of such words are in this thesis. These are not philosophically robust definitions, nor are they meant to be, as the intended audience is primarily physicists.

By ``ontology'', what is meant is all or part of some actual or conceivable final and fundamental description of reality. That is, supposing that there is an objective reality, this is described by the ontology. Not everybody will agree that such a thing needs to exist, of course. As noted above, ``epistemic'' describes a feature that is part of an agent's subjective knowledge about a system and may often take the form of probability distributions (or measures) over ontic features. By ``causality'', what is meant is the study of causal relations where one event or physical value has some actual influence on the occurrence of another, as distinct from mere correlations.

There are vast philosophical literatures on each of these topics, none of which can be effectively engaged with here. It is, however, hoped that the contents of this thesis may be able to inform some work in these philosophical fields. So the intention here is not to ignore the valuable philosophical work of these areas at all, but simply to note that they must unfortunately lie well beyond the scope of this thesis.

An inevitability of working in these areas of quantum foundations is some philosophic\-ally-charged language. Despite this, all of the results presented here should be fairly philo\-sophically-neutral. That is, they should be valuable to people of most common philosophical leanings. On the simplest level, this is because the results are mathematical in nature. For example, the results in Chaps.~\ref{ch:SO}, \ref{ch:MR}, based on the ontological models framework [Sec.~\ref{sec:IN:ontological-models}], are very general regardless of philosophy, since they apply to any potential underlying physical theory that can be cast as an ontological model---a very large class of theories indeed.

\section{Thesis Overview}

This thesis is about ontology in quantum theory. More specifically, about obtaining a better understanding of the types of ontology that are plausible given the apparent correctness of quantum theory. This is done by carefully obtaining precise results about the nature of quantum ontology under well-defined assumptions. The study of the ontology of quantum states is relatively mature and so this thesis can build on that work to prove some very specific statements. In particular: that almost all superposition states must be real (in a well-defined sense); that quantum theory is incompatible with most types of macro-realism; and that under reasonable assumptions any conceivable overlap between many quantum states on an ontological level must be small if not zero. States are not, however, the only feature of ontology. A newer field of study is that of causes in quantum theory. This thesis constructs an expressive framework for this study by building on the specific, but important, example of a complete common cause. Since causes are typically considered to be ontological features, this framework provides a new setting for discussing ontology in quantum causality. While it is much less typical to consider time travel to the past to be real, it remains a possibility and, even beyond that possibility, time travel to the past provides an interesting playground in which to test and compare ideas in quantum ontology. This thesis therefore does so by comparing different models of quantum time travel and constructing a novel one that addresses their shortcomings. Combining these approaches to quantum ontology, a fuller and more nuanced understanding can be achieved.

In Chap.~\ref{ch:SO}, the ontology of quantum states is studied. After discussing how different ontologies are classified, the current state of theorems both ruling out and demonstrating the possibilities of various classes of ontology is discussed. This discussion identifies two main shortcomings common to many such theorems, to be addressed by new results presented later in the chapter. First, is proved that almost all superpositions must be real (in a well-defined sense) for any ontology of a quantum system of dimension $d>3$. The same techniques used are then applied to proving that no quantum states can be \emph{$\psi$-epistemic} [Sec.~\ref{sec:SO:specific-states-ontology-definitions}] and that in large-dimensional quantum systems, many quantum states must be close to ontologically distinct. The effect of these results is to show that potential epistemic uncertainty over the exact ontological state of a quantum system is unlikely to have much power for explaining any features of quantum theory. However, all of these new results assume that quantum predictions are exactly correct and are therefore not immediately relevant to experiments, so a proof-of-concept result for error-tolerant extensions is also given. Finally, it is shown how these foundational techniques can be applied to problems in quantum information theory. In particular, it is shown that exponentially many classical resources are required to simulate a quantum channel. Similar results have appeared in the literature before, but the method used here has some key advantages. For instance, it is a significantly simpler proof than previous results and it can also be easily extended to generate potentially better bounds from new classical error-correction codes in the future.

In Chap.~\ref{ch:MR}, the focus remains within the ontology of quantum states, but shifts slightly to study \emph{macro-realism}. Macro-realism is a particular ontological property that is most commonly associated with the Leggett-Garg inequalities and corresponding no-go theorem. After introducing macro-realism in detail, the Leggett-Garg argument is briefly reviewed and some loopholes in it are identified. That is, the Leggett-Garg argument can only show that quantum theory is incompatible with one of three types of macro-realism. This is followed by a new result showing that quantum systems of $d>3$ dimensions are incompatible with two of those three types of macro-realism, improving on the Leggett-Garg argument by using the methods of Chap.~\ref{ch:SO}. Since both this result and the result showing that superpositions are real are intolerant to experimental error, these results are then brought together and given error-tolerant variations. The main conclusions of Chaps.~\ref{ch:SO}, \ref{ch:MR} are then discussed in terms of their impact and potential for further work.

Large parts of chapters~\ref{ch:SO}, \ref{ch:MR} are based on Refs.~\cite{Allen16a,AllenMaroney+16}. In particular, results of Secs.~\ref{sec:MR:intro}--\ref{sec:MR:main-theorem} are the result of collaboration with Owen Maroney and Stefano Gogioso. The communication work in Sec.~\ref{sec:SO:communication} is joint work with Jonathan Barrett.

In Chap.~\ref{ch:CI} quantum causality is studied. After introducing the topic and its primary motivations the first major problem for quantum causality is identified. That is, the problem of finding a satisfactory version of Reichenbach's principle for a quantum universe [Sec.~\ref{sec:CI:introduction-Bell}]. An appropriate quantum Reichenbach's principle is then carefully justified by developing a notion of quantum conditional independence. This quantum conditional independence has four definitions which are proved to be equivalent. Each definition naturally generalises a corresponding definition for classical conditional independence, lending strength to the proposed quantum Reichenbach's principle that results. These new definitions are then illustrated with some examples and generalised from conditional independence of two systems to conditional independence of $k\geq 2$ systems.

Chapter~\ref{ch:CM} extends this approach to quantum causality to the general case of causal models. Classically, causal models form a general framework based on Reichenbach's principle that allows systematic study of causal relationships. These classical causal models are briefly introduced before being used to motivate an analogous generalisation to quantum causal models. A proposed definition of these quantum causal models is then given that generalises quantum Reichenbach's principle from Chap.~\ref{ch:CI} in a sensible way. Some examples then illustrate this new framework, including that of Bell's theorem. The constructions of Chaps.~\ref{ch:CI}, \ref{ch:CM} are then discussed and compared to alternative approaches to formalising the study of quantum causality. In particular, it is noted that the solid foundation provided by the quantum Reichenbach's principle places confidence in the robustness of the approach taken here.

The results of Chaps.~\ref{ch:CI}, \ref{ch:CM} first appeared in Ref.~\cite{AllenBarrett+17} and are the result of joint work with Jonathan Barrett, Dominic Horsman, Ciar\'{a}n Lee, and Robert Spekkens.

A somewhat different approach is taken in Chap.~\ref{ch:TT}. There, quantum ontology is studied through the lens of possible time travel to the past. The most common way to motivate this is through noting that \emph{closed timelike curves} (CTCs) permitted by general relativity would allow for such time travel. After motivating the approach, previous attempts to model quantum time travel using quantum circuits are briefly covered. There are two such models, called D-CTCs and P-CTCs [Sec.~\ref{sec:TT:review}]. Between them, they highlight the interactions between ontology and possible non-linear extensions of quantum theory. In particular, that non-linearity added to quantum theory forces one to be more specific about ontology to consistently describe a system. Both the successes and shortcomings of D- and P-CTCs are then used to construct two classes of new models for time travel in quantum theory. From these classes, one model---dubbed T-CTCs---is fully fleshed out and compared to D- and P-CTCs at length. These findings are then discussed, with particular attention paid to the roles of non-linearity and ontological understanding in the models.

Finally, Chap.~\ref{ch:CO} summarises the results of the thesis and lays them out in the context of the further work that they suggest. In particular, the possibilities for combining the approaches to quantum ontology taken in this thesis will be discussed.

Some notational conventions used throughout this thesis have been summarised starting on page~\pageref{sec:notation}. As is common in modern literature on quantum foundations, the focus will be on finite-dimensional quantum systems so finite-dimensional Hilbert spaces may often be assumed.

\chapter{Ontology of Quantum States and Superpositions} \label{ch:SO}

\section{Theorising About State Ontology} \label{sec:SO:introduction}

When considering the ontology of quantum systems the usual questions centre on the ontology of quantum \emph{states}. Questions such as: What properties must the ontological states of the system have, or not have? How can the structure of the space of these ontological states relate to preparations of quantum states? and How similar must this ontological state space be to the quantum state spaces, $\mathcal{P}(\mathcal{H})$ and $\mathcal{D}(\mathcal{H})$?

Many physicists would, for example, be much more comfortable if quantum phenomena could be explained with a concise and elegant realist ontology for states. Some certainly believe that existing explicit realist interpretations---such as Bohmian mechanics \cite{BohmHiley95,Bohm52a,Bohm52b,deBroglie27} or Everettian interpretations \cite{DeWittGraham73,SaundersBarrett+10}---achieve this, but such opinions are hardly uncontroversial \cite{SaundersBarrett+10,Goldstein16}. But even regardless of one's view on interpretations, it is interesting to consider whether certain ontological features can be ruled out \emph{a priori} as being incompatible with the predictions of quantum theory. Moreover, there has also been considerable cross-pollination between foundational ontology results and information-theoretic results \cite{Montina12,BarrettHardy+05,PerryJain+15,LiuPerry+16,Montina13,Montina15,MontinaWolf16,Leifer14b}, as discussed in Sec.~\ref{sec:SO:communication}.

Some of these questions will be addressed and answered in this chapter. In doing so, the results will fall into a tradition of ``ontology theorems'' in quantum foundations \cite{Leifer14b}. In particular, this chapter will concentrate on the ontology of quantum superposition states, proving that almost all quantum superpositions must be ``real'' in a well-defined sense [Sec.~\ref{sec:SO:superpositions-are-real}]. The techniques used to prove this main result will then be adapted for two purposes. First, to address some common shortcomings of many of the current ontology theorems [Sec.~\ref{sec:SO:specific-states-ontology}]. Second, to provide a simple way to exponentially bound the classical resources required to simulate a quantum channel [Sec.~\ref{sec:SO:communication}].

This chapter will also lay much of the groundwork for Chap.~\ref{ch:MR}, where the general question of quantum state ontology is applied specifically to the case of macro-realism.

The methods of this chapter and the next are all based on the mathematical framework of ontological models. Because of this, before introducing the relevant background it will be prudent to first introduce this framework in Sec.~\ref{sec:IN:ontological-models}. This framework itself is independent of quantum theory, however since the focus here is exclusively on quantum systems it will normally be assumed that ontological models exactly reproduce quantum predictions. As well as this, it will be assumed that all quantum systems are finite-dimensional, with the infinite-dimensional case being discussed briefly in Sec.~\ref{sec:MR:summary-and-discussion}.

Most of the material in this chapter and the next overlaps with that published in Refs.~\cite{Allen16a,AllenMaroney+16}, with the notable exceptions of Secs.~\ref{sec:SO:communication}, \ref{sec:MR:error-tolerant-argument}.

\subsection{The Ontological Models Framework} \label{sec:IN:ontological-models}

Debates about the ontology of quantum states are at least as old as quantum theory itself. In more recent years, this subject has benefited from a standardisation of definitions and concepts into the \emph{ontological models} framework \cite{HarriganRudolph07,HarriganSpekkens10,Leifer14b}. The framework of ontological models has been expressly developed to make discussions about ontology in physics precise and is the natural arena for such discussions. It will therefore be useful to lay down the mathematics of this framework before discussing the motivations and history behind the work in this chapter as it will enable a much easier and more precise discussion.

An ontological model is exactly that: a bare-bones model for the underlying ontology of some physical system. Since the ontological models framework is so sparse, very many conceivable ontologies can be understood as ontological models. The system may also be correctly described by some other, higher, theory (or framework of theories)---such as Newtonian mechanics or quantum theory---in which case the possible ontological models considered should be constrained to reproduce the predictions of that theory. By combining these constraints with the very general framework of ontological models, interesting and general conclusions can be drawn about the nature of the ontology. It is important to note that, while ontological models are normally used to discuss quantum ontology, the framework itself is entirely independent from quantum theory.

In this section the framework of ontological models will be defined and introduced. First on its own and then as applied to quantum theory in both the absence and presence of possible experimental error. This account of ontological models is based on the one given in Ref.~\cite{Leifer14b}.

The framework of ontological models relies on just two core assumptions: (1) that the system of interest has some ontic state $\lambda$ representing the entirety of the actual state of affairs of the system and (2) that standard probability theory may be applied to these states. That is, at any given time the entire ontology of a system is given by its ontic state. Together, these bring us to consider the ontology of some physical system as represented by some measurable space $(\Lambda,\Sigma)$ of ontic states $\lambda\in\Lambda$ which the system might occupy ($\Sigma$ being a sigma algebra of measurable subsets of $\Lambda$). The requirement that ontic states occupy a measurable space simply guarantees that sensible probabilities can be defined over them.

In the lab, a system can be prepared, transformed, and measured in certain ways. Each of these operational processes needs to be describable in the ontological model for it to be capable of describing the system.

Preparation must result in the system ending up in some ontic state $\lambda$, though the exact state need not be known. Thus, each use of an operational preparation $P$ gives rise to some \emph{preparation measure} $\mu$ over $\Lambda$ which is a probability measure ($\mu(\emptyset)=0$, $\mu(\Lambda)=1$). For every measurable subset $\Omega\in\Sigma$, $\mu(\Omega)$ gives the probability that the resulting $\lambda$ is in $\Omega$. The set of all such preparation measures for some $P$ is $\Delta_P$. Note how the measure can vary between uses of the same preparation.

Similarly, an operational transformation $T$ of the system will generally change the ontic state from $\lambda^{\prime}\in\Lambda$ to a new $\lambda\in\Lambda$. Recalling that the ontic state $\lambda^{\prime}$ represents the entirety of the actual state of affairs before the transformation, then the final state can only depend on $\lambda^{\prime}$ (and not the preparation method or any previous ontic states, except as mediated through $\lambda^{\prime}$). The transformations must therefore be described as \emph{stochastic maps} $\gamma$ on $\Lambda$. A stochastic map consists of a probability measure $\gamma(\cdot|\,\lambda^{\prime})$ for each initial ontic state, such that for any measurable $\Omega\in\Sigma$, $\gamma(\Omega\,|\,\lambda^{\prime})$ is the probability that the final $\lambda$ lies in $\Omega$ given that the initial state was $\lambda^{\prime}$ \footnote{These stochastic maps, viewed as a set of functions $\gamma(\Omega\,|\cdot)\,:\,\Lambda\rightarrow[0,1]$ (one for each $\Omega\in\Sigma$), must be \emph{measurable functions}. That is, for any measurable set $\mathcal{S}\subseteq[0,1]$ and any $\gamma(\Omega\,|\,\lambda^{\prime})$, then $\{\lambda^{\prime}\in\Lambda\,:\,\gamma(\Omega\,|\,\lambda^{\prime})\in\mathcal{S}\}\subseteq\Lambda$ is a measurable set.}. The set of all stochastic maps corresponding to some $T$ is $\Gamma_T$.

Finally, a measurement $M$ may give rise to some outcome $E$. Again, which outcome is obtained can only depend on the current ontic state $\lambda^{\prime}$. Therefore a measurement $M$ gives rise to a conditional probability distribution\footnote{These probability distributions, viewed as functions $\Lambda\rightarrow[0,1]$, must also be measurable functions.} $\mathbb{P}_M(E\,|\,\lambda^{\prime})$, sometimes called a \emph{response function}. The set of all such response functions corresponding to $M$ is $\Xi_M$. For this thesis it is only necessary to consider measurements that have countable sets of possible outcomes $E$.

Putting these parts together: given a system where a preparation $P$ is performed followed by some transformation $T$ and some measurement $M$ then the ontological model for that system must have some preparation measure $\mu\in\Delta_P$, stochastic map $\gamma\in\Gamma_T$, and conditional probability distribution $\mathbb{P}_M \in \Xi_M$ such that the probability of obtaining outcome $E$ is
\begin{equation} \label{eq:IN:ontological-model-probability}
\mathbb{P}_M (E\,|\,\nu) = \int_{\Lambda}\mathrm{d}\nu(\lambda)\,\mathbb{P}_M(E\,|\,\lambda)
\end{equation}
where 
\begin{equation} \label{eq:IN:ontological-model-transformation}
\nu(\Omega)\eqdef\int_{\Lambda}\mathrm{d}\mu(\lambda)\,\gamma(\Omega\,|\,\lambda)
\end{equation}
is the effective preparation measure obtained by preparation $P$ followed by transformation $T$. This is schematically illustrated in Fig.~\ref{fig:SO:ontological-models}.

\begin{figure}
\begin{centering}
\includegraphics[scale=1,angle=0]{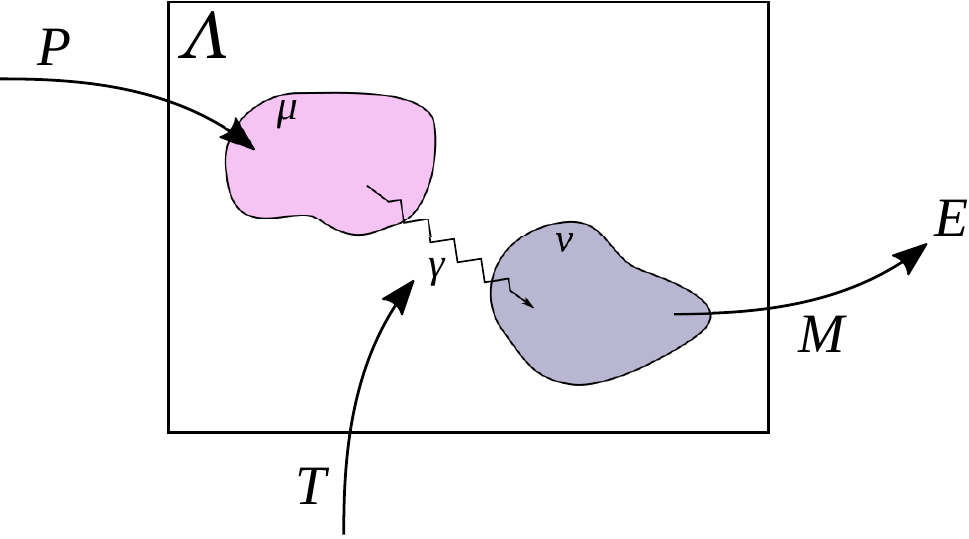}
\par\end{centering}
\protect\caption{Illustration of the basic prepare-transform-measure procedure in an ontological model. The box represents the entire ontic state space $\Lambda$, while shaded regions represent the ontic states that can be prepared by the indicated preparation measures. Some preparation $P$ is performed, resulting in an ontic state according to $\mu\in\Delta_P$. Some transformation $T$ is performed, resulting in some stochastic map $\gamma\in\Gamma_T$, effectively transforming $\mu$ to $\nu$ as in Eq.~(\ref{eq:IN:ontological-model-transformation}). Finally, some measurement $M$ is made, resulting in some outcome $E$ according to probabilities given by Eq.~(\ref{eq:IN:ontological-model-probability}).}
\label{fig:SO:ontological-models}
\end{figure}

Note that ontological models must be closed under transformations. That is, for any preparation $\mu$ and transformation $\gamma$ in the model then the preparation $\nu$ defined by Eq.~(\ref{eq:IN:ontological-model-transformation}) must also exist in the model (since a preparation followed by a transformation is itself a type of preparation).

It is useful to think of stochastic maps acting on preparation measures as an emergent property of them acting on the underlying ontic state space. If an ontic state is sampled from $\mu$ and then transformed via $\gamma$ the effect is the same as sampling from the measure $\nu$ defined by Eq.~(\ref{eq:IN:ontological-model-transformation}). Rather that having to say this in such a cumbersome way, the notation $\mu\transto{\gamma}\nu$ will be used to express the same thing and may be read ``$\mu$ transforms to $\nu$ under $\gamma$''. The full definition of this notation is, however, via Eq.~(\ref{eq:IN:ontological-model-transformation}) as above.

Summarising, an ontological model for some physical system:
\begin{enumerate}
\item defines a measurable space $(\Lambda,\Sigma)$ of ontic states for the system;
\item for each possible transformation $T$ defines a set of stochastic maps $\gamma\in \Gamma_T$ from $\Lambda$ to itself;
\item for each possible preparation $P$ defines a set of preparation measures $\mu\in\Delta_P$ over $\Lambda$, ensuring closure under the actions of the stochastic maps as in Eq.~(\ref{eq:IN:ontological-model-transformation});
\item for each possible measurement $M$ defines a set of response functions $\mathbb{P}_M \in\Xi_M$ over the outcomes given $\lambda\in\Lambda$;
\end{enumerate}
and then produces probabilities for measurement outcomes via Eqs.~(\ref{eq:IN:ontological-model-probability}, \ref{eq:IN:ontological-model-transformation}). 

It should be noted that ontological models are usually presented using probability \emph{distributions} rather than the more mathematically involved measures used here. However, as noted in Ref.~\cite{Leifer14b}, this simplification precludes many reasonable ontological models, including the archetypal Beltrametti-Bugajski model \cite{BeltramettiBugajski95} (see also Sec.~\ref{sec:SO:desired-ontologies}). The more accurate approach is taken here both for the accuracy itself and to serve as a resource of how to construct proofs in measure-theoretic ontological models as such proofs are rarely seen in the literature.

The above definition for ontological models \emph{per se} does not have much useful structure. The power is found when the probabilities given by Eqs.~(\ref{eq:IN:ontological-model-probability}, \ref{eq:IN:ontological-model-transformation}) are constrained to match those given by other theories known to accurately describe the system (such as quantum theory) or by experiments. To this end, consider applying ontological models to systems accurately described by quantum theory.

An ontological model is defined by some ontic state space $(\Lambda,\Sigma)$ as well as the relevant preparation measures, stochastic maps, and conditional probability distributions. For a quantum system these must include at least the following. For each state $|\psi\rangle\in\mathcal{P}(\mathcal{H})$ there must be a set $\Delta_{|\psi\rangle}$ of preparation measures $\mu_{|\psi\rangle}$---potentially at least one for each distinct experimental procedure for preparing $|\psi\rangle$. Similarly, for each unitary operator $U$ on $\mathcal{H}$ there is a set $\Gamma_{U}$ of stochastic maps $\gamma_{U}$ and for each basis measurement $M=\{|i\rangle\}_{i=0}^{d}$ there is a set $\Xi_{M}$ of conditional probability distributions $\mathbb{P}_{M}$---again, potentially at least one stochastic map/probability distribution for each experimental procedure for transforming/measuring.

Since the quantum system is accurately described by quantum theory, the ontological model must also reproduce the predictions of quantum theory. That is, for any $|\psi\rangle\in\mathcal{P}(\mathcal{H})$,
$\mu\in\Delta_{|\psi\rangle}$, $U,$ $\gamma\in\Gamma_{U}$, basis
$M$, and $\mathbb{P}_{M}\in\Xi_{M}$ it is required that
\begin{equation} \label{eq:IN:ontological-model-quantum-probability}
|\langle i|U|\psi\rangle|^{2}=\int_{\Lambda}\mathrm{d}\nu(\lambda)\,\mathbb{P}_{M}(|i\rangle\,|\,\lambda),\quad\forall|i\rangle\in M
\end{equation}
where $\mu\transto{\gamma}\nu$ is defined as in Eq.~(\ref{eq:IN:ontological-model-transformation}). Note also that $\nu\in\Delta_{U|\psi\rangle}$ since preparing the quantum state $|\psi\rangle$ (via any ontological preparation $\mu\in\Delta_{|\psi\rangle}$) followed by performing the quantum transformation $U$ (via any $\gamma\in\Gamma_{U}$) is simply a way to prepare the quantum state $U|\psi\rangle$.

So Eq.~(\ref{eq:IN:ontological-model-quantum-probability}) must hold when quantum theory is known to accurately describe the system. What of the case where quantum theory only approximately describes the system? Suppose, for instance, that the probabilities predicted by quantum theory are accurate to within $\pm\epsilon$ for some given $\epsilon\in(0,1]$. It follows that the ontological model need only reproduce these approximate predictions, so Eq.~(\ref{eq:IN:ontological-model-quantum-probability}) is replaced by
\begin{equation} \label{eq:IN:ontological-model-approx-quantum-probability}
|\langle i|U|\psi\rangle|^{2} + \epsilon \geq \int_{\Lambda}\mathrm{d}\nu(\lambda)\,\mathbb{P}_{M}(|i\rangle\,|\,\lambda) \geq |\langle i|U|\psi\rangle|^{2} - \epsilon, \quad\forall|i\rangle\in M
\end{equation}
where, again, $\nu\in\Delta_{U|\psi\rangle}$ is defined as in Eq.~(\ref{eq:IN:ontological-model-transformation}).

This mathematical framework forms the foundation of many modern results in the ontology of quantum states. It may be seen as a successor to, and extension of, the \emph{hidden variables} models used historically \cite[and references therein]{HarriganRudolph07}. Recall that, unless otherwise stated, quantum predictions are assumed to be exactly accurate throughout this chapter and Chap.~\ref{ch:MR}. In other words, Eq.~(\ref{eq:IN:ontological-model-quantum-probability}) will be assumed on top of the bare ontological models framework unless otherwise stated.

\subsection{The Desire for Simpler Ontologies} \label{sec:SO:desired-ontologies}

Given the definitions and mathematical background of the ontological models framework, it is now time to consider the motivation behind wanting a relatively simple ontology for quantum states. Whether or not you believe it to be possible, this motivation is quite easy to understand.

The simplest way to capture textbook quantum theory in an ontological model is to: have the pure quantum states \emph{be} the ontic states $\Lambda\cong\mathcal{P}(\mathcal{H})$, have the pure state preparation measures give unit probability to the corresponding ontic state, and have transformations and measurements act on $\Lambda$ exactly as they do on $\mathcal{P}(\mathcal{H})$ in textbook quantum theory. This transliteration of quantum theory to an ontological model is called the \emph{Beltrametti-Bugajski model} \cite{BeltramettiBugajski95,HarriganRudolph07,Leifer14b}. 

The Beltrametti-Bugajski model is often seen as an unattractive ontology for several reasons. It requires an uncountable ontic state space for any non-trivial system, even when there are only two distinguishable preparations. It also forces one to talk about ontology globally, due to entanglement. Finally, it contains a lot of redundancy: the properties of the quantum state $|\psi\rangle = \alpha|0\rangle + \beta|1\rangle$ are entirely inherited from properties of $|0\rangle$ and $|1\rangle$, it therefore seems ontologically extravagant to describe $|\psi\rangle$ entirely separately from $|0\rangle$ and $|1\rangle$ on an ontological level.

So how does one respond to this easy, but rather ugly, understanding of state ontology? One way is to simply accept that the ontology of quantum states is just like this or to argue that it is not as ugly as it might seem. That view might lead one towards Everettian interpretations of quantum theory, for example. Another is to deny the need for a realist ontology of quantum states altogether---a position variously called ``anti-realist'', ``neo-Copenhagen'', and ``instrumentalist'' \cite{Leifer14b}---a route with its own conceptual and philosophical hurdles. A third way (and the one of most interest in this thesis) is to seek a more elegant realist ontology that might underlie quantum theory.

This third approach is sometimes called that of the ``epistemic realist'' and naturally leads to using ontological models because of their extreme generality for realist ontologies. By using the fact that preparing some quantum state may leave some uncertainty over the ontic state (the preparation measure can give non-zero probability to a large number of ontic states), the epistemic realist can hope to explain many features of quantum systems as arising from this ``ontological uncertainty''. These features include the indistinguishability of non-orthogonal states, no-cloning, stochasticity of measurement outcomes, and the exponential increase in state complexity with increasing system size \cite{Spekkens07}.

In particular, the epistemic realist can hope to use the fact that both quantum states and ontological model preparations can ``overlap''. Two quantum states in $\mathcal{P}(\mathcal{H})$ overlap by an amount quantified by the Born rule (equivalently, their inner product). The corresponding ontic overlap occurs when preparations for different quantum states can prepare some of the same ontic states.

A typical motivation given for the epistemic realist view is that such ontic overlaps might naturally explain the indistinguishability of non-orthogonal quantum states \cite{Spekkens07,BarrettCavalcanti+14,Leifer14a}. If two non-orthogonal quantum states have finite ontic overlap, then sometimes they will prepare the same ontic states and, since the ontic state describes the entire ontology of the system, there is no way to tell them apart. By looking at the probability of this occurring, the ability of ontic overlaps to account for indistinguishability can be quantified [Sec.~\ref{sec:SO:anti-distinguishability}]. This is illustrated in Fig.~\ref{fig:SO:ontic-overlap}. More thorough discussions of the way that epistemic realist explanations can explain puzzles in quantum foundations can be found in Refs.~\cite{Spekkens07,Spekkens14}.

\begin{figure}
\begin{centering}
\includegraphics[scale=1,angle=0]{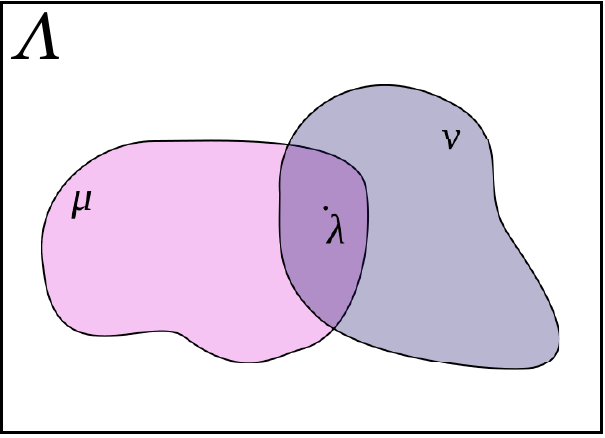}
\par\end{centering}
\protect\caption{Illustration of the epistemic realist explanation for indistinguishability of non-orthogonal quantum states. Suppose that $\mu\in\Delta_{|\psi\rangle}$ and $\nu\in\Delta_{|\phi\rangle}$ are preparation measures over ontic state space $\Lambda$ for non-orthogonal quantum states $|\psi\rangle$ and $|\phi\rangle$. If they overlap, as shown, there are ontic states such as the $\lambda\in\Lambda$ illustrated that may be prepared by both. In this case, there is no way to distinguish which preparation, $\mu$ or $\nu$, was performed, naturally explaining why there is no way to perfectly distinguish between $|\psi\rangle$ and $|\phi\rangle$.}
\label{fig:SO:ontic-overlap}
\end{figure}

The epistemic realist perspective on the foundations of quantum theory is not only philosophically attractive but also appears to be tenable. There are theories that explain the quantum state in an epistemically realist manner that reproduce interesting and large subsets of quantum theory including many characteristically quantum features \cite{Spekkens14,Spekkens07,BartlettRudolph+12,JenningsLeifer15}. There are also several explicit ontological models exactly describing whole isolated quantum systems that make very good use of ontological uncertainty, though only in $d=2$ dimensions as discussed below.

A particularly powerful result for the epistemic realist would be to find an explanation of certain quantum superposition states as statistical effects due to ontic overlap. Superpositions are behind quantum interference, the uncertainty principle, wave-particle duality, entanglement, Bell non-locality \cite{Bell87}, and the probable increased computational power of quantum theory \cite{JozsaLinden03}. Perhaps most alarmingly, superpositions give rise to the measurement problem, so captivatingly illustrated by the ``Schr\"{o}dinger's cat'' thought experiment \cite{Gibbins87}. Explaining superpositions with ontological uncertainty would therefore go a long way to explaining these features of quantum theory.

In quantum theory, if $|\psi\rangle$ is a superposition over, say, $|0\rangle$ and $|1\rangle$ then all of the properties of $|\psi\rangle$ are inherited directly from $|0\rangle$ and $|1\rangle$, mediated by the amplitudes of the superposition. This is in a very similar way to how a probability distribution over classical states inherits all of its properties from the underlying classical states, mediated by the probabilities. This situation would naturally be explained if the superposition $|\psi\rangle$ was just a statistical effect over the ontic states corresponding to $|0\rangle$ and $|1\rangle$, exactly as ontic overlap could explain indistinguishability. A full mathematical treatment of this intuition is deferred until Sec.~\ref{sec:SO:ontic-superpositions}.

If the epistemic realist programme is successful in finding an elegant underlying ontology that naturally explains any of these quantum features, it will become a very attractive proposition indeed. The primary problem for the epistemic realist is that no such ontological model has yet been found. Rather, the epistemic realist sits between the success of models that partially reconstruct quantum theory \cite{Spekkens14,Spekkens07,BartlettRudolph+12,JenningsLeifer15} and the ontology theorems that constrain their ability to reconstruct the rest, discussed in Sec.~\ref{sec:SO:previous-theorems}.

\subsection{Classifying Ontologies} \label{sec:SO:classifying-ontologies}

In order to effectively discuss the types of ontological models are or are not compatible with quantum theory, it is necessary to identify some classes of ontological models worth discussing. By finding examples of ontological models in some classes compatible with quantum theory and proving the impossibility of such models in other classes, a clearer prognosis for the epistemic realist perspective emerges. The main classifications from the literature used in this chapter are: $\psi$-ontic, $\psi$-epistemic, maximally $\psi$-epistemic, and various contextualities.

An ontological model for a quantum system is \emph{$\psi$-ontic} if and only if each ontic state can only be prepared by a single quantum state. That is, if one were able to see the ontology of the system directly, there would be no ambiguity as to which quantum state was prepared. The Beltrametti-Bugajski model is trivially $\psi$-ontic, as the ontic states \emph{are} the quantum states, but it is possible to consider others (including Bohmian mechanics) which are $\psi$-ontic but quantum state preparations can result in more than one ontic state.

The opposite of $\psi$-ontic is \emph{$\psi$-epistemic}. That is, an ontological model is $\psi$-epistemic if and only if it is not $\psi$-ontic. A $\psi$-epistemic model is something of a minimum requirement for the epistemic realist since $\psi$-ontic models leave no room for explaining anything other than indeterminism by ontological uncertainty. If there were a theorem proving that all ontological models that reproduce quantum statistics must be $\psi$-ontic then the epistemic realist programme would certainly be dead. However, no such theorem can exist as there are ontological models for every finite dimension which are both $\psi$-epistemic and compatible with quantum theory \cite{LewisJennings+12,AaronsonBouland+13}.

So $\psi$-ontic delineates one extreme of the spectrum of conceivable ontological models. The other end is marked by \emph{maximally $\psi$-epistemic} models. Recall that both quantum states and ontological model preparations can overlap. An ontological model is maximally $\psi$-epistemic if and only if the ontic overlap entirely accounts for the Born rule overlap \cite{Maroney12a,LeiferMaroney13,Leifer14a,Ballentine14}. As will become clear in Sec.~\ref{sec:SO:specific-states-ontology-definitions}, it is impossible for ontic overlaps to be any larger than this, hence such models are ``maximally'' $\psi$-epistemic.

A maximally $\psi$-epistemic ontology would be ideal for the epistemic realist, as the large amount of ontological uncertainty between quantum states would allow for powerful explanations \cite{Leifer14b,Maroney12a,LeiferMaroney13,Leifer14a,Ballentine14}. Such ontological models do exist in $d=2$ dimensions, such as the Kochen-Specker model \cite{KochenSpecker67,HarriganRudolph07}, but no such model can exist for $d>2$, as discussed in Sec.~\ref{sec:SO:previous-theorems}.

These three notions---$\psi$-ontic, $\psi$-epistemic, and maximally $\psi$-epistemic---form a coarse first-order classification of the ontological models of interest to the epistemic realist. Precise definitions of each will be deferred until Sec.~\ref{sec:SO:specific-states-ontology-definitions} once the appropriate mathematical background has been covered. While this classification is far from nuanced, it serves as an excellent starting point to discuss finer distinctions within the class of $\psi$-epistemic models.

For example, the concepts behind the $\psi$-epistemic/ontic dichotomy can also be used to discuss the reality of quantum superpositions. Consider the example of Schr\"{o}dinger's cat. Schr\"{o}dinger's cat is set up to be in a superposition of $|{\rm dead}\rangle$ and $|{\rm alive}\rangle$ quantum states. The epistemic realist (and probably the cat) would ideally prefer the ontic state of the cat to only ever be one of ``dead'' or ``alive'' (\emph{viz.,} only in ontic states accessible when preparing either the $|{\rm dead}\rangle$ or $|{\rm alive}\rangle$ quantum states). In that case, the cat's apparent quantum superposition would be \emph{epistemic}---there would be nothing ``real'' about the superposition state not already captured by $|{\rm dead}\rangle$ and $|{\rm alive}\rangle$. Conversely, if there are ontic states which can only obtain when the cat is in a quantum superposition (and never when the cat is in either quantum $|{\rm dead}\rangle$ or $|{\rm alive}\rangle$ states), then the superposition is unambiguously \emph{ontic}: there are ontological features which correspond to that superposition but not to non-superpositions, so that superposition is real. Precise mathematical definitions will be deferred until Sec.~\ref{sec:SO:ontic-superpositions}.

It is not yet known to what degree epistemic superpositions are compatible with quantum theory. Some of the epistemic realist theories reproducing subsets of quantum theory noted above do include epistemic superpositions. The question of the reality of superpositions in quantum theory is, therefore, very much open.

Obviously quantum superpositions are different from proper mixtures of basis states. The question here is rather whether quantum superpositions can be understood as distributions over some subset of underlying ontic states, where each such ontic state is also accessible by preparing some basis state.

A neat toy example is found in Spekkens' toy theory \cite{Spekkens07}, where the ``toy-bit'' reproduces a subset of qubit behaviour. A toy bit consists of four ontic states, \textbf{$a,b,c,d$}, and four possible preparations, $|0),|1),|+),|-)$, which are analogous to the correspondingly named qubit states. Each preparation corresponds to a uniform probabilistic distribution over exactly two ontic states: $|0)$ is a distribution over $a$ and $b$; $|1)$ a distribution over $c$ and $d$; $|+)$ over $a,c$; and $|-)$ over $b,d$. Full details of how these states behave and how they reproduce qubit phenomena are described in Ref.~\cite{Spekkens07}. For the purposes here, it suffices to note that all ontic states corresponding to the superposition states $|+)$ and $|-)$ are also ontic states corresponding to either $|0)$ or $|1)$---this toy-bit has nothing on the ontological level which can be identified as a superposition. The toy theory superpositions are epistemic. Toy models such as this therefore lend credibility to the idea that quantum superpositions themselves might, in a similar way, fail to have an ontological basis.

Finally, it is common to identify ontological models with various features collectively referred to as ``contextuality'' \cite{Spekkens05,HarriganRudolph07,LeiferMaroney13}. Contextuality was first introduced with what is now known as ``Kochen-Specker contextuality'' \cite{KochenSpecker67,Held13}. However, Ref.~\cite{Spekkens05} showed how contextuality can be more broadly be thought of in terms of operational equivalence. Loosely, contextuality refers to when operationally indistinguishable situations are described differently in the ontological model. There are various types of contextuality of interest in various situations (including Kochen-Specker contextuality) but the only one necessary for this thesis is \emph{preparation contextuality}. Roughly, a model is preparation contextual if it contains separate descriptions for preparations that are operationally equivalent. That is, the ontology depends on the ``context'' of the preparation in a way that is not operationally discernible. A precise mathematical definition will be given in Sec.~\ref{sec:SO:specific-states-ontology-definitions}.

Contextuality of any kind is traditionally thought of as undesirable for the epistemic realist but it is certainly not fatal. It has even been argued that contextuality should be expected in an appropriate epistemic realist theory \cite{Ballentine14}. 

While the classifications introduced here are relatively coarse, they are good starting points for identifying more subtle classifications, some of which will be discussed in Sec.~\ref{sec:SO:specific-states-ontology}. Before getting there though, it is appropriate first to review the main results in the literature concerning the classifications already introduced.

\subsection{Previous Ontology Theorems} \label{sec:SO:previous-theorems}

As noted above, there is a tradition of ``ontology theorems'' in quantum foundations. These are typically no-go theorems, proving that certain classes of ontological models can never be compatible with the predictions of quantum theory. The most famous ontology theorem is Bell's theorem \cite{Bell87}, which primarily concentrates on locality. However, the ontological models framework as defined here---and in much of the recent literature---avoids talking about composition of local systems, preferring to concentrate on single systems where issues of locality and composition do not arise. The primary reason for this is the PBR theorem.

The PBR theorem \cite{PuseyBarrett+12} proves that all ontological models for $d>2$ must be $\psi$-ontic if they satisfy the \emph{preparation independence postulate} (PIP). The PIP is a reasonable extra restriction on the structure of ontological models for multipartite systems that is not present in the bare ontological models framework. The PBR theorem is therefore a very powerful blow to the epistemic realist. However, it does problematically depend on the PIP, which has been challenged \cite{Mansfield16,EmersonSerbin+13,Hall11,SchlosshauerFine14,Wallden13} (especially as being similar to Bell locality, which is already ruled out by Bell's theorem) and without it the PBR theorem is impotent. It is for this reason that recent work has avoided the PIP and related issues by looking only at single systems in order to seek more conclusive ontology results. This will be the perspective taken here.

While $\psi$-epistemic models are not possible in multipartite systems with the PIP, they are possible for single systems without the PIP. As noted above, explicit $\psi$-epistemic models have been presented for every finite dimension in Refs.~\cite{LewisJennings+12,AaronsonBouland+13}. 

The obvious next question is whether maximally $\psi$-epistemic models are possible without the PIP? In $d=2$ dimensions they are, as shown explicitly by the Kochen-Specker model \cite{KochenSpecker67,HarriganRudolph07}. However, they are not possible for $d>2$. There are several disparate ways of proving this \cite{HarriganRudolph07,Ballentine14,LeiferMaroney13,Maroney12a,BarrettCavalcanti+14,Branciard14,Leifer14a}. One of these proofs even existed before the idea of ``maximally $\psi$-epistemic'' \cite{HarriganRudolph07,Ballentine14}. Reference \cite{LeiferMaroney13} proved that ``maximally $\psi$-epistemic'' implies Kochen-Specker contextuality and is therefore ruled out by the Kochen-Specker theorem \cite{KochenSpecker67,Held13}. The majority of these theorems, however, prove their results by bounding ontic overlaps to be less than maximal for $d>2$ \cite{Maroney12a,BarrettCavalcanti+14,Branciard14,Leifer14a,Ballentine14}. As a result these theorems are able to rule out more than simply maximally $\psi$-epistemic models; each also rules out some subset of non-maximally $\psi$-epistemic models, depending on the exact bound found. In particular, the stated aim is often to show that certain ontic overlaps must be small and therefore, whilst $\psi$-epistemic ontologies are possible, using them to explain the indistinguishability of quantum states is implausible. The overlap bounds derived each also tend to zero in some limit, suggesting negligible ontic overlap in those limits (often in large dimensions).

These overlap bound ontology theorems are of most interest to this chapter. As numerical inequalities their conclusions can be more nuanced in ruling out a range of ontological models. They are also likely to be more amenable to experimental investigation \cite{NiggMonz+15,RingbauerDuffus+15,Knee16}. This contrasts especially with results based on the Kochen-Specker theorem, which has finite-precision loopholes \cite{Meyer99,Kent99,CliftonKent00,BarrettKent04}.

The above is not an exhaustive list, nor a thorough discussion, of the current state of ontology theorems. Only those directly relevant to the approach taken in this chapter were included. The majority of other overlap theorems, including Refs.~\cite{ColbeckRenner13b,ColbeckRenner12a,ColbeckRenner11,ColbeckRenner12b,Hardy04,Montina08,Hardy13,PatraPironio+13,AaronsonBouland+13}, are more specialised, often making extra assumptions on top of the ontological models formalism. Several are reviewed at length in Ref.~\cite{Leifer14b}. However, no previous result has yet tackled the ontology of quantum superposition states directly.

Before discussing some shortcomings of these results, it is worth noting that one will often hear that ontological models for quantum systems must be preparation contextual \cite{Spekkens05,LeiferMaroney13,Leifer14b}. However, such results only rule out preparation non-contextuality for \emph{mixed} quantum states---\emph{viz.} they prove that there must exist mixed quantum states where the ontological description of the preparation can vary. Preparation non-contextuality for \emph{pure} quantum states is certainly possible \cite{LewisJennings+12,AaronsonBouland+13,Leifer14b} and often assumed without question. This is likely because the contextuality/non-contextuality distinction for pure state preparations is often irrelevant \cite{Leifer14b}. However, as will be seen in Sec.~\ref{sec:SO:specific-states-ontology}, this will not be the case here and discussion of pure state preparation contextuality will be required.

\subsection{Limitations and Loopholes} \label{sec:SO:limitations}

Of the ontology theorems mentioned in Sec.~\ref{sec:SO:previous-theorems}, those which bound ontic overlaps for $d>2$ are probably the greatest threat to the epistemic realist view. References~\cite{Maroney12a,BarrettCavalcanti+14,Branciard14,Leifer14a,Ballentine14} each prove a bound of this type: a set of quantum states is constructed and an inequality on ontological overlaps is proved to hold for at least one pair from the set, bounding their ontic overlaps to be less than the corresponding Born rule overlap. Trivially, this implies that the ontological model for the system cannot be maximally $\psi$-epistemic. The authors further argue that when these upper bounds become small (typically in large dimensions) then explaining indistinguishability of quantum states by ontological uncertainty becomes implausible. Indeed, if the epistemic realist is hoping to use the ontic overlap to explain quantum features, then it seems unlikely that a very small overlap could have much explanatory power. All of these theorems, however, share at least the following two shortcomings.

First, the proofs are non-constructive proofs of existence. That is, they conclude that there is \emph{some} pair of quantum states with a small ontic overlap, but no guarantee is made as to \emph{which} pair. This may not seem problematic, but consider the following. If a theorem only requires that I incorporate at least one pair of quantum states with small ontic overlap, then I am not prevented from postulating a model where exactly one pair has small ontic overlap, but all other pairs overlap maximally. Such a model would be indistinguishable from a maximally $\psi$-epistemic model, since the probability of encountering exactly that pair of quantum states is zero. Such theorems, as stated, are therefore very weak restrictions on the types of ontological model that are compatible with quantum theory.

The above is a caricature of the loophole and by examining the proofs of the theorems in Refs.~\cite{Maroney12a,BarrettCavalcanti+14,Branciard14,Leifer14a,Ballentine14} it seems possible to use the unitary symmetry of quantum theory to extend them and obtain conclusions that require more pairs of quantum state to have less-than-maximal ontic overlap. However, this has not been done rigorously so the loophole remains, making them much weaker than they seem at face value. Even if such improvements were made, they would still not guarantee that any particular pair of quantum states has less-than-maximal overlap, which may leave significant room for the epistemic realist to explain quantum experiments using ontological uncertainty simply by assuming that the quantum states used in the experiment do have large ontic overlap.

The second shortcoming is less of a loophole and more of a limitation. The overlap bounds proved in Refs.~\cite{Maroney12a,BarrettCavalcanti+14,Branciard14,Leifer14a,Ballentine14} each approach zero in some limit (usually as $d\rightarrow \infty$). However, as noted in Ref.~\cite{Leifer14b} as one approaches these limits, the sets of states considered also approach orthogonality. Orthogonal quantum states are perfectly distinguishable and therefore trivially must have zero ontic overlap. This casts doubt on these results having any meaning at all in those limits where they at first appear to be most powerful. Even being close to the limit makes the results seem less impressive, it is easier to believe that quantum states that are close to orthogonality have very small ontic overlap than those that are nearly collinear.

These theorems are therefore less powerful than they initially may seem and far from conclusive. In particular, it would be much more convincing to have a theorem that could identify particular pairs of quantum states which have bounded ontic overlap and/or a theorem that identified ontic overlaps approaching zero for quantum states with fixed inner product.

\subsection{Chapter Overview} \label{sec:SO:overview}

In this chapter the epistemic realist view of quantum states will be interrogated, with a particular focus on the ontology of superpositions.

In Sec.~\ref{sec:SO:asymmetric-overlap} the asymmetric overlap and anti-distinguishable quantum states will be introduced. This is necessary mathematical groundwork for the results that follow. This will include proving several lemmas at a level of rigour that has not yet been achieved in the literature.

Section~\ref{sec:SO:superpositions-are-real} will discuss the ontology of quantum superposition states. A theorem will be presented showing that for a $d > 3$ dimensional quantum system, almost all quantum superpositions with respect to any given orthonormal basis must be ontic. That is, the epistemic realist must include superpositions in their ontology explicitly.

The techniques used will then be adapted in Sec.~\ref{sec:SO:specific-states-ontology} to address the shortcomings of current ontology theorems noted in Sec.~\ref{sec:SO:limitations}. In particular, theorems will be proved bounding ontic overlap for large numbers of specific pairs of quantum states and demonstrating bounds that approach zero overlap in the large-$d$ limit, even while the quantum states have fixed inner product. In order to adapt the method of the previous theorem, a very mild form of preparation non-contextuality will be assumed, which will be defended as a natural assumption in Sec.~\ref{sec:SO:justifying-prep-noncontextuality}. The feasibility of extending these results to be error-tolerant will be tackled in Sec.~\ref{sec:SO:error-tolerance}.

It is fairly common for ontology theorems to inspire results in quantum information and communication. In Sec.~\ref{sec:SO:communication} this example will be followed and a simple method for exponentially bounding the ability of classical systems to simulate quantum communication will be demonstrated.

Chapter~\ref{ch:MR} will continue directly from this work, applying the same techniques to the problem of ``macro-realism'' in quantum theory. As a result, Sec.~\ref{sec:MR:error-tolerant-argument} will present an error-tolerant variant theorem that can rule out epistemic superpositions. This result will have to be logically weaker than the theorem presented in this chapter to gain error tolerance. The results of both chapters will then be discussed together in Sec.~\ref{sec:MR:summary-and-discussion}.

\section{The Asymmetric Overlap} \label{sec:SO:asymmetric-overlap}

As discussed in Sec.~\ref{sec:SO:previous-theorems}, many ontology theorems proceed by bounding the ontic overlap between quantum states. When a quantum state is prepared, some ontic states can obtain and some cannot. Loosely speaking, the ontic overlap between two quantum states is made up of those ontic states that can obtain by preparing either. To precisely discuss and derive ontology theorems, it is necessary to quantify this notion in some way. 

In quantum theory, the overlap between any pair of $|\psi\rangle,|\phi\rangle\in\mathcal{P}(\mathcal{H})$ is quantified by the Born rule probability $|\langle\phi|\psi\rangle|^{2}$. That is, for a system prepared in state $|\psi\rangle$ the probability for it to behave (for all intents and purposes) like it was prepared in state $|\phi\rangle$ is $|\langle\phi|\psi\rangle|^{2}$.

Adapting this logic to an ontological model for the quantum system, consider the probability that a system prepared according to measure $\mu$ behaves like it was prepared according to $\nu$. That is, the probability that the ontic state obtained from $\mu$ could also have been obtained from $\nu$. This quantity is called the \emph{asymmetric
overlap} and is mathematically defined \cite{Ballentine14,LeiferMaroney13,Maroney12a}
\begin{equation} \label{eq:SO:asymmetric-definition}
\varpi(\nu\,|\,\mu)\eqdef\inf\{\mu(\Omega)\,:\,\Omega\in\Sigma,\,\nu(\Omega)=1\},
\end{equation}
recalling that the infimum of a subset of real numbers is the greatest lower bound of that set. This is because a preparation of $\nu$ has unit probability of producing a $\lambda$ from each measurable subset $\Omega\subseteq\Lambda$ that satisfies $\nu(\Omega)=1$. Therefore minimising $\mu(\Omega)$ with respect to $\Omega$ gives the desired probability. The notation used here deliberately borrows from conditional probabilities and $\varpi(\nu\,|\,\mu)$ may be read ``asymmetric overlap with $\nu$ \emph{given} a preparation of $\mu$''.

It is convenient to slightly overload the terminology and notation and define the asymmetric overlap for more general quantities. First, define the asymmetric overlap between a preparation measure $\mu$ and some \emph{quantum state} $|\phi\rangle$ as the probability that preparing $\mu$ will produce a $\lambda$ obtainable by preparing $|\phi\rangle$. This corresponds to 
\begin{equation} \label{eq:SO:asymmetric-state-definition}
\varpi(|\phi\rangle\,|\,\mu)\eqdef\inf\left\{ \mu(\Omega)\;:\;\Omega\in\Sigma,\;\nu(\Omega)=1,\;\forall\nu\in\Delta_{|\phi\rangle}\right\}.
\end{equation}

The next useful generalisation is the asymmetric overlap of some preparation measure $\mu$ with \emph{two} quantum states $|0\rangle,|\phi\rangle$. This can be thought of as the \emph{union} of the overlaps expressed by $\varpi(|\phi\rangle\,|\,\mu)$ and $\varpi(|0\rangle\,|\,\mu)$ and is mathematically defined as
\begin{equation}
\varpi(|0\rangle,|\phi\rangle\,|\,\mu)\eqdef\inf\left\{ \mu(\Omega)\;:\;\Omega\in\Sigma,\;\nu(\Omega)=\chi(\Omega)=1,\;\forall\nu\in\Delta_{|\phi\rangle},\chi\in\Delta_{|0\rangle}\right\} .
\end{equation}
So $\varpi(|0\rangle,|\phi\rangle\,|\,\mu)$ is the probability that sampling from $\mu$ produces a $\lambda$ obtainable by preparing either $|0\rangle$ or $|\phi\rangle$.

Finally, this can be extended to the asymmetric overlap of $\mu$ with a \emph{set} of quantum states $\mathcal{S}\subseteq\mathcal{P}(\mathcal{H})$ given preparation of $\mu$ in the obvious way
\begin{equation} \label{eq:SO:asymmetric-multipartite-definition}
\varpi(\mathcal{S}\,|\,\mu)\eqdef\inf\left\{ \mu(\Omega)\;:\;\Omega\in\Sigma,\;\nu(\Omega)=1,\;\forall\nu\in\Delta_{|\phi\rangle},\;\forall|\phi\rangle\in\mathcal{S}\right\} .
\end{equation}

Clearly, the asymmetric overlap is not the only sensible way to quantify ontic overlaps. Its main advantages are a natural motivation and ontological interpretation and it will be used extensively in what follows. Another popular choice is the symmetric overlap, which will be used briefly in Sec.~\ref{sec:SO:error-tolerance}.

The remainder of this section will flesh out the asymmetric overlap's properties. This will unfortunately be somewhat dry and mathematical but, once in place, will significantly reduce the complexity of the proofs later in the chapter.

\subsection{Properties of the Asymmetric Overlap} \label{sec:SO:asymmetric-properties}

Many of the properties of the asymmetric overlap derived here will not greatly surprise those familiar with using ontological models. However, they have not yet been achieved at this level of rigour using the measure-theoretic approach to ontological models, necessitating full proofs.

While this rigour is important, it can often obscure the more intuitive reasons that the results hold. To offset this, most proofs presented below will start with a rough-but-intuitive argument followed by the more cumbersome-but-correct measure theory.

It is convenient to begin with the following definition to simplify the notation for the rest of this section.

\begin{definition} \label{def:SO:k-bar}
For any measurable function $g:\Lambda \rightarrow [0,1]$ let
\begin{equation}
\k(g) \eqdef \ker(1 - g) = \{ \lambda\in\Lambda \,:\, g(\lambda) = 1 \} \in \Sigma.
\end{equation}
\end{definition}

Next, this technical lemma forms the basis of proofs of many of the following properties.

\begin{lemma} \label{lem:SO:measurable}
Given any measurable function $f:\Lambda\rightarrow[0,1]$ and preparation measure $\nu$ satisfying $\int_\Lambda \d\nu(\lambda)f(\lambda) = 1$, then $\nu(\k(f)) = 1$.
\end{lemma}
\begin{proof}
Roughly, this lemma simply affirms that if the average $f(\lambda)$ according to $\nu$ is unity, then the probability that $f(\lambda)=1$ according to $\nu$ is also unity.

If $\int_{\Lambda}\mathrm{d}\nu(\lambda)\,f(\lambda)=1$ then 
\begin{eqnarray} 
1 & = & \int_{\k(f)} \d\nu(\lambda) \,f(\lambda) + \int_{\Lambda\setminus\k(f)} \d\nu(\lambda) \,f(\lambda) \\
 & = & \nu(\ker\bar{f}) + \int_{\Lambda\setminus\k(f)} \d\nu(\lambda) \,f(\lambda) \label{eq:SO:lemma-measurable-proof-expanded}
\end{eqnarray}
since if $\lambda\in\k(f)$ then $f(\lambda)=1$. 

Suppose that the the second term in Eq.~(\ref{eq:SO:lemma-measurable-proof-expanded}) is non-zero. Since $f(\lambda) < 1$ for all $\lambda\in\Lambda\setminus\k(f)$, then $\int_{\Lambda\setminus\k(f)} \d\nu(\lambda) \,f(\lambda) < \nu(\Lambda\setminus\k(f))$. This further implies
\begin{eqnarray}
\nu(\k(f)) + \nu(\Lambda\setminus\k(f)) & > & 1 \\
\nu(\Lambda) & > & 1
\end{eqnarray}
which is a contradiction as $\nu(\Lambda)=1$ by definition. Therefore the second term in Eq.~(\ref{eq:SO:lemma-measurable-proof-expanded}) must be zero and Eq.~(\ref{eq:SO:lemma-measurable-proof-expanded}) implies $1 = \nu(\k(f))$ as desired.
\end{proof}

Having established this technical background, the following shows a fundamental property of the asymmetric overlap: it is upper-bounded by the Born rule probability.

\begin{lemma} \label{lem:SO:asymmetric-basic-bound}
For any pair of pure quantum states $|\psi\rangle,|\phi\rangle\in\mathcal{P}(\mathcal{H})$ and for any preparation $\mu\in\Delta_{|\psi\rangle}$, then
\begin{equation}
\varpi( |\phi\rangle \,|\, \mu ) \leq |\langle\phi | \psi\rangle |^2.
\end{equation}
From this it immediately follows that $\varpi(\nu\,|\,\mu) \leq |\langle\phi|\psi\rangle|^2$ for every $\nu\in\Delta_{|\phi\rangle}$.
\end{lemma}
\begin{proof}
Roughly, this holds because almost all ontic states obtained by preparing $|\phi\rangle$ must also return $|\phi\rangle$ in any measurement where that is an option. $|\langle\phi|\psi\rangle|^2$ is the probability of getting outcome $|\phi\rangle$ when $|\psi\rangle$ has been prepared. This must therefore occur at least as often as obtaining an ontic state accessible by $|\phi\rangle$ when preparing $|\psi\rangle$. This will now be properly proved.

Consider preparing $|\phi\rangle$ via any $\nu\in\Delta_{|\phi\rangle}$ and then performing some quantum measurement $M_\phi \ni |\phi\rangle$. Since the ontological model reproduces quantum probabilities [Eq.~(\ref{eq:IN:ontological-model-quantum-probability})] then
\begin{equation}
\int_\Lambda \d\nu(\lambda)\, \mathbb{P}_{M_\phi}(|\phi\rangle\,|\,\lambda) = 1.
\end{equation}
Letting $g(\lambda) \eqdef \mathbb{P}_{M_\phi}(|\phi\rangle\,|\,\lambda)$, Lem.~\ref{lem:SO:measurable} shows that $\nu(\k(g))=1$ for every $\nu\in\Delta_{|\phi\rangle}$.

Now consider preparing $|\psi\rangle$ via any $\mu\in\Delta_{|\psi\rangle}$ and then measuring with the same $M_\phi$. By Eq.~(\ref{eq:IN:ontological-model-quantum-probability})
\begin{equation}
|\langle\phi | \psi\rangle|^2 = \int_\Lambda \d\mu(\lambda)\,g(\lambda) \geq \int_{\k(g)} \d\mu(\lambda)\,g(\lambda) = \mu(\k(g)).
\end{equation}
having used that $g(\lambda\in\k(g))=1$ in the final step. Recalling that $\nu(\k(g))=1$ for every $\nu\in\Delta_{|\phi\rangle}$, then $\varpi(|\phi\rangle\,|\,\mu) \leq \mu(\k(g))$ by definition. Combining these inequalities gives the desired general result. Applying these inequalities to any particular $\nu\in\Delta_{|\phi\rangle}$ gives the specific case $\varpi(\nu\,|\,\mu) \leq |\langle\phi|\psi\rangle|^2$.
\end{proof}

The next property demonstrates that the asymmetric overlap is non-increasing under transformations.

\begin{lemma} \label{lem:SO:asymmetric-unitary}
Let unitary $U$ satisfy $U|0\rangle = |\phi\rangle$ and $\mu^\prime \transto{\gamma} \mu$ for some $\gamma\in\Gamma_U$, then
\begin{equation}
\varpi( |\phi\rangle \,|\, \mu ) \geq \varpi( |0\rangle \,|\, \mu^\prime ).
\end{equation}
\end{lemma}
\begin{proof}
Roughly, this is because ontic states preparable by $\mu^\prime$ will be mapped onto ontic states preparable by $\mu$ by $\gamma$ (since preparing $\mu^\prime$ then applying $\gamma$ is equivalent to preparing $\mu$). Similarly, $\gamma$ maps ontic states preparable by $|0\rangle$ onto ontic states preparable by $|\phi\rangle$. Thus, ontic states in the overlap of $\mu^\prime$ and $|0\rangle$ will be mapped onto the overlap of $\mu$ and $|\phi\rangle$, implying that the overlap as measured by $\varpi$ cannot decrease under the action of $\gamma$.

It suffices to prove that for any measurable $\Omega \in \Sigma$ satisfying $\nu(\Omega) = 1,\forall \nu \in \Delta_{|\phi\rangle}$ there exists some measurable $\Omega^\prime \in \Sigma$ such that $\chi(\Omega^\prime)=1,\forall \chi\in\Delta_{|0\rangle}$ and $\mu(\Omega) \geq \mu^\prime(\Omega^\prime)$.

For any $\chi\in\Delta_{|0\rangle}$ there is some $\nu\in\Delta_{|\phi\rangle}$ such that $\chi \transto{\gamma} \nu$. So for any such $\Omega$ Eq.~(\ref{eq:IN:ontological-model-quantum-probability}) gives
\begin{equation}
1 = \nu(\Omega) = \int_\Lambda \d\chi(\lambda)\,\gamma(\Omega\,|\,\lambda).
\end{equation}
Letting $g(\lambda) \eqdef \gamma(\Omega\,|\,\lambda)$ then by Lem.~\ref{lem:SO:measurable} this implies $\chi(\k(g)) = 1$ for all $\chi\in\Delta_{|0\rangle}$. Therefore $\Omega^\prime \eqdef \k(g)$ is a valid choice of $\Omega^\prime$.

Similarly, therefore
\begin{equation}
\mu(\Omega) = \int_\Lambda \d\mu^\prime (\lambda)\,\gamma(\Omega\,|\,\lambda) = \int_\Lambda \d\mu^\prime(\lambda)\,g(\lambda) \geq \int_{\Omega^\prime} \d\mu^\prime(\lambda)\,g(\lambda) = \mu^\prime(\Omega^\prime)
\end{equation}
recalling that $g(\lambda\in\k(g)) = 1$ by definition.
\end{proof}

The following property is perhaps the easiest to both understand and prove. It relates arbitrary multipartite overlaps, Eq.~(\ref{eq:SO:asymmetric-multipartite-definition}), to the underlying individual overlaps.

\begin{lemma} \label{lem:SO:asymmetric-Boole}
For any finite set $\mathcal{S}\subset\mathcal{P}(\mathcal{H})$ of quantum states and any preparation measure $\mu$
\begin{equation} \label{eq:SO:lemma-asymmetric-Boole}
\sum_{|i\rangle\in\mathcal{S}} \varpi( |i\rangle \,|\, \mu ) \geq \varpi(\mathcal{S} \,|\, \mu ).
\end{equation}
\end{lemma}
\begin{proof}
The overlap $\varpi(\mathcal{S} \,|\, \mu )$ is the probability of a disjunction of events and each $\varpi( |i\rangle \,|\, \mu )$ is a probability of one of those events. The result is therefore a simple application of Boole's inequality.
\end{proof}

This last property relates tripartite asymmetric overlaps to certain quantum measurements. It is a little arbitrary, but will be repeatedly used in the theorems that follow and is therefore useful to prove separately here.

\begin{lemma} \label{lem:SO:asymmetric-triple-measurement}
Consider quantum states $|\psi\rangle, |\phi\rangle, |0\rangle$ and orthonormal basis $\mathcal{B}\supset\{|a\rangle,|b\rangle\}$. Suppose that, when written as superpositions over $\mathcal{B}$, $|\psi\rangle$ and $|\phi\rangle$ only have common support on $|a\rangle$ and $|b\rangle$ (\emph{viz.} they are orthogonal except for components in the linear subspace spanned by $|a\rangle,|b\rangle$). Similarly, suppose that $|\psi\rangle$ and $|0\rangle$ only have common support on $|a\rangle$ and $|b\rangle$. Then for every $\mu\in\Delta_{|\psi\rangle}$
\begin{equation}
\varpi( |\phi\rangle, |0\rangle \,|\, \mu ) \leq \mathbb{P}_\mathcal{B} ( |a\rangle\vee|b\rangle \,|\, |\psi\rangle )
\end{equation}
for some basis measurement in $\mathcal{B}$.
\end{lemma}
\begin{proof}
The gist of the proof is that if ontic state $\lambda$ is accessible by preparing two quantum states then it may only return measurement results that are compatible with both preparations. So if $\lambda$ is accessible by preparing both $|\psi\rangle$ and $|\phi\rangle$, then it may only return $|a\rangle$ or $|b\rangle$ in the measurement of $\mathcal{B}$. Similarly, if $\lambda$ is accessible by preparing both $|\psi\rangle$ and $|0\rangle$, then it may only return $|a\rangle$ or $|b\rangle$ for a similar measurement. Thus, if one of these $\lambda$s is obtained in a preparation of $|\psi\rangle$ then the measurement result is necessarily either $|a\rangle$ or $|b\rangle$. This will now be fully fleshed out.

Consider a measurement $M_{\mathcal{B}}$ of the basis $\mathcal{B}$ and let $f(\lambda)\eqdef\mathbb{P}_{\mathcal{B}}(|a\rangle\,|\,\lambda)+\mathbb{P}_{\mathcal{B}}(|b\rangle\,|\,\lambda)$. Let $A\subset\mathcal{B}$ be the set of remaining basis states on which $|\phi\rangle$ has support (\emph{i.e.} $|\phi\rangle$ is a superposition over $|a\rangle$, $|b\rangle$, and $A$) and define $g_\phi (\lambda) \eqdef \mathbb{P}_{\mathcal{B}}(A\,|\,\lambda)$. Similarly define $g_0 (\lambda) \eqdef \mathbb{P}_\mathcal{B}(B\,|\,\lambda)$ for $B\subset\mathcal{B}$ the remaining basis states over which $|0\rangle$ has support.

By these definitions, the following quantum probabilities are known for any $\nu\in\Delta_{|\phi\rangle}$ and $\chi\in\Delta_{|0\rangle}$ by Eq.~(\ref{eq:IN:ontological-model-quantum-probability})
\begin{eqnarray}
\int_{\Lambda}\d\nu(\lambda)\left(f(\lambda)+g_\phi(\lambda)\right) & = & 1,\\
\int_{\Lambda}\d\chi(\lambda)\left(f(\lambda)+g_0(\lambda)\right) & = & 1.
\end{eqnarray}
Therefore by Lem.~\ref{lem:SO:measurable} it follows that $\nu(\k(f+g_\phi))=\chi(\k(f+g_0))=1$. Since $\k(f+g_\phi)\subseteq\k(f+g_\phi+g_0)\supseteq\k(f+g_0)$ it follows that $\Omega\eqdef\k(f+g_\phi+g_0)$ is a measurable subset of $\Lambda$ for which $\nu(\Omega)=\chi(\Omega)=1,\,\forall\nu\in\Delta_{|\phi\rangle},\chi\in\Delta_{|0\rangle}$.

The desired quantum probability is given by
\begin{equation}
\mathbb{P}_\mathcal{B}(|a\rangle,|b\rangle\,|\,|\psi\rangle)=\int_{\Lambda}\d\mu(\lambda)\,f(\lambda).
\end{equation}
Since $|a\rangle$ and $|b\rangle$ are the \emph{only} basis states where $|\psi\rangle$ and $|\phi\rangle$ have common support, then $\int_\Lambda\d\mu(\lambda)g_\phi(\lambda) = 0$. Similarly, $\int_\Lambda\d\mu(\lambda)g_0(\lambda) = 0$ and so it follows that
\begin{eqnarray}
\mathbb{P}_\mathcal{B}(|a\rangle,|b\rangle\,|\,|\psi\rangle) & = & \int_\Lambda\d\mu(\lambda)\left(f(\lambda)+g_\phi(\lambda)+g_0(\lambda)\right)\\
 & \geq & \int_{\k(f+g_\phi+g_0)}\d\mu(\lambda)\left(f(\lambda)+g_\phi(\lambda)+g_0(\lambda)\right)\\
 & = & \mu\left(\k(f+g_\phi+g_0)\right)=\mu(\Omega).
\end{eqnarray}
Since, by definition $\mu(\Omega)$ upper bounds $\varpi(|\phi\rangle,|0\rangle\,|\,\mu)$
this completes the proof.
\end{proof}

This does not quite conclude the properties of the asymmetric overlap required for this chapter. Before continuing to the main result it will be necessary to consider how asymmetric overlaps interact with so-called anti-distinguishable quantum states.

\subsection{Anti-distinguishability} \label{sec:SO:anti-distinguishability}

Quantum states are perfectly distinguishable---that is, there is a measurement telling them apart with certainty---if and only if they are mutually orthogonal. Distinguishable states must be also ontologically distinct in order to satisfy Eq.~(\ref{eq:IN:ontological-model-probability}): if the preparation measures non-trivially overlapped, then any ontic states in that overlap would fail to return \emph{any} consistent outcome in a distinguishing measurement. This makes distinguishable states too restrictive to be very helpful in ontology theorems. The opposite and more subtle concept of \emph{anti-distinguishability} is much more useful in discussions of ontic overlaps\footnote{Anti-distinguishability was introduced in Ref.~\cite{CavesFuchs+02a} under the name ``PP-incompatibility'' and was given the more informative name of anti-distinguishability in Ref.~\cite{Leifer14b}.}.

A finite set of quantum states $\{|\psi\rangle,|\phi\rangle,...\}\subset\mathcal{P}(\mathcal{H})$ is anti-distinguishable if and only if there exists a measurement $M=\{E_{\neg\psi},E_{\neg\phi},...\}$ such that 
\begin{equation} \label{eq:SO:anti-distinguishing-measurement}
\langle\psi|E_{\neg\psi}|\psi\rangle=\langle\phi|E_{\neg\phi}|\phi\rangle=...=0,
\end{equation}
\emph{i.e.} the measurement can tell, with certainty, one state from the set that was \emph{not} prepared. This is the ``opposite'' to distinguishable since for distinguishable sets there is a measurement that can tell, with certainty, one state which \emph{was} prepared. It has been proved\footnote{This result was proved in Ref.~\cite{CavesFuchs+02a} but Ref.~\cite{BarrettCavalcanti+14} points out and corrects a typographical error in their result (the original had the second inequality as a strict inequality, which is incorrect).} that if some inner products $a=|\langle\phi|\psi\rangle|^{2}$, $b=|\langle0|\psi\rangle|^{2}$, $c=|\langle0|\phi\rangle|^{2}$ satisfy
\begin{equation} \label{eq:SO:anti-distinguishable-criterion}
a+b+c<1,\quad(1-a-b-c)^{2}\geq4abc,
\end{equation}
then the triple $\{|\psi\rangle,|\phi\rangle,|0\rangle\}$ must be anti-distinguishable by a projective measurement.

For the purposes of ontological models, the main utility of anti-distinguishable sets is that they exclude intersections of certain overlaps. This can also be viewed in relation to Lem.~\ref{lem:SO:asymmetric-Boole}: anti-distinguishable triples guarantee that Eq.~(\ref{eq:SO:lemma-asymmetric-Boole}) holds with equality. This will now be proved, together with a rough-but-intuitive argument.

\begin{lemma} \label{lem:SO:anti-distinguishing-equality}
For any anti-distinguishable triple of quantum states $\{|\psi\rangle,|\phi\rangle,|0\rangle\}$ Lem.~\ref{lem:SO:asymmetric-Boole} holds with equality, that is for all $\mu\in\Delta_{|\psi\rangle}$
\begin{equation} \label{eq:SO:lemma-anti-distinguishing-equality}
\varpi( |0\rangle, |\phi\rangle\, |\, \mu ) = \varpi( |0\rangle\,|\,\mu ) + \varpi( |\phi\rangle\,|\,\mu ).
\end{equation}
\end{lemma} 
\begin{proof}
Roughly speaking, this result is fairly easy to see. Probabilities $\varpi(|\phi\rangle\,|\,\mu)$ and $\varpi(|0\rangle\,|\,\mu)$ correspond to the events of preparing $|\psi\rangle$ via $\mu$ and getting an ontic state compatible with preparing $|\phi\rangle$ and $|0\rangle$ respectively. Suppose that these events are not mutually exclusive, \emph{viz.} suppose that $\mu$ can prepare an ontic state compatible with preparing \emph{both} $|\phi\rangle$ and $|0\rangle$ at the same time. Such an ontic state would not be able to return a consistent outcome for the anti-distinguishing measurement: being compatible with preparations of all three states rules out all measurement outcomes. This is a contradiction, so the events must indeed be mutually exclusive and therefore the probabilities sum to the probability of their disjunction, which is exactly what Eq.~(\ref{eq:SO:lemma-anti-distinguishing-equality}) requires. This will now be rigorously proved.

Recall that $\{|\psi\rangle,|\phi\rangle,|0\rangle\}$ is an anti-distinguishable triple if and only if there is some quantum measurement $M$ with three outcomes $E_{\neg\psi},E_{\neg\phi},E_{\neg0}$ such that the outcome of getting $E_{\neg\psi}$ from a system prepared in state $|\psi\rangle$ is zero and similarly for the other state/outcome pairs. By Eq.~(\ref{eq:IN:ontological-model-quantum-probability}) it therefore follows that for all $\mu\in\Delta_{|\psi\rangle},\nu\in\Delta_{|\phi\rangle},\chi\in\Delta_{|0\rangle}$
\begin{eqnarray}
\int_{\Lambda}\d\mu(\lambda)\,\mathbb{P}_{M}(E_{\neg\psi}|\,\lambda) & = & 0, \label{eq:SO:lemma-anti-distinguishable-measurement-1}\\
\int_{\Lambda}\d\nu(\lambda)\,\mathbb{P}_{M}(E_{\neg\phi}|\,\lambda) & = & 0, \label{eq:SO:lemma-anti-distinguishable-measurement-2}\\
\int_{\Lambda}\d\chi(\lambda)\,\mathbb{P}_{M}(E_{\neg0}|\,\lambda) & = & 0.\label{eq:SO:lemma-anti-distinguishable-measurement-3}
\end{eqnarray}

To prove that Lem.~\ref{lem:SO:asymmetric-Boole} holds with equality it suffices to show that given any $\Omega\in\Sigma$ for which $\nu(\Omega)=\chi(\Omega)=1$ for all $\nu\in\Delta_{|\phi\rangle},\chi\in\Delta_{|0\rangle}$, there exists some $\Omega^{\prime},\Omega^{\prime\prime}\in\Sigma$ for which $\nu(\Omega^{\prime})=\chi(\Omega^{\prime\prime})=1$ for all $\nu\in\Delta_{|\phi\rangle},\chi\in\Delta_{|0\rangle}$ and
\begin{equation}
\mu(\Omega)\geq\mu(\Omega^{\prime})+\mu(\Omega^{\prime\prime}).\label{eq:SO:lemma-anti-distinguishability-equality-to-prove}
\end{equation}
This, together with Lem.~\ref{lem:SO:asymmetric-Boole} itself, would prove the desired result since the right-hand side bounds $\varpi(|\phi\rangle\,|\,\mu)+\varpi(|0\rangle\,|\,\mu)$ from above.

To prove that Eq.~(\ref{eq:SO:lemma-anti-distinguishability-equality-to-prove}) holds define the following measurable functions from $\Lambda$ to $[0,1]$
\begin{eqnarray}
g_{\psi}(\lambda) & \eqdef & \mathbb{P}_{M}(E_{\neg\phi}|\,\lambda)+\mathbb{P}_{M}(E_{\neg0}|\,\lambda),\\
g_{\phi}(\lambda) & \eqdef & \mathbb{P}_{M}(E_{\neg\psi}|\,\lambda)+\mathbb{P}_{M}(E_{\neg0}|\,\lambda),\\
g_{0}(\lambda) & \eqdef & \mathbb{P}_{M}(E_{\neg\psi}|\,\lambda)+\mathbb{P}_{M}(E_{\neg\phi}|\,\lambda).
\end{eqnarray}
Using the fact that, for any $\lambda\in\Lambda$, the sum of probabilities of outcomes for any measurement must be unity it follows that $\mathbb{P}_{M}(E_{\neg\psi}|\lambda)=1-g_{\psi}(\lambda)$ and similarly for $|\phi\rangle$ and $|0\rangle$. Therefore Eqs.~(\ref{eq:SO:lemma-anti-distinguishable-measurement-1}, \ref{eq:SO:lemma-anti-distinguishable-measurement-2}, \ref{eq:SO:lemma-anti-distinguishable-measurement-3}) are equivalent to 
\begin{eqnarray}
\int_{\Lambda}\d\mu(\lambda)\,g_{\psi}(\lambda) & = & 1,\\
\int_{\Lambda}\d\nu(\lambda)\,g_{\phi}(\lambda) & = & 1,\\
\int_{\Lambda}\d\chi(\lambda)\,g_{0}(\lambda) & = & 1.
\end{eqnarray}
By Lem.~\ref{lem:SO:measurable} it immediately follows that $\nu(\k(g_{\phi}))=\chi(\k(g_{0}))=1$ where, recall,
$\nu$ and $\chi$ are arbitrary measures from $\Delta_{|\phi\rangle}$ and $\Delta_{|0\rangle}$ respectively.

With these definitions, consider $\mu(\Omega)$ for any $\Omega\in\Sigma$ satisfying $\nu(\Omega)=\chi(\Omega)=1$ for all $\nu\in\Delta_{|\phi\rangle},\chi\in\Delta_{|0\rangle}$.
\begin{eqnarray}
\mu(\Omega) & = & \int_{\Omega}\d\mu(\lambda)\\
 & = & \int_{\Omega}\d\mu(\lambda)\Bigl(g_\phi(\lambda) + g_0(\lambda) - \mathbb{P}_{M}(E_{\neg\psi}|\,\lambda)\Bigr)
\end{eqnarray}
follows by definition of $g_{\phi,0}$. The last term vanishes by Eq.~(\ref{eq:SO:lemma-anti-distinguishable-measurement-1}), so
\begin{equation}
\mu(\Omega) = \int_{\Omega}\d\mu(\lambda)\,g_{\phi}(\lambda)+\int_{\Omega}\d\mu(\lambda)\,g_{0}(\lambda).
\end{equation}
By restricting the domain of integration
\begin{eqnarray}
\mu(\Omega) & \geq & \int_{\Omega\cap\k(g_{\phi})}\d\mu(\lambda)\,g_{\phi}(\lambda)+\int_{\Omega\cap\k(g_{0})}\d\mu(\lambda)\,g_{0}(\lambda)\\
 & = & \mu\left(\Omega\cap\k(g_{\phi})\right)+\mu\left(\Omega\cap\k(g_{0})\right)
\end{eqnarray}
recalling that $\forall\lambda\in\k(g_{\phi})$, $g_{\phi}(\lambda)=1$ (and similarly for $g_{0}$). Note that as both $\Omega$ and $\k(g_{\phi})$ are measure-one according to any $\nu\in\Delta_{|\phi\rangle}$, it follows that their intersection also satisfies $\nu\left(\Omega\cap\k(g_{\phi})\right)=1$. Similarly, $\chi\left(\Omega\cap\k(g_{0})\right)=1$. Thus what has been proved is that given any $\Omega\in\Sigma$ such that $\nu(\Omega)=\chi(\Omega)=1$ for all $\nu\in\Delta_{|\phi\rangle},\chi\in\Delta_{|0\rangle}$ there exist measurable sets $\Omega^{\prime}\eqdef\Omega\cap\k(g_{\phi})$ and $\Omega^{\prime\prime}\eqdef\Omega\cap\k(g_{0})$ satisfying $\nu(\Omega^{\prime})=\chi(\Omega^{\prime\prime})=1$ for all $\nu\in\Delta_{|\phi\rangle},\chi\in\Delta_{|0\rangle}$ and
\begin{equation}
\mu(\Omega)\geq\mu(\Omega^{\prime})+\mu(\Omega^{\prime\prime}).
\end{equation}
This is exactly what was to be proved.
\end{proof}

Together, Lems.~\ref{lem:SO:asymmetric-basic-bound}--\ref{lem:SO:anti-distinguishing-equality} form the mathematical scaffolding necessary to prove the main results of this chapter that follow.

\section{Must Superpositions be Real?} \label{sec:SO:superpositions-are-real}

With the technical background of asymmetric overlaps and anti-distinguishable sets it is now possible to properly address the main question of this chapter. Is it possible to have an ontology for quantum systems where superpositions are simply statistical effects---where they are not real?

\subsection{Distinguishing Ontic Superpositions} \label{sec:SO:ontic-superpositions}

In order to even precisely formulate this question, one first needs to carefully define what is meant by ``real'' superpositions. This was touched on in Sec.~\ref{sec:SO:classifying-ontologies} but will now be done precisely.

Before even getting to the ontology, it is important to note that superpositions are defined with respect to some orthonormal basis. To be exact, $|\psi\rangle\in\mathcal{P}(\mathcal{H})$ is \emph{a superposition with respect to orthonormal basis $\mathcal{B}$} if and only if $|\psi\rangle$ has non-zero inner product with more than one state in $\mathcal{B}$. This is equivalent to saying that $|\psi\rangle\not\in\mathcal{B}$.

Now, suppose that Bob believes superpositions with respect to some particular basis $\mathcal{B}$ are just a statistical effect. That is, he believes that superpositions have no ontology of their own but are emergent effects from the statistics. In terms of ontological models, this means that an ontic state obtained by preparing $|\psi\rangle\not\in\mathcal{B}$ cannot be independent of $\mathcal{B}$. That is, Bob believes that there is no new ontology required to describe $|\psi\rangle$ that was not already present when describing $\mathcal{B}$. In this case, he is compelled to say that whenever he prepares $|\psi\rangle$ the probability of getting an ontic state obtainable by preparing some state from $\mathcal{B}$ is unity:
\begin{equation} \label{eq:SO:epistemic-superposition}
\varpi(\mathcal{B}\,|\,\mu) = 1,\;\forall\mu\in\Delta_{|\psi\rangle}.
\end{equation}

To be consistent with the language of Sec.~\ref{sec:SO:classifying-ontologies}, call $|\psi\rangle$ an \emph{epistemic superposition with respect to $\mathcal{B}$} if and only if Eq.~(\ref{eq:SO:epistemic-superposition}) holds. Conversely, if Eq.~(\ref{eq:SO:epistemic-superposition}) is violated then there is a finite probability that preparing $|\psi\rangle$ results in novel ontic states not accounted for by $\mathcal{B}$. Therefore call $|\psi\rangle$ an \emph{ontic superposition with respect to $\mathcal{B}$} if and only if it is not epistemic.

\subsection{Almost All Superpositions are Real} \label{sec:SO:superpositions-are-real-subsection}

Superpositions epistemic with respect to some basis would have considerable explanatory power. Any seemingly bizarre quantum effect based on using such a superposition would have a neat underlying explanation in terms of the ontology of basis states. So, are such things possible?

Perhaps unfortunately, no. Indeed, the no-go result proved below is just about the strongest rejection of epistemic superpositions one could think of. It is relatively easy, for example, to prove that ``not every quantum superposition is epistemic''. What is proved in the following two theorems is the much stronger statement that almost every superposition with respect to any given orthonormal basis is ontic.

This result is presented in two theorems. The first contains the main technical argument.

\begin{theorem}
\label{thm:SO:no-maximal-overlap-for-basis}
Consider a $d>3$ dimensional quantum system and any orthonormal basis $\mathcal{B}$ of $\mathcal{H}$. For any $|\psi\rangle\in\mathcal{P}(\mathcal{H})$ such that there is a $|0\rangle\in\mathcal{B}$ satisfying $|\langle0|\psi\rangle|^2\in(0,\frac{1}{2})$, then there exists some other $|i\rangle\in\mathcal{B}$ and $\mu\in\Delta_{|\psi\rangle}$ such that
\begin{equation}
\label{eq:SO:thm:no-maximal-overlap-for-basis}
\varpi(|i\rangle\,|\,\mu) \neq |\langle i|\psi\rangle|^2.
\end{equation}
\end{theorem}

\begin{proof}
The proof proceeds by contradiction. To that end, assume that 
\begin{equation}
\label{eq:SO:proof-technical-background-contradiction-assumption}
\varpi(|j\rangle\,|\,\mu)=|\langle j|\psi\rangle|^2,\quad \forall|j\rangle\in\mathcal{B},\; \forall\mu\in\Delta_{|\psi\rangle}.
\end{equation}

Consider $|0\rangle\in\mathcal{B}$ as described in the theorem's statement. By choosing the global phase of $|\psi\rangle$ appropriately, another orthonormal basis $\mathcal{B}^\prime = \{|0\rangle\}\cup\{|i^\prime\rangle\}_{i=1}^{d-1}$ can be defined such that
\begin{equation} \label{eq:SO:proof-psi-no-maximal-overlap}
|\psi\rangle = \alpha|0\rangle + \beta|1^\prime\rangle + \tau|2^\prime\rangle,
\end{equation}
where $\alpha\in(0,\frac{1}{\sqrt{2}})$ and $\beta\eqdef\sqrt{2}\alpha^2$. Note that $\alpha = \langle 0|\phi\rangle$ can be taken to be real and positive without loss of generality, as states in $\mathcal{P}(\mathcal{H})$ are equivalent up to global phase factors. This is always possible since $|\alpha|^2 + |\beta|^2 = \alpha^2 (1 + 2\alpha^2) < 1$ for every such $\alpha$. With respect to the same $\mathcal{B}^\prime$ define
\begin{equation} \label{eq:SO:proof-phi-no-maximal-overlap}
|\phi\rangle \eqdef \delta|0\rangle + \eta|1^\prime\rangle + \kappa|3^\prime\rangle,
\end{equation}
where $\delta\eqdef 1 - 2\alpha^2$ and $\eta\eqdef\sqrt{2}\alpha$. This is always possible since $|\delta|^2 + |\eta|^2 = (1 - 2\alpha^2)^2 + 2\alpha^2 < 1$ for all $\alpha$ and thus an appropriate $\kappa$ always exists.

The above construction has been chosen such that 
\begin{itemize}
\item $|\langle 0|\psi\rangle|^2 = \alpha^2 = |\langle\phi | \psi\rangle|^2$ so that there exists a unitary operator $U$ for which $U|0\rangle=|\phi\rangle$ and $U|\psi\rangle=|\psi\rangle$; and
\item the inner products $|\langle 0|\psi\rangle|^2$, $|\langle \phi|\psi\rangle|^2$, $|\langle 0|\phi\rangle|^2$ satisfy Eq.~(\ref{eq:SO:anti-distinguishable-criterion}) and, therefore, the triple $\{|\psi\rangle,|\phi\rangle,|0\rangle\}$ is anti-distinguishable.
\end{itemize}

Choose any preparation measure $\mu^\prime\in\Delta_{|\psi\rangle}$ and any stochastic map $\gamma\in\Gamma_U$. Let $\mu\in\Delta_{U|\psi\rangle}=\Delta_{|\psi\rangle}$ be the preparation measure such that $\mu^\prime \transto{\gamma}\mu$. Then 
\begin{eqnarray}
\varpi(|\phi\rangle,|0\rangle\,|\,\mu) &= & \varpi(|\phi\rangle\,|\,\mu) + \varpi(|0\rangle\,|\,\mu) \\
&\geq & \varpi(|0\rangle\,|\,\mu^\prime) + \varpi(|0\rangle\,|\,\mu) \\
\label{eq:SO:proof-technical-background-overlap-to-alpha}
&= & 2|\langle 0|\psi\rangle|^2 = 2\alpha^2
\end{eqnarray}
where the first line follows from anti-distinguishability and Lem.~\ref{lem:SO:anti-distinguishing-equality}, the second from Lem.~\ref{lem:SO:asymmetric-unitary}, and the third from Eq.~(\ref{eq:SO:proof-technical-background-contradiction-assumption}).

Now consider that in any basis measurement $M$ of $\mathcal{B}^\prime$, $|0\rangle$ and $|1^\prime\rangle$ are the only outcomes compatible with preparations of either both $|\psi\rangle$ \& $|\phi\rangle$ or both $|\psi\rangle$ \& $|0\rangle$. So by Lem.~\ref{lem:SO:asymmetric-triple-measurement}, $\mathbb{P}_M (|0\rangle\vee|1^\prime\rangle\,|\,|\psi\rangle) \geq \varpi(|0\rangle,|\phi\rangle\,|\,\mu)$. Therefore, in order to reproduce quantum predictions and satisfy Eq.~(\ref{eq:IN:ontological-model-quantum-probability}),
\begin{equation}
\label{eq:SO:proof-technical-background-measurement-to-alpha}
\varpi(|0\rangle,|\phi\rangle\,|\,\mu) \leq |\langle0|\phi\rangle|^2 + |\langle1^\prime|\psi\rangle|^2 = \alpha^2 + 2\alpha^4.
\end{equation}

Combining Eqs.~(\ref{eq:SO:proof-technical-background-overlap-to-alpha}, \ref{eq:SO:proof-technical-background-measurement-to-alpha}) one finds $\alpha\geq\frac{1}{\sqrt{2}}$. This is a contradiction, since $\alpha\in(0,\frac{1}{\sqrt{2}})$ by construction. Therefore, Eq.~(\ref{eq:SO:proof-technical-background-contradiction-assumption}) must be false, implying that there exists some $|i\rangle\in\mathcal{B}$ for which $\varpi(|i\rangle\,|\,\mu)\neq|\langle i |\psi\rangle|^2$ for some $\mu\in\Delta_{|\psi\rangle}$.
\end{proof}

Theorem~\ref{thm:SO:no-maximal-overlap-for-basis} is the central idea that allows the no-go proof for epistemic superpositions. Comparing it to Lem.~\ref{lem:SO:asymmetric-basic-bound}, Thm.~\ref{thm:SO:no-maximal-overlap-for-basis} states that the ontic overlap between $|i\rangle\in\mathcal{B}$ and $|\psi\rangle\not\in\mathcal{B}$ cannot be maximal. That is, like many of the results discussed in Sec.~\ref{sec:SO:previous-theorems}, this is a bound on ontic overlaps at heart.

Using Thm.~\ref{thm:SO:no-maximal-overlap-for-basis} the main result of this chapter can be proved fairly easily.

\begin{theorem}
\label{thm:SO:superpositions-are-ontic}
Consider a quantum system of dimension $d>3$ and define superpositions with respect to any orthonormal basis $\mathcal{B}$. Almost all quantum superposition states $|\psi\rangle\not\in\mathcal{B}$ are ontic.
\end{theorem}

\begin{proof}
Let $|\psi\rangle\not\in\mathcal{B}$ be any superposition state with respect to $\mathcal{B}$, such that $\exists|0\rangle\in\mathcal{B}$ for which $|\langle 0|\psi\rangle|^2 \in (0,\frac{1}{2})$. By Thm~\ref{thm:SO:no-maximal-overlap-for-basis} it follows that $\varpi(|i\rangle\,|\,\mu) \neq |\langle i|\psi\rangle|^2$ for some $|i\rangle\in\mathcal{B}$ and some $\mu\in\Delta_{|\psi\rangle}$ and further by Lem.~\ref{lem:SO:asymmetric-basic-bound} that $\varpi(|i\rangle\,|\,\mu) < |\langle i |\psi\rangle |^2$. Using this, with Lem.~\ref{lem:SO:asymmetric-basic-bound} again, also gives
\begin{equation}
\sum_{|j\rangle\in\mathcal{B}} \varpi( |j\rangle \,|\, \mu ) < \sum_{|j\rangle\neq|i\rangle} |\langle j|\psi\rangle |^2 + |\langle i |\psi\rangle |^2 = 1.
\end{equation}

Finally, note that Lem.~\ref{lem:SO:asymmetric-Boole} requires that $\varpi( \mathcal{B}\,|\,\mu ) \leq \sum_{|j\rangle\in\mathcal{B}} \varpi( |j\rangle\,|\,\mu)$ and therefore it is found that
\begin{equation}
\varpi(\mathcal{B}\,|\,\mu) < 1.
\end{equation}
This directly contradicts Eq.~(\ref{eq:SO:epistemic-superposition}) and therefore shows that $|\psi\rangle$ is an ontic superposition with respect to $\mathcal{B}$. However, $|\psi\rangle$ could be any superposition so long as $|\langle 0|\psi\rangle|^2 \in (0,\frac{1}{2})$ for some $|0\rangle\in\mathcal{B}$. This is true for all that are not exact $50:50$ superpositions over two states of $\mathcal{B}$. Indeed, it is true of almost all states in $\mathcal{P}(\mathcal{H})$. This completes the proof.
\end{proof}

This is a powerful result. It applies to any superposition state over any basis, so long as it is not an exact $50:50$ superposition and $d>3$. It therefore shows that quantum theory is grossly logically incompatible with epistemic superpositions. 

This is the primary result of this chapter. In the sections that follow, it will be seen how variations on this theorem form other strong restrictions on the character of ontology of quantum systems, as well as how they powerfully constrain the ability of classical resources to simulate quantum channels.

\section{Must any Specific States be Real?} \label{sec:SO:specific-states-ontology}

\subsection{State-Specific Ontology} \label{sec:SO:specific-states-ontology-definitions}

The primary categories of realist ontologies for quantum systems were briefly introduced in Sec.~\ref{sec:SO:classifying-ontologies}. These are: $\psi$-ontic, $\psi$-epistemic, and maximally $\psi$-epistemic (a small subset of $\psi$-epistemic). Now with the machinery of asymmetric overlaps, these can be defined formally.

$\psi$-ontic ontological models were introduced by saying that the ontic state uniquely identifies the quantum state that was prepared---there is no ontic overlap. Clearly, this means that when preparing any $|\psi\rangle$ there is zero probability of getting an ontic state obtainable by preparing any other $|\phi\rangle\neq|\psi\rangle$. In the language of asymmetric overlaps, $\psi$-ontic models can therefore be precisely defined as those satisfying
\begin{equation} \label{eq:SO:psi-ontic}
\varpi(\nu\,|\,\mu) = 0, \quad \forall\nu\in\Delta_{|\phi\rangle},\mu\in\Delta_{|\psi\rangle},|\psi\rangle\neq|\phi\rangle,
\end{equation}
or equivalently $\varpi(|\phi\rangle\,|\,\mu)=0$ for all $\mu\in\Delta_{|\psi\rangle}$ and $|\psi\rangle\neq|\phi\rangle$. Since $\psi$-epistemic models are precisely those which are not $\psi$-ontic, this also mathematically defines $\psi$-epistemic models.

Maximally $\psi$-epistemic models can be formalised similarly. Section~\ref{sec:SO:classifying-ontologies} introduced maximally $\psi$-epistemic models as those where the ontic overlap accounts for all of the Born rule overlap. As noted in Sec.~\ref{sec:SO:asymmetric-overlap}, the Born rule overlap $|\langle\phi|\psi\rangle|^2$ measures the probability that $|\psi\rangle$ will act like $|\phi\rangle$ and the asymmetric overlap does the same for ontological preparations. Clearly then, a model is maximally $\psi$-epistemic if and only if 
\begin{equation} \label{eq:SO:maximally-psi-epistemic}
\varpi(\nu\,|\,\mu) = |\langle\phi|\psi\rangle|^2, \quad \forall\mu\in\Delta_{|\psi\rangle},\,\nu\in\Delta_{|\phi\rangle}.
\end{equation}
By Lem.~\ref{lem:SO:asymmetric-basic-bound}, this means that all asymmetric overlaps must be maximal, hence the name.

In Secs.~\ref{sec:SO:classifying-ontologies}--\ref{sec:SO:limitations} it was discussed how these mark the extremes of possible ontological models for quantum systems. On the face of it, one might expect \emph{some} $\psi$-epistemic ontological model (such as in Ref.~\cite{BarrettCavalcanti+14}) to exist, since Eq.~(\ref{eq:SO:psi-ontic}) is such a strong condition. Similarly, it is relatively easy to rule out all maximally $\psi$-epistemic models for $d>2$ since Eq.~(\ref{eq:SO:maximally-psi-epistemic}) is so restrictive. The ontology theorems that bound overlaps---discussed in Secs.~\ref{sec:SO:previous-theorems}, \ref{sec:SO:limitations}---attempt to bridge the gap between these extremes by deriving quantitative bounds on overlap measures like the asymmetric overlap. But different papers often use different measures and they all suffer from shortcomings discussed in Sec.~\ref{sec:SO:limitations}.

Another way to reach a more nuanced discussion is to apply the ontic/epistemic dichotomy not just to the ontological models as a whole, but to the quantum states themselves too. That is, a \emph{quantum state $|\psi\rangle$ is $\psi$-ontic} if and only if it has no ontic overlap with any other quantum state $|\phi\rangle\neq|\psi\rangle$: $\varpi(|\phi\rangle\,|\,\mu) = 0$ for all $\mu\in\Delta_{|\psi\rangle}$. Correspondingly, \emph{$|\psi\rangle$ is $\psi$-epistemic} if and only if it is not $\psi$-ontic. Clearly, these are much more fine-grained requirements than the equivalents for the entire ontological model.

A specific state can similarly be defined as maximally $\psi$-epistemic in two ways. First, \emph{$|\psi\rangle$ is maximally $\psi$-epistemic with respect to $|\phi\rangle$} if and only if $\varpi(\nu\,|\,\mu)=|\langle\phi|\psi\rangle|^2$ for all $\nu\in\Delta_{|\phi\rangle}$ and $\mu\in\Delta_{|\psi\rangle}$. This is very fine-grained, identifying exactly which pair of quantum states overlap maximally. Second, \emph{$|\psi\rangle$ is maximally $\psi$-epistemic in the ontological model} if and only if it is maximally $\psi$-epistemic with respect to all states in $\mathcal{P}(\mathcal{H})$. This is clearly less fine-grained, identifying a point between a maximally $\psi$-epistemic model and a state that is maximally $\psi$-epistemic with respect to only one other.

Contextuality was also briefly introduced in Sec.~\ref{sec:SO:classifying-ontologies}. Reference~\cite{Spekkens05} shows how many types of contextuality can be precisely defined using the ontological models framework. For this thesis, only preparation contextuality needs to be considered in detail.

Preparation contextual models are those that describe operationally equivalent preparations differently. Therefore, for a quantum system, preparation non-contextuality requires that all preparations for the same quantum state $|\psi\rangle$ (all of which are operationally equivalent) are described with the same preparation measure; \emph{i.e.} every $\Delta_{|\psi\rangle}$ is a singleton. Recall that preparation non-contextuality is often simply (and implicitly) assumed for pure states, while preparation contextuality is necessary for mixed states [Sec.~\ref{sec:SO:previous-theorems}].

Similarly to the $\psi$-ontic models, the definition of preparation non-contextual models is very strong but can be used as a starting point for more subtle forms of contextuality. For example, one can consider preparations that are (non)-contextual with respect to certain sets of preparations. One such form of preparation contextuality is pure state \emph{preparation (non)-contextuality with respect to stabiliser unitaries} of some given state $|\psi\rangle$.

For any quantum state $|\psi\rangle\in\mathcal{P}(\mathcal{H})$, the set of \emph{stabiliser unitaries}, $\mathcal{S}_{|\psi\rangle}$, is formed of those $U$ on $\mathcal{H}$ which leave $|\psi\rangle$ unaffected: $U|\psi\rangle=|\psi\rangle$. That is, $\mathcal{S}_{|\psi\rangle}$ is the stabiliser subgroup of the unitary group with respect to $|\psi\rangle$.

Preparations of $|\psi\rangle$ are non-contextual with respect to stabiliser unitaries of $|\psi\rangle$ if and only if preparations that differ only by the action of a $U\in\mathcal{S}_{|\psi\rangle}$ are identical. In other words, $\mu\transto{\gamma}\mu$ for every $\mu\in\Delta_{|\psi\rangle}$ and $\gamma\in\Gamma_U$. Therefore, to be preparation non-contextual with respect to stabiliser unitaries simply requires that these unitaries cannot affect the distribution over $\Lambda$ for a preparation of $|\psi\rangle$. This may seem an oddly specific thing to assume. However Sec.~\ref{sec:SO:justifying-prep-noncontextuality} will argue that it is quite a natural assumption. Moreover, pure state preparation non-contextuality is often assumed wholesale without question and this is a much weaker assumption.

These definitions start to fill in the $\psi$-epistemic gulf between $\psi$-ontic models and maximally $\psi$-epistemic models. Using them as waypoints facilitates a cleaner discussion of exactly which types of $\psi$-epistemic models are and are not possible. In particular, the theorems that follow will use them to address the shortcomings discussed in Sec.~\ref{sec:SO:limitations}.

\subsection{No States Can Be \texorpdfstring{$\psi$}{psi}-Epistemic} \label{sec:SO:specific-states-ontology-results}

The first shortcoming of previous ontology theorems, as noted in Sec.~\ref{sec:SO:limitations}, is that they only prove the \emph{existence} of a single pair of quantum states that are not maximally $\psi$-epistemic, without being able to identify those states. Even if the theorems were extended to prove that more than one pair is not maximally $\psi$-epistemic, they would still not identify which states fail to be maximally $\psi$-epistemic. This is a loophole for the epistemic realist, who can still postulate ontological models where the vast majority of states of interest are maximally $\psi$-epistemic without violating these theorems.

The following theorem addresses this loophole directly. By making the additional assumption of preparation non-contextuality with respect to stabiliser unitaries it proves that any given $|\psi\rangle\in\mathcal{P}(\mathcal{H})$ is not maximally $\psi$-epistemic with respect to very many other states. In particular, no individual state can be maximally $\psi$-epistemic at all.

\begin{theorem} \label{thm:SO:no-maximally-epistemic}
Consider a $d>3$ dimensional quantum system and any pair of quantum states $|\psi\rangle,|0\rangle\in\mathcal{P}(\mathcal{H})$ that satisfy $|\langle 0|\psi\rangle|^2 \in (0,\frac{1}{2})$. Assume that preparations of $|\psi\rangle$ are non-contextual with respect to stabiliser unitaries of $|\psi\rangle$. For any preparation measure $\mu\in\Delta_{|\psi\rangle}$, the asymmetric overlap is bounded by
\begin{equation} \label{eq:SO:thm:no-maximally-epistemic}
\varpi(|0\rangle\,|\,\mu)\leq|\langle0|\psi\rangle|^{2}\left(\frac{1}{2}+|\langle0|\psi\rangle|^{2}\right)<|\langle0|\psi\rangle|^{2}.
\end{equation}
In particular, this implies that the asymmetric overlap $\varpi(|0\rangle\,|\,\mu)$ is strictly less than maximal. Therefore $|\psi\rangle$ is not maximally $\psi$-epistemic with respect to $|0 \rangle$ for any such pair of states in $d>3$ dimensions satisfying $|\langle 0|\psi\rangle|^2\in(0,\frac{1}{2})$.
\end{theorem}

The proof closely follows that of Thm.~\ref{thm:SO:no-maximal-overlap-for-basis} though it must be proved separately due to the difference in assumptions. The full proof can be found in Appendix~\ref{app:proof:thm:no-maximally-epistemic}. An immediate corollary is that no quantum state $|\psi\rangle\in\mathcal{P}(\mathcal{H})$ can be $\psi$-epistemic for $d>3$.

The epistemic realist must therefore either accept preparation contextuality with respect to stabiliser unitaries, or accept that reality is a long way from being maximally $\psi$-epistemic. The extra non-contextuality assumption is very mild. As noted above, it is much weaker than general pure-state preparation contextuality, which is very often assumed. Moreover, Sec.~\ref{sec:SO:justifying-prep-noncontextuality} will argue that it is a natural assumption for any minimally realistic ontological model.

Theorem~\ref{thm:SO:no-maximally-epistemic} improves upon the previous results in as far as ruling out maximally $\psi$-epistemic ontologies and closing loopholes for ontologies that are close to maximally $\psi$-epistemic. However, as a bound on the ontic overlap, Eq.~(\ref{eq:SO:thm:no-maximally-epistemic}) is rather weak. The next theorem aims instead to establish a restrictive quantitative bound on ontic overlaps by a similar method.

\begin{theorem} \label{thm:SO:no-overlap-large-d}
Consider a $d>3$ dimensional quantum system and any pair $|\psi\rangle,|0\rangle\in\mathcal{P}(\mathcal{H})$ that satisfy $\alpha \eqdef |\langle 0|\psi\rangle| \in (0,\frac{1}{4})$. Assume that preparations of $|\psi\rangle$ are non-contextual with respect to stabiliser unitaries of $|\psi\rangle$. For any preparation measure $\mu\in\Delta_{|\psi\rangle}$, the asymmetric overlap must satisfy
\begin{eqnarray} \label{eq:SO:thm:no-overlap-large-d}
\varpi(|0\rangle\,|\,\mu) & \leq & \alpha^{2}\left(\frac{1+2\alpha}{d-2}\right)\\
\lim_{d\rightarrow\infty}\varpi(|0\rangle\,|\,\mu) & = & 0
\end{eqnarray}
and so becomes arbitrarily small as $d$ increases even as $\alpha$ is held constant.
\end{theorem}

The proof strategy here is similar to that of Thm.~\ref{thm:SO:no-maximally-epistemic} but modified to make use of the higher available dimensions. It can also be found in Appendix~\ref{app:proof:thm:no-maximally-epistemic}. Note that this modification necessarily weakens the bound compared to Thm.~\ref{thm:SO:no-maximally-epistemic} in low dimensions.

While the $1/d$ scaling of Eq.~(\ref{eq:SO:thm:no-overlap-large-d}) is relatively weak compared to previous overlap bounds (some of which scale exponentially \cite{Leifer14a,Branciard14}), Thm.~\ref{thm:SO:no-overlap-large-d} overcomes both limitations noted in Sec.~\ref{sec:SO:limitations}. First, Thm.~\ref{thm:SO:no-overlap-large-d} bounds the overlap for many specific pairs of quantum states, just like Thm.~\ref{thm:SO:no-maximally-epistemic}. Second, as the dimension increases, the bound tightens even while the inner products of the states remains constant. In particular then, this implies that in large-dimensional systems many specific pairs of states can only barely overlap at all.

\subsection{Towards Error-Tolerance and Experiments} \label{sec:SO:error-tolerance}

Thus far, it has been assumed that quantum statistics must be exactly reproduced by a valid ontological model [Eq.~(\ref{eq:IN:ontological-model-quantum-probability})]. However, it is impossible to exactly verify this. At most, experiments can demonstrate that quantum probabilities hold to within some finite additive error $\epsilon\in(0,1]$, as in Eq.~(\ref{eq:IN:ontological-model-approx-quantum-probability}). It is therefore necessary to consider \emph{error-tolerant } versions of the above theorems.

Unfortunately, the asymmetric overlap is an error-intolerant quantity. That is, for any quantum system there is an ontological model that satisfies Eq.~(\ref{eq:IN:ontological-model-approx-quantum-probability}) for any given $\epsilon\in(0,1]$ but for which every asymmetric overlap is unity. 

One can construct such an approximate model by modifying the Beltrametti-Bugajski model \cite{BeltramettiBugajski95} (that is, quantum theory re-phrased as an ontological model, Sec.~\ref{sec:SO:desired-ontologies}). Simply adjust the preparation measures so that each has probability $\epsilon$ of preparing a completely random state (according to, for example, the Haar measure \cite{ZyczkowskiSommers01}) and probability $1-\epsilon$ of acting as usual. Thus, the model will differ from quantum predictions with probability at most $\epsilon$. However, now every preparation measure can prepare any state with finite probability, so the asymmetric overlaps are all unity.

There is an alternative overlap measure, the symmetric overlap $\omega(|\psi\rangle,|\phi\rangle)$ \cite{Maroney12a,BarrettCavalcanti+14,Leifer14a,Branciard14,Leifer14b}, that does not have this problem---it is robust to small errors. The symmetric overlap is based on distinguishability of measures.

Suppose you are given some $\lambda\in\Lambda$ obtained by sampling from either $\mu$ or $\nu$ (each with equal \emph{a priori} probability). Consider using the optimal strategy to guess which of $\mu,\nu$ was used. The \emph{symmetric overlap of the measures $\mu$ and $\nu$} is defined as twice the probability of guessing incorrectly (and therefore takes values in $[0,1]$). This is known to correspond to \cite{Maroney12a,BarrettCavalcanti+14,Leifer14b}
\begin{equation}
\omega(\mu,\nu)\eqdef\inf\left\{ \mu(\Omega) + \nu(\Lambda\setminus \Omega)\,:\,\Omega\in\Sigma \right\}.
\end{equation}
As the names suggest, $\omega(\mu,\nu)$ is necessarily symmetric in its arguments and $\varpi(\nu\,|\,\mu)$ is not.

Just as with the asymmetric overlap, it is useful to slightly overload the notation and define the \emph{symmetric overlap of the states $|\psi\rangle$ and $|\phi\rangle$} as
\begin{equation} \label{eq:SO:quantum-state-symmetric}
\omega(|\psi\rangle,|\phi\rangle)\eqdef\sup_{\mu\in\Delta_{|\psi\rangle},\nu\in\Delta_{|\phi\rangle}}\omega(\mu,\nu).
\end{equation}

Quantum theory provides an upper bound on the symmetric overlap, since any quantum procedure for distinguishing $|\psi\rangle,|\phi\rangle$ is also a method for distinguishing $\mu\in\Delta_{|\psi\rangle},\nu\in\Delta_{|\phi\rangle}$ in an ontological model. As $\frac{1}{2}\left(1-\sqrt{1-|\langle\phi|\psi\rangle|^{2}}\right)$ is the minimum average error probability when distinguishing $|\psi\rangle,|\phi\rangle$ within quantum theory\footnote{By using the Helstrom measurement \cite{WaldherrDada+12,BarrettCavalcanti+14}.} it follows that $\omega(\mu,\nu)\leq1-\sqrt{1-|\langle\phi|\psi\rangle|^{2}}$ for every $\mu\in\Delta_{|\psi\rangle},\nu\in\Delta_{|\phi\rangle}$ and so 
\begin{equation} \label{eq:SO:symmetric-state-basic-bound}
\omega(|\psi\rangle,|\phi\rangle)\leq1-\sqrt{1-|\langle\phi|\psi\rangle|^{2}}.
\end{equation}
This is analogous to Lem.~\ref{lem:SO:asymmetric-basic-bound} for the asymmetric overlap.

With this new machinery, Thm.~\ref{thm:SO:no-overlap-large-d} can be modified to bound the symmetric overlap while only assuming that quantum probabilities are reproduced to within some finite additive error.

\begin{theorem} \label{thm:SO:small-symmetric-overlap-large-d}
Consider the assumptions of Thm.~\ref{thm:SO:no-overlap-large-d}, but only assume that the ontological model reproduces quantum probabilities to within $\pm\epsilon$ for some $\epsilon\in(0,1]$, as in Eq.~(\ref{eq:IN:ontological-model-approx-quantum-probability}). The symmetric overlap must satisfy
\begin{equation} \label{eq:SO:thm:small-symmetric-overlap-large-d}
\omega(|0\rangle,|\psi\rangle) \leq \alpha^2 \left( \frac{1 + 2\alpha}{d - 2} \right) + \frac{(3d^2 - 7d)}{2(d-2)}\epsilon.
\end{equation}
This bound is tighter than Eq.~(\ref{eq:SO:symmetric-state-basic-bound}) for $d>5$ and sufficiently small $\epsilon$.
\end{theorem}

The proof strategy of Thm.~\ref{thm:SO:no-overlap-large-d} is closely related to the properties of the asymmetric overlap. Since Thm.~\ref{thm:SO:small-symmetric-overlap-large-d} adapts this method to the symmetric overlap, the fit between method and overlap measure is much less comfortable. This also makes the proof itself---provided in Appendix~\ref{app:proof:thm:small-symmetric-overlap-large-d}---unfortunately long, ugly\footnote{It should be noted that the proof in Appendix~\ref{app:proof:thm:small-symmetric-overlap-large-d} uses the simpler non-measure-theoretic version of ontological models in order to keep it as short and legible as possible, the measure-theoretic version would be even longer and uglier.}, and resulting in a looser bound [Eq.~(\ref{eq:SO:thm:small-symmetric-overlap-large-d})] than its asymmetric counterpart Eq.~(\ref{eq:SO:thm:no-overlap-large-d}).

The significance of Thm.~\ref{thm:SO:small-symmetric-overlap-large-d} is twofold. First, it allows many of the same conclusions as Thm.~\ref{thm:SO:no-overlap-large-d} but in a context with finite error: it demonstrates that many specific pairs of quantum states cannot be maximally $\psi$-epistemic and must have small ontic overlap in the large-$d$ limit (without those states approaching orthogonality) for sufficiently small error. Being error-tolerant, it opens this conclusion up to experimental investigation. Second, it is as a first-step and proof-of-concept for error-tolerance for Thms.~\ref{thm:SO:superpositions-are-ontic}, \ref{thm:SO:no-maximally-epistemic}. It does not immediately imply that almost all superpositions are real, but by demonstrating how Thm.~\ref{thm:SO:no-overlap-large-d}'s arguments can be made robust against small error it suggests that error-tolerant versions of the other theorems of this chapter should also be possible. It also proves that no pure state can be individually maximally $\psi$-epistemic in an error-tolerant way.

Even so, an error-tolerant version of Thm.~\ref{thm:SO:superpositions-are-ontic} would require the definition of ``ontic superposition'' to be modified, since it is currently defined in terms of the asymmetric overlap. This is tackled in Sec.~\ref{sec:MR:error-tolerant-superpositions-and-mr}, where such a re-definition is provided and error-tolerant variations on Thms.~\ref{thm:SO:superpositions-are-ontic}, \ref{thm:MR:no-ESMR-EMMR} are presented.

Theorem~\ref{thm:SO:small-symmetric-overlap-large-d} is probably most valuable as a proof-of-concept for experimental applicability rather than forming the basis of a concrete experimental proposal itself. The mismatch between the proof strategy and the symmetric overlap probably makes the result very non-optimal. For example, take the case where $\alpha = 0.245$ and $d=6$. This makes the bound of Eq.~(\ref{eq:SO:thm:small-symmetric-overlap-large-d}) $\approx 0.0224 + 8.25\epsilon$, while the bound of Eq.~(\ref{eq:SO:symmetric-state-basic-bound}) $\approx 0.0305$. In this case, Eq.~(\ref{eq:SO:thm:small-symmetric-overlap-large-d}) is an improved bound for $\epsilon \lesssim 0.0009$. Such experimental accuracy does not seem completely infeasible with current technology \cite{RingbauerDuffus+15}, but this is nevertheless a challenging experiment to simply improve on the easy bound of Eq.~(\ref{eq:SO:symmetric-state-basic-bound}). Of course, as $d$ increases precise experiments become more challenging, meaning that the error term of Eq.~(\ref{eq:SO:thm:small-symmetric-overlap-large-d}) is likely to increase super-linearly with $d$.

\subsection{Justifying Preparation Contextuality for Stabiliser Unitaries} \label{sec:SO:justifying-prep-noncontextuality}

Theorems~\ref{thm:SO:no-maximally-epistemic}--\ref{thm:SO:small-symmetric-overlap-large-d} all assume that preparations are non-contextual with respect to stabiliser unitaries for given states, as defined in Sec.~\ref{sec:SO:specific-states-ontology-definitions}. This is an extra assumption beyond the bare ontological models framework. Can such an assumption be justified? What follows is a heuristic argument aiming to do exactly that.

The ontological models framework combines postulated fundamental objective ontology with operational notions. The fundamental ontology is reflected in the idea that ontic states represent actual states of affairs, independently of any other theories an observer might use to describe the same system. On the other hand, the only way to reason about this largely-unspecified ontological level is operationally: how does it respond to preparations, transformations, and measurements that can actually be performed? 

An assumption of non-contextuality is an assumption about these operational bridges between our capabilities and the ontology. With this perspective, extra assumptions of non-contextuality can be justified by arguing that they are part of any sensible operational understanding of ontological models.

Any specific operational method for preparing some state $|\psi\rangle\in\mathcal{P}(\mathcal{H})$ may be thought of as a black box which the system is fed into. When the system is fed out of the box, it is promised that the box has prepared the system in state $|\psi\rangle$ according to some specific method. In terms of ontological models, the preparation method corresponds to a measure $\mu\in\Delta_{|\psi\rangle}$ and any such method $\mu$ can be considered in terms of such a box.

Suppose you design some experiment which involves preparing $|\psi\rangle$ via a method corresponding to $\mu$. Scientists implementing that experiment would acquire the corresponding black box to be sure that the method is indeed used. Once prepared, the system will need to be presented to other pieces of apparatus. However, there will always be variation in how the system is treated between preparation and the action of any other apparatus, any amount of motion or passage of time or other (seemingly innocuous) treatment amounts to applying some unitary $U$ to the system. Each scientist will, no doubt, be careful to ensure that the system is not disturbed from its preparation state, so it can be safely assumed that any such $U$ is a stabiliser unitary $U\in\mathcal{S}_{|\psi\rangle}$. However, the point remains that some unknown $U\in\mathcal{S}_{|\psi\rangle}$ is inevitably applied to the system after preparation via $\mu$, and this can never be perfectly accounted for.

Therefore, to analyse the result of the experiment, you have to allow for some unknown $U\in\mathcal{S}_{|\psi\rangle}$ to by applied (via some unknown $\gamma\in\Gamma_{U}$) after preparation of $|\psi\rangle$ via $\mu$. As a result, on this minimally realistic operational level, an arbitrary preparation distribution $\mu$ can never be prepared unscathed; you have to account for the inevitable, unknown, subsequent stabiliser unitary. It is therefore prudent to have the effective preparation measure that you use to describe the experiment be one that is non-contextual with respect to such transformations, allowing the experiment to still be analysed despite the application of an unknown $U\in\mathcal{S}_{|\psi\rangle}$.

One must be careful to only consider operational features that are not, even in principle, impossible to reliably perform. Since the sets of preparation measures for any given quantum states are, in the end, operational in character, one may safely restrict to preparation measures that satisfy certain sensible realistic requirements. The above heuristic argument aims to establish pure state preparation non-contextuality with respect to stabiliser unitaries as such a realistic requirement. It is, however, only a heuristic argument and is not rigorous. In particular, no strong reason has been given for including \emph{all} stabiliser unitaries.

\section{Communication Bounds from Ontology Results} \label{sec:SO:communication}

Foundational ontology results often suggest broader implications in many areas. The most obvious example is, of course, Bell's theorem, which has found applications across all areas of quantum theory as well as inspiring many related results is disparate places \cite{BertlmannZeilinger02,BertlmannZeilinger17}. Not all foundations results can claim quite such outstanding success, of course. Most commonly they can find applications in quantum information, particularly with relation to communications tasks \cite{Montina12,BarrettHardy+05,PerryJain+15,LiuPerry+16,Montina13,Montina15,MontinaWolf16,Leifer14b}. In this section that pattern will be repeated, with the techniques used above to find ontology results applied to a communication task in quantum information.

\subsection{Relating Ontology and Communication} \label{sec:SO:communication-from-ontology}

The framework of ontological models has a very natural links to communication tasks in quantum information \cite{PerryJain+15,Montina15}. Indeed, in certain cases direct parallels are known between ontological models and communication protocols \cite{Montina12}. When seeking to apply the techniques of new ontology results to quantum information, communication tasks are therefore a natural place to start.

One of the simplest communication tasks with links to ontological models is the \emph{finite communication} (FC) protocol where Alice and Bob simulate a noiseless $n_q$-qubit quantum channel with a finite $n_c$-bit noiseless classical channel. In this task, Alice is given a description of some quantum state $|\psi\rangle$ from a $d$-dimensional system and Bob is given a description of a quantum measurement $M$ on the same system. The task is for Bob to output some outcomes $E\in M$ with probabilities compatible with quantum theory over many runs, with possibly different $|\psi\rangle$ and $M$ each time.

Clearly, this task can be achieved with one-way communication from Alice to Bob using a noiseless \emph{quantum} channel of enough qubits $n_q$---Alice sends $|\psi\rangle$ to Bob, who simply measures it. To simulate this with \emph{classical} resources, Alice and Bob are given access to a shared random value $r\in\mathcal{R}$ (taking a new value each run according to some probability measure $\varrho(r)$) and noiseless channel with which Alice can send classical messages $c\in\mathcal{C}$ to Bob from a finite set $|\mathcal{C}| < \infty$. The question is: what quantities of classical resources $|\mathcal{C}|$ and $|\mathcal{R}|$ are required for exact simulation of the $n_q$-qubit quantum channel? Equivalently, how many classical bits $n_c \geq \log|\mathcal{C}|$, $n_r \geq \log |\mathcal{R}|$ are required to simulate a given number of qubits $n_q \geq \log d$? This protocol is schematically illustrated in Fig.~\ref{fig:SO:FC-protocol}.

\begin{figure}
\begin{centering}
\includegraphics[scale=0.35,angle=0]{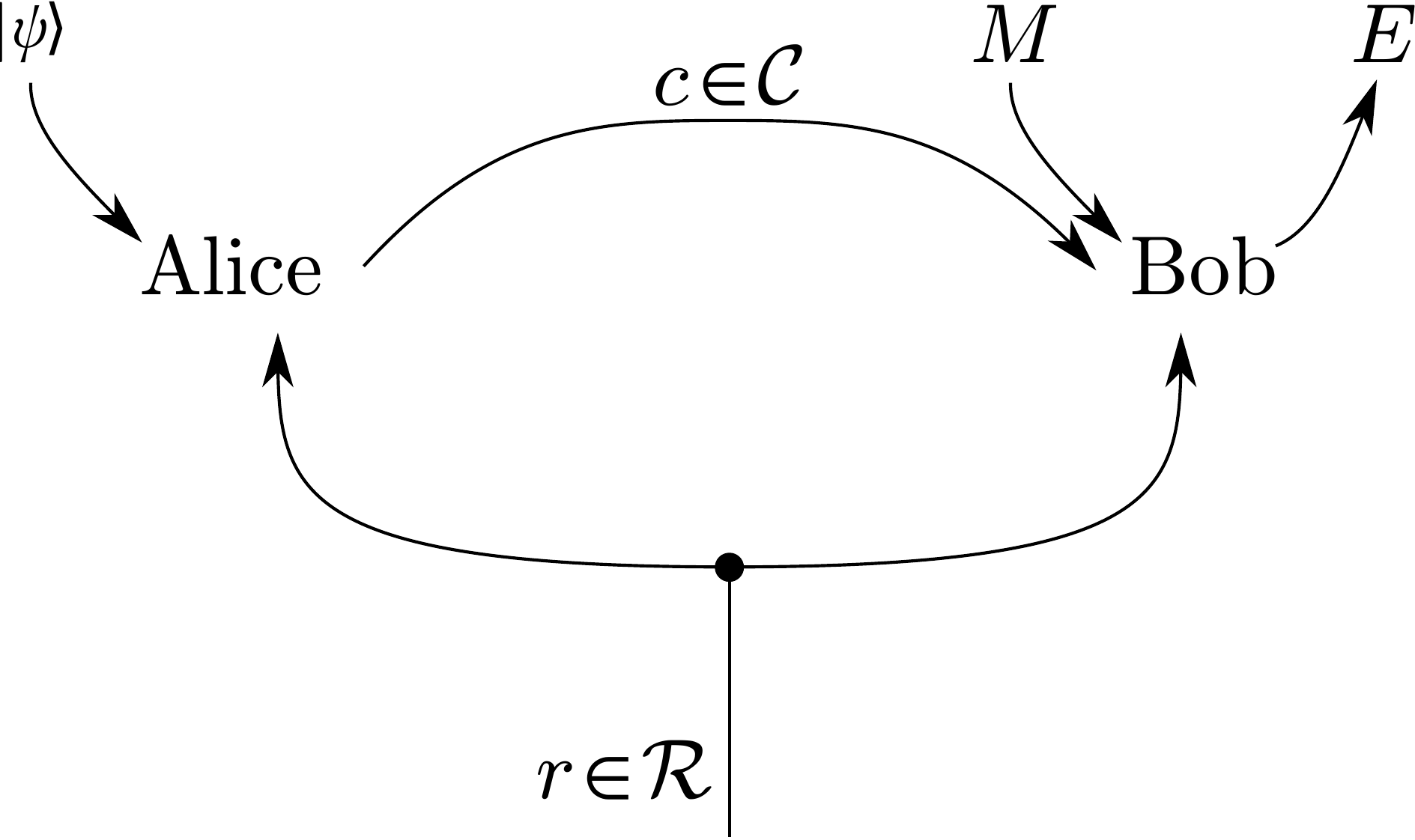}
\par\end{centering}
\protect\caption{Schematic representation of the finite communication (FC) protocol for simulating a quantum channel with classical resources. Alice recieves a description of quantum state $|\psi\rangle$ and Bob recieves a description of quantum measurement $M$. They share a random value $r$ from a set of possibilities $\mathcal{R}$. Alice sends a classical message $c\in\mathcal{C}$ to Bob through a noiseless channel, who then outputs measurement outcome $E\in M$ based on $c$ and $r$. }\label{fig:SO:FC-protocol}
\end{figure}

It is easy to see that such a protocol implies a basic ontological model for the same quantum system. The ontic state space is $\Lambda = \mathcal{C}\times\mathcal{R}$ with ontic states $\lambda = (c,r)$. If Alice sends message $c$ with probability given by $\mathbb{P}(c\,|\,|\psi\rangle,r)$ then the corresponding preparation measure for $|\psi\rangle$ is $\mu_{|\psi\rangle}(\lambda) = \mathbb{P}(c\,|\,|\psi\rangle,r)\varrho(r)$. If Bob outputs measurement outcome $E\in M$ with probability $\mathbb{P}(E\,|\,c,r)$ then this is exactly the corresponding response function $\mathbb{P}_M (E\,|\,\lambda)$. This ontological model is missing a description of transformations, which could be provided by extending the FC protocol to include another agent who performs a transformation, but such an extension will not be required here.

A couple of properties of FC protocols can be stated immediately. First, the excess baggage theorem requires ontological models (even those missing transformations) to have infinite ontic state spaces in order to exactly reproduce quantum predictions \cite{Hardy04}. Therefore, there must be an infinite amount of shared random data $|\mathcal{R}| \geq \infty$ for any such protocol. As a result, the focus is normally on how many classical bits $n_c \geq \log |\mathcal{C}|$ are required to exactly simulate an $n_q$-qubit quantum system given arbitrary shared random data $\mathcal{R}$. This approach to FC protocols will be used here.

Given this simplification, one can always assume that all randomness in the protocol comes from the shared random data. That is, any FC protocol for $n_q$ qubits with $n_c$ bits where Alice and Bob can act stochastically implies the existence of an equivalent FC protocol for the same $n_q$ and $n_c$ where Alice and Bob act deterministically given $|\psi\rangle$, $M$, and $r\in\mathcal{R}$. This is because Alice and Bob can simply obtain any required randomness from $r$, as $\mathcal{R}$ is arbitrarily large.

Alice's task can therefore be quite simple in an FC protocol. She simply assigns quantum states to classical messages $c\in\mathcal{C}$ based on $r$ in a deterministic way. Her strategy becomes that of an infinite look-up table. Bob's strategy can be similarly described.

Such simulation tasks are not just interesting from a foundational view, but are also important in quantum information \cite{BuhrmanCleve+10}. Considering distributed computing, for example, it is useful to know what sort of advantages quantum channels can offer over classical channels. More generally, an optimal classical FC protocol would provide a natural measure of the power of quantum channels \cite{Montina11b}.

The best known example of an FC protocol is from Ref.~\cite{TonerBacon03}. There, an explicit FC protocol is given that simulates a $n_q = 1$ qubit channel with subsequent projective measurement using exactly $n_c = 2$ bits of classical communication (of course, with infinite shared random data). However, no similarly general protocols are known for any quantum dimension $d > 2$. Some partial protocols are known to exist---such as in Ref.~\cite{Montina13}, which demonstrates how to construct an FC protocol for arbitrary $d$ if Bob can only perform two-outcome projective measurements---but it is not known whether such a protocol is possible for arbitrary projective (or, more generally POVM measurements) even for $d=3$.

What is known, however, are certain lower-bounds for the required $n_c$ as $n_q$ increases. These will be briefly reviewed in the next section, before a new bound is given with a simple proof based on the results from earlier in the chapter.

\subsection{A Simple Exponential Bound}	\label{sec:SO:communication-bound}

Several lower-bounds on the classical bits $n_c$ required to simulate an $n_q$-qubit quantum channel are known. These require $n_c$ to scale exponentially with $n_q$ and may therefore be seen as ``anti-Holevo'' results: while the Holevo bound states that a $n_q$-qubit quantum state can store at most $n_q$ classical bits, these results effectively show that to store $n_q$ qubits one requires at least $\mathcal{O}(2^{n_q})$ classical bits\footnote{This observation is taken from Ref.~\cite{Montanaro16}, where the ``anti-Holevo'' moniker is attributed to Tony Short.}.

The first example is from Ref.~\cite{BrassardCleve+99}, following the work of Ref.~\cite{BuhrmanCleve+98}, where it was proved that $n_c \geq c2^{n_q}$ classical bits are required to exactly simulate $n_q$ qubits for some constant $c \approx 0.01$. This was improved in Ref.~\cite{Montina11b}, where the same asymptotic lower bound was obtained with $c = 0.293$ and a postulated improvement (based on a plausible but unproved conjecture) was presented for a $2^{n_q} - 1$ lower bound.

Taking a somewhat different approach, Ref.~\cite{Montanaro16} proved an $\mathcal{O}(2^{n_q})$ asymptotic lower bound. This has the key advantage of bounding approximate, as well as exact, simulations and also applying to two-way classical communication between Alice and Bob. No precise constant factors for the bound are given, however.

The following theorem shows how the methods used to prove the ontology results in this chapter can also provide a comparable lower bound for the exact FC protocol. A key advantage of this theorem is its simplicity. Even a cursory glance at Refs.~\cite{BrassardCleve+99,Montina11b,Montanaro16} will show that their proofs are often very mathematically involved. It is comparatively simple to prove the following exponential bound.

\begin{theorem} \label{thm:SO:communication-bound}
For any $n\in\mathbb{Z}^+$, there is a quantum system of $n_q = \mathcal{O}(\log n)$ qubits such that any FC protocol exactly simulating an $n_q$-qubit channel requires at least $n_c$ classical bits of communication bounded by
\begin{equation}
n_c \geq 2^{ n_q + \mathcal{O}(1) } - 1.
\end{equation}
\end{theorem}

\begin{proof}
This results follows quickly by borrowing a result from the study of quantum fingerprinting \cite{BuhrmanCleve+01}. In quantum fingerprinting, the aim is to relate $n$-bit classical bit strings to quantum states that have bounded Born rule overlap in such a way that $n$ scales exponentially in the number of qubits required.

The exact result required here is from Ref.~\cite[Thm.~2]{BuhrmanCleve+01}. This establishes that, for any $n\in\mathbb{Z}^+$ there is a set of $N \eqdef 2^n$ quantum states $\mathcal{F} = \{|F_x\rangle\}_x$ of $n_q = \log n + \mathcal{O}(1)$ qubits such that $|\langle F_x |F_y\rangle|^2 \leq \frac{1}{4}$ whenever $x\neq y$. It is simple to see that these Born rule overlaps ensure that every triple $\{|F_x\rangle, |F_y\rangle, |F_z\rangle\}$ of unequal states from this set is anti-distinguishable by Eq.~(\ref{eq:SO:anti-distinguishable-criterion}).

Consider an FC protocol for this system of $\log n + \mathcal{O}(1)$ qubits. For any value $r\in\mathcal{R}$, Alice can assign a maximum of two states from $\mathcal{F}$ to each message $c\in\mathcal{C}$. If she were to assign three states from $\mathcal{F}$ to the same message $c\in\mathcal{C}$ then Bob would not be able to correctly simulate the corresponding anti-distinguishable measurement. Therefore, the number of messages required for the protocol must be $|\mathcal{C}| \geq N/2$.

Since $N = 2^n$ is a power of two, the number of required classical bits is $n_c \geq \log |\mathcal{C}| = n - 1$. It therefore immediately follows that $n_c \geq 2^{n_q + \mathcal{O}(1)} - 1$.
\end{proof}

It is clear that this bound comes from an easily-understood property of quantum states. That is, there exist sets of quantum states with bounded Born-rule overlaps that are exponential in the Hilbert space dimension (that is, $\mathcal{F}$ from the theorem). It is this fact that directly prevents storing quantum states in a small number of classical bits even with shared random data.

This proof also has methodological advantages over the previous bounds found in Refs.~\cite{BrassardCleve+99,Montina11b,Montanaro16}, which will now be discussed.

\subsection{Generating Further Bounds} \label{sec:SO:further-communication-bounds}

The proof of Thm.~\ref{thm:SO:communication-bound} is very easy to understand compared to its peers, while giving a comparable bound. But more than this, careful examination of the proof reveals a general method for deriving such bounds from classical error correction codes. This suggests that the proof of Thm.~\ref{thm:SO:communication-bound} may be more important than the result itself, as simply by finding appropriate classical codes this proof will generate more precise bounds.

As noted above, the key component in the proof is a set $\mathcal{F}$ of quantum states such that $N = |\mathcal{F}|$ is exponential in the Hilbert space dimension $d$ and yet all pairs of states from $\mathcal{F}$ have bounded inner product. Existence of an appropriate such set was taken from quantum fingerprinting, specifically Ref.~\cite[Thm. 2]{BuhrmanCleve+01}.

This points to the first way this proof can generate more bounds. If one finds such a quantum fingerprinting set $\mathcal{F}$ where Born-rule overlaps are bounded $\leq\frac{1}{4}$ and with good scaling of $N$ with $d$ then this yields a bound exactly as in Thm.~\ref{thm:SO:communication-bound}. I am, however, unaware of any better sets for this purpose than those used above.

The second way that the proof can generate more bounds is perhaps more interesting. Existence of the $\mathcal{F}$ used in Thm.~\ref{thm:SO:communication-bound} is known due the existence of certain classical error-correction codes. Specifically, any code $E:\{0,1\}^n \rightarrow \{0,1\}^m$ with $m = \mathcal{O}(n)$ such that the Hamming distance between unequal code words is greater than $m/4$ yields a quantum fingerprinting set $\mathcal{F}$ with the properties required by Thm.~\ref{thm:SO:communication-bound}. The states of this set are in a $(d = m)$-dimensional Hilbert space and of the form 
\begin{equation} \label{eq:SO:fingerprinting-states}
\frac{1}{\sqrt{m}} \sum_{\alpha = 1}^m (-1)^{E_\alpha (x)}|\alpha\rangle,\quad\forall x\in\{0,1\}^n
\end{equation}
where $\{|\alpha\rangle\}_\alpha$ is an arbitrary orthonormal basis of the system and $E_\alpha (x)$ is the $\alpha$th bit of the code word $E(x)$ \cite{BuhrmanCleve+01}. Clearly, there are $2^n = \mathcal{O}(2^m) = \mathcal{O}(2^d)$ states in this set and they can be verified to have Born rule overlap $\leq \frac{1}{4}$.

Any classical error correcting code of this form (or a sufficiently similar form) will generate an exponential communication bound via the proof of Thm.~\ref{thm:SO:communication-bound}. Reference~\cite{BuhrmanCleve+01} notes an existence proof for such codes, but specific examples would give specific numerical FC protocol bounds.

This means that Thm.~\ref{thm:SO:communication-bound} has the potential to address one of the shortcomings of previous exponential bounds for FC protocols. That is, those bounds are not particularly useful in low-dimensional settings. For example, the $n_c \geq 0.293 \times 2^{n_q}$ bound only starts to exceed the trivial bound of $n_c \geq n_q$ at $d = 2^{n_q} = 16$ dimensions. Such low-dimensional cases are important for potential experiments, where it is currently unfeasible to test very large-dimensional systems.

So although Thm.~\ref{thm:SO:communication-bound} does not currently provide an exact bound in any dimension, by finding appropriate classical error-correction codes exact bounds will immediately follow. This is in marked contrast to the proofs of existing bounds. In Ref.~\cite{Montanaro16} only an asymptotic bound is given without a precise form, while in Refs.~\cite{BrassardCleve+99,Montina11b} the proofs do not appear to admit easy modification to provide tighter bounds.

\section{Summary} \label{sec:SO:summary}

This chapter tackled the ontology of quantum states directly. By considering what types of realist ontologies are possible, one is naturally led to using the ontological models framework, introduced in Sec.~\ref{sec:IN:ontological-models}. Broadly, there are two types of realist ontology for quantum states: those where the ontology is similar to the description in quantum theory and those where ontological uncertainty enables potentially more elegant ontologies. The second type is preferred by the ``epistemic realist'' [Sec.~\ref{sec:SO:desired-ontologies}] and raises the question as to what exactly can be gained through ontological uncertainty in terms of elegance and explanatory power.

The family of results that attempt to answer these questions are known as ``ontology theorems'' [Sec.~\ref{sec:SO:previous-theorems}]. Due to the restrictions from Bell's theorem and the PBR theorem, modern ontology theorems typically aim to apply to any single-system ontological model, without specifying how individuals combine into multipartite systems. Within this scope, $\psi$-epistemic models that reproduce quantum theory have been exhibited but maximally $\psi$-epistemic models have been proved to be impossible for $d>2$ dimensions. This leaves the question as to exactly how close to maximally $\psi$-epistemic ontological models can get before they need to violate quantum predictions. The ``overlap theorems'' [Sec.~\ref{sec:SO:previous-theorems}] that aim to restrict epistemic ontological models severely all share some shortcomings which appear to leave the epistemic realist plenty of room to consider ontologies that are close to maximally $\psi$-epistemic, though not exactly [Sec.~\ref{sec:SO:limitations}].

One aspect of quantum state ontology that had not been considered in detail is the ontology of superposition states. Since quantum superpositions inherit all of their properties from their underlying basis states, it is natural to ask whether their ontology can similarly only depend on ontic states accessible to the basis. This question was considered in detail in Sec.~\ref{sec:SO:superpositions-are-real}, where Thm.~\ref{thm:SO:superpositions-are-ontic} proves that, for $d>3$, almost every superposition is necessarily ontic. That is, any epistemic realist account of quantum theory must include ontic features corresponding to superposition states and the unfortunate cat cannot be put out of its misery.

By adapting the methods of Thm.~\ref{thm:SO:superpositions-are-ontic}, it was possible to obtain more general overlap ontology theorems in Sec.~\ref{sec:SO:specific-states-ontology}. By making a very mild extra non-contextuality assumption, Thms.~\ref{thm:SO:no-maximally-epistemic}, \ref{thm:SO:no-overlap-large-d} were proved---between them addressing the shortcomings of other overlap theorems noted in Sec.~\ref{sec:SO:limitations}. In making explicit use of an assumption beyond the bare ontological models framework, they are technically weaker results than some previous ontology theorems. However, the assumption is weak and also arguably natural [Sec.~\ref{sec:SO:justifying-prep-noncontextuality}]. In particular, it is much weaker than an assumption that is very often implicitly made. Between them, Thms.~\ref{thm:SO:no-maximally-epistemic}, \ref{thm:SO:no-overlap-large-d} prove that quantum theory is incompatible with very many pairs of states being maximally $\psi$-epistemic (in particular, any given quantum state cannot be individually maximally $\psi$-epistemic) and that in large-dimensional systems many pairs of states cannot have any substantial ontic overlap at all.

Theorems~\ref{thm:SO:superpositions-are-ontic}--\ref{thm:SO:no-overlap-large-d} were all based on the asymmetric overlap, introduced in detail in Sec.~\ref{sec:SO:asymmetric-overlap}. This way of quantifying ontic overlaps is easily understood and has a clear motivation, but is unfortunately intolerant to error. The problem of how to adapt these results to be robust to experimental error was therefore considered in Sec.~\ref{sec:SO:error-tolerance} and Thm.~\ref{thm:SO:small-symmetric-overlap-large-d}, which adapts the methods of Thm.~\ref{thm:SO:no-overlap-large-d} to apply instead to the error-tolerant symmetric overlap. The resulting bound is somewhat weaker, due to the mismatch between methodology and overlap measure, but gains error tolerance. It still, however, rules out any given quantum state from being maximally $\psi$-epistemic in $d>5$ dimensions and provides a bound on ontic overlaps that approaches zero in large dimensions without the quantum states needing to approach orthogonality (for small enough error).

Often, new results in the foundations of quantum theory can be used to obtain parallel results in quantum information. In Sec.~\ref{sec:SO:communication}, the question of how the methods used in this chapter might affect communication abilities was considered. The result, Thm.~\ref{thm:SO:communication-bound}, was a simple argument proving that an exponentially large classical channel is required to perfectly simulate a quantum channel, even when arbitrary pre-shared random data is available. While this result does not yet improve on the best known such bounds \cite{Montanaro16,Montina11b} (asymptotically, they are equivalent), it does provide a general recipe from which these bounds can be generated, given appropriate classical error-correction channels.

This chapter proved some powerful restrictions on the types of ontological model that can reproduce quantum predictions, paying particular attention to the reality of quantum superpositions. In the next, these ideas will be applied to the concept of macro-realism, which will also lead to an error tolerant variation on Thm.~\ref{thm:SO:superpositions-are-ontic} in Sec.~\ref{sec:MR:error-tolerant-argument}. A full discussion of the meaning and impact of these results will therefore be deferred until Sec.~\ref{sec:MR:summary-and-discussion}.

\chapter{Ontology and Macro-realism} \label{ch:MR}

Chapter~\ref{ch:SO} began to address the ontology of quantum states using ontological models. The aim was to derive ontology theorems in the tradition of Bell and PBR---exactly the task the ontological models framework was designed for (being a refinement of Bell's original approach to ``hidden variable'' models). This chapter will use those same methods to tackle the issue of macro-realism, which (despite the name) is not normally considered using ontological models. The main result of this will be a no-go theorem for macro-realism in quantum theory that closes loopholes in the original approach due to Leggett and Garg.

\section{The Meaning of Macro-realism} \label{sec:MR:intro}

The concept of macro-realism was introduced to the study of quantum theory by Leggett \& Garg alongside their eponymous inequalities \cite{LeggettGarg85}. The \emph{Leggett-Garg inequalities} (LGIs) are inequalities on observed measurement statistics that are derived by assuming a particular form of macro-realism and can be violated by measurements on quantum systems. The purpose of the LGIs is therefore to prove that quantum theory and macro-realism are incompatible. However, since its introduction the exact meaning of ``macro-realism'' has been the subject of debate \cite{Ballentine87,LeggettGarg87,Leggett88,Leggett02a,Leggett02b,KoflerBrukner13,MaroneyTimpson17}. The purpose of this section is to clarify the meaning of macro-realism, though for the sake of brevity some details will have to be omitted. A more thorough account can be found in Ref.~\cite{MaroneyTimpson17}.

\subsection{Introducing Macro-realism} \label{sec:MR:introducing-MR}

Macro-realism is an ontological position. Loosely, macro-realism is the philosophical requirement that certain ``macroscopic'' quantities always possess definite values. As such, macro-realism is a great candidate for analysis with ontological models.

By using ontological models it is possible to illuminate and classify various definitions of macro-realism precisely (Sec.~\ref{sec:MR:sub-classes}, following Ref.~\cite{MaroneyTimpson17}). This analysis will reveal some fundamental loopholes in the Leggett-Garg argument for the incompatibility of quantum theory and macro-realism [Sec.~\ref{sec:MR:loopholes}]. In particular, it will show that violation of the LGIs serves only to rule out one sub-category of macro-realist models and that there are other macro-realist models of quantum theory which are compatible with the LGIs. These loopholes are not experimental but logical; the only way to close them is to fundamentally change the argument.

The main result of this chapter is a stronger theorem for the incompatibility between quantum theory and macro-realism [Sec.~\ref{sec:MR:main-theorem}]. This theorem closes a loophole in the Leggett-Garg argument by establishing that quantum theory is incompatible with a larger subset of macro-realist models. It does not prove incompatibility of quantum theory with \emph{all} macro-realist models since that is impossible due to existing counter-examples [Sec.~\ref{sec:MR:loopholes}]. The theorem proceeds in a very different manner than the Leggett-Garg argument and is related to Thm.~\ref{thm:SO:superpositions-are-ontic}. It thereby circumvents many of the controversies of the original Leggett-Garg approach.

Macro-realism is of interest to experimentalists as well as theorists. There has been a surge of recent work on experimental verification of LGI violation \cite{WangKnee+17,KneeKakuyanagi+16,HuffmanMizel16,ZhouHuelga+15,KneeSimmons+12} and in particular on noise-tolerance and closing experimental loopholes. At face value, the main theorem presented in Sec.~\ref{sec:MR:main-theorem} will not be suitable for experimental investigation, but Sec.~\ref{sec:MR:error-tolerant-argument} will follow one route to error-tolerance for experiments. This will also enable an error-tolerant variation of Thm.~\ref{thm:SO:superpositions-are-ontic}. Further discussion of the experimental relevance of these results will be deferred until Sec.~\ref{sec:MR:summary-and-discussion}.

It should be noted that mathematically there is no meaning to the stipulation that macro-realism is about ``macroscopic'' quantities, as opposed to other physical quantities that aren't ``macroscopic''. Philosophically, however, it is easy to understand the desire for macro-realism applying to ``macroscopic'' quantities. The types of physical quantity that humans experience are all considered macroscopic and they certainly appear to possess definite values. On the other hand, it is much easier to imagine that microscopic quantities that aren't directly observed behave in radically different ways. So while there is nothing in the structure of quantum theory to pick-out ``macroscopic'' versus ``microscopic'', the motivation for considering macro-realism does come from considering macroscopic quantities, hence the name.

\subsection{Defining Macro-realism} \label{sec:MR:definition}

Exactly what is meant by ``macro-realism'' has been a subject of contention ever since its introduction alongside the LGIs. This controversy has fed into more recent work on understanding the violation of the LGIs \cite{ClementeKofler15,ClementeKofler16,Bacciagaluppi15}. In Ref.~\cite{MaroneyTimpson17}, uses of the term ``macro-realism'' are analysed and the concept is illuminated using ontological models. One result of that paper is that the ``macro-realism'' intended by Leggett and Garg, as well as many subsequent authors, can be made precise in a reasonable way with the definition:

\begin{quote}
``A macroscopically observable property with two or more distinguishable values available to it will at all times determinately possess one or other of those values.'' \cite{MaroneyTimpson17}
\end{quote}

Throughout this chapter, ``distinguishable'' will be taken to mean ``in principle perfectly distinguishable by a single measurement in the noiseless case''. Note that macro-realism is defined with respect to some specific property $Q$. A macro-realist model generally will be macro-realist for some properties and not others. This property will have values $\{q\}$ and to be ``observable'' must correspond to at least one measurement $M_{Q}$ with corresponding outcomes $E_{q}$ which faithfully reveals the underlying macro-realist value $q$.

Reference~\cite{MaroneyTimpson17} fleshes out this definition using ontological models and as a result describes three sub-categories of macro-realism. In order to discuss these it will be necessary to first define an \emph{operational eigenstate} in ontological models.

An operational eigenstate $Q_{q}$ of any value $q$ of an observable property $Q$ is a set of preparation procedures $\{P_{q}\}$. This set is defined so that immediately following any $P_{q}$ with any measurement of the quantity $Q$ will result in the outcome $E_{q}$ with certainty. In other words, an operational eigenstate is simply an extension of the concept of a quantum eigenstate to ontological models: the preparations which, when appropriately measured, always return a particular value of a particular property. Note that if two values $q,q^{\prime}$ have operational eigenstates then they can sensibly be called ``distinguishable'', since any system prepared in a corresponding operational eigenstate can be identified to have one value and not the other with certainty.

\section{The Leggett-Garg Inequalities} \label{sec:MR:LGIs}

The LGIs are inequalities on the outcomes of certain experiments. The aim of the Leggett-Garg argument is to derive them by assuming macro-realism so that if measurements in quantum theory violate these inequalities then quantum theory must be incompatible with macro-realism. Quantum theory certainly predicts measurements that violate the LGIs and so must be incompatible with at least one of the assumptions needed to derive them. Whether or not the Leggett-Garg argument proves the incompatibility of quantum theory and macro-realism therefore rests on exactly what assumptions are required to derive the LGIs.

\subsection{Outline of the Argument} \label{sec:MR:LGI-outline}

The Leggett-Garg argument has a similar structure to Bell's theorem and the LGIs themselves also bear striking resemblance to some Bell inequalities. Operationally, however, the approaches of Bell and Leggett-Garg are quite different. The LGIs only require a single system that is measured several times in sequence. A thorough and complete discussion of LGI derivations is inappropriate here. What follows is rather a sketch sufficient to give context to the remarks later in the chapter. A more extensive discussion can be found in, \emph{e.g.}, Refs.~\cite{MaroneyTimpson17,EmaryLambert+14}.

Consider measuring a two-valued property $Q$ for a system at three times $t_1 < t_2 < t_3$ in sequence. Label the outcome values $Q^{[123]}_{1,2,3} \in \{-1,+1\}$ for the first, second, and third measurements respectively. The superscript $[123]$ labels this as an experiment where measurements are performed at all three times. On any run of this experiment, it is simple to verify that $Q^{[123]}_1 Q^{[123]}_2 + Q^{[123]}_1 Q^{[123]}_3 + Q^{[123]}_2 Q^{[123]}_3$ can only equal $-1$ or $3$. Clearly, taking the average over many runs gives
\begin{equation}
-1 \leq \langle Q^{[123]}_1 Q^{[123]}_2 \rangle + \langle Q^{[123]}_1 Q^{[123]}_3 \rangle + \langle Q^{[123]}_2 Q^{[123]}_3 \rangle \leq 3.
\end{equation}

Now assume that $Q$ is a macro-realist quantity. This means that $Q$ always has some value, whether or not it is measured. The same argument can therefore also be run where no measurements are actually made. Letting $Q_{1,2,3}^\ast$ be the underlying values of $Q$ in an experiment where no measurements are made, then $Q^\ast_1 Q^\ast_2 + Q^\ast_1 Q^\ast_3 + Q^\ast_2 Q^\ast_3$ can only equal $-1$ or $3$.

Since the underlying macro-realist values of $Q$ can always be revealed faithfully by a measurement, it follows that $Q^\ast_1 = Q^{[123]}_1$. However, it does not immediately follow that $Q^\ast_2 = Q^{[123]}_2$, while it does follow that $Q^\ast_2 = Q^{[23]}_2$ (\emph{i.e.}, where no measurement is performed at $t_1$).  The reason for this is simple, the very act of measuring at $t_1$ could change the underlying value of $Q$ at $t_2$ compared to not having measured at $t_1$.  Similar comments hold for $t_3$. The underlying value of $Q$ revealed by a measurement will generally depend on whether any measurements have occurred before.

So consider a final assumption: \emph{non-invasive measurability}. Suppose that measurements of $Q$ do not affect subsequent underlying values of $Q$. Now it follows, for example, that $Q^\ast_2 = Q^{[12]}_2$ since the measurement, or not, at $t_1$ has been assumed not to affect the underlying value at $t_2$ revealed by $Q^{[12]}_2$. With this non-invasive measurability, one obtains the LGIs
\begin{equation} \label{eq:MR:LGIs}
-1 \leq \langle Q^{[12]}_1 Q^{[12]}_2 \rangle + \langle Q^{[13]}_1 Q^{[13]}_3 \rangle + \langle Q^{[23]}_2 Q^{[23]}_3 \rangle \leq 3.
\end{equation}

To recap: the LGIs of Eq.~(\ref{eq:MR:LGIs}) are inequalities on the outcomes of three different types of experiments labelled $[12]$, $[13]$, and $[23]$. For experiment $[13]$, measurements of $Q$ are made at $t_1$ and $t_3$ only and similarly for $[12]$ and $[23]$. Equation~(\ref{eq:MR:LGIs}) has been derived by assuming both macro-realism and non-invasive measurability of $Q$. There are many quantum experiments of this form that can violate Eq.~(\ref{eq:MR:LGIs}) \cite{LeggettGarg85} so quantum theory must be incompatible with either macro-realism or non-invasive measurability for $Q$.

That non-invasive measurability is required to derive LGIs has been known since their introduction \cite{LeggettGarg85}. Some go so far as to include non-invasive measurability as part of their definition of macro-realism to avoid having to deal with it explicitly, relegating definitions like that in Sec.~\ref{sec:MR:definition} to ``macro-realism \emph{per se}''. The final part of the Leggett-Garg argument has always been to contend that this non-invasive measurability is a necessary consequence of macro-realism.

While this concludes the sketch of the traditional Leggett-Garg argument, the problem of deriving non-invasive measurability from macro-realism will be returned to in Sec.~\ref{sec:MR:loopholes}.

\subsection{Sub-classes of Macro-realism} \label{sec:MR:sub-classes}

To clearly discuss exactly what is ruled out by the Leggett-Garg argument---and how it might be improved---it is necessary to understand three sub-categories of macro-realism. These were identified in Ref.~\cite{MaroneyTimpson17} by considering the definition of Sec.~\ref{sec:MR:definition} in terms of ontological models for the system. The resulting sub-categories are then categories of ontological models with particular properties.

The three sub-categories of macro-realism for some quantity $Q$ are:
\begin{enumerate}

\item \emph{Operational eigenstate mixture macro-realism (EMMR) --} The only preparations in the model are operational eigenstates of $Q$ or statistical mixtures of operational eigenstates. That is, every preparation measure can be written in the form $\nu=\sum_{q}\sum_{i}c_{q,i}\mu_{q,i}$ where each $\mu_{q,i}$ is a preparation measure for an operational eigenstate for $q$ and $\{c_{q,i}\}$ are positive reals summing to unity. Note that this means the space of ontic states $\Lambda$ need only include ontic states accessible by preparing some operational eigenstate of $Q$, as no other ontic states can ever be prepared. 

\item \emph{Operational eigenstate support macro-realism (ESMR) --} Like EMMR, every ontic state $\lambda\in\Lambda$ is accessible by preparing some operational eigenstate but, unlike EMMR, there are preparation measures in the model that are not statistical mixtures of operational eigenstate preparations for $Q$. That is, if $\Omega\in\Sigma$ satisfies $\mu_{q}(\Omega)=1$ for every operational eigenstate preparation $\mu_{q}$ of $Q$, then every preparation measure $\nu$ in the model also satisfies $\nu(\Omega)=1$. Moreover, the model has at least one preparation measure not in the mixture form required by EMMR. In other words, if you're certain to prepare an ontic state from some subset $\Omega$ when preparing operational eigenstates, then you're also certain to prepare an ontic state from $\Omega$ from any other preparation measure in the model.

\item \emph{Supra eigenstate support macro-realism (SSMR) --} Every ontic state $\lambda$ in the model will produce some specific value $q_{\lambda}$ of $Q$ when a measurement of $Q$ is made, but some of those ontic states are not accessible by preparing any operational eigenstates of $Q$. That is, for every $\lambda\in\Lambda$ there is some value $q_{\lambda}$ of $Q$ such that $\mathbb{P}_{M}(E_{q_{\lambda}}\,|\,\lambda)=1$ whenever $Q$ is measured. Moreover, there exists some $\Omega\in\Sigma$ and preparation measure $\nu$ such that $\nu(\Omega)>0$ while $\mu_{q}(\Omega)=0$ for every operational eigenstate preparation measure $\mu_{q}$.

\end{enumerate}
To help unpack these definitions, they are illustrated in Fig.~\ref{fig:MR:definitions}.

\begin{figure}
\includegraphics[width=1\columnwidth]{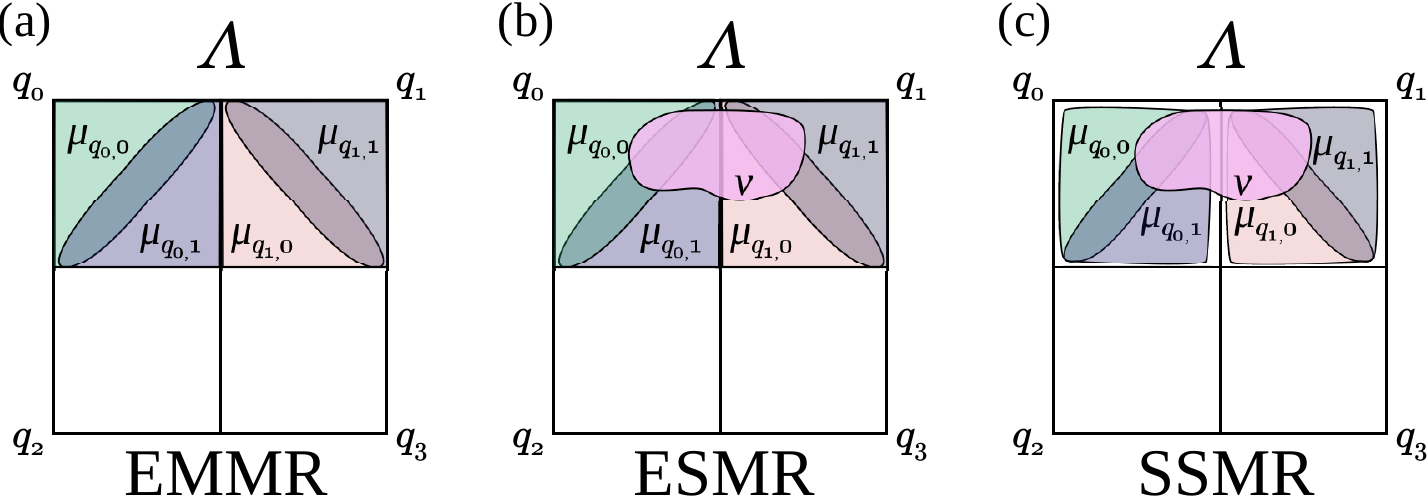}

\protect\caption{Illustration of the three sub-categories of macro-realism as defined in the text. In each case the large square represents the whole ontic state space $\Lambda$, the four smaller squares indicate those subspaces of ontic states associated with each value $q_{0..3}$ of some quantity $Q$, and the shaded regions represent those ontic states accessible by preparing some select preparation measures. Only a few preparation measures are shown while more would exist in a real model, in particular measures have been left out of the $q_2$ and $q_3$ boxes to avoid over-cluttering the figures.\protect \\
(a) illustrates EMMR, where the squares for each $q_{i}$ contain all ontic states preparable via some operational eigenstate preparation $\mu_{q_{i},j}$ and all other allowed preparation measures are simply statistical mixtures of these, \emph{e.g. $\nu=\frac{1}{3}\left(\mu_{q_{0},0}+\mu_{q_{0},1}+\mu_{q_{1},0}\right)$} is permissible.\protect \\
(b) illustrates ESMR, where the state space is exactly as in EMMR, but now more general preparation measures, such as the $\nu$ illustrated, are permitted.\protect \\
(c) illustrates SSMR, where now every $\lambda$ in the box for $q_{i}$ must produce outcome $q_{i}$ in any appropriate measurement of $Q$, but the operational eigenstates no longer fill these boxes. That is, there are ontic states that lie outside the preparations for operational eigenstates. General preparation measures over the boxes, like $\nu$, are still permitted.}
\label{fig:MR:definitions}

\end{figure}

In each of these cases, every ontic state $\lambda$ (up to possible measure-zero sets of exceptions) is associated with a specific value $q_{\lambda}$ of $Q$, such that it can be sensibly said that $\lambda$ ``possesses'' $q_{\lambda}$. This is why they are all considered types of macro-realism. Consider this for each case in turn.

In an EMMR model, every preparation can be read as a probabilistic choice between operational eigenstate preparations. Depending on which preparation is chosen, the resulting ontic state $\lambda$ therefore ``possesses'' the value $q$ for the operational eigenstate.

In an ESMR model, every ontic state $\lambda$ can be prepared by an operational eigenstate of exactly one value of $Q$ (up to measure-zero sets of exceptions). Similarly, therefore, each ontic state ``possesses'' the corresponding value of $Q$.

In SSMR models the link between each $\lambda$ and the corresponding $q_{\lambda}$ is explicit. Each $\lambda$ ``possesses'' the value $q_\lambda$ for which $\mathbb{P}_M(E_{q_\lambda}\,|\,\lambda) = 1$, as required by the definition. That is, $\lambda$ ``possesses'' the value which it must return with certainty in any appropriate measurement.

Note that these three sub-categories of macro-realism are defined such that they are mutually exclusive, but they still have a natural hierarchy to them. ESMR can be seen as a less restrictive variation on EMMR, since you can make an EMMR model into an ESMR model simply by including a single preparation measure that is not a statistical mixture of operational eigenstate preparations (the ontic state space and everything else can remain unchanged). Similarly, SSMR can be seen as a less restrictive variation on ESMR. In ESMR, every ontic state $\lambda$ can be obtained by preparing an operational eigenstate preparation for a value of $Q$. By definition of operational eigenstate it follows that a measurement of $Q$ will therefore return some specific value for each ontic state (up to measure-zero sets of exceptions), which is the primary requirement on the ontic states for SSMR.

\subsection{Macro-realism for Quantum Systems} \label{sec:MR:quantum-MR}

Just like ontological models, the concept of macro-realism is---and always should be---logically independent of quantum theory. So far this section has presented macro-realism in this general way. In order to proceed to a precise discussion of the loopholes in the Leggett-Garg argument, it is now necessary to bring macro-realism, ontological models, and quantum theory together.

To do this, consider what can count as a ``macroscopically observable'' quantity $Q$. To be observable $Q$ must correspond to some quantum measurement $M_{Q}$. Therefore, there is some orthonormal basis $\mathcal{B}_{Q}$ so that for each value $q$ of $Q$ the corresponding outcome of $M_{Q}$ is a state in $\mathcal{B}_{Q}$. In order to make sense of the definitions $Q$ must also have operational eigenstates for each value $q$ of $Q$. Fortunately this is straightforward in quantum theory: the states in $\mathcal{B}_{Q}$ are exactly the operational eigenstates of $Q$. Moreover, because the elements of $\mathcal{B}_{Q}$ are orthogonal it follows that preparations corresponding to different values $q,q^{\prime}$ of $Q$ are distinguishable.

So in quantum theory, macro-realism for quantity $Q$ simply means that there is a basis $\mathcal{B}_Q$ each element of which is an (operational) eigenstate of $Q$. This is gained simply by applying the definition of macro-realism from Sec.~\ref{sec:MR:definition} to quantum theory, the three different sub-categories of macro-realism noted in Sec.~\ref{sec:MR:sub-classes} correspondingly define three possible sub-categories of macro-realism in quantum theory.

\subsection{Loopholes in the Leggett-Garg Argument} \label{sec:MR:loopholes}

The aim of the LGIs has always been to rule out macro-realist ontologies for quantum theory, where the inequalities are violated. However, in light of the above precise definition of macro-realism some loopholes in the argument can be identified. These loopholes all have to do with the attempt to derive non-invasive measurability from macro-realism as noted in Sec.~\ref{sec:MR:LGI-outline}.

The first loophole is that violation of the LGIs cannot rule out SSMR models of quantum systems. Indeed, no argument that rests on compatibility with quantum predictions can completely rule out SSMR models since there exists a well-known SSMR model for quantum systems that reproduces all quantum predictions: Bohmian mechanics \cite{Bohm52a,Bohm52b,deBroglie27}.

To see that Bohmian mechanics implies an SSMR ontological model consider, for example, the Bohmian description of a single spinless point particle in three-dimensional space (the argument for more general systems is analogous). Bohmian mechanics has the ontic state as a pair $\lambda=(\vec{r},|\psi\rangle)\in\mathbb{R}^{3}\times\mathcal{P}(\mathcal{H})$ where $\vec{r}$ is the actual position of the particle and $|\psi\rangle$ is the quantum state (or ``pilot wave''). Note that the quantum state is part of the ontology here. The ``macroscopically observable property'' is the position of the particle, $\vec{r}$, and any sharp measurement of position will reveal the true value of $\vec{r}$ with certainty. Thus, for any ontic state $\lambda$ there is some value of the macroscopically observable property (that is, $\vec{r}$) which is obtained with certainty from any appropriate measurement. Thus, Bohmian mechanics provides an SSMR ontological model. 

The second loophole is that LGI violation is unable to rule out ESMR ontological models. This also has a counter-example in the form of the Kochen-Specker model for the qubit \cite{KochenSpecker67}, which is an ontological model satisfying ESMR\footnote{Strictly speaking, the Kochen-Specker model was not defined with a post-measurement update rule and so cannot deal with sequences of measurements (and therefore Leggett-Garg experiments). However, it is simple to append the obvious update rule ``prepare a new state corresponding to the measurement outcome'' and this fixes the issue.}. The Kochen-Specker model exactly reproduces quantum predictions for $d=2$ dimensional systems. As the LGIs can be written in $d=2$ the Kochen-Specker model must therefore violate them.

The key point is \emph{why} these counter-examples evade the Leggett-Garg argument. As discussed in Sec.~\ref{sec:MR:LGI-outline}, to derive the LGIs one needs to assume non-invasive measurability. If it is not possible to derive or demonstrate non-invasive measurability as a consequence of some type of macro-realism, then it is not possible to derive the LGIs from that type of macro-realism and the Leggett-Garg argument does not apply. It turns out that there is no way to do this by assuming either SSMR or ESMR and therefore Leggett-Garg arguments generically have loopholes for these types of macro-realism \cite{MaroneyTimpson17}. Bohmian mechanics and the Kochen-Specker model are both examples: they contain measurement disturbance that violates the non-invasive measurability assumption, while still satisfying SSMR and ESMR respectively. So the crux is that both SSMR and ESMR models can include measurements that don't disturb the distribution over $\Lambda$ if the system is prepared in an operational eigenstate, but still disturb the distribution over $\Lambda$ for systems prepared in other ways.

On the other hand, EMMR requires that all preparations are statistical mixtures of operational eigenstates. If one can demonstrate experimentally that operational eigenstates are not disturbed by some measurement, then no preparations can be disturbed by that measurement in an EMMR model. It is this which prevents EMMR models from violating the LGIs. A more extensive discussion of this point can be found in Ref.~\cite{MaroneyTimpson17}.

Recent experiments \cite{KneeKakuyanagi+16,HuffmanMizel16} following Ref.~\cite{WildeMizel12} have sought to address the ``clumsiness'' loophole in the Leggett-Garg argument by dropping the assumption of non-invasive measurability, replacing it with control experiments that experimentally serve the same purpose. These approaches follow the Leggett-Garg argument quite closely and show that the disturbance on a general preparation cannot be explained in terms of the disturbances on a statistical mixture of operational eigenstates. As a result, they are still only capable of ruling out EMMR models, as noted in Ref.~\cite{KneeKakuyanagi+16}.

So the Leggett-Garg proof only rules out EMMR macro-realism and leaves loopholes for SSMR and ESMR. Moreover, the loophole for SSMR models cannot be fully plugged by any proof because Bohmian mechanics exists as a counter-example. Similarly, the loophole for ESMR cannot be fully plugged in $d=2$ dimensions since the Kochen-Specker model exists as a counter-example. This leaves a clear question: can the ESMR loophole can be closed by another theorem for any $d>2$? Answering this question requires a different approach to the Leggett-Garg argument, one that doesn't use any assumptions about measurement disturbance.

\section{A Stronger Theorem Against Macro-realism} \label{sec:MR:main-theorem}

The Leggett-Garg argument, in its usual form, does not make use of anything like ontological models. Macro-realism is, however, an ontological position. The success of ontological models in investigating other such positions (such as excess baggage, locality, and $\psi$-ontology) suggests that it may also be useful when investigating macro-realism. Moreover, the three sub-categories of macro-realism of Sec.~\ref{sec:MR:sub-classes} were identified in Ref.~\cite{MaroneyTimpson17} by using an ontological model approach. Therefore, this section will take the ontological techniques developed in Chap.~\ref{ch:SO} and apply them to macro-realism to obtain a no-go result in quantum theory that is stronger than the traditional Leggett-Garg argument.

\subsection{Theorem}

As discussed in Sec.~\ref{sec:MR:loopholes}, the Leggett-Garg argument can rule out EMMR macro-realist models for quantities $Q$ with $n>1$ distinguishable values. The following theorem improves on this by ruling out both EMMR and ESMR macro-realist models for $n>3$.

\begin{theorem} \label{thm:MR:no-ESMR-EMMR}
Quantum theory is incompatible with ESMR or EMMR macro-realist models for quantities $Q$ with $n>3$ distinguishable values.
\end{theorem}

\begin{proof}
In ESMR and EMMR ontological models, every ontic state in the model can be obtained through preparing some operational eigenstate of $Q$. From this it follows that, when preparing any quantum state $|\psi\rangle$ through any $\mu\in\Delta_{|\psi\rangle}$, the probability of getting an ontic state the could have been prepared by preparing a state in $\mathcal{B}_Q$ is unity. Mathematically, $\varpi(\mathcal{B}_Q\,|\,\mu) = 1$. Comparing this to Eq.~(\ref{eq:SO:epistemic-superposition}) it is seen that every $|\psi\rangle\not\in\mathcal{B}_Q$ must therefore be an epistemic superposition over $\mathcal{B}_Q$.

Clearly, if $Q$ has $n>3$ distinguishable values then $d = |\mathcal{B}_Q| > 3$. So by using Thm.~\ref{thm:SO:superpositions-are-ontic}, one reaches a contradiction.

The assumptions that went in to reaching this contradiction were:
\begin{enumerate}[label=(\alph*)]
\item the ontology satisfies ESMR or EMMR;

\item the ``macroscopically observable quantity'' $Q$ has $n>3$ distinguishable values;
 
\item the states preparations and measurements used in the proof of Thm.~\ref{thm:SO:superpositions-are-ontic} are possible; and

\item the ontological model reproduces quantum measurement predictions, as in Eq.~(\ref{eq:IN:ontological-model-quantum-probability}).
\end{enumerate}
Assumptions (a--b) are about the underlying ontological model, whereas assumptions (c--d) are implications of standard quantum theory. The conclusion must therefore be that both ESMR and EMMR ontologies for $n>3$ are impossible, or quantum theory is incorrect.
\end{proof}

Since the Leggett-Garg argument is only able to prove incompatibility between EMMR and quantum theory, this is a strict improvement for $n>3$.

\subsection{Relation to Other Ontology Results} \label{sec:MR:relation-between-results}

The proof of this stronger no-go theorem directly used Thm.~\ref{thm:SO:superpositions-are-ontic} to get its result. Some comments are therefore in order about the exact relationship between Thms.~\ref{thm:SO:superpositions-are-ontic}, \ref{thm:MR:no-ESMR-EMMR}.

The obvious difference between the two is in the quantifiers. To prove Thm.~\ref{thm:SO:superpositions-are-ontic} a contradiction is required for almost all superposition states. Comparing this to the proof of Thm.~\ref{thm:MR:no-ESMR-EMMR}, one sees that the latter only requires a contradiction for a single superposition state. From this point of view, Thm.~\ref{thm:MR:no-ESMR-EMMR} is weaker than Thm.~\ref{thm:SO:superpositions-are-ontic} which is why the latter is used to prove the former rather than the other way around. Indeed, once it has been established that ESMR/EMMR ontologies imply epistemic superpositions, it is almost a corollary.

Clearly, the two theorems, in their current forms, are getting at a similar point. It even appears that ESMR/EMMR macro-realism might be simply be subsumed under the banner of epistemic superpositions. But there are also divergences. Foremost of these is that ``macro-realism'' can (and should) be defined independently from quantum theory (as in Secs.~\ref{sec:MR:definition}--\ref{sec:MR:sub-classes}) while the very quantum notion of ``superpositions'' is required to even talk about ``epistemic superpositions''. The two concepts are of very different character.

This difference in character becomes very important when considering extensions beyond Thms.~\ref{thm:SO:superpositions-are-ontic}, \ref{thm:MR:no-ESMR-EMMR}, especially to statements amenable to experimental tests. Both results in their current forms concentrate on properties of pure states. This is entirely natural for Thm.~\ref{thm:SO:superpositions-are-ontic}, since superpositions are pure states. But as macro-realism is defined independently from quantum theory, it would be appropriate to extend Thm.~\ref{thm:MR:no-ESMR-EMMR} using mixed state preparations and more general quantum measurements. More discussion on possible appropriate extensions to each of the theorems will be in Sec.~\ref{sec:MR:summary-and-discussion}.

Neither of these theorems are intended to be final, but rather starting points that point to new lines of research. The most obvious next direction is establishing error-tolerant variations that allow for experimental investigation. One method for doing this will follow in the next section.

\section{An Error-Tolerant Argument} \label{sec:MR:error-tolerant-argument}

Theorems~\ref{thm:SO:superpositions-are-ontic}, \ref{thm:MR:no-ESMR-EMMR} are the main results of their respective chapters. However, both proofs rely on requiring ontological models to exactly reproduce quantum theory, that is Eq.~(\ref{eq:IN:ontological-model-quantum-probability}). This means that, in their current form, neither is amenable to experimental investigation. There is always noise and error in experimental results. Experimental results against ESMR/EMMR or epistemic superpositions would carry much more weight, so it is prudent to seek error-tolerant variations of these theorems.

In Sec.~\ref{sec:SO:error-tolerance} an error-tolerant variation on Thm.~\ref{thm:SO:no-overlap-large-d} was presented and used as a proof-of-concept for error-tolerance of the other theorems of Chap.~\ref{ch:SO}. This is a valuable result (and provides scope for experimental tests in its own right) but by too closely following the same proof strategy it does not address the core hurdle to error tolerance for the main results.

As noted in Sec.~\ref{sec:SO:error-tolerance}, the primary reason that these results are not error-tolerant is that they are based on the asymmetric overlap, which becomes useless when there is finite error in reproducing quantum statistics. Worse, the mathematical formalisation of ``epistemic superposition'', Eq.~(\ref{eq:SO:epistemic-superposition}), is the key fact used in proofs of both Thms.~\ref{thm:SO:superpositions-are-ontic} and \ref{thm:MR:no-ESMR-EMMR} and is itself in terms of the asymmetric overlap. Thus, to properly find error-tolerant variations on Thms.~\ref{thm:SO:superpositions-are-ontic} and \ref{thm:MR:no-ESMR-EMMR}, this formalisation needs to be generalised to be error-tolerant.

\subsection{Error-Tolerance and the \texorpdfstring{$\epsilon$}{epsilon}-Asymmetric Overlap} \label{sec:MR:error-tolerant-superpositions-and-mr}

In order to generalise the formalisation of epistemic superpositions and ESMR/EMMR macro-realism, it is necessary to first generalise the asymmetric overlap.

Consider the definitions of the asymmetric overlap in Eqs.~(\ref{eq:SO:asymmetric-definition}, \ref{eq:SO:asymmetric-state-definition}, \ref{eq:SO:asymmetric-multipartite-definition}). Each definition extremises over measurable subsets $\Omega\in\Sigma$ that satisfy $\nu(\Omega)=1$ for some preparation measure(s) $\nu$. This condition can be relaxed by considering \emph{$\epsilon$-typical subsets}. For any $\epsilon\in[0,1)$, $\Omega\in\Sigma$ is $\epsilon$-typical for preparation measure $\nu$ if and only if $\nu(\Omega)>1-\epsilon$. That is, the probability of $\nu$ preparing an ontic state \emph{outside} any given $\epsilon$-typical subset $\Omega$ is $\epsilon$. Despite the name, a given $\epsilon$-typical subset $\Omega$ does not identify the ontic states inside it as typical as such, but rather it identifies the ontic states \emph{outside} it as \emph{atypical} (where (a)typicality is quantified by $\epsilon$).

$\epsilon$-typical subsets can be used to generalise the asymmetric overlap $\varpi$ to the \emph{$\epsilon$-asymmetric overlap} $\varpi_\epsilon$, defined for preparation measures as
\begin{equation}
\varpi_\epsilon (\nu\,|\,\mu) \eqdef \inf\{ \mu(\Omega)\,:\,\Omega\in\Sigma,\;\nu(\Omega)>1-\epsilon\}.
\end{equation}
That is, $\varpi_\epsilon(\nu\,|\,\mu)$ is a lower bound probability that $\mu$ will produce an ontic state that is in some $\epsilon$-typical subset for $\nu$.  Clearly, by taking $\epsilon = 0$ the $0$-typical asymmetric overlap is the asymmetric overlap of Eq.~(\ref{eq:SO:asymmetric-definition}).

How should $\varpi_\epsilon$ be interpreted? Recalling that all ontic states outside an $\epsilon$-typical subset $\Omega$ are atypical for $\nu$, one sees that $1 - \mu(\Omega)$ is a probability that $\mu$ will produce an ontic state that is definitely atypical for $\nu$. Therefore, $1 - \varpi_\epsilon (\nu\,|\,\mu)$ is the upper bound probability that $\mu$ will produce an ontic state that is definitely atypical for $\nu$. This double-negative interpretation of the $\epsilon$-asymmetric overlap is the precise one, but is also clumsy. In effect, it says that $\varpi_\epsilon (\nu\,|\,\mu)$ is the probability that $\mu$ will produce an ontic state that is typical for $\nu$ (where typicality is quantified by $\epsilon$). However, the nature of measure theory means that this must be understood as a double-negative.

When describing a quantum system, Eq.~(\ref{eq:SO:asymmetric-state-definition}) defines the ordinary asymmetric overlap with a quantum state. The same can be done for the $\epsilon$-asymmetric overlap:
\begin{equation}
\varpi_\epsilon (|\phi\rangle\,|\,\mu) \eqdef \inf\{ \mu(\Omega)\,:\,\Omega\in\Sigma,\;\nu(\Omega)>1-\epsilon,\;\forall\nu\in\Delta_{|\phi\rangle}\}.
\end{equation}
Finally, the definition can also be extended to overlapping with a \emph{set} of quantum states $\mathcal{S}\subseteq\mathcal{P}(\mathcal{H})$, as in Eq.~(\ref{eq:SO:asymmetric-multipartite-definition}),
\begin{equation}
\varpi_\epsilon (\mathcal{S}\,|\,\mu)\eqdef\inf\left\{ \mu(\Omega)\;:\;\Omega\in\Sigma,\;\nu(\Omega)>1-\epsilon,\;\forall\nu\in\Delta_{|\phi\rangle},\;\forall|\phi\rangle\in\mathcal{S}\right\}.
\end{equation}

Now defined, the $\epsilon$-asymmetric overlap can be used to identify an error-tolerant formalisation of epistemic superpositions.

As discussed in Sec.~\ref{sec:SO:ontic-superpositions}, an epistemic superposition is one that is just a statistical effect, one where the ontology is all accounted for by the underlying basis. That is, somebody who believes $|\psi\rangle$ is an epistemic superposition over $\mathcal{B}$ should believe that any ontic state obtained by preparing $|\psi\rangle$ could also be obtained by preparing some state from $\mathcal{B}$. In the case of finite error this is clearly an inconsequential belief as there is always some finite probability of preparing any outlandish ontic state in any preparation. An appropriate modified belief would be that any preparation of $|\psi\rangle$ will, with high probability, produce an ontic state that is \emph{typical} for at least one state in $\mathcal{B}$. That is, the vast majority of the time the resulting ontic state will be one that can typically be obtained by at least one state in $\mathcal{B}$. Mathematically, this can be characterised with the $\epsilon$-asymmetric overlap:
\begin{equation} \label{eq:MR:typical-epistemic-superposition}
\varpi_\eta ( \mathcal{B}\,|\,\mu ) > 1 - \tau, \quad \forall\mu\in\Delta_{|\psi\rangle}
\end{equation}
for small values $\eta$ and $\tau$ which quantify the strength of the belief. Clearly, taking $\tau=\eta=0$ this reduces to the noiseless case of Eq.~(\ref{eq:SO:epistemic-superposition}).

A similar argument can be made for ESMR/EMMR macro-realism for a quantity $Q$, generalising the proof of Thm.~\ref{thm:MR:no-ESMR-EMMR}. The ESMR/EMMR macro-realist believes that every ontic state preparable by any $|\psi\rangle\in\mathcal{P}(\mathcal{H})$ can also be obtained by preparing a state from a basis $\mathcal{B}_Q$. In the case of finite error this is a trivial belief, since there's a finite probability of preparing any ontic state from any quantum state. A reasonable macro-realist would therefore rather believe that when preparing any $|\psi\rangle$ there's a high probability that the ontic state obtained will be \emph{typical} for at least one state in $\mathcal{B}_Q$. Mathematically, this implies Eq.~(\ref{eq:MR:typical-epistemic-superposition}) again, with $\mathcal{B} = \mathcal{B}_Q$.
	
\subsection{Properties of the \texorpdfstring{$\epsilon$}{epsilon}-Asymmetric Overlap}

Just as with the asymmetric overlap, it will be necessary to prove some properties of the $\epsilon$-asymmetric overlap before proceeding to the main theorem. The properties here are mostly generalisations of those proved in Sec.~\ref{sec:SO:asymmetric-properties}. The reader may find it easier to first read the theorem in the next section, referring back here only when the properties are referenced.

For notational convenience, it is useful to generalise Def.~\ref{def:SO:k-bar} in the following way.

\begin{definition} \label{def:MR:k-bar-e}
For any measurable function $g:\Lambda\rightarrow[0,1]$ let 
\begin{equation}
\k_\epsilon (g) \eqdef \{ \lambda \in\Lambda\, :\, g(\lambda)\geq 1 - \epsilon \} \in\Sigma
\end{equation}
for any $\epsilon\in[0,1)$.
\end{definition}

Next, this lemma generalises Lem.~\ref{lem:SO:measurable} and will similarly provide a link between integrals over a measure and measurable subsets.

\begin{lemma} \label{lem:MR:meas-function-to-subset}
Given any measurable function $f:\Lambda\rightarrow[0,1]$ and preparation measure $\nu$ satisfying $\int_\Lambda \d\nu(\lambda)f(\lambda) \geq 1 - \delta$ then 
\begin{equation}
\nu(\k_\kappa(f)) > 1 - \frac{\delta}{\kappa}
\end{equation}
for any $0 \leq \delta \leq \kappa \leq 1$.
\end{lemma}
\begin{proof}
For any $\kappa\in[0,1)$ then by assumption
\begin{equation}
1 - \delta \leq \int_{\k_\kappa(f)} \d\nu(\lambda)\,f(\lambda) + \int_{\Lambda\setminus\k_\kappa(f)} \d\nu(\lambda)\,f(\lambda).
\end{equation}
Considering the first term, clearly $\int_{\k_\kappa(f)}\d\nu(\lambda)\,f(\lambda) \leq \nu(\k_\kappa(f))$ as $f(\lambda)\leq 1$. For the second term, note $f(\lambda) < 1 - \kappa$ when $\lambda\in\Lambda\setminus\k_\kappa(f)$, giving
\begin{equation}
1 - \delta < \nu(\k_\kappa(f)) + (1 - \kappa)\nu(\Lambda\setminus\k_\kappa(f))
\end{equation}
which, recalling $\nu(\Lambda)=1$, implies the desired result.
\end{proof}

This first property does not have an analogue for the asymmetric overlap. It simply notes that some $\epsilon$-asymmetric overlaps are more restrictive than others (depending on the values of $\epsilon$).

\begin{lemma} \label{lem:MR:e-asymmetric-less-restrictive}
For any set $\mathcal{S}$ of quantum states and any preparation measure $\mu$,
\begin{equation}
\varpi_p (\mathcal{S}\,|\,\mu) \geq \varpi_q (\mathcal{S}\,|\,\mu)
\end{equation}
for any $0 \leq p \leq q < 1$.
\end{lemma}
\begin{proof}
For any $\Omega_p$ that is $p$-typical for $\mathcal{S}$ then
\begin{equation}
\nu_i(\Omega_p) > 1 - p \geq 1 - q, \quad \forall\nu_i\in\Delta_{|i\rangle},\;\forall|i\rangle\in\mathcal{S}
\end{equation}
so $\Omega_p$ is also $q$-typical for $\mathcal{S}$ and by definition $\varpi_p (\mathcal{S}\,|\,\mu) \geq \varpi_q (\mathcal{S}\,|\,\mu)$.
\end{proof}

With this technical background, the following property generalises Lem.~\ref{lem:SO:asymmetric-basic-bound}, showing how measurement probabilities simply bound the $\epsilon$-asymmetric overlap.

\begin{lemma} \label{lem:MR:e-asymmetric-to-probability}
For any state $|\phi\rangle\in\mathcal{P}(\mathcal{H})$ and measurement $M$ with $|\phi\rangle$ as an outcome, if $\mathbb{P}_M(|\phi\rangle\,|\,\nu) \geq 1 - \sigma$ for some $\sigma\in[0,1)$ and all $\nu\in\Delta_{|\phi\rangle}$ then 
\begin{equation}
\varpi_\eta (|\phi\rangle\,|\,\mu) \leq \frac{\mathbb{P}_M(|\phi\rangle\,|\,\mu)}{1 - \sigma/\eta},\quad \forall\eta\in(\sigma,1)
\end{equation}
for any preparation measure $\mu$, where $\mathbb{P}_M(|\phi\rangle\,|\,\mu)$ is obtained from the ontological model via Eq.~(\ref{eq:IN:ontological-model-probability}).
\end{lemma}
\begin{proof}
For some $\eta\in(\sigma,1)$, let $\Omega\eqdef \k_{\sigma/\eta}(\mathbb{P}_M(|\phi\rangle\,|\,\lambda))$. For all $\nu\in\Delta_{|\phi\rangle}$,
\begin{equation}
1 - \sigma \leq \mathbb{P}_M(|0\rangle\,|\,\nu) = \int_\Lambda \d\nu(\lambda)\,\mathbb{P}_M(|\phi\rangle\,|\,\lambda)
\end{equation}
implies that $\nu(\Omega) > 1 - \eta$ by Lem.~\ref{lem:MR:meas-function-to-subset}. Therefore, $\Omega$ is $\eta$-typical for $|\phi\rangle$, giving
\begin{eqnarray}
\varpi_\eta(|\phi\rangle\,|\,\mu) & \leq & \mu(\Omega) \\
 & \leq & \frac{ \int_\Omega \d\mu(\lambda)\,\mathbb{P}(|\phi\rangle\,|\,\lambda) }{ 1 - \sigma/\eta } \\
 & \leq & \frac{ \mathbb{P}_M( |\phi\rangle\,|\,\mu ) }{ 1 - \sigma/\eta }.
\end{eqnarray}
The second line follows as $\mathbb{P}_M(|\phi\rangle\,|\,\lambda) \geq (1 - \sigma/\eta)$ for $\lambda\in\Omega$ and the third by expanding the range of integration.
\end{proof}

Lemma~\ref{lem:SO:asymmetric-unitary} establishes how the asymmetric overlap changes under a unitary transformation, this next lemma does the same for the $\epsilon$-asymmetric overlap.

\begin{lemma} \label{lem:MR:e-asymmetric-transformation}
Let unitary $U$ satisfy $U|0\rangle = |\phi\rangle$ and $\mu^\prime\transto{\gamma}\mu$ for some $\gamma\in\Gamma_U$, then 
\begin{equation}
\varpi_\epsilon (|\phi\rangle\,|\,\mu) \geq \left( 1 - \frac{\epsilon}{\delta} \right) \varpi_\delta (|0\rangle\,|\,\mu^\prime)
\end{equation}
for any $0 \leq \epsilon \leq \delta < 1$.
\end{lemma}
\begin{proof}
Since $\varpi_\epsilon$ is defined as a greatest lower bound, it suffices to prove that for every $\epsilon$-typical $\Omega\in\Sigma$ for $|\phi\rangle$ there exists a $\delta$-typical $\Omega^\prime\in\Sigma$ for $|0\rangle$ such that
\begin{equation}
\mu(\Omega) \geq \left( 1 - \frac{\epsilon}{\delta} \right) \mu^\prime(\Omega^\prime).
\end{equation}

For any $\chi\in\Delta_{|0\rangle}$, consider the $\nu\in\Delta_{|\phi\rangle}$ such that $\chi\transto{\gamma}\nu$. Then for any $\epsilon$-typical $\Omega\in\Sigma$ for $|\phi\rangle$ it follows that
\begin{equation}
1 - \epsilon < \nu(\Omega) = \int_\Lambda \d\chi(\lambda)\,\gamma(\Omega\,|\,\lambda).
\end{equation}
Therefore, by Lem.~\ref{lem:MR:meas-function-to-subset}, $\chi( \k_\kappa(\gamma(\Omega\,|\,\lambda)) ) > 1 - \frac{\epsilon}{\kappa}$ for any $\kappa\in[\epsilon,1]$. Recalling that $\chi\in\Delta_{|0\rangle}$ is arbitrary, this means that $\Omega^\prime \eqdef \k_\kappa(\gamma(\Omega\,|\,\lambda))$ is $(\epsilon/\kappa)$-typical for $|0\rangle$.

Considering then $\mu(\Omega)$,
\begin{eqnarray}
\mu(\Omega) &= & \int_\Lambda \d\mu^\prime(\lambda)\,\gamma(\Omega\,|\,\lambda) \geq \int_{\Omega^\prime} \d\mu^\prime(\lambda)\,\gamma(\Omega\,|\,\lambda) \\
&\geq & (1 - \kappa)\mu^\prime(\Omega^\prime)
\end{eqnarray}
having noted that $\gamma(\Omega\,|\,\lambda)\geq(1 - \kappa)$ for $\lambda\in\Omega^\prime$. Letting $\delta \eqdef \epsilon/\kappa$ completes the proof.
\end{proof}

This next property generalises Lem.~\ref{lem:SO:asymmetric-Boole} by providing a relationship between the overlap with a set of quantum states and the individual overlaps with each state.

\begin{lemma} \label{lem:MR:multi-partite-e-asymmetric-breakup}
For any finite set $\mathcal{S}\subset\mathcal{P}(\mathcal{H})$ of quantum states, any preparation measure $\mu$, and any $\eta\in(0,1)$
\begin{equation}
\sum_{|i\rangle\in\mathcal{S}} \varpi_\eta (|i\rangle\,|\,\mu) \geq \varpi_\eta (\mathcal{S}\,|\,\mu).
\end{equation}
\end{lemma}
\begin{proof}
For each $|i\rangle\in\mathcal{S}$ let $\Omega_i \in\Sigma$ be an $\eta$-typical subset for $|i\rangle$. Consider $\Omega\eqdef\cup_{|i\rangle\in\mathcal{S}}\Omega_i$ which must also be $\eta$-typical for each $|i\rangle\in\mathcal{S}$ as $\nu_i(\Omega)\geq\nu_i(\Omega_i)$. Therefore by definition $\varpi_\eta (\mathcal{S}\,|\,\mu) \leq \mu(\Omega)$. By Boole's inequality,
\begin{equation}
\varpi_\eta (\mathcal{S}\,|\,\mu) \leq \mu(\Omega) \leq \sum_{|i\rangle} \mu(\Omega_i)
\end{equation}
from which the result follows by recalling that all $\Omega_i$ are arbitrary $\eta$-typical subsets.
\end{proof}

Finally, in Sec.~\ref{sec:SO:anti-distinguishability} anti-distinguishable sets of quantum states were defined. Lemma~\ref{lem:SO:anti-distinguishing-equality} then showed how anti-distinguishable triples affect asymmetric overlaps. In the following lemma, the corresponding statement is proved for the $\epsilon$-asymmetric overlap and \emph{approximately} anti-distinguishable sets (that is, the ``anti-distinguishing'' measurement must be accurate to within $\pm\epsilon$).

\begin{lemma} \label{lem:MR:antisymmetric-asymmetric-overlap}
Consider a quantum measurement $M = \{E_{\neg\psi}, E_{\neg\phi}, E_{\neg0} \}$ such that for some $\epsilon\in[0,1)$
\begin{equation}
\int_\Lambda \d\mu(\lambda)\,\mathbb{P}_M(E_{\neg\psi}\,|\,\lambda) \leq \epsilon,\quad \forall\mu\in\Delta_{|\psi\rangle}
\end{equation}
and similarly for $|\phi\rangle$ and $|0\rangle$. Then for any $\mu\in\Delta_{|\psi\rangle}$ and $\kappa\in(2\epsilon,1)$,
\begin{equation}
\varpi_{\epsilon/\kappa} (|0\rangle,|\phi\rangle\,|\,\mu) \geq (1 - \kappa)\left( \varpi_{2\epsilon/\kappa}(|0\rangle\,|\,\mu) + \varpi_{2\epsilon/\kappa}(|\phi\rangle\,|\,\mu) \right) - \epsilon.
\end{equation}
\end{lemma}
\begin{proof}
Since $\varpi_\epsilon$ is defined as a greatest lower bound, it suffices to prove that for every $\Omega\in\Sigma$ that is $(\epsilon/\kappa)$-typical for both $|0\rangle$ and $|\phi\rangle$ there exists $\Omega^\prime\in\Sigma$ that is $(2\epsilon/\kappa)$-typical for $|\phi\rangle$ and $\Omega^{\prime\prime}$ that is $(2\epsilon/\kappa)$-typical for $|0\rangle$ such that
\begin{equation}
\mu(\Omega) \geq (1-\kappa)\left( \mu(\Omega^\prime) + \mu(\Omega^{\prime\prime}) \right) - \epsilon.
\end{equation}

Start by defining the measurable function $g_\phi (\lambda) \eqdef \mathbb{P}_M(E_{\neg\psi}|\lambda) + \mathbb{P}_M(E_{\neg0}|\lambda)$. Define $g_0$ similarly. 

By assumption, $\int_\Lambda \d\nu(\lambda)\,g_\phi(\lambda) \geq 1 - \epsilon$ for all $\nu\in\Delta_{|\phi\rangle}$ so, by Lem.~\ref{lem:MR:meas-function-to-subset}, $\k_\kappa(g_\phi)$ is $(\epsilon/\kappa)$-typical for $|\phi\rangle$. Similarly, $\k_\kappa(g_0)$ is $(\epsilon/\kappa)$-typical for $|0\rangle$.

So consider any $\Omega\in\Sigma$ that is $(\epsilon/\kappa)$-typical for both $|\phi\rangle$ and $|0\rangle$. For every $\nu\in\Delta_{|\phi\rangle}$ then
\begin{equation}
\nu(\Omega \cap \k_\kappa(g_\phi)) \geq \nu(\Omega) + \nu(\k_\kappa(g_\phi)) - 1 > 1 - \frac{2\epsilon}{\kappa}.
\end{equation}
Therefore $\Omega^\prime \eqdef \Omega \cap \k_\kappa(g_\phi)$ is $(2\epsilon/\kappa)$-typical for $|\phi\rangle$. By a similar argument $\Omega^{\prime\prime} \eqdef \Omega \cap \k_\kappa ( g_0 )$ is $(2\epsilon/\kappa)$-typical for $|0\rangle$.

Finally, it therefore follows that
\begin{eqnarray}
\mu(\Omega) &= & \int_\Omega \d\mu(\lambda) \left( g_\phi(\lambda) + g_0(\lambda) - \mathbb{P}_M(E_{\neg\psi}\,|\,\lambda) \right) \\
&\geq & \int_{\Omega^\prime} \d\mu(\lambda)\,g_\phi(\lambda) + \int_{\Omega^{\prime\prime}} \d\mu(\lambda)\,g_0(\lambda) - \epsilon \\
&\geq & (1-\kappa)\left( \mu(\Omega^\prime) + \mu(\Omega^{\prime\prime}) \right) - \epsilon.
\end{eqnarray}
The second line follows from the upper bound on $\int_\Lambda \d\mu(\lambda)\,\mathbb{P}_M(E_{\neg\psi}|\lambda)$ and by restricting the range of the other integrals. The third line follows by noting that $g_\phi$ and $g_0$ are lower-bounded in $\Omega^\prime$ and $\Omega^{\prime\prime}$ respectively. This completes the proof.
\end{proof}

\subsection{Error-Tolerant Results} \label{sec:MR:error-tolerance-results}

With the technical background of $\epsilon$-asymmetric overlaps, the proof of Thm.~\ref{thm:SO:no-maximal-overlap-for-basis} can be adapted to obtain an error-tolerant variant. Since the proofs of Thm.~\ref{thm:SO:superpositions-are-ontic} and Thm.~\ref{thm:MR:no-ESMR-EMMR} use Thm.~\ref{thm:SO:no-maximal-overlap-for-basis}, this will provide an error-tolerant argument against both epistemic superpositions and ESMR/EMMR macro-realism in quantum theory.

\begin{theorem}
\label{thm:MR:no-large-overlap-for-basis}
Consider a $d>3$ dimensional quantum system described by some ontological model. Assume one can experimentally demonstrate quantum probabilities to within some $\pm\epsilon\in(0,1]$ as in Eq.~(\ref{eq:IN:ontological-model-approx-quantum-probability}). Let $\mathcal{B}$ be any orthonormal basis of $\mathcal{H}$ and $|\psi\rangle\in\mathcal{P}(\mathcal{H})$ be any pure state such that $|\langle 0|\psi\rangle| = \alpha \in (0,\frac{1}{\sqrt{2}})$ for some $|0\rangle\in\mathcal{B}$. For any choices of $0 < \tau \leq \eta < \frac{1}{16}$ there exist choices of $\epsilon > 0$ such that
\begin{equation}
\varpi_\eta (\mathcal{B}\,|\,\mu) \leq 1 - \tau
\end{equation}
for some $\mu\in\Delta_{|\psi\rangle}$, for a finite range of $\alpha\in(0,\frac{1}{\sqrt{2}})$.
\end{theorem}

\begin{proof}
The proof proceeds by contradiction. To this end, assume that 
\begin{equation} \label{eq:MR:proof-no-large-overlap-contradiction-assumption}
\varpi_\eta (\mathcal{B}\,|\,\mu) > 1 - \tau
\end{equation}
for all $\mu\in\Delta_{|\psi\rangle}$ where $0<\tau\leq\eta<\frac{1}{16}$.

Consider the construction of $|\psi\rangle,|\phi\rangle,|0\rangle$, $\mathcal{B}^\prime\ni|0\rangle$, and $U$ used in the proof of Thm.~\ref{thm:SO:no-maximal-overlap-for-basis} [Eqs.~(\ref{eq:SO:proof-psi-no-maximal-overlap}, \ref{eq:SO:proof-phi-no-maximal-overlap})]. Let $M_A = \{E_{\neg\psi},E_{\neg\phi},E_{\neg 0}\}$ be a quantum anti-distinguishing measurement for $\{|\psi\rangle,|\phi\rangle,|0\rangle\}$ (\emph{i.e.} the measurement that would be anti-distinguishing in the absence of error), $M$ be a basis measurement for $\mathcal{B}$, and let $M^\prime$ be a basis measurement for $\mathcal{B}^\prime$. Recall that $\langle0|\psi\rangle = \alpha \in (0,\frac{1}{\sqrt{2}})$.

The assumption that quantum probabilities are reproduced to within $\pm\epsilon$ variously implies the following.
\begin{eqnarray}
\mathbb{P}_{M_A}(E_{\neg\psi}\,|\,\mu) \leq & \epsilon, \quad & \forall\mu\in\Delta_{|\psi\rangle} \\
\mathbb{P}_{M_A}(E_{\neg\phi}\,|\,\nu) \leq & \epsilon, \quad & \forall\nu\in\Delta_{|\phi\rangle} \\
\mathbb{P}_{M_A}(E_{\neg0}\,|\,\chi) \leq & \epsilon, \quad & \forall\chi\in\Delta_{|0\rangle} \\
\mathbb{P}_{M^\prime}(|0\rangle\vee|1^\prime\rangle\,|\,\chi) \geq & 1 - \epsilon, \quad & \forall\chi\in\Delta_{|0\rangle} \\
\mathbb{P}_{M^\prime}(|0\rangle\vee|1^\prime\rangle\vee|3^\prime\rangle\,|\,\nu) \geq & 1 - \epsilon, \quad & \forall\nu\in\Delta_{|\phi\rangle} \\
\mathbb{P}_{M^\prime}(|3^\prime\rangle\,|\,\mu) \leq & \epsilon, \quad & \forall\mu\in\Delta_{|\psi\rangle} \label{eq:MR:proof-no-large-overlap-measurement-psi-3} \\
\mathbb{P}_M(|i\rangle\,|\,\chi_i) \geq & 1 - \epsilon, \quad & \forall\chi_i\in\Delta_{|i\rangle},\;\forall|i\rangle\in\mathcal{B}
\end{eqnarray}

Now proceed to derive the contradiction. Choosing any $\mu\in\Delta_{|\psi\rangle}$, Lem.~\ref{lem:MR:antisymmetric-asymmetric-overlap} implies that
\begin{equation} \label{eq:MR:proof-no-large-overlap-antisymmetric-bit}
\varpi_{\epsilon/\kappa}(|0\rangle,|\phi\rangle\,|\,\mu) \geq (1 - \kappa)\left( \varpi_{2\epsilon/\kappa}(|\phi\rangle\,|\,\mu) + \varpi_{2\epsilon/\kappa}(|0\rangle\,|\,\mu) \right) - \epsilon
\end{equation}
for any $\kappa\in(2\epsilon,1)$. Let $\delta\in[2\epsilon/\kappa,1)$, then Lem.~\ref{lem:MR:e-asymmetric-less-restrictive} gives $\varpi_{2\epsilon/\kappa}(|0\rangle\,|\,\mu) \geq \varpi_\delta(|0\rangle\,|\,\mu)$. This, together with Lem.~\ref{lem:MR:e-asymmetric-transformation}, gives
\begin{equation} \label{eq:MR:proof-no-large-overlap-transformation-bit}
\varpi_{\epsilon/\kappa}(|0\rangle,|\phi\rangle\,|\,\mu) \geq (1 - \kappa)\left( \left( 1 - \frac{2\epsilon}{\kappa\delta} \right)\varpi_{\delta}(|0\rangle\,|\,\mu^\prime) + \varpi_{\delta}(|0\rangle\,|\,\mu) \right) - \epsilon
\end{equation}
where $\mu^\prime \transto{\gamma} \mu$ for some choice of $\gamma\in\Gamma_U$. Note that for the first term to contribute non-trivially, it is required that $\delta > 2\epsilon/\kappa$.

The next step is to relate the $\varpi_\delta$ terms on the right hand side to their quantum probabilities. For any $\bar{\mu}\in\Delta_{|\psi\rangle}$, Lem.~\ref{lem:MR:multi-partite-e-asymmetric-breakup} and Eq.~(\ref{eq:MR:proof-no-large-overlap-contradiction-assumption}) imply
\begin{equation}
\sum_{|i\rangle\in\mathcal{B}} \varpi_\eta(|i\rangle\,|\,\bar{\mu}) \geq \varpi_\eta(\mathcal{B}\,|\,\bar{\mu}) > 1 - \tau.
\end{equation}
Using Lem.~\ref{lem:MR:e-asymmetric-to-probability} and recalling that $|0\rangle\in\mathcal{B}$ this further implies that
\begin{equation}
\varpi_\eta(|0\rangle\,|\,\bar{\mu}) > 1 - \tau - \frac{\sum_{|i\rangle\neq|0\rangle}\mathbb{P}_M (|i\rangle\,|\,\bar{\mu})}{1 - \epsilon/\eta}
\end{equation}
so long as $\eta > \epsilon$. Lastly, noting that $\alpha^2 - \epsilon \leq \mathbb{P}_M(|0\rangle\,|\,\bar{\mu}) = 1 - \sum_{|i\rangle\neq|0\rangle}\mathbb{P}_M (|i\rangle\,|\,\bar{\mu})$ gives
\begin{equation} \label{eq:MR:proof-no-large-overlap-single-measurement-bit}
\varpi_\eta(|0\rangle\,|\,\bar{\mu}) > \frac{ \alpha^2 - \epsilon - \tau - \frac{\epsilon}{\eta}(1 - \tau) }{ 1 - \epsilon/\eta }
\end{equation}
for all $\bar{\mu}\in\Delta_{|\psi\rangle}$. In particular, this holds both when $\bar{\mu}=\mu$ and $\bar{\mu}=\mu^\prime$ from Eq.~(\ref{eq:MR:proof-no-large-overlap-transformation-bit}).

The final step before putting everything together is to relate the bipartite $(\epsilon/\kappa)$-asymmetric overlap (left hand side of Eq.~(\ref{eq:MR:proof-no-large-overlap-transformation-bit})) to a corresponding quantum measurement probability. Let $f(\lambda) \eqdef \mathbb{P}_{M^\prime}(|0\rangle\,|\,\lambda) + \mathbb{P}_{M^\prime}(|1^\prime\rangle\,|\,\lambda)$ and $g(\lambda) \eqdef \mathbb{P}_{M^\prime}(|3^\prime\rangle\,|\,\lambda)$. From the above measurement probabilities, Lem.~\ref{lem:MR:meas-function-to-subset} gives that
\begin{eqnarray}
\chi(\k_\kappa(f)) > & 1 - \frac{\epsilon}{\kappa}, \quad & \forall\chi\in\Delta_{|0\rangle} \\
\nu(\k_\kappa(f+g)) > & 1 - \frac{\epsilon}{\kappa}, \quad & \forall\nu\in\Delta_{|\phi\rangle} 
\end{eqnarray}
since it is already required that $\kappa>2\epsilon>\epsilon$. Because $\k_\kappa(f) \subseteq \k_\kappa(f+g)$ it follows that $\Omega\eqdef \k_\kappa(f+g)$ is a measurable subset of $\Lambda$ which is $(\epsilon/\kappa)$-typical for both $|\phi\rangle$ and $|0\rangle$. Consider then
\begin{equation}
\mathbb{P}_{M^\prime}(|0\rangle\vee|1^\prime\rangle\,|\,\mu) = \int_\Lambda \d\mu(\lambda)f(\lambda).
\end{equation}
By Eq.~(\ref{eq:MR:proof-no-large-overlap-measurement-psi-3}) it follows that $\int_\Lambda \d\mu(\lambda)\,g(\lambda) \leq \epsilon$ so
\begin{eqnarray}
\mathbb{P}_{M^\prime}(|0\rangle\vee|1^\prime\rangle\,|\,\mu) & \geq & \int_\Lambda \d\mu(\lambda)\left( f(\lambda) + g(\lambda) \right) - \epsilon \\
 & \geq & \int_{\Omega} \d\mu(\lambda)\left( f(\lambda) + g(\lambda) \right) - \epsilon \\
 & \geq & ( 1 - \sigma )\mu(\Omega) - \epsilon \\
 & \geq & ( 1 - \sigma )\varpi_{\epsilon/\kappa}(|\phi\rangle,|0\rangle\,|\,\mu) - \epsilon.
\end{eqnarray}
Where the final line follows as $\Omega$ is $(\epsilon/\kappa)$-typical for both $|\phi\rangle$ and $|0\rangle$. Using that quantum probabilities are reproduced to within $\pm\epsilon$ and noting the construction of Eq.~(\ref{eq:SO:proof-psi-no-maximal-overlap}) gives
\begin{equation} \label{eq:MR:proof-no-large-overlap-multi-measurement-bit}
\alpha^2( 1 + 2\alpha^2 ) + 2\epsilon \geq (1 - \kappa)\varpi_{\epsilon/\kappa}(|0\rangle,|\phi\rangle\,|\,\mu).
\end{equation}

Combining Eqs.~(\ref{eq:MR:proof-no-large-overlap-antisymmetric-bit}, \ref{eq:MR:proof-no-large-overlap-single-measurement-bit}, \ref{eq:MR:proof-no-large-overlap-multi-measurement-bit}) shows that the assumption-towards-contradiction Eq.~(\ref{eq:MR:proof-no-large-overlap-contradiction-assumption}) implies
\begin{equation} \label{eq:MR:proof-no-large-overlap-full-inequality}
\alpha^2( 1 + 2\alpha^2 ) + \epsilon( 3 - \kappa ) > \frac{2(1 - \kappa)^2(1 - \frac{\epsilon}{\kappa\eta})}{1 - \epsilon/\eta}\left( \alpha^2 - \epsilon - \tau - \frac{\epsilon}{\eta}(1 - \tau) \right)
\end{equation}
where $\delta = \eta$ has been taken. Recall that the derivation of this inequality required that $\eta > 2\epsilon/\kappa$ and $\eta > \epsilon$, and that $\kappa > 2\epsilon$ is chosen.

To obtain a contradiction and complete the proof, it must be demonstrated that the inequality of Eq.~(\ref{eq:MR:proof-no-large-overlap-full-inequality}) is violated for appropriate choices of parameters. This is much easier if one first simplifies Eq.~(\ref{eq:MR:proof-no-large-overlap-full-inequality}).

To simplify, first assume that $\eta \geq \tau$, so that Eq.~(\ref{eq:MR:proof-no-large-overlap-contradiction-assumption}) implies $\varpi_\eta(\mathcal{B}\,|\,\mu) > 1-\eta$. Therefore, Eq.~(\ref{eq:MR:proof-no-large-overlap-full-inequality}) holds with $\tau=\eta$. Another simplifying assumption comes by taking $\kappa = \sqrt{\epsilon}$. Under these assumptions, Eq.~(\ref{eq:MR:proof-no-large-overlap-full-inequality}) becomes
\begin{equation}
\alpha^2( 1 + 2\alpha^2 ) + \kappa^2( 3 - \kappa ) > \frac{2(1 - \kappa)^2(\eta - \kappa)}{\eta - \kappa^2}\left( \alpha^2 - \eta - \frac{\kappa^2}{\eta} \right)
\end{equation}
where $\eta > 2\kappa$, $\eta > \kappa^2$ and $\kappa < \frac{1}{2}$.

Observe that for any fixed value of $\eta$, a value of $\epsilon = \kappa^2$ may be chosen to be arbitrarily small. In the limit of $\epsilon\rightarrow 0$ the inequality becomes
\begin{equation}
\alpha^2( 2\alpha^2 - 1 ) + 2\eta > 0,
\end{equation}
which is violated by a finite range of $\alpha^2\in(0,\frac{1}{2})$ for every value of $\eta\in(0,\frac{1}{16})$.

Consider turning this logic around. For any value of $\eta\in(0,\frac{1}{16})$, if one assumes Eq.~(\ref{eq:MR:proof-no-large-overlap-contradiction-assumption}) with $\eta \geq \tau$ then there exists a value of $\epsilon > 0$ such that one can reach a contradiction by violating Eq.~(\ref{eq:MR:proof-no-large-overlap-full-inequality}). This completes the proof.
\end{proof}

As stated, Thm.~\ref{thm:MR:no-large-overlap-for-basis} is somewhat formal and opaque. A little unpacking is required to see how it provides error-tolerant statements against both epistemic superpositions and ESMR/EMMR macro-realism.

First, consider Bob who believes that some superposition $|\psi\rangle\in\mathcal{P}(\mathcal{H})$ should be epistemic with respect to some orthonormal basis $\mathcal{B}\not\ni|\psi\rangle$ of a $d > 3$ dimensional system. As discussed in Sec.~\ref{sec:MR:error-tolerant-superpositions-and-mr}, Bob should be compelled to make a statement along the lines of: ``any preparation of $|\psi\rangle$ will produce an ontic state that is $\eta$-typical for at least one state in $\mathcal{B}$ with probability at least $1 - \tau$''. The values of $\eta$ and $\tau$ are left up to Bob and quantify exactly how macro-realist his belief is (relative to the noise in available experiments), but they should be small if $|\psi\rangle$ is to be approximately epistemic. In summary, Bob believes that
\begin{equation}
\varpi_\eta (\mathcal{B}\,|\,\mu) > 1 - \tau, \quad \forall\mu\in\Delta_{|\psi\rangle}.
\end{equation}

Alice, who has access to a lab, can contradict Bob's belief in the following way. Suppose that $0 < \tau \leq \eta < \frac{1}{16}$ and Alice can demonstrate that quantum probabilities are correct to within $\pm\epsilon\in(0,1)$. By Thm.~\ref{thm:MR:no-large-overlap-for-basis} then for sufficiently small $\epsilon$ there is some continuous region $\mathcal{R}_\epsilon \subset (0,1)$ such that Bob is contradicted if $|\langle 0|\psi\rangle|^2 \in \mathcal{R}_\epsilon$ for any $|0\rangle\in\mathcal{B}$.

In fact, given specific values for $\eta$, $\tau$, and $\alpha$, Alice can probably do better. Looking at the proof of Thm.~\ref{thm:MR:no-large-overlap-for-basis}, Alice simply needs to find an appropriate $\epsilon$ such that Eq.~(\ref{eq:MR:proof-no-large-overlap-full-inequality}) can be violated (with valid choices for the other parameters in the inequality) and then demonstrate that quantum probabilities hold to within $\pm\epsilon$. Note that, in this more general case, Alice doesn't necessarily have to assume that $\tau \leq \eta$.

So why bother with the first, less general, statement if violation of Eq.~(\ref{eq:MR:proof-no-large-overlap-full-inequality}) is the more powerful result? Simply put, it is not clear \emph{a priori} that Eq.~(\ref{eq:MR:proof-no-large-overlap-full-inequality}) can be violated for reasonable choices of the parameters. The first statement serves to demonstrate that the contradiction is possible, while the second is strictly more general.

Theorem~\ref{thm:MR:no-large-overlap-for-basis} can be used in a similar way to rule out ESMR or EMMR ontologies in an error-tolerant way.

Consider Clare, who believes in ESMR or EMMR macro-realism for some quantity $Q$ of a quantum system with $n > 3$ distinguishable values. Similarly to Bob---and as discussed in Sec.~\ref{sec:MR:error-tolerant-superpositions-and-mr}---Clare should be compelled to say ``any preparation of any $|\psi\rangle\in\mathcal{P}(\mathcal{H})$ will produce an ontic state that is $\eta$-typical for at least one operational eigenstate of $Q$ with probability at least $1 - \tau$.'' As with Bob, the exact values of $\eta$ and $\tau$ are up to Clare, but should be small if Clare seriously believes in ESMR/EMMR macro-realism.

As noted in Sec.~\ref{sec:MR:quantum-MR}, the operational eigenstates of $Q$ in a quantum system form an orthonormal basis $\mathcal{B}_Q$. This means that Alice believes that
\begin{equation} \label{eq:MR:Clares-belief}
\varpi_\eta (\mathcal{B}_Q\,|\,\mu) > 1 - \tau, \quad \forall\mu\in\Delta_{|\psi\rangle},\;\forall|\psi\rangle\in\mathcal{P}(\mathcal{H}).
\end{equation}

It is clear that Alice can contradict Clare's belief in a similar way to how she contradicted Bob. Specifically, if $0 < \tau \leq \eta < \frac{1}{16}$ then there is an $\epsilon\in(0,1)$ such that if Alice can demonstrate that quantum probabilities hold to within $\pm\epsilon$ then Clare is proved wrong by Thm.~\ref{thm:MR:no-large-overlap-for-basis}. More generally given $\tau$ and $\eta$, if Alice can violate Eq.~(\ref{eq:MR:proof-no-large-overlap-full-inequality}) for any valid values of $\alpha$, $\kappa$, and $\epsilon$, then demonstrating that quantum probabilities hold to within $\pm\epsilon$ will contradict Clare.

Note that, since Clare must believe Eq.~(\ref{eq:MR:Clares-belief}) for all pure states $|\psi\rangle\in\mathcal{P}(\mathcal{H})$, it is easier for Alice to contradict Clare than Bob, whose choice of $|\psi\rangle$ is arbitrary.

In this way, Thm.~\ref{thm:MR:no-large-overlap-for-basis} can be used as the basis for error-tolerant variations on both Thm.~\ref{thm:SO:superpositions-are-ontic} and Thm.~\ref{thm:MR:no-ESMR-EMMR}. This has been achieved by generalising the asymmetric overlap to the $\epsilon$-asymmetric overlap in a natural way by using $\epsilon$-typical subsets. The proof (while rather involved) is still arguably much simpler than that of Thm.~\ref{thm:SO:small-symmetric-overlap-large-d}, which established an error-tolerant version of Thm.~\ref{thm:SO:no-overlap-large-d} by using the symmetric overlap. This is because the asymmetric overlap has much more in common with the $\epsilon$-asymmetric overlap than the symmetric overlap, so the proof strategy can more more easily adapted. The main disadvantage of working with the $\epsilon$-asymmetric overlap is a proliferation of free parameters, so rather than getting a clean contradiction, one finds an inequality Eq.~(\ref{eq:MR:proof-no-large-overlap-full-inequality}) with five parameters which must be violated to obtain a contradiction. The proof of Thm.~\ref{thm:MR:no-large-overlap-for-basis} demonstrates that this is possible at least in some regime---specifically, when  $0 < \tau \leq \eta < \frac{1}{16}$ for some ranges of $\epsilon$ and $\alpha$---but violations of Eq.~(\ref{eq:MR:proof-no-large-overlap-full-inequality}) may also be possible outside of that regime.

\section{Summary and Discussion} \label{sec:MR:summary-and-discussion}

This chapter addressed the incompatibility of macro-realism and quantum theory by using ontological models and building on the work and techniques of Chap.~\ref{ch:SO}. After a quick summary of this chapter, the results and outstanding questions from these two chapters will be discussed below.

The exact meaning of ``macro-realism'' has always been somewhat controversial. It was therefore appropriate to start by discussing what macro-realism means and identifying an appropriate specific definition in Sec.~\ref{sec:MR:intro}. The history of macro-realism is inseparable from the Leggett-Garg inequalities and their use in the Leggett-Garg argument for the incompatibility of quantum theory and macro-realism. This argument was outlined briefly in Sec.~\ref{sec:MR:LGI-outline}, with a focus on the assumptions required to derive the inequalities.

The definition of macro-realism used in Sec.~\ref{sec:MR:definition} is from Ref.~\cite{MaroneyTimpson17}, where the history of ``macro-realism'' is analysed and the definition clarified using ontological models. This ontological model analysis led to the identification of three sub-categories of macro-realism presented in Sec.~\ref{sec:MR:sub-classes}: EMMR, ESMR, and SSMR. By applying these definitions to quantum theory, one finds fundamental loopholes in the Leggett-Garg argument. In particular, Leggett-Garg is only able to prove incompatibility of EMMR models with quantum theory, leaving loopholes for ESMR and SSMR as discussed in Sec.~\ref{sec:MR:loopholes}. Moreover, the existence of counter-examples in the form of Bohmian mechanics and the Kochen-Specker model show that no theorem can rule out all SSMR models or all ESMR models for $d=2$ dimensions respectively.

Papers on the Leggett-Garg argument, including those addressing the clumsiness loophole \cite{WildeMizel12,KneeKakuyanagi+16,HuffmanMizel16}, have concentrated on $d=2$ dimensional systems. As a result of closely following the Leggett-Garg assumptions, they are still unable to rule out any models outside of EMMR.

With this background, Sec.~\ref{sec:MR:main-theorem} proceeded to apply the methods of Chap.~\ref{ch:SO} to macro-realism, obtaining Thm.~\ref{thm:MR:no-ESMR-EMMR} as the main result of this chapter. This proved the incompatibility of quantum theory with all ESMR and EMMR models for observables with $n>3$ distinguishable values---a result which is therefore stronger than the Leggett-Garg argument. 

Finally, Sec.~\ref{sec:MR:error-tolerant-argument} approached the problem of finding error-tolerant variations of Thms.~\ref{thm:SO:superpositions-are-ontic}, \ref{thm:MR:no-ESMR-EMMR} with a view towards experimental investigations. Any such variation requires a generalisation of the mathematical formalisations of epistemic superpositions and ESMR/EMMR macro-realism. This was done by introducing the $\epsilon$-asymmetric overlap: a generalisation of the asymmetric overlap that is error-tolerant. This enabled a proof of Thm.~\ref{thm:MR:no-large-overlap-for-basis}, which can be used to obtain error-tolerant variations of both Thm.~\ref{thm:SO:superpositions-are-ontic} and Thm.~\ref{thm:MR:no-ESMR-EMMR}.

Between Thm.~\ref{thm:MR:no-ESMR-EMMR} and the Leggett-Garg argument, the only possibilities for macro-realism compatible with quantum theory that remain are SSMR models (such as Bohmian mechanics) or ESMR models for $d=2,3$ dimensions (such as the Kochen-Specker model). One clear question for further investigation is whether it is possible to rule out any subsets of these remaining models. For example, Bohmian mechanics is a $\psi$-ontic theory and it may therefore be possible prove the incompatibility of quantum theory and all SSMR ontologies that aren't $\psi$-ontic. This, together with the result presented here, would essentially say that to be macro-realist you must have an ontology consisting of the full quantum state plus extra information. Many would consider this a very strong argument against macro-realism.

Theorems~\ref{thm:SO:superpositions-are-ontic}--\ref{thm:SO:no-overlap-large-d} paint a similarly grim outlook for the epistemic realist, especially one who seeks to explain any superposition-based phenomena using ontological uncertainty. In particular, Thm.~\ref{thm:SO:superpositions-are-ontic} proves that, for $d>3$, almost all superpositions defined with respect to any given basis $\mathcal{B}$ must be real. Therefore, any epistemic realist account of quantum theory must include ontic features corresponding to superposition states and the unfortunate cat cannot be put out of its misery. Theorems~\ref{thm:SO:no-maximally-epistemic}, \ref{thm:SO:no-overlap-large-d} then proceed to provide new bound on ontic overlaps for largely arbitrary pairs of quantum states. In particular from Thm.~\ref{thm:SO:no-maximally-epistemic} one can conclude that almost no quantum states can be $\psi$-epistemic. From Thm.~\ref{thm:SO:no-overlap-large-d}, one can conclude that in any moderately large system a large number of pairs of non-orthogonal states cannot overlap significantly, making it unlikely that such overlaps can satisfactorily explain quantum features. These are both much stronger statements than previous results in the same vein have been able to achieve, circumventing the shortcomings noted in Sec.~\ref{sec:MR:loopholes}. However, it should be noted that Thms.~\ref{thm:SO:no-maximally-epistemic}, \ref{thm:SO:no-overlap-large-d} do depend on an extra assumption, though one that is mild.

The next stage for all of these results will be to develop experimental tests, which require detailed error-tolerant analyses. Theorem~\ref{thm:SO:small-symmetric-overlap-large-d} begins to develop such an analysis based on Thm.~\ref{thm:SO:no-overlap-large-d}. An experiment based on this would require demonstration of small errors in probabilities for a wide range of measurements on a $d>5$ dimensional system.

The error-tolerant variations on Thms.~\ref{thm:SO:superpositions-are-ontic}, \ref{thm:MR:no-ESMR-EMMR} provided by Thm.~\ref{thm:MR:no-large-overlap-for-basis} also open the door for experimental investigation. While this variation uses the $\epsilon$-asymmetric overlap, there are potentially other routes to error-tolerant variations, including one currently in development \cite{HermensMaroney17}. This demonstrates how Thms.~\ref{thm:SO:superpositions-are-ontic}, \ref{thm:MR:no-ESMR-EMMR} are just the first steps towards new ways of investigating epistemic superpositions and macro-realism respectively.

It is interesting to note that experiments based on this result will be an entirely new avenue for tests of macro-realism. Experimental tests based on the Leggett-Garg argument will always have certain features and difficulties in common (such as the clumsiness loophole briefly mentioned in section \ref{sec:MR:loopholes}). However, since the approach of this chapter is so different in character one can expect the resulting experiments to be similarly different, hopefully avoiding many of the difficulties common to Leggett-Garg while requiring challenging new high-precision tests of quantum theory in $d>3$ Hilbert spaces.

One should note that in this chapter the ``macro'' quantity $Q$ was taken to correspond to a measurement in basis $\mathcal{B}_{Q}$ in the quantum case. A more general approach might allow $Q$ to correspond to a POVM measurement instead. That is, for each value $q$ of $Q$ there would be some POVM element $E_{q}$ and the operational eigenstates $|\psi\rangle$ of $q$ would be those satisfying $\langle\psi|E_{q}|\psi\rangle=1$. It does not seem unlikely that the results presented here could be fairly directly extended to such a case and this would be another interesting avenue for further work. Such an extension would likely add significant complexity to the proofs without changing the fundamental ideas, however.

Extensions of Thm.~\ref{thm:MR:no-ESMR-EMMR} along these lines would likely start to further separate out the work on macro-realism from that on epistemic superpositions. As noted in Sec.~\ref{sec:MR:relation-between-results}, superpositions are pure states, so extensions looking for new results there will likely stay within the realm of pure-state quantum theory. This is in contrast to macro-realism, which is defined independently from quantum theory, so extensions in that direction may well lead to mixed quantum theory and POVM measurements.

The methods used for the foundational results of Chap.~\ref{ch:SO} were also adapted to prove communication bounds in quantum information in Sec.~\ref{sec:SO:communication}. The result was Thm.~\ref{thm:SO:communication-bound} which provides a bound on the ability for classical resources to perfectly simulate quantum channels which asymptotically matches the best known results from the literature \cite{Montanaro16,Montina11b,BrassardCleve+99} while using a substantially simpler argument. In particular, it was proved that at least $2^{n_q + \mathcal{O}(1)} - 1$ bits of classical communication is required for the simulation of a noiseless $n_q$-qubit quantum channel, even when using arbitrary shared random data. Beyond its simplicity, the proof of Thm.~\ref{thm:SO:communication-bound} has two key advantages over other proofs for similar bounds. First, it identifies an underlying reason for the result: it is a consequence of the existence of exponentially large sets of quantum states that have small mutual Born-rule overlaps. Second, the proof identifies a general strategy for deriving such bounds. If one found different sets of fingerprinting states and/or classical error-correction codes one could simply re-run the proof to potentially produce more powerful results [Sec.~\ref{sec:SO:further-communication-bounds}].

It is worth noting that Thm.~\ref{thm:SO:communication-bound} was developed independently from the exponential bound for the same problem presented in Ref.~\cite{Montanaro16} and preceded the publication of that bound. However, as it stands the bound of Ref.~\cite{Montanaro16} has two advantages over Thm.~\ref{thm:SO:communication-bound}. First, it applies to approximate simulations as well as perfect simulations of quantum channels. Second, it also bounds simulations using two-way classical communication. The problem of extending the proof of Thm.~\ref{thm:SO:communication-bound} to apply to approximate simulations is a topic of ongoing research which is currently unready for publication but remains very promising nonetheless. Simulations using two-way communication are less related to ontological models and therefore of less relevance to this thesis, but it is an interesting open problem as to whether the techniques of Thm.~\ref{thm:SO:communication-bound} could also be applied to this more general setting.

As a final note on Thm.~\ref{thm:SO:communication-bound}, one can understand the simplicity and power of the proof as coming from the combination of two facts. First, the existence of exponentially large sets of quantum states with bounded Born-rule overlap discussed above. Second, the property of anti-distinguishability and its implications for ontological models and simulations [Sec.~\ref{sec:SO:anti-distinguishability}]. The result is the existence of exponentially large sets of quantum states such that all triple subsets are anti-distinguishable. It seems likely that this fact can be fruitfully used to prove results in many areas of quantum information, potentially including quantum cryptography and quantum algorithms.

Finally, it has been assumed throughout Chaps.~\ref{ch:SO}, \ref{ch:MR} that all quantum systems are finite-dimensional for clarity. If such extensions are needed, it should not be conceptually difficult to extend any of these results to the infinite case. In particular, Thm.~\ref{thm:SO:no-overlap-large-d} strongly suggests that in the infinite case one would find zero overlap between many pairs of quantum states, which would be a very powerful restriction on many types of ontological models.

\chapter{Quantum Conditional Independence and Quantum Common Causes} \label{ch:CI}

\section{The Need for Quantum Conditional Independence}

The ontology of quantum theory is normally discussed in terms of quantum states, as in Chaps.~\ref{ch:SO}, \ref{ch:MR}. However, scientists also typically consider causal influences to be real. If quantum theory challenges our notions of ontology of the states of systems, then it is natural to ask how it might also challenge our causal notions. The simplest scenario where this becomes apparent is in the case of complete common causes and in the related concept of conditional independence. This chapter analyses these fully in a quantum universe. Much of the work presented here has been published in Ref.~\cite{AllenBarrett+17}.

\subsection{The Reality of Causal Influences}

It is quite typical to think of causal influences as being ``real'' (without wishing to go into specifics as to what that might mean). Indeed, it is common to hear that the business of science is primarily to discover what these real causal influences are \cite{Pearl09,SpirtesGlymour+01}. In particular, the whole science of causal discovery algorithms is predicated on the idea that there exist objective causes that can be discovered. If even the possibility of ontological causal influences is taken seriously in a quantum universe, then there needs to be a framework in which to discuss quantum causal influences.

It is important, especially for this chapter and the next, to distinguish causal notions from probabilistic or statistical ones. Causation is about influences between events, while probability and statistics describe if and how often events (or combinations of events) occur.

\emph{Reichenbach's principle} \cite{Reichenbach56} might be viewed as a cornerstone of causal reasoning. It is often paraphrased as ``no correlation without causation''. Thus providing a link between the probabilistic/statistical idea of ``correlation'' and that of causation. Reichenbach's principle is so well ingrained in both everyday and scientific thought that it is rarely specifically discussed outside of philosophy.

To illustrate, if the dog barks every time the post arrives people naturally assume that it is the arrival of the post (the sound it makes, the smell of the deliverer, \emph{etc.}) that causes the dog to bark, rather than a mere coincidence. Similarly, if every time the cat runs away the dog is scolded it is likely that both events are caused by the dog barking rather than either a coincidence or one event being a cause for the other. More seriously, without some general commitment to Reichenbach's principle (or something sufficiently similar) it would be impossible to regard two experiments as independent due to a lack of causal mechanism between them. This would be catastrophic for science in general.

So Reichenbach's principle calls for causal links between any correlated events $A$ and $B$. What are these links? They must be at least one of the following: $A$ is a cause for $B$, $B$ is a cause for $A$, or there is some common cause for both $A$ and $B$. The last of these is the most interesting. Reichenbach's principle requires a quantitative restriction on the probabilities if the link is solely via a common cause: conditional independence. A full statement and discussion of Reichenbach's principle is deferred until Sec.~\ref{sec:CI:Reichenbach}.

\subsection{Causation in Bell Experiments} \label{sec:CI:introduction-Bell}

Despite the central role of causal explanations in science and life, certain quantum experiments seem to elude such explanations. The best known of these are Bell experiments---\emph{viz.} experiments of the type seen in Bell's theorem \cite{Bell66}. In such experiments, Alice prepares a pair of systems that are distributed to Bob and Clare who perform independent measurements on those systems. These experiments are purposefully constructed to suggest the only reasonable causal link between the measurement outcomes of Bob and Clare is via a common cause (presumably involving Alice). Bell's theorem proceeds to show that certain quantum experiments of this form can produce statistics which violate Reichenbach's principle for any conceivable common cause. The findings of Bell's theorem have been thoroughly experimentally verified \cite{HensenBernien+15,ShalmMeyer-Scott+15,GiustinaVersteegh+15}.

Bell's theorem is most often discussed in terms of non-locality, however recent interpretations have put its causal implications front-and-centre \cite{WoodSpekkens15,CavalcantiLal14,HensonLal+14,PienaarBrukner15}. In particular, Ref.~\cite{WoodSpekkens15} considers the statistics observed in Bell experiments in terms of classical causal models. It is shown that no classical causal model---even those with exotic properties such as retrocausality and superluminal signalling---can account for these statistics without undesirable fine tuning. Fine tuning is typically considered enough to rule out a proposed causal explanation \cite{Pearl09}. Note that, at least from a causal perspective, this is strictly stronger than Bell's theorem. Bell's theorem assumes that all causal links except common causes are ruled out by the experimental structure, while in Ref.~\cite{WoodSpekkens15} all causal structures are considered.

The conclusion must be that quantum theory demands a revision of our ideas of causality---either to allow and explain fine tuning or to revise what is permitted as a common cause. At very least, Bell experiments demand a causal explanation. It has been suggested \cite{CavalcantiLal14,WoodSpekkens15} that an appropriate revision could preserve a form of Reichenbach's principle while rejecting fine tuning.

In this chapter, such a revised quantum Reichenbach's principle is motivated and presented. Specifically, this chapter mostly deals with the case where quantum systems $B$ and $C$ are in the causal future of system $A$. The aim is to characterise the quantum channels from $A$ to $BC$ that can occur when $A$ is a complete common cause for $B$ and $C$. Since Reichenbach's principle uses conditional independence to characterise common causes [Sec.~\ref{sec:CI:Reichenbach}], this is done by generalising to a natural quantum conditional independence. Because of this limited scope, the quantum conditional independence thus defined only needs to make sense when $A$ is in the causal past of $BC$. Compare this with classical conditional independence, which can hold between random variables with any causal relationships. Quantum conditional independence can be defined in several equivalent ways, each of which naturally generalises a corresponding expression for classical conditional independence. Moreover, both quantum conditional independence and quantum Reichenbach presented here reduce to the classical cases in appropriate limits. These properties strongly suggest that the definitions here are the correct way of generalising these classical concepts to natural quantum analogues.

The equivalence of four ways to define quantum conditional independence will be proved in Sec.~\ref{sec:CI:proving-QCI-conditions}. A generalisation to channels from $A$ to $k>2$ outputs will then be given in Sec.~\ref{sec:CI:k-outputs-generalisation}.

Of course, there is much more to the classical Reichenbach's principle and classical conditional independence than channels from $A$ to $BC$. Indeed, there is a whole framework of classical causal models that generalises these concepts. The problem of making the corresponding generalisations for quantum Reichenbach's principle and quantum conditional independence will be tackled in Chap.~\ref{ch:CM}.

For clarity of presentation in this chapter and the next, all random variables and graphs will be taken to be finite and all quantum systems finite-dimensional. It is not anticipated that any conceptual changes would be needed to extend the results to the infinite cases.

\subsection{Metaphysically-Neutral Quantum Conditional Independence}

Naturally, the subject matter of this chapter and the next involves some discussion of metaphysical concepts. However, the results themselves can be viewed somewhat independently from the metaphysics. Indeed, an aim in writing these chapters is to make them as metaphysically-neutral as the results allow.

As in the rest of this thesis, it is suspected that the results presented here can inform discussions of causation and ontology in metaphysics, but this thesis does not pretend to contain those discussions. In particular, no metaphysical positions need to be subscribed to in order to derive and understand the results presented here.

For example, the above discussion used the assumption that causal influences are fundamentally real to motivate the work in this chapter. However, it is likely that a philosopher who argues against this position would still find value in the definition of quantum conditional independence and quantum common cause derived here. Even for that philosopher, Bell's theorem provides strong reasons for quantum common causes.

In short, the metaphysical content of this chapter and the next (while necessary) is always secondary to the physical and mathematical content.

\section{Classical Conditional Independence and Reichenbach's Principle} \label{sec:CI:Reichenbach}

\subsection{Reichenbach's Principle in Two Parts} \label{sec:CI:Reichenbach-bipartite}

Reichenbach's principle \cite{Reichenbach56} can be usefully thought of as comprising two parts (as identified in Ref.~\cite{CavalcantiLal14}). The first, the \emph{qualitative part}, is the claim ``no correlation without causation'' most usually associated with Reichenbach. This is supplemented by the \emph{quantitative part} which characterises common causes and provides the link to the probabilistic notion of conditional independence.
 
In full, the qualitative part of Reichenbach's principle is: if two physical variables $Y$ and $Z$ are statistically dependent, then there should be a causal explanation of this fact such that at least one of the following is true: 
\begin{enumerate}[label=(\alph*)]
\item $Y$ is a cause of $Z$;
\item $Z$ is a cause of $Y$; or
\item there is a third variable $X$ which is a common cause of both $Y$ and $Z$.
\end{enumerate}
These causal influences may be indirect (mediated by other variables) and each of these cases may be true many times over (multiple chains of causation or multiple common causes). If none of these causal influences exist, then $Y$ and $Z$ are \emph{ancestrally independent}: they have no causal ancestor in common. In this way, an alternative phrasing for the qualitative part is: ancestral independence of $Y$ and $Z$ implies statistical independence $\mathbb{P}(Y,Z) = \mathbb{P}(Y)\mathbb{P}(Z)$. 

The quantitative part of Reichenbach's principle applies to the case where $X$ and $Y$ share a common cause but neither is a direct cause of the other (that is, only case (c) from the qualitative part is true). $X$ is called the \emph{complete common cause} for $Y$ and $Z$ when $X$ is the union of all common causes and there is no cause (direct or indirect) from $Y$ to $Z$ or vice versa. The quantitative part states that in such a case $Y$ and $Z$ are conditionally independent given their complete common cause $X$, that is
\begin{equation}
\mathbb{P}(Y,Z|X) = \mathbb{P}(Y|X) \mathbb{P}(Z|X).
\end{equation}

Note that Reichenbach's principle therefore forms a two-way link between causal and statistical notions. First, if a statistical condition is met (statistical dependence), then at least one causal condition must be met (the qualitative part). Second, in the case of a particular causal condition being met (simple common cause), a further statistical condition must be met for consistency (conditional independence).

\subsection{Justifying the Quantitative Part} \label{sec:CI:justifying-classical-Reichenbach}

Unsurprisingly, justification for Reichenbach's principle is a delicate problem in philosophy. Basic questions from across the metaphysics of causation and probability play into this issue. Nonetheless, it is useful to present one way to justify the quantitative part given the qualitative part, as this can then be used to motivate its quantum counterpart. Really, this ``justification'' is just for illustrative and motivational purposes. It will involve temporarily making a quite strong metaphysical assumption. However, Reichenbach's principle (and its quantum generalisation in Sec.~\ref{sec:CI:quantum-Reichenbach}) stand apart from such assumptions.

For the sake of definiteness, suppose the position of a determinist is taken. So uncertainty is understood as arising from ignorance about deterministic dynamics or from initial conditions. The task is to use this assumption and the qualitative part of Reichenbach's principle to show that if $X$ is the complete common cause for $Y$ and $Z$ then $\mathbb{P}(Y,Z|X) = \mathbb{P}(Y|X)\mathbb{P}(Z|X)$.

A classical channel describing the influence of random variable $X$ on $Y$ is given by a probability distribution $\mathbb{P}(Y|X)$. The determinist will always view this as the result of a deterministic \emph{function} $f:X \times \Lambda \rightarrow Y$ for some unknown variable $\Lambda$. Any such channel can always be dilated in this way.

\begin{definition}
A classical deterministic \emph{dilation} of a classical channel $\mathbb{P}(Y|X)$ is given by some function $f:X\times\Lambda\rightarrow Y$ for random variable $\Lambda$ taking values $\lambda$ with probability distribution $\mathbb{P}(\Lambda)$ such that
\begin{equation}
\mathbb{P}(Y | X = x) = \sum_\lambda \mathbb{P}( \lambda ) \delta(Y,f(x,\lambda))
\end{equation}
where $\delta(Y,y^\prime) = 1$ for $Y=y^\prime$ and $0$ otherwise.
\end{definition}

Consider the case where $X$ is a complete common cause for $Y$ and $Z$, illustrated in Fig.~\ref{fig:CI:X-to-YZ}. Let $f = (f_Y, f_Z)$ be a classical dilation of $\mathbb{P}(Y,Z|X)$ for functions $f_Y : \Lambda\times X\rightarrow Y$ and $f_Z : X\times\Lambda\rightarrow Z$. Since $X$ is a \emph{complete} common cause, $\Lambda$ must split into a pair of variables $\Lambda = \Lambda_Y \times \Lambda_Z$ such that $\Lambda_Y$ only influences $Y$ and $\Lambda_Z$ only influences $Z$. If $\Lambda$ were not of this form then it would be a new common cause for $Y$ and $Z$ which is not screened through $X$, violating the assumption that $X$ is a complete common cause. Therefore, the dilation takes the form of $f_Y : \Lambda_Y \times X  \rightarrow Y$ and $f_Z : X \times \Lambda_Z \rightarrow Z$ such that
\begin{equation} \label{eq:CI:justifying-classical-Reichenbach-dilation}
\mathbb{P}(Y, Z | X = x) = \sum_{\lambda_Y, \lambda_Z} \mathbb{P}( \lambda_Y, \lambda_Z ) \delta(Y,f_Y(\lambda_Y,x)) \delta(Z,f_Z(x,\lambda_Z))
\end{equation}
as illustrated in Fig.~\ref{fig:CI:X-to-YZ-dilated}.

\begin{figure}
\begin{centering}
\includegraphics[scale=0.35,angle=0]{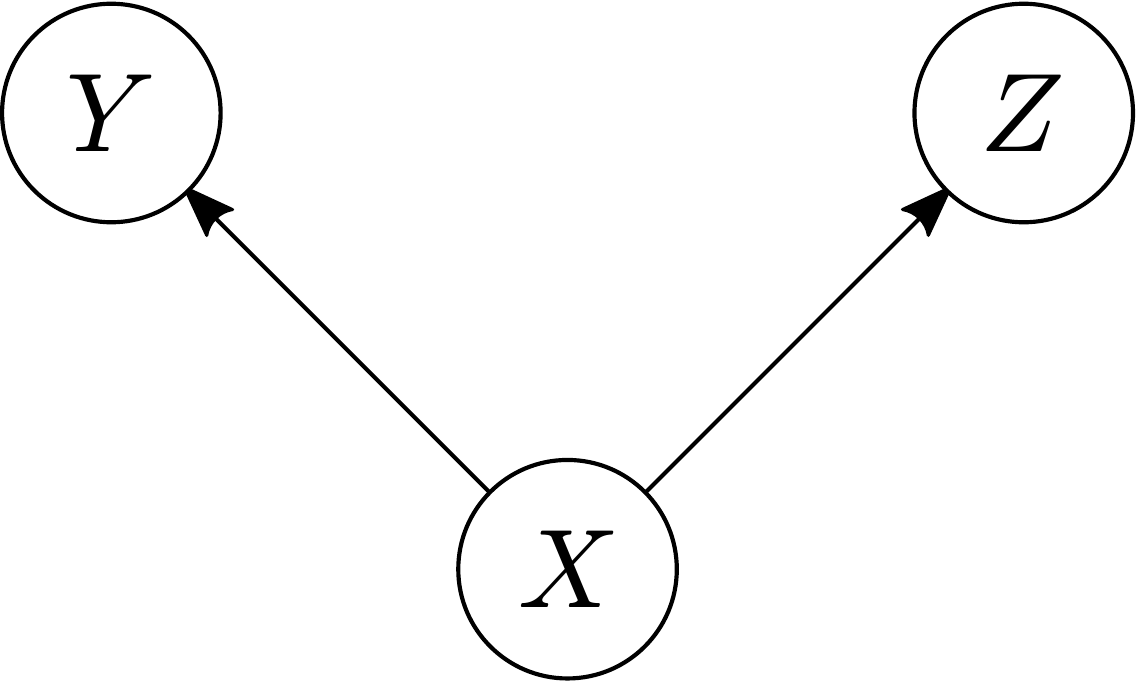}
\par\end{centering}
\protect\caption{A causal structure represented as a directed acyclic graph depicting that $X$ is the complete common cause of $Y$ and $Z$.}\label{fig:CI:X-to-YZ}
\end{figure}

\begin{figure}
\begin{centering}
\includegraphics[scale=0.35,angle=0]{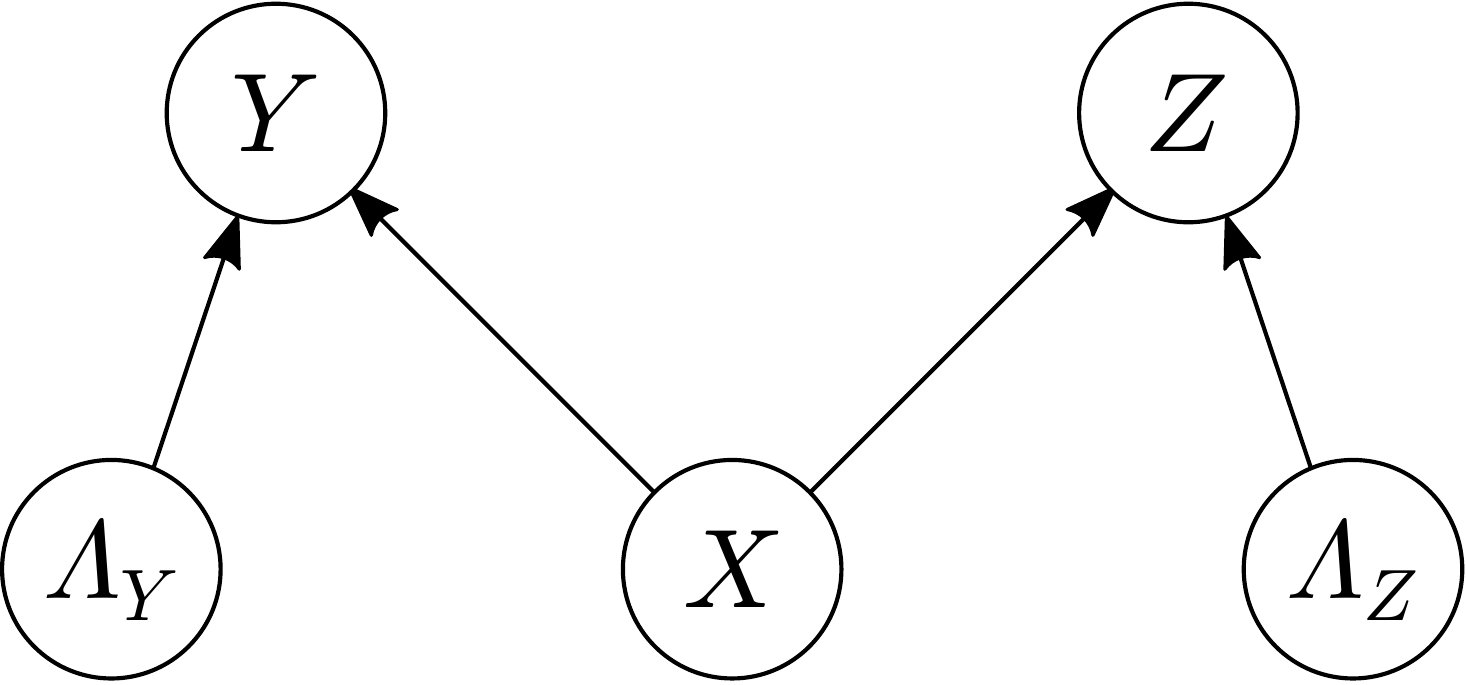}
\par\end{centering}
\protect\caption{The causal structure of Fig.~\ref{fig:CI:X-to-YZ}, expanded so that $Y$ and $Z$ each has a latent variable as a causal parent in addition to $X$ so that both $Y$ and $Z$ can be made to depend functionally on their parents.}\label{fig:CI:X-to-YZ-dilated}
\end{figure}

Using the qualitative part of Reichenbach's principle, ancestral independence of $\Lambda_Y$ and $\Lambda_Z$ implies their statistical independence $\mathbb{P}(\Lambda_Y,\Lambda_Z) = \mathbb{P}(\Lambda_Y)\mathbb{P}(\Lambda_Z)$. Equation~(\ref{eq:CI:justifying-classical-Reichenbach-dilation}) then immediately implies $\mathbb{P}(Y,Z|X) = \mathbb{P}(Y|X)\mathbb{P}(Z|X)$.

A well-known converse statement is also worth noting: any classical channel $\mathbb{P}(Y,Z|X)$ satisfying $\mathbb{P}(Y,Z|X) = \mathbb{P}(Y|X)\mathbb{P}(Z|X)$ admits a deterministic dilation where $X$ is the complete common cause of $Y$ and $Z$ \cite{Pearl09}.

This completes the illustrative justification for the quantitative part of Reichenbach's principle. Note that the mathematics of classical deterministic dilations is independent of the metaphysical assumptions. In fact, just by acknowledging the logical possibility of a deterministic dilation of $\mathbb{P}(Y,Z|X)$ one can say the following.

\begin{definition} \label{def:CI:compatible-classical-dilation}
A classical channel $\mathbb{P}(Y,Z|X)$ is \emph{compatible with $X$ being the deterministic common cause for $Y$ and $Z$} if and only if there is a deterministic dilation of the channel of the form
\begin{equation} \label{eq:CI:def:compatible-classical-dilation}
\mathbb{P}(Y, Z | X = x) = \sum_{\lambda_Y, \lambda_Z} \mathbb{P}( \lambda_Y ) \mathbb{P}( \lambda_Z ) \delta(Y,f_Y(\lambda_Y,x)) \delta(Z,f_Z(x,\lambda_Z)).
\end{equation}
\end{definition}

With this definition, the following can immediately be stated.

\begin{theorem} \label{thm:CI:classical-compatibility}
Given a classical channel $\mathbb{P}(Y,Z|X)$, the following are equivalent:
\begin{enumerate}
\item $\mathbb{P}(Y,Z|X)$ is compatible with $X$ being the deterministic common cause for $Y$ and $Z$.
\item $\mathbb{P}(Y,Z|X) = \mathbb{P}(Y|X)\mathbb{P}(Z|X)$.
\end{enumerate}
\end{theorem}

That (1) $\Rightarrow$ (2) follows immediately from Eq.~(\ref{eq:CI:def:compatible-classical-dilation}). That (2) $\Rightarrow$ (1) is proved in Ref.~\cite{Pearl09}, as noted above.

The reverse implication, (2) $\Rightarrow$ (1), shows that common cause is a possible causal explanation of conditional independence. It is important to note that common cause is only one \emph{possible} causal explanation, not the \emph{only} causal explanation. However, in the case where $Y$ and $Z$ are in the causal future of $X$ the other causal explanations will involve fine tuning, suggesting that a complete common cause explanation is probably the best explanation.

Theorem~\ref{thm:CI:classical-compatibility} can be read in two ways: as a summary of the justification presented above, or simply as a pair of equivalent definitions for conditional independence. The reading depends on how condition (1) is understood. If condition (1) is taken as a causal statement, then Thm.~\ref{thm:CI:classical-compatibility} summarises the above justification. However, if condition (1) is taken as simply a formal probabilistic statement it can be read as a definition of conditional independence.

\section{Quantum Conditional Independence and Quantum Reichenbach} \label{sec:CI:quantum-Reichenbach}

Having thoroughly introduced Reichenbach's principle classically, these ideas can be used to motivate a natural quantum generalisation. Before proceeding to the results some useful notation for quantum channels should be introduced.

\subsection{The Choi-Jamio\l{}kowski Isomorphism} \label{sec:CI:cj-isomorphism}

A quantum channel from system $A$ to $B$ is normally given as a completely positive trace preserving (CPTP) map from density operators on $A$ to density operators on $B$, written as $\mathcal{E}_{B|A}:\mathcal{D}(\mathcal{H}_A) \rightarrow \mathcal{D}(\mathcal{H}_B)$. Alternatively, the same channel can equivalently be expressed as a density operator in $\mathcal{D}(\mathcal{H}_B \otimes \mathcal{H}_A)$ using the Choi-Jamio\l{}kowski isomorphism \cite{Jamiolkowski72,Choi75}. There are several variations on the exact mathematical conventions of this isomorphism \cite{LeiferSpekkens13,Leifer11}, though morally they're more-or-less interchangeable. The notation convention used in this chapter and the next is
\begin{equation} \label{eq:CI:cj-isomorphism}
\rho_{B|A} \eqdef \sum_{i,j} \mathcal{E}_{B|A} \left( |i\rangle_A \langle j| \right) \otimes |i\rangle_{A^\ast} \langle j|
\end{equation}
where $\{|i\rangle_A\}_i$ is some orthonormal basis for $\mathcal{H}_A$ and $\{|i\rangle_{A^\ast}\}_i$ is the dual basis in the dual space $\mathcal{H}_{A^\ast}$. Strictly, $\rho_{B|A}$ is therefore an operator on $\mathcal{H}_B \otimes \mathcal{H}_{A^\ast}$ in this particular version of the isomorphism. This choice of definition has two key advantages over the alternatives: $\rho_{B|A}$ is a positive operator and the same regardless of the choice of basis used to define it. Note also that it is normalised such that $\Tr_B \rho_{B|A} = \mathbbm{1}_{A^\ast}$ and $\Tr\rho_{B|A} = d_A$.

As $\rho_{B|A}$ is an equivalent way to express the channel $\mathcal{E}_{B|A}$, it follows that there must be a way to express $\rho_B = \mathcal{E}_{B|A}(\rho_A)$ in terms of $\rho_{B|A}$ and $\rho_A$. This is done by defining the \emph{linking operator}
\begin{equation} \label{eq:CI:linking-operator}
\tau^\id_A \eqdef \sum_{i,j} |i\rangle_{A^\ast}\langle j| \otimes |i\rangle_A \langle j|
\end{equation}
on $\mathcal{H}_{A^\ast}\otimes\mathcal{H}_A$ where, again, $\{|i\rangle_A\}_i$ is an arbitrary orthonormal basis and $\{|i\rangle_{A^\ast}\}_i$ is its dual. This will also be an important operator for general quantum causal models [Sec.~\ref{sec:CM:predictions}]. Using this, it is easy to verify that
\begin{equation}
\rho_B = \mathcal{E}_{B|A}(\rho_A) = \Tr_{A^\ast A} \left( \rho_{B|A} \tau^\id_A \rho_A \right).
\end{equation}

Note the similarities between $\rho_{B|A}$ for a quantum channel and $\mathbb{P}(Y|X)$ for a classical channel \cite{LeiferSpekkens13}. In particular, note that in a classical decohering limit where the quantum systems and channels are all diagonal in some particular choices of bases, then the main diagonal of $\rho_{B|A}$ is exactly $\mathbb{P}(Y|X)$ and all expressions reduce to their classical counterparts.

As in the rest of the thesis, a density operator written with missing system subscripts indicates the partial trace over those systems. So, for example, given a channel from $AB$ to $CD$, $\rho_{CD|\cdot} \eqdef \Tr_{AB}\rho_{CD|AB}$ and, by definition, $\rho_{\cdot|AB} = \mathbbm{1}_{AB}$. It will also often be convenient to renormalise a channel $\rho_{B|A}$ to unity. This is written with a caret, $\hat{\rho}_{B|A} \eqdef \rho_{B|A}/d_A$.

\subsection{Justifying Quantum Reichenbach} \label{sec:CI:justifying-quantum-Reichenbach}

This subsection follows the classical example of Sec.~\ref{sec:CI:justifying-classical-Reichenbach} to develop a natural quantum version of Reichenbach's principle via a definition for quantum conditional independence.

Following the example of Ref.~\cite{CavalcantiLal14}, the qualitative part of Reichenbach's principle can be applied to quantum theory almost unchanged: if quantum systems $B$ and $C$ are correlated then this implies a causal connection of at least one of the three forms listed in Sec.~\ref{sec:CI:Reichenbach-bipartite}. The only difference is to clarify that ``correlated'' here means that there exist independent bipartite measurements on $B$ and $C$ that have correlated statistics. Alternatively, this can be phrased as: ancestral independence of $B$ and $C$ implies no correlated local statistics.

The challenge, as noted in Sec.~\ref{sec:CI:introduction-Bell}, is to identify quantum conditional independence in the $A\rightarrow BC$ scenario and use it to define a quantum quantitative part for Reichenbach's principle. That is, if $A$ is a complete common cause for $B$ and $C$, as illustrated in Fig.~\ref{fig:CI:A-to-BC}, what restrictions should hold for the channel $\rho_{BC|A}$, corresponding to $\mathbb{P}(Y,Z|X) = \mathbb{P}(Y|X)\mathbb{P}(Z|X)$?

\begin{figure}
\begin{centering}
\includegraphics[scale=0.35,angle=0]{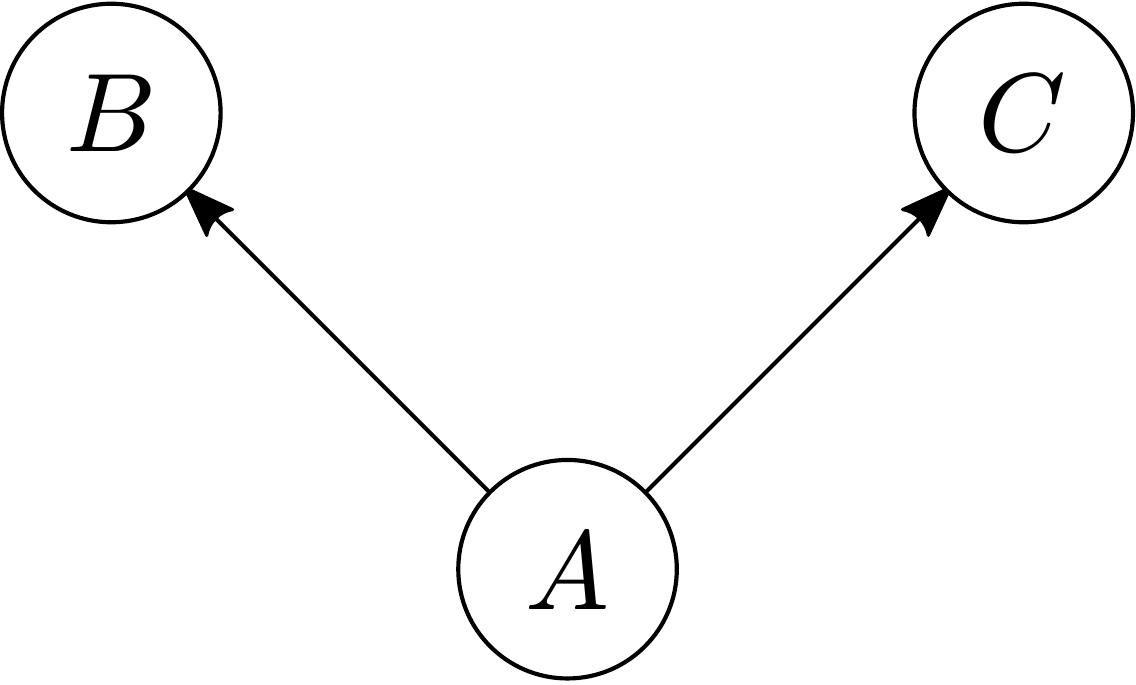}
\par\end{centering}
\caption{A causal structure depicting three quantum systems with $A$ the complete common cause of $B$ and $C$.}\label{fig:CI:A-to-BC}
\end{figure}

It may be tempting to first seek a quantum analogue of the joint distribution $\mathbb{P}(X,Y,Z)$. However, textbook quantum theory does not provide such an analogue when some systems are causally dependent on others and there are serious reasons to believe that such an analogue might be impossible \cite{HorsmanHeunen+16}. This is the reason for concentrating on the channels $\mathbb{P}(Y,Z|X)$ and $\rho_{BC|A}$ and seeking conditions that apply to them directly.

In Sec.~\ref{sec:CI:justifying-classical-Reichenbach} a justification of the quantitative part of Reichenbach's principle from the qualitative part was given by temporarily assuming that classical dynamics is fundamentally deterministic. Suppose, by analogy, that quantum dynamics is taken to be fundamentally unitary. What characterisation of $\rho_{BC|A}$ in Fig.~\ref{fig:CI:A-to-BC} follows from that? Just as in Sec.~\ref{sec:CI:justifying-classical-Reichenbach}, this is a temporary assumption in order to motivate the final result, which will stand apart from it.

It is well known that any quantum channel $\rho_{B|A}$ (equivalently, $\mathcal{E}_{B|A}$) can be viewed as arising from underlying unitary dynamics via the Stinespring dilation \cite{NielsenChuang00}.

\begin{definition}
A \emph{unitary dilation} of a quantum channel $\mathcal{E}_{B|A}$ is given by some unitary $U$ on $\mathcal{H}_B \otimes \mathcal{H}_F \cong \mathcal{H}_A \otimes \mathcal{H}_L$ for some ancillary system $L$ in state $\rho_L$ such that
\begin{equation}
\mathcal{E}_{B|A}( \cdot ) = \Tr_F \left( U(\cdot \otimes \rho_L ) U^\dagger \right),
\end{equation}
where $F$ is some system of dimension $d_F = d_A d_L / d_B$.
\end{definition}

Applying this to the common cause situation of Fig.~\ref{fig:CI:A-to-BC} gives
\begin{equation}
\rho_{BC|A} = \Tr_{FLL^\ast} \left( \rho^U_{BFC|AL} \tau^\id_L \rho_L \right)
\end{equation}
where $\rho^U_{BFC|AL}$ is the Choi-Jamio\l{}kowski operator for the unitary and $\tau_L$ is defined in Eq.~(\ref{eq:CI:linking-operator}). Note that $F$ is required so that the input and output dimensions match, but is unimportant and will normally be traced out.

Classically, in Sec.~\ref{sec:CI:justifying-classical-Reichenbach}, it was argued that the ancilla variable $\Lambda$ must split into two independent variables to preserve $X$ as a complete common cause. This implicitly uses the obvious idea that if a deterministic causal relationship $f:X\times \bar{X} \rightarrow Y$ cannot be equivalently written as some other function $f^\prime :X\rightarrow Y$, then $\bar{X}$ must have a causal influence on $Y$. How could it not? Its value non-trivially affects $Y$.

In unitary quantum theory the corresponding condition is slightly less obvious, but can be made precise as follows.

\begin{definition} \label{def:CI:quantum-no-causal-influence}
For a unitary channel $\rho_{B\bar{B}|A\bar{A}}$, system $\bar{A}$ has \emph{no causal influence} on system $B$ if and only if the marginal output state at $B$ is independent of any extra operations applied to the input $\bar{A}$ system before applying $\rho_{B\bar{B}|A\bar{A}}$.
\end{definition}

This definition captures the idea that interventions on $\bar{A}$ (including preparing some specific input state) cannot affect the local output at $B$ and therefore is a sensible notion for ``no causal influence''. Similar properties of unitary channels have been studied before in different contexts, in particular in Ref.~\cite{SchumacherWestmoreland05} under the guise of ``non-signalling'' and in Refs.~\cite{BeckmanGottesman+01,EggelingSchlingemann+02} for ``semi-causal'' unitaries. An equivalent definition, more directly useful to the treatment here, is that $\bar{A}$ has no causal influence on $B$ if and only if the partial trace of the channel satisfies\footnote{To see this, consider the marginal output at $B$ gained by preceding $\rho_{B\bar{B}|A\bar{A}}$ with a channel that discards the input at $\bar{A}$ and replaces it with $\mathbbm{1}/d_{\bar{A}}$. This must be the same as the marginal output from the original channel, so Eq.~(\ref{eq:CI:cj-isomorphism}) gives the stated result.} $\rho_{B|A\bar{A}} = \rho_{B|A} \otimes \mathbbm{1}_{\bar{A}^\ast}$.

Just as in Sec.~\ref{sec:CI:justifying-classical-Reichenbach}, the assumption that $A$ is a \emph{complete} common cause for $B$ and $C$ therefore implies that $L$ must factorise into ancestrally independent $L_B$ and $L_C$ where $L_B$ has no causal influence on $C$ and $L_C$ has no causal influence on $B$. It follows that the unitary channel $U$ (followed by tracing out $F$) is of the form illustrated in Fig.~\ref{fig:CI:A-to-BC-dilated}. Using the qualitative part of Reichenbach's principle, ancestral independence of $L_B$, $A$, and $L_C$ implies that the input state to $\rho^U_{BFC|L_B A L_C}$ factorises as $\rho_{L_B A L_C} = \rho_{L_B} \rho_A \rho_{L_C}$.

\begin{figure}
\begin{centering}
\includegraphics[scale=0.35,angle=0]{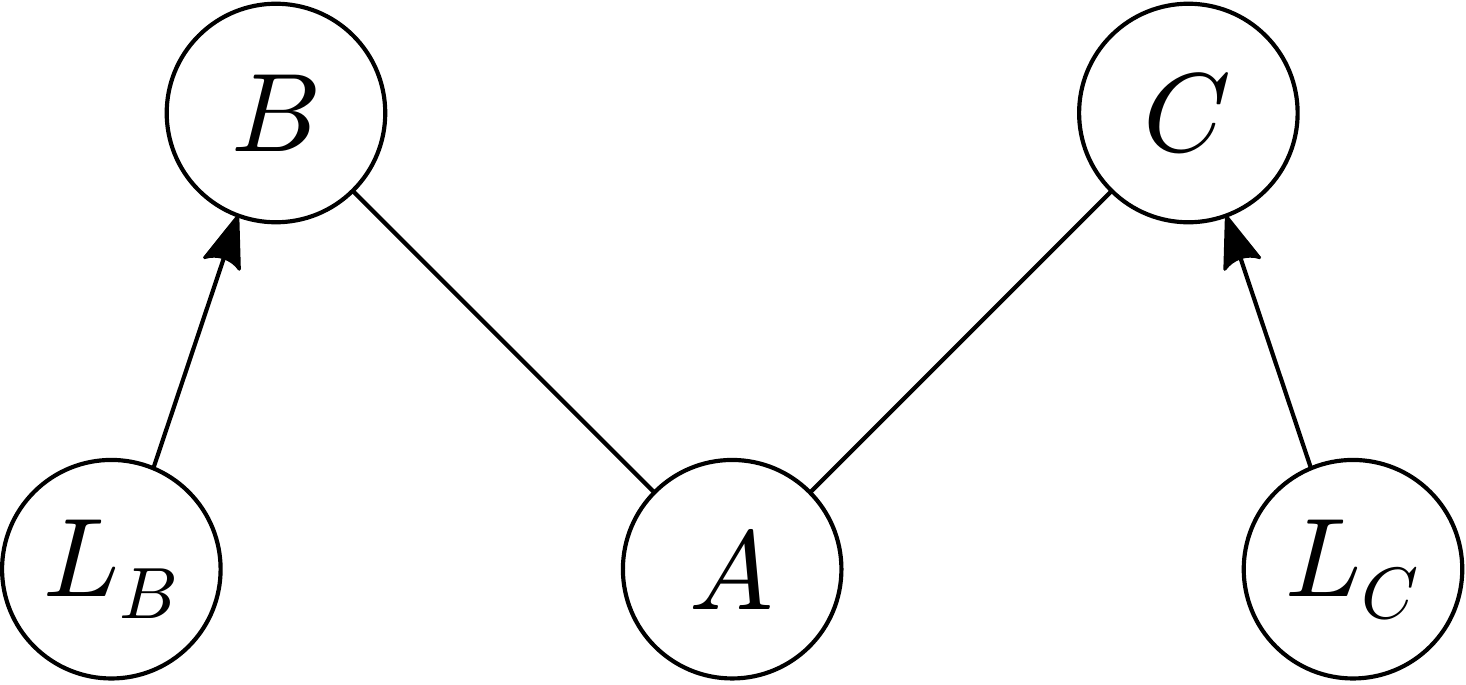}
\par\end{centering}
\caption{The causal structure of Fig.~\ref{fig:CI:A-to-BC}, dilated so that $B$ and $C$ each has a latent system as a causal parent in addition to $A$. $B$ and $C$ (together with $F$, not shown) depend unitarily on their parents $L_B$, $A$, and $L_C$.}\label{fig:CI:A-to-BC-dilated}
\end{figure}

Similarly to the classical case, temporarily assuming that quantum dynamics is unitary has led to the following suggestion for characterising quantum common cause channels.

\begin{definition} \label{def:CI:compatible-unitary-common-cause}
A quantum channel $\rho_{BC|A}$ is \emph{compatible with $A$ being the unitary common cause for $B$ and $C$} if and only if there is a unitary dilation $U$ of the channel of the form
\begin{equation}
\rho_{BC|A} = \Tr_{FL_B L^\ast_B L_C L^\ast_C} \left( \rho^U_{BFC|L_B AL_C} \tau^\id_{L_B} \tau^\id_{L_C} \rho_{L_B} \rho_{L_C} \right)
\end{equation}
where $L_B$ has no causal influence on $C$ and $L_C$ has no causal influence on $B$ in $U$.
\end{definition}

Despite the specific way it was reached, this is a very natural quantum condition. It strongly suggests the following quantum analogue of Thm.~\ref{thm:CI:classical-compatibility}.

\begin{proposition} \label{prop:CI:quantum-conditional-independence}
The following are equivalent:
\begin{enumerate}
\item $\rho_{BC|A}$ is compatible with $A$ being the unitary complete common cause for $B$ and $C$.
\item $\rho_{BC|A} = \rho_{B|A} \rho_{C|A}$.
\end{enumerate}
\end{proposition}

This proposition is proved as part of the more general Thm.~\ref{thm:CI:QCI-full} in Sec.~\ref{sec:CI:proving-QCI-conditions}. Note that condition (2) has no ordering ambiguity as Hermiticity guarantees that $[\rho_{B|A},\rho_{C|A}] = 0$. It is not difficult to verify that both conditions of Prop.~\ref{prop:CI:quantum-conditional-independence} reduce to their corresponding classical statements in any decohering limit (where all states and channels are diagonal in some choices of ``classical'' bases).

That Prop.~\ref{prop:CI:quantum-conditional-independence} holds [Sec.~\ref{sec:CI:proving-QCI-conditions}] and is strongly analogous to Thm.~\ref{thm:CI:classical-compatibility} suggests the following definition for quantum conditional independence (at least where $A$ is in the causal past of $B$ and $C$).

\begin{definition} \label{def:CI:quantum-conditional-independence}
For systems $A$, $B$, and $C$ related by quantum channel $\rho_{BC|A}$, the outputs are \emph{quantum conditionally independent} given the input if and only if $\rho_{BC|A} = \rho_{B|A} \rho_{C|A}$.
\end{definition}

Finally, with these definitions in hand, it is easy to state the quantitative part of quantum Reichenbach's principle: in the case where $B$ and $C$ are correlated due exclusively to some complete common cause $A$, then $B$ and $C$ are quantum conditionally independent given $A$ for the channel $\rho_{BC|A}$. Combined with the qualitative part discussed above, this completes the definition of quantum Reichenbach's principle.

So in Prop.~\ref{prop:CI:quantum-conditional-independence}, the quantitative part of quantum Reichenbach's principle is motivated by the (1) $\Rightarrow$ (2) statement. The converse, (2) $\Rightarrow$ (1), should be useful for causal inference, just as in Thm.~\ref{thm:CI:classical-compatibility}

These definitions will form the motivation for most of the work in Chap.~\ref{ch:CM}. For the remainder of this chapter, these definitions will be expanded and applied in various ways.

\subsection{Alternative Expressions for Quantum Conditional Independence} \label{sec:CI:alternative-QCI}

There are many equivalent ways to define conditional independence of a channel $\mathbb{P}(Y,Z|X)$ classically. Two of these are given in Thm.~\ref{thm:CI:classical-compatibility}: there exists a deterministic dilation that factorises between $Y$ and $Z$; and $\mathbb{P}(Y,Z|X) = \mathbb{P}(Y|X)\mathbb{P}(Z|X)$. Two other useful definitions are as follows.

The first is that for every input $\mathbb{P}(X)$ to the channel, the joint distribution over input and output values $\mathbb{P}(Y=y,Z=z|X = x)\mathbb{P}(X = x)$ has vanishing conditional mutual information $I(Y:Z|X) = 0$. Equivalently, this need only hold when the input distribution is uniform $\mathbb{P}(X=x) = 1/|X|$, so that $I(Y:Z|X) = 0$ on $\hat{\mathbb{P}}(Y,Z,X) \eqdef \mathbb{P}(Y,Z|X)/|X|$, where $|X|$ is the cardinality of $X$.

The second definition is that $\mathbb{P}(Y,Z|X)$ is mathematically equivalent to the channel obtained by doing the following: (i) copy the input $X$, (ii) apply channel $\mathbb{P}(Y|X)$ to only one copy of $X$, (iii) apply channel $\mathbb{P}(Z|X)$ to only the other copy.

That these are both equivalent definitions to $\mathbb{P}(Y,Z|X) = \mathbb{P}(Y|X)\mathbb{P}(Z|X)$ are standard results and easily verified. Mathematically, there is no reason to consider one definition as ``more fundamental'' than the others (that depends on your philosophical bent). The definitions used in Thm.~\ref{thm:CI:classical-compatibility} were only given first because they are the closest to the motivational narrative being used. Another approach might have led to presenting other definitions first.

Like the conditions of Thm.~\ref{thm:CI:classical-compatibility}, these two definitions also have natural quantum counterparts. They provide two new ways to define quantum conditional independence for a channel, equivalent to those in Prop.~\ref{prop:CI:quantum-conditional-independence}. Just as in the classical case, none of these definitions need be considered more fundamental than any other.

\begin{proposition} \label{prop:CI:QCI-equivalents}
The following are equivalent to the conditions of Prop.~\ref{prop:CI:quantum-conditional-independence} and to each other:
\begin{enumerate}
\item[3.] $I(B:C|A) = 0$ for the quantum conditional mutual information evaluated on the (positive, trace-one) operator $\hat{\rho}_{BC|A}$.
\item[4.] The Hilbert space  of $A$ has a decomposition $\mathcal{H}_A = \bigoplus_i \mathcal{H}_{A_i^L} \otimes \mathcal{H}_{A_i^R}$ for which $\rho_{BC|A} = \sum_i \left( \rho_{B|A_i^L} \otimes \rho_{C|A_i^R} \right)$, where for each $i$, $\rho_{B|A_i^L}$ is a quantum channel from $A_i^L$ to $B$ and $\rho_{C|A_i^R}$ is a channel from $A_i^R$ to $C$.
\end{enumerate}
\end{proposition}

The proof of this, with Prop.~\ref{prop:CI:quantum-conditional-independence}, is in Sec.~\ref{sec:CI:proving-QCI-conditions}. It is easy to verify that condition (3) reduces to the classical case in a decohering limit when everything diagonalises in some choice of basis. Recall that $\hat{\rho}_{BC|A}$ is just $\rho_{BC|A}$ re-normalised to be trace-one.

Condition (4) deserves significantly more discussion, provided in Secs.~\ref{sec:CI:circuits}, \ref{sec:CM:summary-and-discussion}. For now, it suffices to note that any classical channel which copies the input and applies independent channels to each output will take the form of condition (4) when written as a quantum channel. It therefore generalises that alternative definition of conditional independence.

\subsection{Proving Propositions \ref{prop:CI:quantum-conditional-independence} and \ref{prop:CI:QCI-equivalents}} \label{sec:CI:proving-QCI-conditions}

Propositions~\ref{prop:CI:quantum-conditional-independence} and \ref{prop:CI:QCI-equivalents} are more naturally taken together as a single theorem. All four conditions given are equally valid ways to define quantum conditional independence for a channel and combining them into one theorem reflects this. Proving their equivalence together also requires fewer steps. The full theorem is restated here for convenience.

\begin{theorem} \label{thm:CI:QCI-full}
Given a quantum channel $\rho_{BC|A}$, the following are equivalent:
\begin{enumerate}
\item $\rho_{BC|A}$ is compatible with $A$ being the unitary complete common cause for $B$ and $C$.
\item $\rho_{BC|A} = \rho_{B|A}\rho_{C|A}$.
\item $I(B:C|A) = 0$ for the quantum conditional mutual information evaluated on the (positive, trace-one) operator $\hat{\rho}_{BC|A}$.
\item The Hilbert space  of $A$ has a decomposition $\mathcal{H}_A = \bigoplus_i \mathcal{H}_{A_i^L} \otimes \mathcal{H}_{A_i^R}$ for which $\rho_{BC|A} = \sum_i \left( \rho_{B|A_i^L} \otimes \rho_{C|A_i^R} \right)$, where for each $i$, $\rho_{B|A_i^L}$ is a quantum channel from $A_i^L$ to $B$ and $\rho_{C|A_i^R}$ is a channel from $A_i^R$ to $C$.
\end{enumerate}
Each of these conditions is an equivalent definition for when \emph{$B$ and $C$ are quantum conditionally independent given $A$} in a channel $\rho_{BC|A}$.
\end{theorem}

This theorem will now be proved in parts. First, consider the following lemma proved in Ref.~\cite{HaydenJozsa+04}.

\begin{lemma}[{\cite[Thm.~6]{HaydenJozsa+04}}] \label{lem:CI:hayden-jozsa}
For any tripartite quantum state $\rho_{ABC}$, the quantum conditional mutual information evaluated on that state vanishes $I(B:C|A)=0$ if and only if the Hilbert space of the $A$ system decomposes as $\mathcal{H}_{A} = \bigoplus_i \mathcal{H}_{A_i^L}\otimes\mathcal{H}_{A_i^R}$, such that
\begin{equation}
\rho_{ABC} = \sum_i p_i \left(\rho_{BA_i^L}\otimes\rho_{CA_i^R}\right),\quad p_i\geq0,\quad\sum_i p_i=1,
\end{equation}
where for each $i$, $\rho_{BA_i^L}$ is a quantum state on $\mathcal{H}_B\otimes \mathcal{H}_{A_i^L}$ and  $\rho_{CA_i^R}$ is a quantum state on $\mathcal{H}_C\otimes \mathcal{H}_{A_i^R}$.
\end{lemma}

This provides the first first part of the proof of Thm.~\ref{thm:CI:QCI-full}.

\begin{proof}[Proof: (3) $\Leftrightarrow$ (4)]
Applying Lem.~\ref{lem:CI:hayden-jozsa} to $\hat{\rho}_{BC|A}$ immediately proves that (4) $\Rightarrow$ (3) and gets most of the way to proving (3) $\Rightarrow$ (4). To complete the proof, note that $\Tr_{BC}\hat{\rho}_{BC|A} = \mathbbm{1}_A/d_A$. It follows that each $\hat{\rho}_{B|A_i^L}$ given by Lem.~\ref{lem:CI:hayden-jozsa} satisfies $\Tr_B \hat{\rho}_{B|A_i^L} = \mathbbm{1}_{A_i^L}/d_{A_i^L}$ and similarly for $\hat{\rho}_{C|A_i^R}$. Therefore each $\hat{\rho}_{B|A_i^L}$ and each $\hat{\rho}_{C|A_i^R}$ are appropriate quantum channels, completing the proof.
\end{proof}

The next piece of the proof is perhaps the most involved as it demonstrates that the unitary dilation implied by condition (1) has the structure required by condition (4).

\begin{proof}[Proof: (1) $\Rightarrow$ (4)]
Let $\rho^U_{BFC|L_B A L_C}$ be the Choi-Jamio\l{}kowski operator for the unitary $U$ as in the Sec.~\ref{sec:CI:justifying-quantum-Reichenbach}. Note that $\rho^U_{BC|A}\neq\rho_{BC|A}$ in general (the latter depends on particular choices for inputs at $L_B$ and $L_C$---\emph{cf.} Def.~\ref{def:CI:compatible-unitary-common-cause}). The proof proceeds by demonstrating that several conditional mutual informations evaluated on $\hat{\rho}^U_{BFC|L_B A L_C} \eqdef \rho^U_{BFC|L_B A L_C}/(d_{L_B}d_{A}d_{L_C})$ vanish. 

The first conditional mutual information of $\hat{\rho}^U_{BFC|L_B A L_C}$ that vanishes is
\begin{equation} \label{eq:CI:proof-1to3-first-entopy-relation}
I(B:FC | L_B A L_C )=0.	
\end{equation}
This follows by expanding in terms of von Neumann entropies
\begin{multline}
I(B:FC | L_B A L_C ) = S(\hat{\rho}^U_{B| L_B A L_C  }) + S(\hat{\rho}^U_{FC| L_B A L_C }) \\ - S(\hat{\rho}^U_{ BFC | L_B A L_C}) - S(\hat{\rho}^U_{\cdot | L_B A L_C }).
\end{multline}
The third term is zero, since $\hat{\rho}^U_{BFC|L_B A L_C}$ is pure by unitarity. The final term is equal to $\log(d_{L_B} d_{A} d_{L_C})$, since $\hat{\rho}^U_{ \cdot | L_B A L_C} = \mathbbm{1}_{(L_B A L_C)^\ast}/(d_{L_B} d_A d_{L_C})$. Noting also that 
$\hat{\rho}^U_{ BFC | \cdot} = \mathbbm{1}_{BFC}/(d_B d_F d_C)$, and using the fact that the von Neumann entropy of the partial trace of a pure state is equal to the von Neumann entropy of the complementary partial trace, one finds that the first two terms equal $\log (d_F d_C)$ and $\log d_B$ respectively. Summing these, Eq.~(\ref{eq:CI:proof-1to3-first-entopy-relation}) follows.

Second
\begin{equation} \label{eq:CI:proof-1to3-second-entopy-relation}
I(L_B : L_C | A )=0,
\end{equation}
which follows immediately from $\hat{\rho}^U_{ \cdot | L_B A L_C} = \mathbbm{1}_{(L_B A L_C)^\ast}/(d_{L_B} d_A d_{L_C})$.

Third,
\begin{equation} \label{eq:CI:proof-1to3-third-entopy-relation}
I(B : L_C | L_B A) = 0.
\end{equation}
As with Eq.~(\ref{eq:CI:proof-1to3-first-entopy-relation}), prove this by expanding
\begin{equation}
I(B : L_C | L_B A)  = S(\hat{\rho}^U_{B | L_B A}) + S(\hat{\rho}^U_{\cdot | L_B A L_C}) - S(\hat{\rho}^U_{B | L_B A L_C}) - S(\hat{\rho}^U_{\cdot | L_B A}).
\end{equation}
The second and fourth terms are entropies of maximally mixed states on their respective systems, which hence sum to $\log d_{L_C}$. Since there is no causal influence from $L_C$ to $B$ in $U$ if follows that $\hat{\rho}^U_{B | L_B A L_C} = \hat{\rho}^U_{B | L_B A} \otimes \mathbbm{1}_{L_C^\ast}/d_{L_C}$. Hence, the third term is $S(\hat{\rho}^U_{B | L_B A}) + \log d_{L_C}$, which gives Eq.~(\ref{eq:CI:proof-1to3-third-entopy-relation}).

Fourth and finally
\begin{equation} \label{eq:CI:proof-1to3-fourth-entopy-relation}
I(C : L_B | A L_C) = 0,
\end{equation}
which follows symmetrically to Eq.~(\ref{eq:CI:proof-1to3-third-entopy-relation}) by using the assumption that there is no influence from $L_B$ to $C$ in $U$. 

Equations~(\ref{eq:CI:proof-1to3-first-entopy-relation}, \ref{eq:CI:proof-1to3-second-entopy-relation}, \ref{eq:CI:proof-1to3-third-entopy-relation}, \ref{eq:CI:proof-1to3-fourth-entopy-relation}) can be used to show that $\hat{\rho}_{BC|A}$ satisfies $I(B:C|A)=0$. This follows from Ref.~\cite[Thm.~4.5]{LeiferPoulin08}, which states that quantum conditional mutual informations on partial traces of a multipartite quantum state satisfy the \emph{semi-graphoid axioms} familiar from the classical formalism of causal networks \cite{Pearl09}. Therefore, the semi-graphoid axioms satisfied in this case are:
\begin{align} \label{eq:CI:proof-1to3-semi-graphoid-1}
\left[I(X:Y|Z) = 0\right]  &\Rightarrow \left[I(Y:X|Z) = 0\right], \\
\left[I(X:YW|Z) = 0\right] &\Rightarrow \left[I(X:Y|Z) = 0\right], \\
\left[I(X:YW|Z) = 0\right] &\Rightarrow \left[I(X:Y|ZW) = 0\right], \\
\left[I(X:Y|Z) = 0 \right]
\wedge \left[I(X:W|YZ) = 0\right] &\Rightarrow \left[I(X:YW|Z) = 0\right]. \label{eq:CI:proof-1to3-semi-graphoid-4}
\end{align}

By applying Eqs.~(\ref{eq:CI:proof-1to3-semi-graphoid-1}--\ref{eq:CI:proof-1to3-semi-graphoid-4}) to Eqs.~(\ref{eq:CI:proof-1to3-first-entopy-relation}, \ref{eq:CI:proof-1to3-second-entopy-relation}, \ref{eq:CI:proof-1to3-third-entopy-relation}, \ref{eq:CI:proof-1to3-fourth-entopy-relation}) one finds
\begin{align}
\left[I(B:FC| L_B A L_C ) = 0 \right] &\Rightarrow \left[I(B : C | L_B A L_C) = 0\right] \\
\left[I(C: L_B | AL_C) = 0 \right] \wedge \left[I(B : C | L_B A L_C) = 0\right] &\Rightarrow \left[I(C : BL_B | AL_C) = 0\right] \\
\left[I(L_B : L_C | A) = 0 \right] \wedge \left[I(L_C : B | L_B A) = 0\right] &\Rightarrow \left[I(L_C : B L_B | A) = 0\right] \\
\left[I(B L_B : L_C | A) = 0\right] \wedge \left[I(B L_B : C | A L_C) = 0\right] &\Rightarrow \left[I(BL_B : CL_C | A) = 0\right]
\end{align}

This shows that condition~(1) implies $I(BL_B : C L_C | A) = 0$, calculated on $\hat{\rho}^U_{BC|L_B AL_C}$. Using Lem.~\ref{lem:CI:hayden-jozsa} gives
\begin{equation} 
\hat{\rho}^U_{BC | L_B A L_C} = \sum_i p_i \left( \hat{\rho}^U_{B|L_B A_i^L} \otimes \hat{\rho}^U_{C|A_i^R L_C} \right),
\end{equation}
for some appropriate decomposition of $\mathcal{H}_A^\ast$ and probability distribution $\{p_i\}_i$. This decomposition, with the fact that $\hat{\rho}^U_{ \cdot | L_B A L_C} = \mathbbm{1}_{(L_B A L_C)^\ast}/(d_{L_B} d_A d_{L_C})$, implies
\begin{equation}
\rho^U_{BC | L_B A L_C} = \sum_i \left( \rho^U_{B|L_B A_i^L} \otimes \rho^U_{C|A_i^R L_C} \right),
\end{equation}
where, for each $i$, the components satisfy $\mathrm{Tr}_B \rho^U_{B|L_B A_i^L} = \mathbbm{1}_{(L_B A_i^L)^\ast}$ and $\mathrm{Tr}_C \rho^U_{C|\lambda_C A_i^R} = I_{(A_i^R L_C)^\ast}$. Definition~\ref{def:CI:compatible-unitary-common-cause} shows that the operator $\rho_{BC|A}$ is obtained by acting with this channel on the input states $\rho_{L_B}$ and $\rho_{L_C}$ respectively for $L_B$  and $L_C$. Finally, therefore
\begin{equation}
\rho_{BC | A} = \sum_i \left( \rho_{B| A_i^L} \otimes \rho_{C|A_i^R } \right),
\end{equation}
where $\mathrm{Tr}_B \rho_{B | A_i^L} = \mathbbm{1}_{(A_i^L)^\ast}$ and $\mathrm{Tr}_C \rho_{C | A_i^R} = \mathbbm{1}_{(A_i^R)^\ast}$, as required.
\end{proof}

To prove the converse, that condition (4) implies (1), one can show how to construct an appropriate unitary dilation.

\begin{proof}[Proof: (4) $\Rightarrow$ (1)]
Each channel $\rho_{B|A_i^L}$ can be dilated to a unitary transformation $V_i$, with ancilla input $L_B$ in a fixed state $\rho_{L_B}$, such that $V_i$ acts on $\mathcal{H}_{L_B}\otimes \mathcal{H}_{A_i^L}$. Similarly, $\rho_{C|A_i^R}$ can be dilated to a unitary transformation $W_i$, with ancilla $L_C$ in a fixed state $\rho_{L_C}$, acting on $\mathcal{H}_{A_i^R}\otimes \mathcal{H}_{L_C}$. By choosing the dimension of $L_B$ large enough, a single system $L_B$ and state $\rho_{L_B}$ can be used for each value of $i$ and similarly for $L_C$.

For each $i$, let $V_i^\prime$ be the unitary that acts as $V_i\otimes \mathbbm{1}_{A_i^R} \otimes \mathbbm{1}_{L_C}$ on the subspace $\mathcal{H}_{L_B} \otimes \mathcal{H}_{A_i}\otimes \mathcal{H}_{L_C}$ and as zero on each subspace $\mathcal{H}_{L_B} \otimes \mathcal{H}_{A_j}\otimes \mathcal{H}_{L_C}$ where $j\ne i$. Similarly, for each $i$ let $W_i^\prime$ be the unitary that acts as $\mathbbm{1}_{L_B}\otimes \mathbbm{1}_{A_i^L}\otimes W_i$ on the subspace $\mathcal{H}_{L_B} \otimes \mathcal{H}_{A_i}\otimes \mathcal{H}_{L_C}$, and as zero on each subspace $\mathcal{H}_{L_B} \otimes \mathcal{H}_{A_j}\otimes \mathcal{H}_{L_C}$ for every $j\ne i$. Using these, define unitaries
\begin{align}
V &\eqdef \sum_i V_i^\prime, \\
W &\eqdef \sum_i W_i^\prime.
\end{align}
Note that: they commute $[V,W]=0$; $V$ acts as the identity on $L_C$; and $W$ acts as the identity on $L_B$.

It can easily be checked that the unitary $WV=VW$ is a unitary dilation for the channel $\rho_{BC|A}$ in the form of condition (4), with ancillae $L_B$ and $L_C$. From the form of $V$ and $W$ it can also be seen that there can be no causal influence from $L_C$ to $B$ in $WV$ (the output at $B$ already exists, by the action of $V$, before $W$ acts on $L_C$ at all). Similarly, since $WV=VW$, there can be no causal influence from $L_B$ to $C$. The unitary $WV$ is therefore the unitary dilation required by condition (1).
\end{proof}

These final two pieces complete the proof of Thm.~\ref{thm:CI:QCI-full}.

\begin{proof}[Proof: (2) $\Rightarrow$ (3)]
Assuming a channel $\rho_{BC|A} = \rho_{B|A}\rho_{C|A}$ it follows that:
\begin{align}
\rho_{BC|A}  = & \exp\left( \log\rho_{B|A} + \log\rho_{C|A} \right) \\
\log\rho_{BC|A}  = & \log\rho_{B|A} + \log\rho_{C|A} \\
\log\rho_{BC|A} + \log\rho_{\cdot|A}  = & \log\rho_{B|A} + \log\rho_{C|A} \\
\log(d_A^{-1}\rho_{BC|A}) + \log(d_A^{-1}\rho_{\cdot|A})  = & \log(d_A^{-1}\rho_{B|A}) + \log(d_A^{-1}\rho_{C|A}).
\end{align}
The first line follows because $[\rho_{B|A},\rho_{C|A}]=0$ (as noted in Sec.~\ref{sec:CI:justifying-quantum-Reichenbach}); the third because $\rho_{\cdot|A}=\mathbbm{1}_{A^\ast}$ and therefore $\log\rho_{\cdot|A}$ is the zero matrix; and the final line by adding $2\log d_A^{-1}$ to both sides. It is proved in Ref.~\cite{Ruskai02} that if $\rho_{XYZ}$ is any trace-one density operator, then $\log\rho_{XYZ} + \log\rho_{Z}  =  \log\rho_{XZ} + \log\rho_{YZ}$ is equivalent to $I(X:Y|Z)=0$. Therefore, the conditional mutual information of $\hat{\rho}_{BC|A}$ vanishes.
\end{proof}

\begin{proof}[Proof: (4) $\Rightarrow$ (2)]
Condition (4) implies 
\begin{align}
\rho_{B|A} & =  \sum_{i} \left( \rho_{B|A_{i}^{L}}\otimes \mathbbm{1}_{(A_{i}^{R})^\ast} \right), \\
\rho_{C|A} & =  \sum_{j} \left( \mathbbm{1}_{(A_{j}^{L})^\ast}\otimes\rho_{C|A_{j}^{R}} \right).
\end{align}
Taking the product, terms where $i\ne j$ have support on orthogonal subspaces and so vanish, therefore
\begin{eqnarray}
\rho_{B|A}\rho_{C|A} & = & \sum_{i,j} \left( \rho_{B|A_{i}^{L}}\otimes \mathbbm{1}_{(A_{i}^{R})^\ast} \right) \left( \mathbbm{1}_{(A_{j}^{L})^\ast}\otimes\rho_{C|A_{j}^{R}} \right) = \sum_{i} \rho_{B|A_{i}^{L}}\otimes\rho_{C|A_{i}^{R}} \\
& = & \rho_{BC|A}.
\end{eqnarray}
\end{proof}

Combining these proof steps, any of the conditions of Thm.~\ref{thm:CI:QCI-full} can be seen to imply any other. Moreover, some of the contents of the proofs are somewhat revealing. For instance, in the proof for (4) $\Rightarrow$ (1) a specific possible structure for the unitary dilation is revealed, one that makes the no-causal-influence properties explicit. Some of the theorems used from other papers also point to places in the literature concerned with similar ideas and which may be fruitful places to apply these results. 

\subsection{Circuits for Common Cause Channels} \label{sec:CI:circuits}

The definitions for quantum conditional independence given in Thm.~\ref{thm:CI:QCI-full} can be considered in the language of quantum circuits. This is useful both to elucidate their meanings and for comparison with the corresponding classical definitions.

Figure~\ref{fig:CI:classical-summary} illustrates several equivalent ways to view a classical conditionally-independ\-ent channel $\mathbb{P}(Y,Z|X)$ (\emph{viz.} one where $X$ could be a complete common cause). Equality (1) of that figure simply repeats that any classical channel can be dilated to a deterministic channel by a function $f$ and random variable $\Lambda$ that absorbs all stochasticity. The other three equalities are more interesting and may be understood as forming a commuting diagram, showing two equivalent ways to get from the top left to bottom right. That is, two different ways of understanding that $\mathbb{P}(Y,Z|X)$ is compatible with $X$ being the deterministic common cause for $Y$ and $Z$.

\begin{figure}
\begin{centering}
\includegraphics[scale=0.35,angle=0]{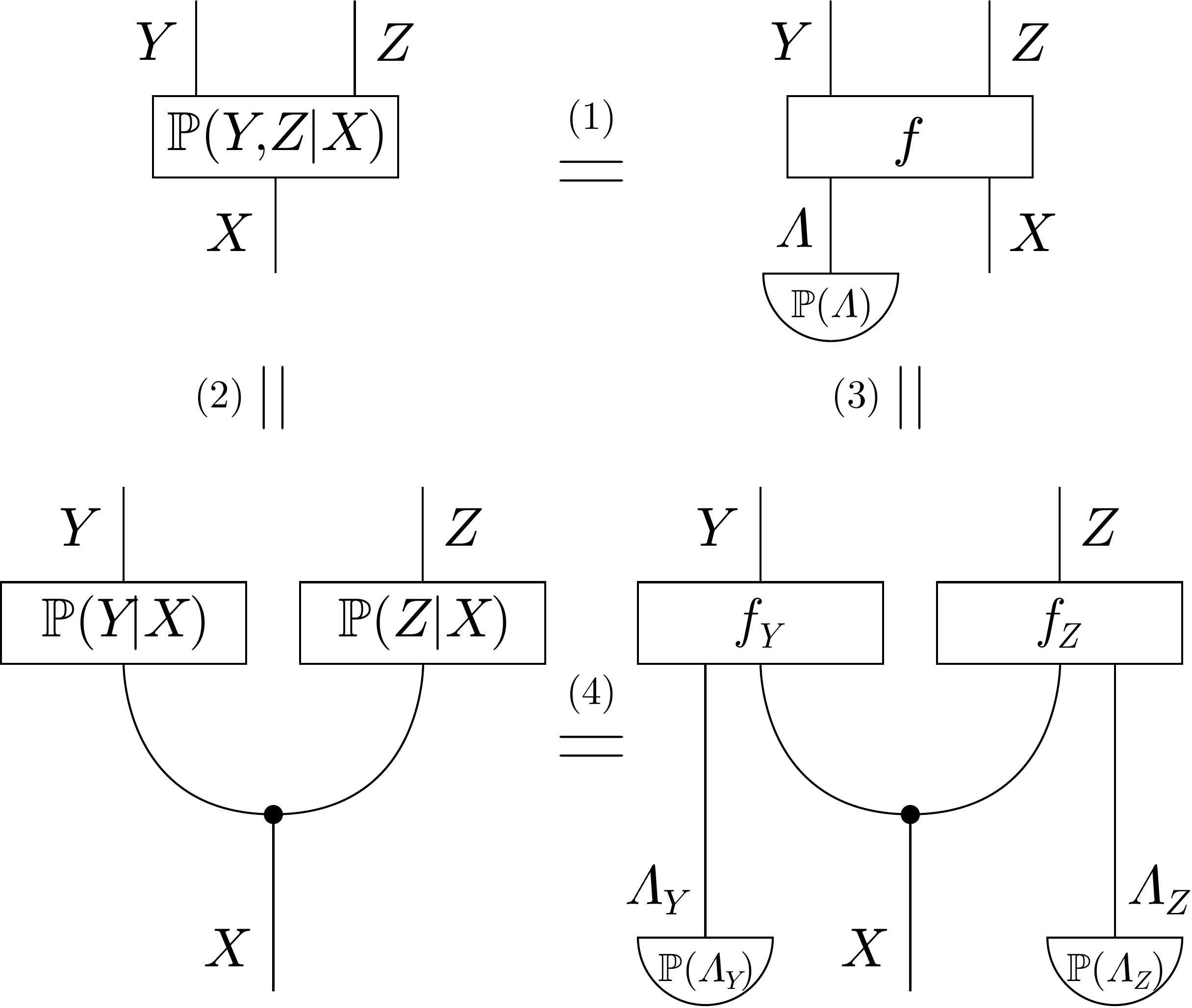}
\par\end{centering}
\caption{Circuit representations of a classical channel $\mathbb{P}(Y,Z|X)$ where $Y$ and $Z$ are conditionally independent given $X$ (equivalently, where $X$ could be a complete common cause for $Y$ and $Z$). The symbol \ccopy{} represents the classical copy operation. Note that equalities (1) and (4) do not depend on conditional independence when viewed in isolation, but they do in the context of the other equalities in the diagram.}
\label{fig:CI:classical-summary}
\end{figure}

Equality (2) of Fig.~\ref{fig:CI:classical-summary} states that if $Y$ and $Z$ are conditionally independent given $X$, then the channel $\mathbb{P}(Y,Z|X)$ may be achieved by first copying $X$ and then applying separate channels $\mathbb{P}(Y|X)$ and $\mathbb{P}(Z|X)$ to each copy. The symbol \ccopy{} is used to represent the copying operation. Equality (4) mirrors equality (1): the channels $\mathbb{P}(Y|X)$ and $\mathbb{P}(Z|X)$ can be individually dilated to functions $f_Y$ and $f_Z$ respectively.

Finally, equality (3) follows as $X$ is a \emph{complete} common cause for $Y$ and $Z$. Together, equalities (1) and (3) illustrate the justification for classical Reichenbach's principle given in Sec.~\ref{sec:CI:justifying-classical-Reichenbach}.

Mirroring this, Fig.~\ref{fig:CI:quantum-summary} illustrates the analogous conditions for quantum conditional independence of channel $\rho_{BC|A}$. In order to do this some non-standard notation in quantum circuits has been introduced in order to capture condition (4) of Thm.~\ref{thm:CI:QCI-full}. This figure may similarly be read as a commuting diagram.

\begin{figure}
\begin{centering}
\includegraphics[scale=0.35,angle=0]{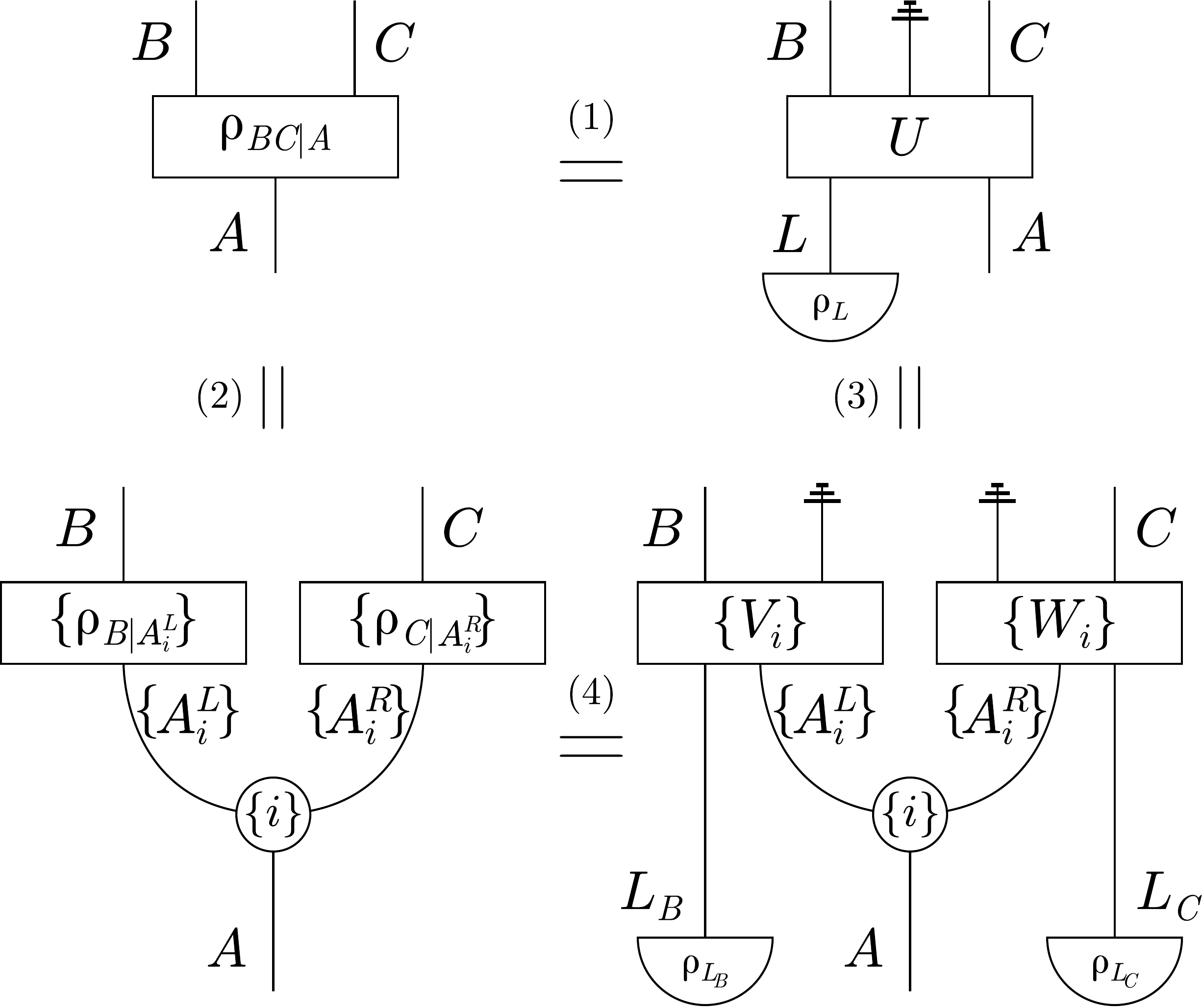}
\par\end{centering}
\caption{Circuit representations of a quantum channel $\rho_{BC|A}$ where $B$ and $C$ are conditionally independent given $A$ (equivalently, where $A$ could be a complete common cause for $B$ and $C$). The meaning of the symbol \qcopy{} is discussed in the text. \trace{} denotes taking the partial trace (discarding) of the indicated subsystem.}
\label{fig:CI:quantum-summary}
\end{figure}

Equality (1) of Fig.~\ref{fig:CI:quantum-summary} reaffirms that $\rho_{BC|A}$ has a unitary dilation $U$. The system being traced-out (the symbol \trace{}) is $F$ from Sec.~\ref{sec:CI:justifying-quantum-Reichenbach}.

Equality (2) expresses condition (4) of Thm.~\ref{thm:CI:QCI-full}. In order to do this the symbol \qcopy{} has been introduced. \qcopy{} may be read as re-expressing the Hilbert space of input system $A$ as a direct sum of factorising spaces $\mathcal{H}_A = \bigoplus_i \mathcal{H}_{A_i^L} \otimes \mathcal{H}_{A_i^R}$. The output wire $\{A_i^L\}$ then carries the left factor of each subspace, which is then acted upon by the relevant channel from $\{\rho_{B|A_i^L}\}$ (similarly for the other output wire). This is non-standard for quantum circuits, where wires normally correspond to a single Hilbert space and adjacent wires represent the tensor product of Hilbert spaces. Here, wires labelled as $\{A_i^L\}$ and $\{A_i^R\}$ only represent a whole Hilbert space when considered together as a pair. This has the advantage of allowing factorised channels $\{\rho_{B|A_i^L}\}$ and $\{\rho_{C|A_i^R}\}$ to act on separate wires. The interpretation of these symbols will be discussed further below.

Equality (4) then applies unitary dilations to each $\rho_{B|A_i^L}$ and $\rho_{C|A_i^R}$, such that all dilations (for each $i$) use a common ancillary input $L_B$ and $L_C$ respectively.

There are two ways to understand equality (3). First, simply as the conjunction of equalities (1), (2), and (4). Second, as expressing that the unitary $U$ has the no-causal-influence properties demanded by Thm.~\ref{thm:CI:QCI-full}. There is no direct way to express the lack of causal influence in a unitary diagrammatically. However, in the proof of Thm.~\ref{thm:CI:QCI-full} the existence of a unitary dilation of the form shown here is seen to be equivalent to the existence of a unitary with appropriate no-causal-influence properties. So the equalities (1) and (3) taken together might be viewed as illustrating the justification of quantum Reichenbach of Sec.~\ref{sec:CI:justifying-quantum-Reichenbach} (as in the classical case) but the correspondence is not immediate.

The introduction of the new symbol \qcopy{} to quantum circuits should, by rights, be accompanied by a thorough description of its meaning. However, in this case it is useful to leave the meaning somewhat ambiguous. For the purposes of this thesis \qcopy{} need only have meaning in the diagrams of Fig.~\ref{fig:CI:quantum-summary}. There are multiple ways to interpret the symbol such that those diagrams are well-defined and it would be premature to pick one above the others before a more general diagrammatic use of \qcopy{} has been found.

One interpretation of \qcopy{} is the \emph{passive} interpretation: it simply re-interprets the input Hilbert space as a direct sum of factorising spaces indexed by $i$, analogously to re-interpreting a single Hilbert space as a tensor product of two factors. Indeed, when $i$ only takes one value \qcopy{} \emph{is} just simple factorisation. In this view, \qcopy{} is entirely reversible and no physical operation occurs until the output wires are acted on. As noted above, the output wires only form a complete Hilbert space when taken together, but allow appropriate sets of channels to act independently on each factor.

Another valid interpretation of \qcopy{} is the \emph{active} interpretation: it represents a von Neumann measurement on $A$ defined by the linear subspaces labelled by $i$, followed by a factorisation of the output system depending on the outcome $i$. This is clearly a non-reversible physical operation.

A conditionally independent channel satisfying Thm.~\ref{thm:CI:QCI-full} is decohering across the linear subspaces labelled by $i$. Under the passive interpretation, this decoherence doesn't take hold until the sets of channels $\{\rho_{B|A_i^L}\}$ and $\{\rho_{C|A_i^R}\}$ are applied. Under the active interpretation on the other hand, the decoherence occurs at \qcopy{} due to the measurement. This allows the subsequent channels, $\rho_{B|A_i^L}$ \emph{etc.}, to be ordinary quantum channels, where the appropriate ones are selected depending on the outcome $i$.

Both interpretations are compatible with the definitions of Thm.~\ref{thm:CI:QCI-full}. The passive interpretation might be preferred as a description of channels that satisfy quantum conditional independence, while the active interpretation gives a concrete way to achieve quantum conditional independence operationally. Of course, there are other possible interpretations too.

If \qcopy{} is only used in Fig.~\ref{fig:CI:quantum-summary}, why should the interpretation matter? Comparing Figs.~\ref{fig:CI:classical-summary} and \ref{fig:CI:quantum-summary} suggests that \qcopy{} is, in some way, analogous to the classical copy \ccopy{}. It represents a general way for two agents to independently act on a single system $A$, just as classically two agents can independently use input $X$ by first copying it. There is certainly potential for fleshing out the relationship between \qcopy{} and \ccopy{} more generally.

The symbol \qcopy{} is most closely related to condition (4) of Thm.~\ref{thm:CI:QCI-full} and both active and passive interpretations can be applied to that definition. Given the above discussion, condition (4) seems to require channels where Bob and Clare can work independently on a single input. Importantly, it shows that there are more general ways of achieving this than simply factorising the quantum input. Condition (4) therefore appears to hint at a general characterisation of when agents can share a single input while working independently.

\subsection{Examples} \label{sec:CI:examples}

Definitions~\ref{def:CI:compatible-unitary-common-cause} and \ref{def:CI:quantum-conditional-independence} respectively defined common causes and conditional independence for quantum channels, which were then used to identify a quantum Reichenbach's principle. It is now time to consider some examples to see how these apply in practice, as well as to double check that the definitions given are reasonable.

\subsubsection{Generic Unitary Transformations} \label{sec:CI:example-unitary}

Consider the circuit of Fig.~\ref{fig:CI:AD-to-BC-unitary}(a) where systems $A$ and $D$ unitarily evolve to $B$ and $C$. Figure~\ref{fig:CI:AD-to-BC-unitary}(b) shows the causal structure one would expect for this basic situation. It would be very troubling indeed if the definition proposed in this chapter forbade $AD$ from being a complete common cause for $B$ and $C$ (that is, if $B$ and $C$ were not quantum conditionally independent given $AD$).

\begin{figure}
\begin{centering}
\includegraphics[scale=0.35,angle=0]{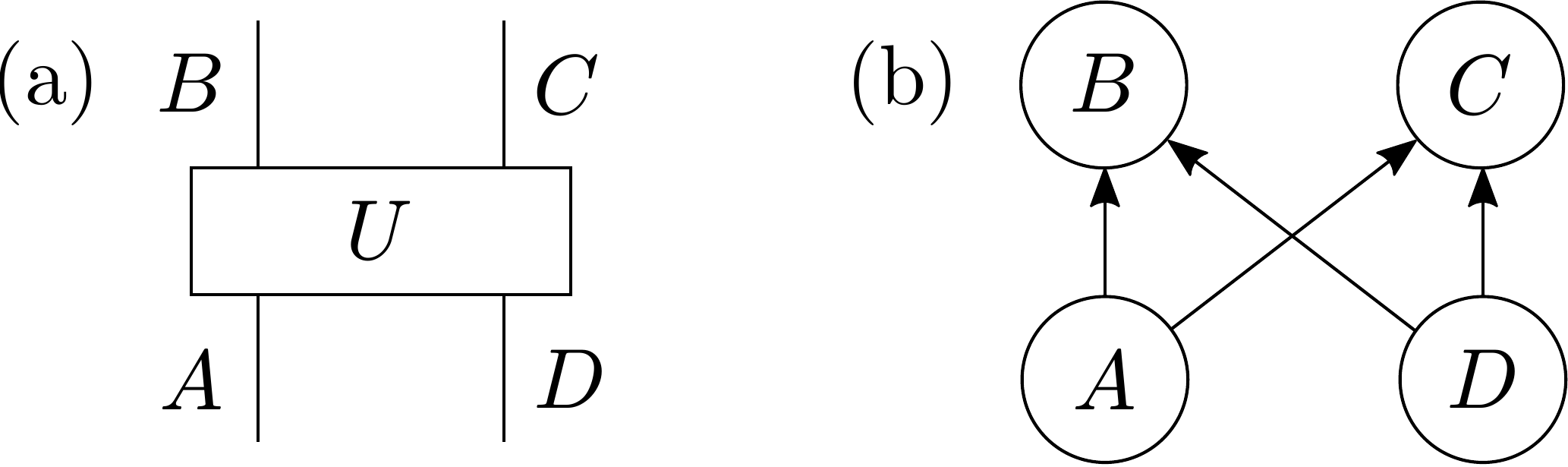}
\par\end{centering}
\caption{(a) A generic unitary transformation from $AD$ to $BC$. (b) The corresponding causal structure, making no assumptions about the properties of $U$. }
\label{fig:CI:AD-to-BC-unitary}
\end{figure}

Fortunately, this is not the case. The channel $\rho_{BC|AD}$ in this example is its own unitary dilation which trivially satisfies Def.~\ref{def:CI:compatible-unitary-common-cause}. Indeed, it is easy to verify directly that the other conditions of Thm.~\ref{thm:CI:QCI-full} are also satisfied. For example, $\hat{\rho}_{BC|AD}$ is a pure state where $AD$ is maximally entangled with $BC$, so $I(B:C|A)=0$. Therefore Def.~\ref{def:CI:quantum-conditional-independence} and Thm.~\ref{thm:CI:QCI-full} pass the most basic of sanity-checks.

\subsubsection{Coherent and Incoherent Copies} \label{sec:CI:example-copies}

Perhaps the simplest non-trivial example of a classical common cause is copying a bit. That is, the classical channel for binary random variables $X,Y,Z \in \{0,1\}$
\begin{equation} \label{eq:CI:classical-bit-copy}
\mathbb{P}(Y,Z|X=x) = \delta(Y,x)\delta(Z,x)
\end{equation}
which simply sets both outputs $Y$ and $Z$ to match the input $X$. Clearly this channel satisfies classical conditional independence and, perhaps most importantly, it is intuitively obvious that this should be the case: $Y$ and $Z$ simply take the value of $X$, so $X$ must explain all correlations between them.

One quantum generalisation for this classical channel is the \emph{incoherent copy} channel for qubits $d_A = d_B = d_C = 2$
\begin{equation}
\alpha |0\rangle_A + \beta |1\rangle_A \rightarrow |\alpha|^2|00\rangle_{BC}\langle 00| + |\beta|^2|11\rangle_{BC}\langle 11|.
\end{equation}
This channel reduces to the classical bit-copy in the case where inputs are diagonal in the $\{|0\rangle,|1\rangle\}$ basis. The Choi-Jamio\l{}kowski state for this channel is
\begin{equation}
\rho_{BC|A}^{\mathrm{inc}} \eqdef |000\rangle_{BCA^\ast}\langle 000| + |111\rangle_{BCA^\ast}\langle 111|.
\end{equation}
It is easy to verify that this satisfies all of the conditions of Thm.~\ref{thm:CI:QCI-full} and therefore satisfies quantum conditional independence according to Def.~\ref{def:CI:quantum-conditional-independence}. This should be unsurprising, the incoherent copy is essentially a classical channel which removes any quantum coherence.

There is, however, another easy quantum generalisation for the classical bit-copy. That is, the \emph{coherent copy} channel
\begin{equation} \label{eq:CI:coherent-copy-channel}
\alpha|0\rangle_A + \beta|1\rangle_A \rightarrow \alpha|00\rangle_{BC} + \beta|11\rangle_{BC}
\end{equation}
defined by Choi-Jamio\l{}kowski operator
\begin{equation}
\rho_{BC|A}^{\mathrm{coh}} \eqdef \left( |000\rangle + |111\rangle \right)_{BCA^\ast}\left( \langle 000| + \langle 111| \right).
\end{equation}
This channel also reduces to the classical bit-copy in the same diagonal cases. However, this channel fails to satisfy the conditions of Thm.~\ref{thm:CI:QCI-full} and therefore $B$ and $C$ are not conditionally independent given $A$. In particular, it is easy to verify that $I(B:C|A)=1$ on $\hat{\rho}_{BC|A}^{\mathrm{coh}}$. $A$ cannot be considered a complete common cause for $B$ and $C$ in the coherent copy.

This should seem troubling. Just as in the classical case, $B$ and $C$ take the value of $A$ in Eq.~(\ref{eq:CI:coherent-copy-channel}), so how can $A$ fail to be a complete common cause? There are at least three ways to understand this apparent discrepancy and see that, far from being a failure of quantum conditional independence, it is exactly what should hold.

The first is to recall the discussion of Sec.~\ref{sec:CI:circuits}, where it was noted that conditionally independent channels capture the idea that two agents can ``act independently'' on a single input. Classically, this is easy: simply copy the input and give one to each agent. In quantum theory, this is achieved by using \qcopy{} in place of the classical copy. The coherent copy creates arbitrary entanglement between two qubits from a single qubit. From this perspective, if Eq.~(\ref{eq:CI:coherent-copy-channel}) satisfied quantum conditional independence then agents would be able to create entanglement by acting independently and without communication on a single qubit. This would certainly be surprising, perhaps even perverse (for one, it would smell of remote entanglement preparation \cite{BennettHayden+05}), and therefore it is encouraging that Thm.~\ref{thm:CI:QCI-full} rules it out.

The second is to consider the deterministic/unitary dilations of the channels Eqs.~(\ref{eq:CI:classical-bit-copy}, \ref{eq:CI:coherent-copy-channel}), shown in Figs.~\ref{fig:CI:classical-bit-copy} and \ref{fig:CI:coherent-copy} respectively.  Classically, the causal influence in a CNOT gate is unidirectional, as shown in Fig.~\ref{fig:CI:classical-bit-copy}(b). The value at $\Lambda$ has no causal influence on $Y$. However, this is not the case for the quantum CNOT gate, which has a back-action (see also Ref.~\cite{SchumacherWestmoreland12}). This is reflected in Fig.~\ref{fig:CI:coherent-copy}, where $L$ has a causal influence on both $B$ and $C$ and therefore is an additional common cause (preventing $A$ from being the complete common cause).

\begin{figure}
\begin{centering}
\includegraphics[scale=0.35,angle=0]{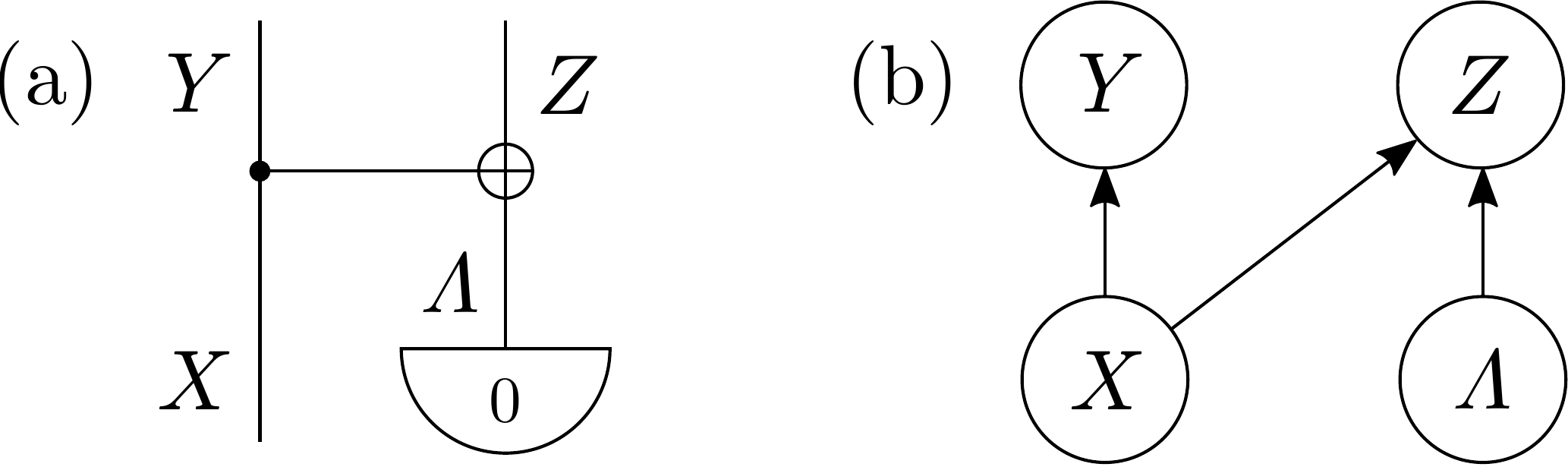}
\par\end{centering}
\caption{(a) Classical deterministic dilation of the bit-copy from $X$ to $Y$ and $Z$. The ancilla $\Lambda$ is prepared with the value $0$ then a CNOT gate is applied, controlled on $X$. (b) The corresponding causal structure of this dilation, note that there is no causal influence from $\Lambda$ to $Y$.}
\label{fig:CI:classical-bit-copy}
\end{figure}

\begin{figure}
\begin{centering}
\includegraphics[scale=0.35,angle=0]{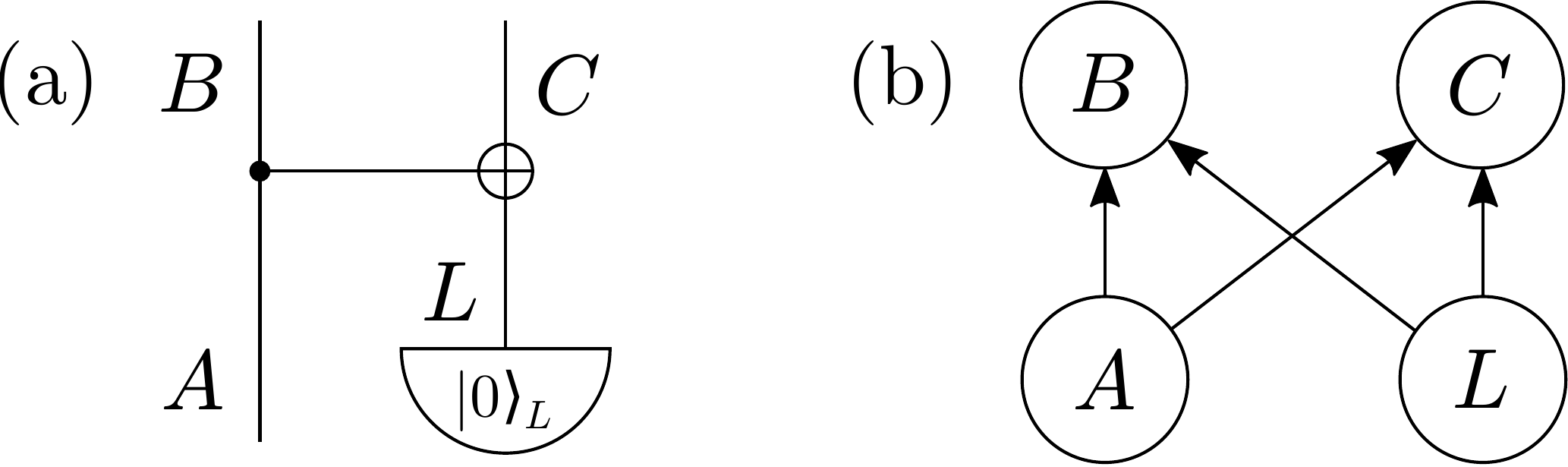}
\par\end{centering}
\caption{(a) Quantum unitary dilation of the coherent copy from $A$ to $BC$. (b) The corresponding causal structure of this dilation, note the explicit back-action from $L$ to $B$.}
\label{fig:CI:coherent-copy}
\end{figure}

These diagrams show just one example of a dilation for each channel (though arguably the most natural), but similar remarks hold for other possible dilations. In particular, it is not possible to find a unitary dilation of Eq.~(\ref{eq:CI:coherent-copy-channel}) where the ancilla does not act as an additional common cause for $B$ and $C$.

One might worry that this back-action should not matter in the quantum CNOT, since the ancilla $L$ always takes the state $|0\rangle_L$. But this does not change the fact that there is no way to achieve Eq.~(\ref{eq:CI:coherent-copy-channel}) unitarily without introducing an additional common cause. Moreover, for those who take pure quantum states to represent ``maximal but incomplete'' information about the system \cite{CavesFuchs+02b,Fuchs02,Spekkens07,Leifer06,Spekkens16} a pure quantum state is not analogous to a classical point distribution\footnote{Indeed, in Spekkens' Toy Model \cite{Spekkens07} (a subset of quantum theory which explicitly models quantum states as maximal but incomplete information and the CNOT as a deterministic map, briefly discussed in Sec.~\ref{sec:SO:classifying-ontologies}) classical conditional independence fails for the coherent copy channel.}.

Finally, the third way to understand why the coherent copy fails to satisfy conditional independence is via the Bayesian updating procedure discussed in Sec.~\ref{sec:CM:Bayesian-updating}. The details of this example are discussed in that section, for now it suffices to say that if $A$ is a complete common cause for $B$ and $C$, then any knowledge about $C$ gained at $B$ should be expressible in terms of knowledge about $A$. Classically, this is always the case for a complete common cause and Sec.~\ref{sec:CM:Bayesian-updating} demonstrates the same for quantum causal models. However, this is not generally possible for the coherent copy: $B$ can (in general) locally learn more about $C$ than can be expressed via $A$. This is a key property of a common cause that fails due to the entanglement between $B$ and $C$.

So the quantum incoherent copy can be the result of a common cause, just as the classical copy channel, while the coherent copy cannot. Reviewing the reasons given above, the core of this discrepancy seems to come from the strength of entanglement that the coherent copy can create and how another qubit is necessary for its creation.

\subsubsection{Bell Experiments}

In Sec.~\ref{sec:CI:introduction-Bell}, Bell experiments were given as one of the main motivations for quantum conditional independence. Simply put, Reichenbach's principle using classical conditional independence fails in these experiments. It is therefore fitting to check how the notion of quantum conditional independence introduced in Sec.~\ref{sec:CI:justifying-quantum-Reichenbach} deals with Bell experiments.

A standard Bell experiment is of the following form. Dave prepares a pair of qubits in the entangled singlet state $|\Psi^-\rangle \eqdef \left( |0\rangle|1\rangle - |1\rangle|0\rangle \right)/\sqrt{2}$ and then gives the whole system to Alice. Alice then noiselessly distributes the qubits, one to Bob and one to Clare. Bob and Clare are each prevented from any means of causal influence on the other for the whole experiment (this is normally done by invoking relativistic locality and ensuring they are spacelike separated). Bob and Clare each independently decide on a choice of measurement (from a pre-determined set of possibilities) on their qubit which they perform, obtaining outcomes $B_r$ and $C_r$ respectively. This scenario is illustrated in Fig.~\ref{fig:CI:bell-experiment}. Note that the slightly unusual step has been taken here of separating the preparation of $|\Psi^-\rangle$ (by Dave) from the distribution of the qubits (by Alice).

The key result from Ref.~\cite{WoodSpekkens15} is that in general there is no classical causal model (and certainly no classical common cause) that explains the correlations between $B_r$ and $C_r$ without fine-tuning. A satisfying result would be that $B_r$ and $C_r$ are appropriately quantum conditionally independent, since the experimental set-up seems to demand an explanation in terms of a common cause from Alice and Dave.

\begin{figure}
\begin{centering}
\includegraphics[scale=0.35]{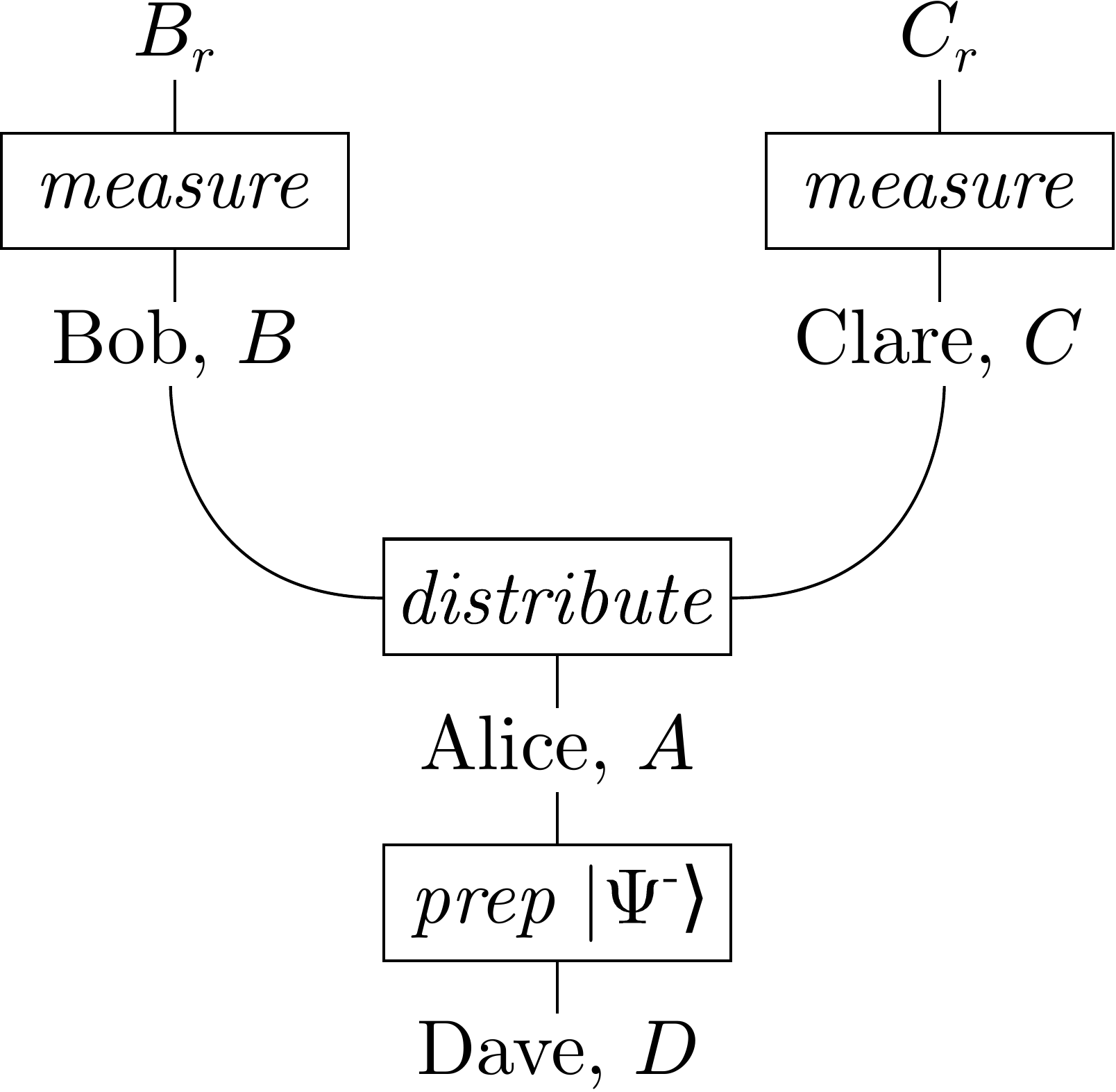}
\par\end{centering}	
\protect\caption{Schematic representation of a Bell experiment as described in the text. The measurements of Bob and Clare are chosen by them independently from a set of predefined choices.}
\label{fig:CI:bell-experiment}
\end{figure}

While $B_r$ and $C_r$ are experimental outcomes, it is necessary to consider them as encoded in quantum systems to apply quantum conditional independence as defined here. That is, the outcomes $B_r$ and $C_r$ are just represented as diagonal density operators of quantum systems in some obvious way.

It immediately follows that $B$ and $C$ are quantum conditionally independent given $A$, since the channel $\rho_{BC|A}$ that describes Alice distributing the qubits to Bob and Clare is simply the identity channel. So conditional independence follows from Sec.~\ref{sec:CI:example-unitary}. Since the measurement procedures of Alice and Bob are entirely factorised, it follows that $B_r$ and $C_r$ must also be quantum conditionally independent given $A$. It is simple to verify from Thm.~\ref{thm:CI:QCI-full} that composing a factorised channel (such as these measurements) after a quantum conditionally independent channel must result in another conditionally independent channel.

However, if one instead asks whether $B$ and $C$ (or $B_r$ and $C_r$) are quantum conditionally independent given $D$, then the answer depends what system Dave starts with to prepare $|\Psi^-\rangle$. If, for example, Dave starts with a single qubit, then the composite channel $\rho_{BC|D}$ is a variant of the coherent copy and is therefore not conditionally independent. However, if Dave starts with a pair of qubits and unitarily evolves them to $|\Psi^-\rangle$, then the composite channel $\rho_{BC|D}$ is unitary and the channel satisfies conditional independence as above.

Should this be troubling? Not at all. This analysis shows that statistics from Bell experiments are quantum conditionally independent because all factorised quantum measurements are quantum conditionally independent, even when made on entangled states. This is as it should be, since conditional independence is a property of the channel and not of the channel-with-input. On the other hand, the cases where $B$ and $C$ (and therefore $B_r$ and $C_r$) are not quantum conditionally independent given $D$ are those where Dave's starting system is too simple to underwrite the correlations. But this also occurs classically: condition on too simple an ancestor and variables that were conditionally independent are rendered dependent again. One only expects conditional independence when conditioning on a complex enough ancestor.

\subsection{Generalisation to \texorpdfstring{$k$}{k} Outputs} \label{sec:CI:k-outputs-generalisation}

Theorem~\ref{thm:CI:QCI-full} only considers quantum channels from a single input $A$ to a pair of outputs $BC$. It is natural to ask whether a similar set of equivalences hold for channels from a single input $A$ to an arbitrary finite set $\{B_l\}_{l=1}^k$ of $k$ outputs. This is more than idle curiosity; it is a first step towards generalising to full quantum causal models, defined in Chap.~\ref{ch:CM}.

Consider such a channel $\rho_{B_1\ldots B_k | A}$, and let $\bar{B}_l$ denote the whole joint output system excepting only $B_l$. Rather than starting from scratch, consider the following natural generalisation of Def.~\ref{def:CI:compatible-unitary-common-cause}. 

\begin{definition} \label{def:CI:compatible-unitary-k-outputs}
A quantum channel $\rho_{B_1\ldots B_k | A}$ is \emph{compatible with $A$ being the unitary complete common cause for $B_1,\ldots, B_k$} if and only if there is
\begin{enumerate}[label=(\alph*)]
\item a unitary dilation $U$ of the channel in terms of ancillae $L_1,\ldots, L_k$ with factorised initial state $\rho_{L_1}\otimes\cdots\otimes\rho_{L_k}$ such that
\item for each $l$, $L_l$ has no causal influence on $\bar{B}_l$.
\end{enumerate}
\end{definition}

Beyond being perhaps the easiest generalisation of Def.~\ref{def:CI:compatible-unitary-common-cause}, there is not yet any reason to suppose that this is a good definition. Its strength comes from the following theorem, which generalises Thm.~\ref{thm:CI:QCI-full}.

\begin{theorem} \label{thm:CI-QCI-multipartite-full}
Given a quantum channel $\rho_{B_1 \ldots B_k | A}$, the following are equivalent:
\begin{enumerate}
\item  $\rho_{B_1\ldots B_k|A}$ is compatible with $A$ being a complete common cause of $B_1,\ldots, B_k$.
\item $\rho_{B_1\ldots B_k | A} = \rho_{B_1|A} \cdots \rho_{B_k|A}$, where $[\rho_{B_l|A}, \rho_{B_m|A}]=0$ for all $l$ and $m$.
\item $I(B_l : \bar{B}_l | A) = 0$ for all $l$ when evaluated on the (positive, trace-one) operator $\hat{\rho}_{B_1\ldots B_k|A}$.
\item The Hilbert space of $A$ has a decomposition $\mathcal{H}_A = \bigoplus_i \mathcal{H}_{A_i^1}\otimes\cdots\otimes\mathcal{H}_{A_i^k}$ for which $\rho_{B_1\ldots B_k|A} = \sum_i \left(\rho_{B_1|A_i^1}\otimes\cdots\otimes\rho_{B_k|A_i^k}\right)$, where for each $i$ and $l$, $\rho_{B|A_i^l}$ is a quantum channel from $A_i^l$ to $B_l$.
\end{enumerate}
Each of these conditions is an equivalent definition for when \emph{$B_1,\ldots,B_k$ are quantum conditionally independent given $A$} in a channel $\rho_{B_1\ldots B_k|A}$.
\end{theorem}

Defining quantum conditional independence with these conditions is done by analogy with the classical and $k=2$ cases. Again, each condition is an equally valid definition and there is no particular \emph{a priori} reason to privilege one above any other. The whole theorem reduces to Thm.~\ref{thm:CI:QCI-full} if $k=2$. The proof follows very similarly to that of Thm.~\ref{thm:CI:QCI-full} and is presented here in parts.

\begin{proof}[Proof: (3) $\Rightarrow$ (2)]
The proof proceeds by induction and repeated use of Thm.~\ref{thm:CI:QCI-full}. To this end, suppose the result holds for $\rho_{B_1\ldots B_n |A}$ for some $n < k$ (the inductive hypothesis) and that $I(B_{n+1} : \bar{B}_{n+1} |A) = 0$ as in condition (3).

By Thm.~\ref{thm:CI:QCI-full}, this implies that
\begin{equation}
\rho_{B_1 \ldots B_{n+1} |A} = \rho_{B_1 \ldots B_n | A}\rho_{B_{n+1}|A}
\end{equation}
where $[\rho_{B_1 \ldots B_n}, \rho_{B_{n+1}}] = 0$. The inductive hypothesis implies
\begin{equation}
\rho_{B_1 \ldots B_n |A} = \rho_{B_1|A} \cdots \rho_{B_n|A}
\end{equation}
where the $\rho_{B_j}$ factors pairwise commute. Together these imply the factorisation of condition (2) for $\rho_{B_1 \ldots B_{n+1}|A}$. To show that $\rho_{B_{n+1}}$ commutes with each $\rho_{B_l}$, simply note that $\rho_{B_1 \ldots B_n} \rho_{B_{n+1}} = \rho_{B_{n+1}} \rho_{B_1 \ldots B_n}$ and trace over all $B_{m\neq l}$. Therefore, using Thm.~\ref{thm:CI:QCI-full} as the base case the result holds for all $k \geq 2$ by induction.
\end{proof}

\begin{proof}[Proof: (2) $\Rightarrow$ (3)]
This follows immediately from Thm.~\ref{thm:CI:QCI-full} by letting $B = B_l$ and $C = \bar{B}_l$ for each $l$.
\end{proof}

\begin{proof}[Proof: (3) $\Rightarrow$ (4)]
Assume that (3) holds, so that $I(B_l : \bar{B}_l | A) = 0$ for all $l$. Since quantum conditional mutual information cannot increase if systems are discarded and is non-negative \cite{NielsenChuang00} it follows that
\begin{equation}
0 = I(B_2 : B_1, B_3, \ldots, B_k | A) = I(B_2 : B_3, \ldots, B_k | A).
\end{equation}

By Thm.~\ref{thm:CI:QCI-full}, $I(B_1 : \bar{B}_1 | A)=0$ implies a decomposition $\mathcal{H}_A = \bigoplus_i \mathcal{H}_{A_i^L} \otimes \mathcal{H}_{A_i^R}$ such that
\begin{eqnarray}
\rho_{B_1 \ldots B_k |A} &= & \sum_i \rho_{B_1 | A_i^L} \otimes \rho_{B_2 \ldots B_k | A_i^R} \\
\Rightarrow \hat{\rho}_{B_2 \ldots B_k | A} &= & \sum_i \frac{p_i}{d_{A_i^L}} \mathbbm{1}_{A_i^L} \otimes \hat{\rho}_{B_2 \ldots B_k | A_i^R}
\end{eqnarray}
where $p_i = d_{A_i^L}d_{A_i^R}/d_{A}$ form a probability distribution. The terms in $\hat{\rho}_{B_2\ldots B_k|A}$ only have support on orthogonal subspaces, so \cite{NielsenChuang00}
\begin{eqnarray} 
S(\hat{\rho}_{B_2 \ldots B_k|A}) &= & H(\{p_i\}) + \sum_i p_i \log d_{A_i^L} + \sum_i p_i S(\hat{\rho}_{{B}_2 \ldots B_k|A_i^R}), \\
S(\hat{\rho}_{B_2 |A}) &= & H(\{p_i\}) + \sum_i p_i \log d_{A_i^L} + \sum_i p_i S(\hat{\rho}_{B_2 |A_i^R}), \\
S(\hat{\rho}_{B_3 \ldots B_k|A}) &= & H(\{p_i\}) + \sum_i p_i \log d_{A_i^L} + \sum_i p_i S(\hat{\rho}_{B_3 \ldots B_k|A_i^R}), \\
S(\hat{\rho}_{\cdot |A}) &= & H(\{p_i\}) + \sum_i p_i \log d_{A_i^L} + \sum_i p_i S(\hat{\rho}_{\cdot |A_i^R}). 
\end{eqnarray}
Substituting into
\begin{equation}
I(B_2:B_3,\ldots,B_k|A) \eqdef S(\hat{\rho}_{B_2|A}) + S(\hat{\rho}_{B_3\ldots B_k|A}) - S(\hat{\rho}_{B_2 \ldots B_k|A}) - S(\hat{\rho}_{\cdot |A})
\end{equation}
the $H(\{p_i\})$ terms and the $\sum_i p_i \log d_{A_i^L}$ terms cancel, leaving 
\begin{equation}
I(B_2:B_3,\ldots,B_k|A) = \sum_i p_i I(B_2:B_3,\ldots,B_k|A_i^R) = 0.
\end{equation}
Therefore $I(B_2:B_3\ldots B_k|A_i^R)=0$ for each $i$ by non-negativity and Thm.~\ref{thm:CI:QCI-full} can be applied to $\rho_{B_2\ldots B_k|A_i^R}$. Iterating this procedure yields the required decomposition.
\end{proof}

\begin{proof}[Proof: (4) $\Rightarrow$ (3)]
This also follows immediately from Thm.~\ref{thm:CI:QCI-full} by letting $B = B_l$ and $C = \bar{B}_l$ for each $l$.
\end{proof}

\begin{proof}[Proof: (1) $\Rightarrow$ (4)]
Assume that condition (1) holds. Note from Def.~\ref{def:CI:quantum-no-causal-influence} that as there is no causal influence from $L_l$ to $\bar{B}_l$ for every $l$, there is also no causal influence from $\bar{L}_l$ to $B_l$ for every $l$. Therefore, by partitioning $B_1 \cdots B_k$ into $B_l$ and $\bar{B}_l$ and partitioning $L_1 \cdots L_k$ into $L_l$ and $\bar{L}_l$, Thm.~\ref{thm:CI:QCI-full} implies that $I(B_l : \bar{B}_l |A) = 0$. Thus, since $l$ is arbitrary, this implies (3) which further implies (4).
\end{proof}

\begin{proof}[Proof: (4) $\Rightarrow$ (1)]
Follows by a straightforward extension of the corresponding proof in Sec.~\ref{sec:CI:proving-QCI-conditions}.
\end{proof}

Together, these	steps can be used to show that any condition of Thm.~\ref{thm:CI-QCI-multipartite-full} implies any other, thus completing the proof.

As well as being interesting in its own right, Thm.~\ref{thm:CI-QCI-multipartite-full} also places further confidence in the quantum conditional independence defined here being correct. The fact that the simplest multipartite generalisation of Thm.~\ref{thm:CI:QCI-full} also holds is a testament to the naturalness of Thm.~\ref{thm:CI:QCI-full} and Def.~\ref{def:CI:quantum-conditional-independence}.

\section{Summary}

This chapter started to consider the nature of ontological causal influences in quantum theory. A natural starting point is Reichenbach's principle, a cornerstone of causal reasoning but one that fails in certain quantum experiments. In particular, classical causal models fail to adequately explain Bell experiments \cite{Bell87,WoodSpekkens15}. Since classical causal models generalise Reichenbach's principle, it is sensible to first find a natural quantum version of the principle before attempting to define a full framework of quantum causal models.

Finding a quantum Reichenbach's principle becomes easier if the original is split into two parts: qualitative and quantitative [Sec.~\ref{sec:CI:Reichenbach-bipartite}]. The qualitative part can be commuted to quantum theory with almost no change. The quantitative part applies specifically to the case of a complete common cause.

One way to derive the classical quantitative part is by assuming fundamentally deterministic dynamics [Sec.~\ref{sec:CI:justifying-classical-Reichenbach}]. An analogous procedure was followed to justify a quantum quantitative part by assuming fundamental unitarity. This gave rise to Def.~\ref{def:CI:quantum-conditional-independence} of quantum conditional independence of $B$ and $C$ given $A$ for a channel $\rho_{BC|A}$. Such quantum conditional independence gave a natural quantum quantitative part and therefore a full quantum Reichenbach's principle [Sec.~\ref{sec:CI:justifying-quantum-Reichenbach}].

Classical conditional independence can be defined in many equivalent ways, each lending itself to different intuitions. It was then proposed that the same can be said for quantum conditional independence of a channel (Props.~\ref{prop:CI:quantum-conditional-independence}, \ref{prop:CI:QCI-equivalents}). In Thm.~\ref{thm:CI:QCI-full}, the four proposed definitions were shown to be equivalent. Each of these is a generalisation of a definition for classical conditional independence and each reduces to the corresponding classical case when states and channels appropriately decohere into some choice of ``classical'' bases.

One of these definitions, condition (4) of Thm.~\ref{thm:CI:QCI-full}, is particularly interesting. It identifies a particular structure that conditionally independent channels must follow. In Sec.~\ref{sec:CI:circuits} it was shown that this structure has an operational interpretation and plays a similar role to the classical copying map. It appears that this structure generally allows two agents to ``act independently'' on a single input.

Theorem~\ref{thm:CI-QCI-multipartite-full} generalised quantum conditional independence to the case of a channel with $k>2$ outputs. It showed that the easy generalisations of the definitions in Thm.~\ref{thm:CI:QCI-full} hold in this more general scenario, giving extra confidence that those definitions are correct.

Some examples were considered in Sec.~\ref{sec:CI:examples}. After verifying that the proposed definitions act appropriately in some basic cases, particular attention was given to Bell experiments since they were a key motivation for developing quantum Reichenbach in the first place. It was shown that the outcomes of Bell experiments can be considered to arise from a quantum common cause, even though Bell's theorem shows they cannot arise from a classical common cause \cite{Bell87}. The conditional independence does depend on what is chosen to be conditioned on, but this is also the case for any classical conditionally independent channel.

Having motivated, defined, and tested quantum Reichenbach's principle, the next step is to follow the classical example and generalise it to a full framework of quantum causal models. This will follow in the next chapter. A full discussion of these results, with open questions for further investigation, is deferred until Sec.~\ref{sec:CM:summary-and-discussion} once the full framework of quantum causal models has been defined.

\chapter{Quantum Causal Models} \label{ch:CM}

\section{The Need for Quantum Causal Models}

The previous chapter began to consider causal influences as an ontological issue in quantum theory. By looking at the simple scenario of a complete common cause it was discovered that a significant revision of the classical ways of describing causation was required. Such a revision was provided for that scenario in the form of a quantum Reichenbach's principle, utilising a notion of quantum conditional independence. In this chapter, these discussions are extended to more general causal scenarios. Much of the work presented in this chapter has been published in Ref.~\cite{AllenBarrett+17}.

\subsection{The Reality of Causal Influences}

Reichenbach's principle is a powerful link between the probabilistic/statistical notion of correlation and the notion of a causal influence between two events. It supports the idea that causal influences are ``real,'' even for those (such as probabilistic Bayesians) who prefer to think of probabilities as purely subjective. 

But Reichenbach's principle only applies to some very simple causal scenarios, with two events of interest. A natural generalisation of the principle is found in the framework of classical causal models \cite{Pearl09,SpirtesGlymour+01}, which describes arbitrary causal structures and their relationship with possible probability distributions. Just as Reichenbach's principle can support the reality of causal influences in its limited realm of applicability, so do classical causal models more generally\footnote{Another important result in this regard is the de Finetti theorem \cite{deFinetti93,DiaconisFreedman80,deFinetti75} (which also has quantum counterparts \cite{BrandaoHarrow13,BarrettLeifer09,ChristandlToner09,ChristandlKoenig+07,CavesFuchs+02c}). Amongst other things, this theorem provides subjective Bayesians justification for behaving \emph{as if} the results of ancestrally independent trails are due to a single unknown probability distribution.}. 

A classical causal model, defined fully in Sec.~\ref{sec:CM:classical-CM}, comes in two parts. First, the \emph{causal structure} simply specifies what causal influences exist between events (or \emph{nodes}). Simple causal structures have already been encountered in the previous chapter, \emph{e.g.} Figs.~\ref{fig:CI:X-to-YZ-dilated}, \ref{fig:CI:coherent-copy}(b). Second, the probability distribution over these events compatible with the causal structure. The framework of classical causal models specifies which causal structures are permissible and which distributions are compatible with any given causal structure. Note the similarity to Reichenbach's principle. The qualitative part simply demands the causal structure take some particular form, while the quantitative part ensures the distribution is compatible. This is no accident: Reichenbach's principle becomes a corollary of the more general framework of causal models.

\subsection{Causation in General Quantum Experiments} \label{sec:CM:general-quantum-experiments}

Given the discussion of the previous chapter, one should expect classical causal models to be insufficient for describing general quantum experiments. It has already been noted in Sec.~\ref{sec:CI:introduction-Bell} that no classical causal model can describe a general Bell experiment without fine-tuning. But Bell experiments are specifically designed to show that a simple classical common cause explanation fails. Are there other types of quantum experiments that also resist explanation by a classical causal model?

Several such experiments have been considered in the literature \cite{GreenbergerHorne+89,BranciardGisin+10,HobanWallman+11,BranciardRosset+12,Fritz12,TavakoliSkrzypczyk+14}. It should perhaps be unsurprising that many are variations or generalisations of Bell experiment set-ups. In any case, the general point is clear: there are many conceivable quantum experiments that resist natural explanation by classical causal models\footnote{Here ``natural'' might mean ``in the absence of fine-tuning'' or ``in the absence of superluminal causal influences'' or similar reasonable restrictions depending on the exact experiment considered.}.

Just as in the case of Bell experiments and Reichenbach's principle, these experiments demand a revision of classical ideas of causality (each in their own way). Classical causal models are simply insufficient to adequately explain many quantum experiments. A satisfactory revision would ideally describe all such experiments without fine tuning, superluminal influences, or other similar undesirable features.

Beyond this, there is another lesson to take from the search for novel Bell-like experiments. Several techniques have been developed for deriving bounds on the observable statistics from particular classical causal models \cite{LeeSpekkens15,WolfeSpekkens+16,Chaves16,RossetBranciard+16}. It has been noted that, given an analogous framework of quantum causal models, these techniques might be adapted to derive similar bounds for achievable quantum statistics from causal models with the same causal structures \cite{ChavesLuft+14a,ChavesMajenz+15,ChavesKueng+15}. Then, by seeing where the quantum bounds differ from the corresponding classical bounds, one could identify experiments where quantum statistics violate classical causal models. A framework of quantum causal models would hopefully enable this to be done systematically.

Such a framework of quantum causal models will be presented in Sec.~\ref{sec:CM:quantum-cm-generalise-reichenbach}. This is found by generalising quantum Reichenbach as given in the previous chapter, just as classical causal models generalise Reichenbach's principle.

The quantum causal models presented here are entirely natively quantum. They are not classical causal models with quantum additions. They do, however, reduce to the classical case in appropriate decohering limits [Sec.~\ref{sec:CM:classical-limit}]. The framework builds on the foundation of quantum Reichenbach from Chap.~\ref{ch:CI} and contains it as a special case. These features strengthen the claim that this is an appropriate definition of quantum causal models and speak to a certain degree of robustness.

Use of the framework will be illustrated in Sec.~\ref{sec:CM:examples} by examples. These include the case of a confounding common cause and simple Bayesian updating across a common cause. Finally, Sec.~\ref{sec:CM:summary-and-discussion} will summarise the chapter and discuss the implications from both this chapter and the previous.

Just as in Chap.~\ref{ch:CI}, all random variables and graphs will be assumed to be finite. Similarly, all quantum systems will be assumed to be finite-dimensional. This is done for clarity and no major conceptual changes needed to extend the results to infinite cases are anticipated.

\subsection{Previous Approaches to Quantum Causal Models} \label{sec:CM:previous-attempts}

This is certainly not the first time that moves in the direction of quantum causal models have been made\footnote{That is probably the work of Refs.~\cite{Tucci95,Tucci12}, where transition amplitudes are taken to replace conditional probabilities from classical causal models.}. This section will review some of this previous and related work to put the results of this chapter into their proper context.

An early approach to causality in quantum theory started small, with a pair of quantum systems $B,C$ acted on by a single quantum operation. This led to defining the properties ``(semi)-causal'' and ``(semi)-localisable'' for such operations, as well as establishing relationships between these properties \cite{BeckmanGottesman+01,EggelingSchlingemann+02,PianiHorodecki+06}. This was extended in Ref.~\cite{SchumacherWestmoreland05} to the tripartite case $BAC$ with very strong results on the structure of ``local'' unitaries. These results are related to the channels with ``no causal influence'' used in Secs.~\ref{sec:CI:justifying-quantum-Reichenbach}, \ref{sec:CI:proving-QCI-conditions}.

References~\cite{LeiferPoulin08,LeiferSpekkens13} made strides towards quantum causal models by viewing density operator quantum theory as a non-commutative generalisation of probability theory. There, density operators took the role of probability distributions and Choi-Jamio\l{}kowski operators took the role of conditional distributions. This is very similar to the way quantum Reichenbach was approached in the previous chapter. Formally, there is a certain amount of overlap, with some results from Ref.~\cite{LeiferPoulin08} being used in proofs in Sec.~\ref{sec:CI:proving-QCI-conditions}. However, Refs.~\cite{LeiferPoulin08,LeiferSpekkens13} primarily aimed to develop ``causally neutral'' quantum causal models that could be used for Bayesian inference, like classical causal models. This is rather different from traditional approaches to quantum theory, where spacelike and timelike relationships between systems are fundamentally different. This focus on quantum-theory-as-probability-theory led to causal models of a different character and the programme has been hampered by difficulties in defining quantum-state-like objects for timelike separated systems \cite{HorsmanHeunen+16}.

The programme of deriving Bell-like bounds in more general causal structures was mentioned in Sec.~\ref{sec:CM:general-quantum-experiments}. This inspired independent formulations of quantum causal models in Ref.~\cite{HensonLal+14} and Ref.~\cite{Fritz16} with the aim of providing frameworks for deriving such bounds. Both approaches are based on some underlying operational theory: the arrows of the causal structure are taken to represent systems of that theory while the nodes (events) represent transformations. The operational framework of \emph{general probabilistic theories} (GPTs, \cite{Hardy01,Hardy11,ChiribellaDAriano+10,ChiribellaDAriano+11}) is the basis for Ref.~\cite{HensonLal+14}, while Ref.~\cite{Fritz16} applies to theories based on symmetric monoidal categories. Both of these are more general operational theories that include classical and quantum as special cases. Reference~\cite{PienaarBrukner15} constructs specifically quantum causal models compatible with Ref.~\cite{Fritz16}.

Both frameworks provide some unification due to their generality. They also prove some novel results (notably, analogues to d-separation theorems in Refs.~\cite{HensonLal+14,PienaarBrukner15}). However, neither approach is natively quantum. They only define conditional independence between observed classical outcomes, rather than between quantum systems as in Thm.~\ref{thm:CI:QCI-full}. In particular, one cannot condition on a quantum node in those frameworks. The framework defined in this chapter, by contrast, will generalise quantum Reichenbach and will apply directly to quantum systems just as that principle did.

The approaches of Refs.~\cite{HensonLal+14,Fritz16,PienaarBrukner15} demonstrate that there are close links between quantum causal models and operational formulations of quantum theory. This is especially true of those operational formulations that consider relativistic causal structure, including the causaloid framework \cite{Hardy09}, the multi-time formalism \cite{AharonovVaidman07,AharonovPopescu+09,AharonovPopescu+14,SilvaGuryanova+14}, quantum combs \cite{ChiribellaDAriano+09,Chiribella12,ChiribellaDAriano+13}, categorical quantum theory \cite{CoeckeKissinger17,CoeckeLal13}, and process matrices \cite{OreshkovCosta+12,AraujoCosta+14}. In particular, many of these frameworks \cite{AharonovPopescu+09,ChiribellaDAriano+09,OreshkovCosta+12} use a pair of isomorphic Hilbert spaces to represent a system at a point in spacetime: an ``input'' system which is received then transformed to an ``output'' system. This has been particularly useful for describing local \emph{interventions}, where the transformation from input to output may be freely chosen. This approach to interventions will be used in the nodes of quantum causal models later.

A similar interventional approach to nodes in causal models has also been considered classically \cite{RichardsonRobins13}, as noted in Ref.~\cite{CostaShrapnel16}. This was done to describe possible counterfactual interventions and similar properties make the quantum interventional approach useful here.

The previous work that most closely resembles the quantum causal models defined here is Ref.~\cite{CostaShrapnel16}. There, a definition for quantum causal models is given based on the process matrix formalism \cite{OreshkovCosta+12,AraujoCosta+14}. As in this chapter, the dual-Hilbert-space interventional method noted above is used for nodes in the causal structure. Nonetheless, there are important differences between quantum causal models as defined here and those of Ref.~\cite{CostaShrapnel16} which will be discussed in Sec.~\ref{sec:CM:summary-and-discussion}.

Perhaps the most important way that the quantum causal models defined in this chapter differ from all previous attempts is that they are generalised from a quantum Reichenbach's principle. As demonstrated in the previous chapter, quantum Reichenbach is robust, general, and natively quantum. If this is the correct way to generalise Reichenbach's principle to quantum systems, then one would expect that its natural generalisation to quantum causal models would be similarly correct. The philosophy of this chapter and the previous is also deliberately neutral and sticks to structures from vanilla quantum theory as far as possible. This has been done to allow the results to be as widely applicable as possible.

\section{Classical Causal Models Generalise Reichenbach's Principle} \label{sec:CM:classical-CM}

\subsection{Classical Causal Models in Two Parts} \label{sec:CM:classical-CM-two-parts}

Like Reichenbach's principle, the framework of classical causal models \cite{Pearl09,SpirtesGlymour+01} neatly divides into two parts. First, \emph{causal structures} between random variables represented by directed acyclic graphs (DAGs). Second, the \emph{Markov condition} which defines when a joint distribution over variables is compatible with any given causal structure.

Causal structures are simply a formal way to express causal relationships between events represented by random variables. A DAG consists of a set of nodes $\{X_i\}_i$ joined by arrows (a directed graph), such that it is impossible to travel from any node back to itself by following arrows (acyclic). It is convenient to use standard genealogical terminology to describe relationships between these nodes. For any node $X_i$, $\Pa(X_i)=\Pa(i)$ is the set of parent nodes for $X_i$. Similarly, $\Ch(X_i)=\Ch(i)$ is the set of children nodes of $X_i$ and $\Desc(X_i)=\Desc(i)$ is the set of descendants of $X_i$ (which conventionally includes $X_i$ itself). Finally, $\Ndesc(X_i)=\Ndesc(i)$ is the set of ``non-descendants'' of $X_i$, \emph{viz.} the complement of $\Desc(i)$.

The Markov condition specifies which joint distributions over these variables are compatible with a given causal structure representing what actually occurred. A joint distribution $\mathbb{P}(\{X_i\}_i)$ is said to be \emph{Markov for a given DAG} if and only if the nodes are those same random variables and the joint distribution can be written in the form
\begin{equation} \label{eq:CM:Markov-condition}
\mathbb{P}(\{X_i\}_i) = \prod_i \mathbb{P}(X_i | \Pa(i))
\end{equation}
(recalling that each $\mathbb{P}(X_i |\Pa(i))$ can be calculated from $\mathbb{P}(\{X_i\}_i)$ once the DAG is known).

This defines the \emph{framework} of causal models. A \emph{specific} causal model is given by: a set of random variables $\{X_i\}_i$ representing events, a DAG with nodes of those random variables, and a set of conditional probability distributions $\mathbb{P}(X_i |\Pa(i))$. One could equivalently just specify a joint distribution that is Markov for the graph, but it is often more practical\footnote{More practical in two ways. First, when constructing a causal model for a given situation the individual conditional probabilities are often simply easier to discover or estimate. Second, it is more economical, as union of the spaces of different conditional distributions is normally much smaller than the space of joint distributions (you always get a Markov joint distribution by specifying conditional distributions, but not every joint distribution is Markov).} to give each conditional distribution which guarantees a Markov joint distribution can be constructed.

This framework generalises Reichenbach's principle. If $X$ and $Y$ are ancestrally independent, then the Markov condition requires that $\mathbb{P}(Y,Z) = \mathbb{P}(Y)\mathbb{P}(Z)$, which is the qualitative part of Reichenbach's principle. If there is a complete common cause $X$ for $Y$ and $Z$ in the causal structure (\emph{e.g.} Fig.~\ref{fig:CI:X-to-YZ}), then the Markov condition guarantees that $\mathbb{P}(Y,Z|X) = \mathbb{P}(Y|X)\mathbb{P}(Z|X)$, which is the quantitative part. Moreover, a loose generalisation of Reichenbach's principle for causal models may be stated as: correlations between variables should have causal explanations in the causal structure.

\subsection{Justifying the Markov Condition} \label{sec:CM:justifying-classical-CM}

In Sec.~\ref{sec:CI:justifying-classical-Reichenbach}, one possible justification for the quantitative part of Reichenbach's principle from the qualitative part was presented. This was done by temporarily assuming fundamental determinism for illustrative purposes. A very similar argument can be used to justify the Markov condition [Eq.~(\ref{eq:CM:Markov-condition})] from the qualitative part of Reichenbach's principle. This is strictly stronger than the argument in Sec.~\ref{sec:CI:justifying-classical-Reichenbach} since, as seen above, the Markov condition implies Reichenbach's principle.

To this end suppose once again that classical dynamics is fundamentally deterministic. The task is to prove that if a situation has a causal structure given by some DAG, then the qualitative part of Reichenbach's principle requires that the distribution over variables is Markov for that DAG. This argument closely mirrors that of Sec.~\ref{sec:CI:justifying-classical-Reichenbach} and so will be covered quickly.

Just as in Sec.~\ref{sec:CI:justifying-classical-Reichenbach}, a determinist will always view the classical maps between nodes in the causal structure to be the result of \emph{deterministic dilations}. That is, for every node $X$ there must, in reality, also be some latent node $\Lambda$ that is a cause for $X$ in a more fundamental causal structure. This $\Lambda$ ensures the classical channel that outputs $X$ is the result of a deterministic dilation with input $\Pa(X)\times\Lambda$. Simply put, any stochasticity in the original causal model is explained by functions including hidden latent variables in a more fundamental causal model.

Moreover, each latent variable must be unique to each node and have no parents. If this were not the case then adding the latent variables would invalidate the original causal structure by introducing new common cause links. In particular therefore, the latent variables are ancestrally independent and thus, by the qualitative part of Reichenbach's principle, their distributions factorise.

Putting this together, one can say the following.

\begin{definition} \label{def:CM:classical-compatibility}
A joint distribution $\mathbb{P}(\{X_i\}_i)$ is \emph{deterministically compatible with a causal structure given by DAG $G$} if and only if:
\begin{enumerate}[label=(\alph*)]
\item the nodes of $G$ are the variables $\{X_i\}$ and
\item there exists causal structure $G^\prime$ obtained from $G$ by adding a node $\Lambda_i$ and a single arrow $\Lambda_i \rightarrow X_i$ for each $X_i$ such that
\item there exist distributions $\mathbb{P}(\Lambda_i)$ and functions\footnote{Here $\Pa(i)$ is being used to refer to the parents of $X_i$ in $G$. Of course, in $G^\prime$ each $\Lambda_i$ is a parent for the corresponding $X_i$.} $f_i : \Pa(i)\times\Lambda_i \rightarrow X_i$ that form deterministic dilations for each channel $\mathbb{P}(X_i |\Pa(X_i))$ derived from the joint distribution.
\end{enumerate}
\end{definition}

Just as in Sec.~\ref{sec:CI:justifying-classical-Reichenbach}, this mathematics is independent from the temporary assumption of determinism used to reach it. By simply entertaining the logical possibility of a deterministic dilations, the following theorem from Ref.~\cite{Pearl09} completes the justification of the Markov condition.

\begin{theorem}[(Ref.~\cite{Pearl09})] \label{thm:CM:classical-compatibility}
Given a joint distribution $\mathbb{P}(\{X_i\}_i)$ and a causal structure the following are equivalent:
\begin{enumerate}
\item $\mathbb{P}(\{X_i\}_i)$ is deterministically compatible with the causal structure.
\item $\mathbb{P}(\{X_i\}_i)$ is Markov for the causal structure, satisfying Eq.~(\ref{eq:CM:Markov-condition}).
\end{enumerate}
\end{theorem}

Theorem~\ref{thm:CM:classical-compatibility} is clearly analogous to Thm.~\ref{thm:CI:classical-compatibility}. It also admits similar readings: both as the justification presented above and as proving that Eq.~(\ref{eq:CM:Markov-condition}) and Def.~\ref{def:CM:classical-compatibility} are equivalent definitions of the Markov condition. As in Sec.~\ref{sec:CI:justifying-classical-Reichenbach}, this depends on whether condition (1) is taken as a causal or probabilistic statement respectively.

Once again, this argument is not meant as a complete first-principles derivation of classical causal models. It is instead simply illustrative and helps to motivate the approach to quantum causal models that follows.

\section{Quantum Causal Models Generalise Quantum Reichenbach} \label{sec:CM:quantum-cm-generalise-reichenbach}

In the previous section, it was seen that the framework of classical causal models can be seen as a generalisation of Reichenbach's principle to more general causal structures. Moreover, one way to justify the crucial Markov condition of classical causal models was shown by assuming fundamental determinism (mirroring the justification of the quantitative part of Reichenbach's principle in Sec.~\ref{sec:CI:justifying-classical-Reichenbach}). 

This will now be used to motivate a definition for quantum causal models. In particular, quantum causal models should generalise both quantum Reichenbach and classical causal models. The definition of quantum causal models that follows naturally achieves both of these things, putting it in a strong position.

\subsection{Defining Quantum Causal Models} \label{sec:CM:quantum-definition}

In Secs.~\ref{sec:CI:justifying-classical-Reichenbach}, \ref{sec:CI:justifying-quantum-Reichenbach}, and \ref{sec:CM:justifying-classical-CM} justifications were given for Reichenbach's principle, quantum Reichenbach, and classical causal models respectively. Each of these proceeded by temporarily assuming either fundamental determinism (in the classical cases) or fundamental unitarity (in the quantum case).

It would be neat, therefore, to continue that pattern here. That is, to justify a definition for quantum causal models by assuming fundamental unitarity. Unfortunately, this is difficult to do \emph{a priori} since it is unclear what sort of mathematical object a quantum analogue for the Markov condition should apply to. For Reichenbach's principle, the quantitative part applies to the classical channel. Correspondingly, for quantum Reichenbach the quantitative part applies to the quantum channel. For classical causal models, the Markov condition applies to the joint probability distribution. Here lies the problem. As already noted, there is no convenient or accepted quantum analogue for a joint probability distribution for causally-related systems \cite{HorsmanHeunen+16}.

The approach taken here, therefore, is to propose a definition that most simply generalises quantum Reichenbach to general causal structures. The resulting definition can then be checked to satisfy other natural criteria.

As with the qualitative part of Reichenbach's principle, the causal structures from classical causal models can be commuted to quantum theory with almost no change. This is because a causal structure is, mathematically, just a DAG. All that remains is to specify what the nodes should correspond to.

Mathematically, the nodes of a classical causal structure are taken to be random variables. Physically, they can be thought of as events/systems in local regions of spacetime which can take several values/physical states. In quantum theory, the most general way to describe the same thing is with a quantum instrument applied to a system. This covers both the interpretation as an event with outcome value and as a system with a state.

Therefore, each node $A_i$ of a quantum causal structure will correspond to two Hilbert spaces: an \emph{input space} $\mathcal{H}_{A_i}^\mathrm{in} = \mathcal{H}_i^\mathrm{in}$ and an \emph{output space} $\mathcal{H}_{A_i}^\mathrm{out} \eqdef \mathcal{H}_i^\mathrm{out} \eqdef (\mathcal{H}_i^\mathrm{in})^*$ which is its dual. The fact that the output Hilbert space is the dual to the input Hilbert space is simply a convenient mathematical convention that will be useful later, morally they can be thought of as ``the same'' spaces. This allows the interpretation that each node $A$ contains the local laboratory of some agent, who performs some quantum instrument with input from $\mathcal{H}_i^\mathrm{in}$ and output to $\mathcal{H}_i^\mathrm{out}$. It also allows for less general interpretations, such as each node simply representing the state of some system when the quantum instrument is just an identity map.

As noted in Sec.~\ref{sec:CM:previous-attempts}, this interventional approach to quantum nodes has been used many times before in the literature. It is well-known as a general and convenient way to allow for both agent interventions and measurement outcomes where they might be needed.

The first part of a quantum causal model is therefore to specify a causal structure as a DAG and interpret the nodes as pairs of input/output Hilbert spaces as above. To complete the causal model, by analogy with a classical causal model, channels between the nodes are required.

The simplest way to generalise quantum Reichenbach to these causal structures is to specify, for each node $A_i$, a quantum channel $\rho_{A_i | \Pa(i)}$ from $\bigotimes_{A_j\in\Pa(i)} \mathcal{H}_j^\mathrm{out}$ to $\mathcal{H}_i^\mathrm{in}$ (using the notational convention from Sec.~\ref{sec:CI:cj-isomorphism}). Then, to ensure that these channels are compatible, simply require that they commute pairwise. That is, for all nodes $A_i$ and $A_j$, require $[\rho_{A_i |\Pa(i)}, \rho_{A_j |\Pa(j)}]=0$. As such, a complete channel from common parents to their children can be taken to be the product of each of these channels $\rho_{A_i | \Pa(i)}\rho_{A_j | \Pa(j)}\cdots$. The fact that each channel commutes with every other ensures that such a product still defines a valid quantum channel via the Choi-Jamio\l{}kowski isomorphism.

It is convenient to wrap all of these channels into a single object. This can easily be done by taking their product. The result is the \emph{model state}
\begin{equation} \label{eq:CM:quantum-Markov}
\sigma \eqdef \prod_i \rho_{A_i | \Pa(i)}
\end{equation}
which is an operator on $\bigotimes_i \left( \mathcal{H}_i^\mathrm{in} \otimes \mathcal{H}_i^\mathrm{out} \right)$ where, recall, each factor commutes with every other.

This completes the definition of a quantum causal model. It consists of a causal structure with local laboratories for the nodes $\{A_i\}_i$ and a set of channels $\rho_{A_i | \Pa(i)}$ which pairwise commute. These channels together define a model state by Eq.~(\ref{eq:CM:quantum-Markov}).

Equivalently, a quantum causal model could be taken as a causal structure together with a model state $\sigma$ over $\bigotimes_i \left( \mathcal{H}_i^\mathrm{in} \otimes \mathcal{H}_i^\mathrm{out} \right)$ which is required to satisfy Eq.~(\ref{eq:CM:quantum-Markov}). This way, Eq.~(\ref{eq:CM:quantum-Markov}) is analogous to the Markov condition and is therefore called the \emph{quantum Markov condition}.

Crucially, this definition generalises quantum Reichenbach. If the causal network is taken to be a complete common cause from $A$ to $BC$ [Fig.~\ref{fig:CI:A-to-BC}], then the quantum Markov condition requires that the channel factorises as $\rho_{BC|A}=\rho_{B|A}\rho_{C|A}$. This is the reason for claiming that this definition is the simplest to generalise quantum Reichenbach. Moreover, it also generalises classical causal models, as shall be shown explicitly in Sec.~\ref{sec:CM:classical-limit}

More work is required to flesh out the meaning of these quantum causal models and how they should be used, but first some remarks should be made about the relationship between this approach and traditional classical causal models.

\subsection{Bayesian vs. do-conditionals} \label{sec:CM:conditionals}

Strictly speaking, an important distinction should be made between two types of conditional in classical causal models, discussed at length in Ref.~\cite{Pearl09}. The \emph{Bayesian conditional} is derived from a joint probability distribution by Bayes' rule $\mathbb{P}(Y|X) \eqdef \mathbb{P}(Y,X)/\mathbb{P}(X)$. This represents the set of probability distributions that should be assigned to $Y$ given that $X$ is found to have any given value. The \emph{do-conditional} $\mathbb{P}(Y|\docn X)$, on the other hand, gives the set of probability distributions for $Y$ given that an agent has intervened and set $X$ to any definite value. Rather than following Bayes' rule, this is found by modifying the causal model to remove any arrows entering $X$ and setting $\mathbb{P}(X)$ to the appropriate point distribution for the specifically set value.

A classical channel from $X$ to $Y$ should strictly, therefore, be written $\mathbb{P}(Y|\docn X)$ and similarly quantum channels $\rho_{B|A}$ are more analogous to do- than Bayesian conditionals. This relates to the problem with joint states over causal structure in quantum theory discussed above \cite{HorsmanHeunen+16}. 

Since Bayesian conditionals are derived from a joint distribution, finding a close quantum analogue for the Bayesian conditional seems difficult at best (finding a joint state over the network is not straightforward). For this reason, no such analogue is given here. Rather, the quantum Markov condition has been given in terms of quantum channels $\rho_{A|\Pa(A)}$, which do not require such a joint state but are more analogous to do-conditionals.

Fortunately, this need not break the analogy between quantum and classical causal models. If a probability distribution $\mathbb{P}(\{X_i\}_i)$ is Markov for a causal network, then for each node $\mathbb{P}(X_i |\Pa(i)) = \mathbb{P}(X_i |\docn\Pa(i))$ \cite{Pearl09}. That is, the Bayesian conditionals match the form of the classical channels in the network. Therefore one can equivalently define classical causal models as causal structures supplemented with do-conditionals $\mathbb{P}(X_i |\docn\Pa(i))$ for each node. It is this definition to which quantum causal models are more closely analogous.

These distinctions have not and will not be used in the bulk of this thesis for brevity. Reference~\cite{CostaShrapnel16} discusses similar issues in the context of its related approach to causal models.
 
Of course the other main difference between how quantum causal models and classical causal models have been defined here is that quantum causal models use an interventionalist approach to nodes, as noted in Sec.~\ref{sec:CM:previous-attempts}. 

\subsection{Making Predictions} \label{sec:CM:predictions}

A quantum causal model would not be much use if it did not give predictions for measurements. Fortunately, once a quantum instrument has been specified at each node joint probabilities for their outcomes can be found.

For a quantum causal model with nodes $\{A_i\}_i$, let $\{\mathcal{E}_i^{k_i} \}_{k_i}$ be the quantum instrument at node $A_i$, expressed as a set of completely positive maps from $\mathcal{H}_i^\mathrm{out}$ to itself. Index $k_i$ labels the classical outcomes of the intervention in that local laboratory. An equivalent way to define these instruments is with their Choi-Jamio\l{}kowski isomorphic states on $\mathcal{H}_i^\mathrm{out}\otimes\mathcal{H}_i^\mathrm{in}$
\begin{equation} \label{eq:CM:quantum-instrument-operator}
\tau_i^{k_i} = \tau_{A_i}^{k_i} \eqdef \sum_{p,q} \mathcal{E}_i^{k_i}\left( |p\rangle_\mathrm{out}\langle q| \right) \otimes |p\rangle_\mathrm{in}\langle q|
\end{equation}
where $\{|p\rangle_\mathrm{in}\}_p$ is an arbitrary orthonormal basis on $\mathcal{H}_i^\mathrm{in}$ and $\{|p\rangle_\mathrm{out}\}_p$ its dual basis in $\mathcal{H}_i^\mathrm{out}$.

In the special case where the instrument is the identity (\emph{i.e.} there is no intervention) then this state is the \emph{linking operator} used in Sec.~\ref{sec:CI:cj-isomorphism}
\begin{equation} \label{eq:CM:linking-operator}
\tau_i^\id = \tau_{A_i}^\id \eqdef \sum_{p,q} |p\rangle_\mathrm{out}\langle q| \otimes |p\rangle_\mathrm{in} \langle q|.
\end{equation}

With this notation, the joint probability for any set of outcomes $\{k_i\}_i$ given choices of quantum instruments is given by
\begin{equation}
\mathbb{P}(\{k_i\}_i) = \Tr \left( \sigma \left( \tau_1^{k_1} \otimes \tau_2^{k_2} \otimes \cdots \right) \right).
\end{equation}
It is easy to verify that this produces the same probabilities as standard quantum theory.

The linking operator allows one to remove or ignore nodes from a causal network. This process is called \emph{linking out} and corresponds to the classical marginalisation process. For example, suppose node $A$ is ignored. This can be represented by linking it out from the model state
\begin{equation} \label{eq:CM:linking-out}
\Ln_A \sigma \eqdef \Tr_{A^\mathrm{in}A^\mathrm{out}}\left( \sigma \tau_A^\id \right).
\end{equation}
The linking operation $\Ln$ may be thought of as a modified partial trace, which includes insertion of the linking operator. Doing this produces a new model state which can be used to obtain outcome probabilities from the other nodes as above.

There is no general reason to expect $\Ln_A \sigma$ to satisfy the quantum Markov condition for an appropriate causal structure even when $\sigma$ did for the original causal structure. However, in the special case where $A$ has no children, then it is easy to verify that $\Ln_A \sigma$ does satisfy the quantum Markov condition for the causal structure obtained by removing node $A$.

\subsection{Classical Limits} \label{sec:CM:classical-limit}

As noted above, it is crucial that any reasonable definition of quantum causal networks be a generalisation of classical causal networks. It shall now be shown that this is the case for quantum causal models as defined in Sec.~\ref{sec:CM:quantum-definition}.

Quantum theory describes a classical situation in a decohering limit. That is, when there is a choice of orthonormal basis for each system such that all states and channels are diagonal with respect to those bases. In such a case, the elements of those bases become the classical states, the quantum states probability distributions over them, and the channels become classical channels.

Consider any decohering limit for a quantum causal network and let all of the quantum instruments be the identity. By taking the product of the model state $\sigma$ with the linking operators and tracing out the input spaces of each node, a state diagonal in these bases is obtained
\begin{equation} \label{eq:CM:classical-limit}
\varsigma \eqdef \Tr_{A_1^\mathrm{in} A_2^\mathrm{in} \cdots} \left( \sigma \left( \tau_1^\id \otimes \tau_2^\id \otimes \cdots \right) \right).
\end{equation}
This is a positive trace-one operator over $\mathcal{H}_{A_1}^\mathrm{out} \otimes \mathcal{H}_{A_2}^\mathrm{out} \otimes \cdots$.

The diagonal entries of $\varsigma$ (in the decohering bases) form a probability distribution over the decohered classical states of each $\mathcal{H}_{A_i}^\mathrm{out}$. The claim is that if $\sigma$ satisfies the quantum Markov condition [Eq.~(\ref{eq:CM:quantum-Markov})] in such a classical limit, then the distributions encoded in $\varsigma$ satisfy the classical Markov condition for the same causal structure. This will now be shown.

Suppose $\sigma$ satisfies Eq.~(\ref{eq:CM:quantum-Markov}). By writing $\sigma$ as a product where each factor $\rho_{A|\Pa(A)}$ appears before the corresponding factor for any of its children, then each operator $\tau^\id_A$ can be commuted through to appear to the immediate right of the corresponding $\rho_{A|\Pa(A)}$. The result is a product with one factor
\begin{equation} \label{eq:CM:classical-limit-factor}
\Tr_{A^\mathrm{in}} \left( \rho_{A|\Pa(A)} \tau_A^\id \right)
\end{equation}
for each node $A$, in any order where parents appear before children.

Suppose the linking operator is expanded in terms of the decohering basis $\{|a\rangle\}_a$ for $A$, so that Eq.~(\ref{eq:CM:classical-limit-factor}) takes the form
\begin{equation}
\sum_{a,b} {}_{A^\mathrm{in}}\langle b|\rho_{A|\Pa(A)}|a\rangle_{A^\mathrm{in}} \otimes |a\rangle_{A^\mathrm{out}} \langle b |.
\end{equation}
Recalling that $\rho_{A|\Pa(A)}$ is diagonal in the decohering bases, all terms vanish except where $b=a$. Letting $\{|\bar{p}\rangle\}_p$ be the product basis of the decohering bases for each node in $\Pa(A)$, the operator $\rho_{A|\Pa(A)}$ can be expanded to find
\begin{equation}
\sum_{a,p} \left\{ \langle a|\mathcal{E}_{A|\Pa(A)} \left(|\bar{p}\rangle_{\Pa(A)^\mathrm{in}} \langle \bar{p}|\right)|a\rangle \right\} \; |a\rangle_{A^\mathrm{out}}\langle a| \otimes |\bar{p} \rangle_{\Pa(A)^\mathrm{out}}\langle \bar{p}|.
\end{equation}
Here the fact that $\mathcal{E}_{A|\Pa(A)}$ only produces non-zero output when the input is diagonal in the $\{|\bar{p}\rangle\}_p$ basis (due to decoherence in the classical bases) has been used. Finally, note that the factor in braces is simply the probability that the classical values of $\Pa(A)$ represented by $|\bar{p}\rangle$ are mapped by $\mathcal{E}_{A|\Pa(A)}$ to the classical value $|a\rangle$ of $A$. That is, $\langle a|\mathcal{E}_{A|\Pa(A)}( |\bar{p}\rangle\langle\bar{p}|)|a\rangle = \mathbb{P}(A=a | \Pa(A)=p)$.

This shows that $\varsigma$ is a product of factors, each diagonal in the product basis of the decohering bases
\begin{equation}
\varsigma = \prod_A \left\{ \sum_{a,p} \mathbb{P}( A = a |\Pa(A) = p ) |a\rangle_{A^\mathrm{out}}\langle a| \otimes |\bar{p}\rangle_{\Pa(A)^\mathrm{out}}\langle\bar{p}| \right\}.
\end{equation}
Each diagonal entry in this basis is a product of conditional probability distributions $\mathbb{P}(A = a | \Pa(A) = p)$, which forms a joint distribution for the variables to take values $a,p,\ldots$ \emph{etc}. Thus, the entries are joint distributions over the graph which are Markov for the graph, as claimed.

In this way, quantum causal models naturally reduce to classical causal models in the appropriate limit with no interventions.

\subsection{Examples} \label{sec:CM:examples}

Quantum causal models have now been defined, shown to predict measurement outcomes, and  shown to reproduce classical causal models in the relevant limits. Since quantum causal models generalise quantum Reichenbach, the examples of Sec.~\ref{sec:CI:examples} are also examples for quantum causal models. What follows are further examples to illustrate the properties of quantum causal models.

\subsubsection{Confounding Common Cause} \label{sec:CM:confounding-common-cause}

Consider a quantum causal model with causal structure shown in Fig.~\ref{fig:CM:confounding-common-cause}(a). Such a model requires channels $\rho_{C|AB}$, $\rho_{B|A}$ and $\rho_A$ to be specified, which must commute pairwise. This produces a model state $\sigma = \rho_{C|AB}\rho_{B|A}\rho_A$ on $\mathcal{H}_C^\mathrm{out} \otimes \mathcal{H}_C^\mathrm{in} \otimes \mathcal{H}_B^\mathrm{out} \otimes \mathcal{H}_B^\mathrm{in} \otimes \mathcal{H}_A^\mathrm{out} \otimes \mathcal{H}_A^\mathrm{in}$. 

\begin{figure}
\begin{centering}
\includegraphics[scale=0.35]{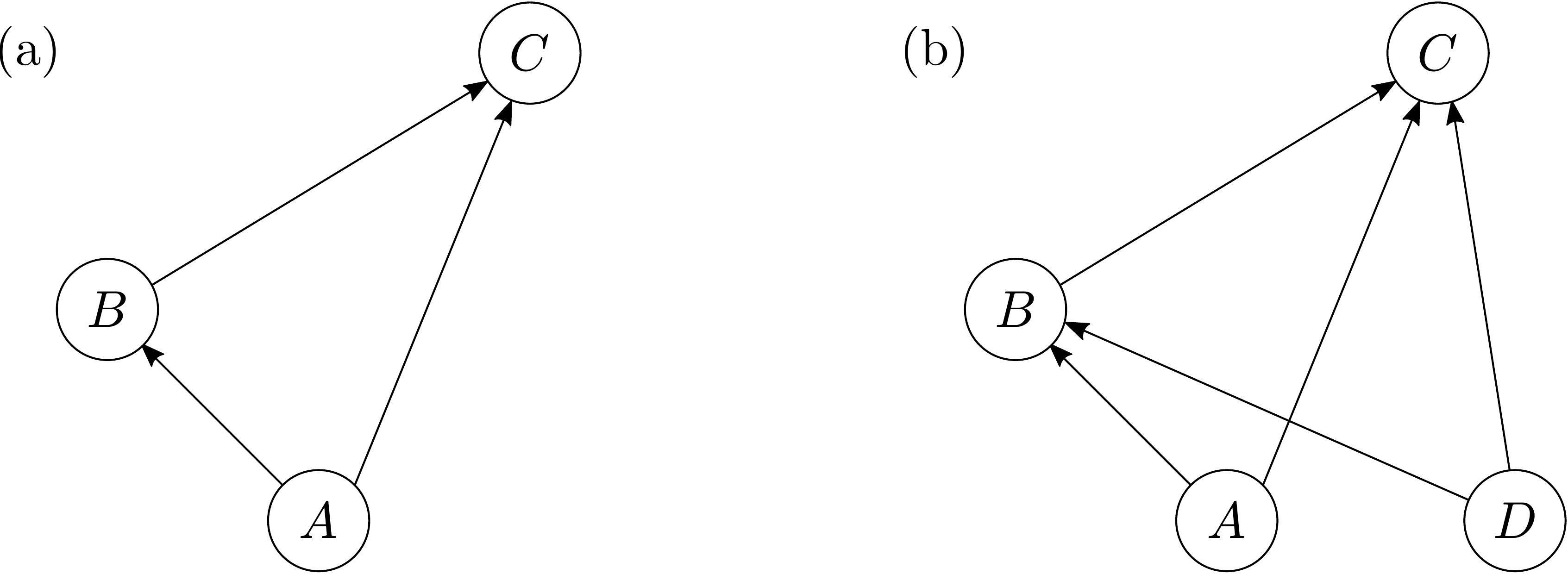}
\par\end{centering}
\protect\caption{Causal structures where $B$ causes $C$, but they also have additional common causes. (a) $A$ is the only common cause (though not a complete common cause due to the arrow from $B$ to $C$). (b) $D$ is an additional common cause.}
\label{fig:CM:confounding-common-cause}
\end{figure}

First, note some properties of these states. Since $A$ has no parents, its ``channel'' is from the trivial system, which is simply a quantum state $\rho_A$ on $\mathcal{H}_A^\mathrm{in}$. Some of the commutation relations, such as $[\rho_{B|A},\rho_A]=0$, follow immediately, since they only act non-trivially on different Hilbert spaces. However, the requirement that $\rho_{C|BA}$ and $\rho_{B|A}$ commute is a significant restriction, as they both act on $\mathcal{H}_A^\mathrm{out}$. By Thm.~\ref{thm:CI:QCI-full}, this requires that there is a decomposition $\mathcal{H}_A^\mathrm{out} = \bigoplus_i \mathcal{H}_{A_i^L}^\mathrm{out} \otimes \mathcal{H}_{A_i^R}^\mathrm{out}$ where $\rho_{C|BA}$ acts only on the right-hand factor in each subspace and $\rho_{B|A}$ only on the left hand factors.

This is very different from \emph{classical} causal models of the same causal structure. It is easy to check that any classical distribution $\mathbb{P}(A,B,C)$ is Markov for Fig.~\ref{fig:CM:confounding-common-cause}(a) interpreted as a classical causal structure. There is no analogue of the strong constraint that $[\rho_{C|BA},\rho_{B|A}]=0$ puts on the channels possible in a quantum causal network. Classically, both $B$ and $C$ may depend arbitrarily on $A$, whereas in a quantum causal network these dependencies must satisfy a compatibility condition required by commutation.

Figure~\ref{fig:CM:confounding-common-cause}(b) shows a variation of this causal network, where $D$ is an additional common cause. This might represent, for example, a system interacting with an environment. $A$, $B$, and $C$ are the system at different times and $D$ is the initial state of the environment. The causal arrow from $A$ to $C$ is necessary since information can flow into the environment from $A$ and back to the system at $C$ (nodes for the environment at later times have been omitted from the diagram).

Clearly, if nodes $A$ and $D$ are considered together then this is the exact same form as Fig.~\ref{fig:CM:confounding-common-cause}(a). Therefore, the Hilbert space $\mathcal{H}_A^\mathrm{out} \otimes \mathcal{H}_D^\mathrm{out}$ must factorise in linear subspaces as above, with the left hand factor going to $B$ and the right hand factor to $C$. This \emph{does not} imply that each individual space $\mathcal{H}_A^\mathrm{out}$ and $\mathcal{H}_D^\mathrm{out}$ must decompose in this way, only that the tensor product of the two does. Suppose that $\rho_{D}$ is fixed in some pure state $|0\rangle\langle 0|$ and consider removing $D$ from the structure (returning to Fig.~\ref{fig:CM:confounding-common-cause}(a)). There is no reason to expect that the resulting model state, found by Eq.~(\ref{eq:CM:linking-out}), satisfies the quantum Markov condition for Fig.~\ref{fig:CM:confounding-common-cause}(a). This is because the channels out of $\mathcal{H}_A^\mathrm{out}$ will not generally have the required decomposition.

Once again, this is in marked contrast with the case of Fig.~\ref{fig:CM:confounding-common-cause}(b) interpreted as a \emph{classical} causal structure. Again, there is no restriction on the form of the classical channels out of $A$ or $D$ (or both together). Moreover, if $D$ is fixed to have some definite value and then is marginalised out of the causal model, the resulting distribution will always be Markov for Fig.~\ref{fig:CM:confounding-common-cause}(a). Of course, it must be, since every distribution over $A$, $B$, and $C$ is Markov for Fig.~\ref{fig:CM:confounding-common-cause}(a).

One way to understand this to consider pure quantum states to contain some irreducible stochasticity. For example, one may prefer to think of them as ``maximal but incomplete information'' about a system (see also Sec.~\ref{sec:CI:example-copies}). In this view, the idea that a pure state of $D$ could still underwrite correlations between $B$ and $C$ is a natural one.

These examples demonstrate a curious thing: that quantum causal models appear to place \emph{more} restrictions on the same causal structure than their classical counterparts. Of course, quantum theory and quantum casual models contain classical theory and classical causal models respectively, so there is no contradiction. However, this shows that causal structure tempers some of the extra power of quantum theory in a way that it does not classically.

\subsubsection{Simple Bayesian Updating with a Complete Common Cause} \label{sec:CM:Bayesian-updating}

Consider again the case of $X$ being a complete classical common cause for $Y$ and $Z$, illustrated in Fig.~\ref{fig:CM:X-to-YZ-with-A-to-BC}(a). Classically, an important property of a complete common cause is that any knowledge about $Z$ gained at $Y$ is expressible via $X$. That is, if an agent learns something about $Y$, then the derived knowledge they gain about $Z$ is entirely the result of derived knowledge they have learned about $X$. The new information ``follows the arrows'' of the causal model (both backwards then forwards). Indeed, properly formalised, this can be viewed as another definition for classical conditional independence.

\begin{figure}
\begin{centering}
\includegraphics[scale=0.35,angle=0]{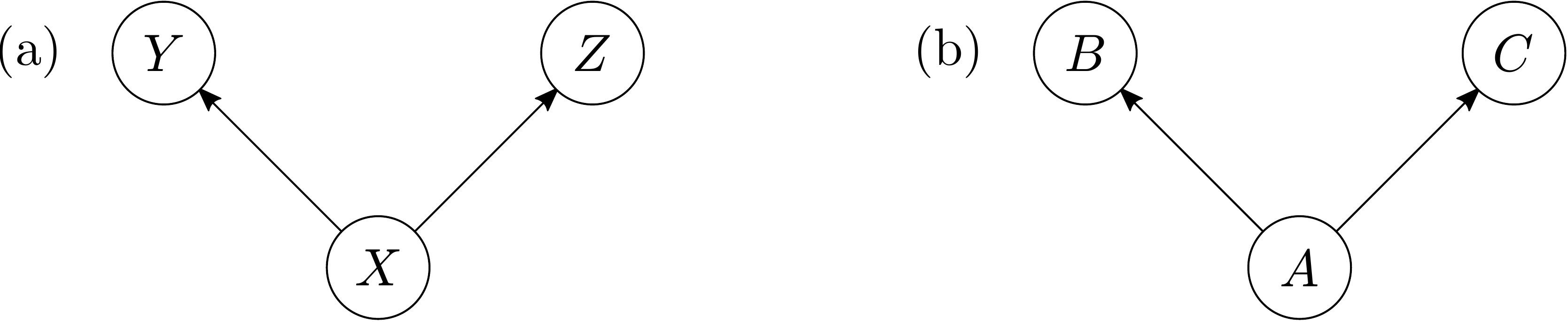}
\par\end{centering}
\caption{Causal structures for simple complete common causes in (a) the classical case and (b) the quantum case.}
\label{fig:CM:X-to-YZ-with-A-to-BC}
\end{figure}

Specifically, the causal model specifies channels $\mathbb{P}(Y|X)$ and $\mathbb{P}(Z|X)$ together with distribution $\mathbb{P}(X)$. Suppose that an agent learns that $Y=y$. This allows them to replace the initial distribution $\mathbb{P}(X)$ with $\tilde{\mathbb{P}}(X) \eqdef \mathbb{P}(X|Y=y)$. Then passing this new distribution through the channel $\mathbb{P}(Z|X)$ they find
\begin{equation}
\mathbb{P}(Z|Y=y) = \sum_{x\in X}\mathbb{P}(Z|X=x)\tilde{\mathbb{P}}(X=x).
\end{equation}
The result is exactly the same as if the agent had simply calculated $\mathbb{P}(Z|Y)$ from $\mathbb{P}(Y,Z,X)$ using Bayes' rule. For this procedure to work generally, it is essential that $Y$ and $Z$ are conditionally independent given $X$.

This is a natural and intuitive property of a complete common cause: it underwrites all correlations and therefore can be used to mediate information gained. It is therefore desirable, perhaps even essential, for the same procedure to work in quantum causal models.

For the corresponding quantum causal model of Fig.~\ref{fig:CM:X-to-YZ-with-A-to-BC}(b), suppose that an agent makes some measurement at $B$, obtaining outcome $k_B$. This corresponds to the operator $\tau_B^{k_B}$ from a quantum instrument, as in Eq.~(\ref{eq:CM:quantum-instrument-operator}). It is well known from standard quantum theory that the agent can consider the system at $C$ to be in an updated state $\sigma_{C|k_B}$ based on this measurement outcome due to ``collapse'' upon measurement. So long as $A$ is a complete quantum common cause for $B$ and $C$, this can be expressed as the result of a modified input state $\tilde{\rho}_A$ being passed into the channel $\rho_{BC|A}$ and then marginalising over $B$. This mirrors exactly the classical updating procedure above.

The updated state at $A$ required to do this is 
\begin{equation}
\tilde{\rho}_A \eqdef \frac{ \rho_A \Tr_{B^\mathrm{in}B^\mathrm{out}}\left( \tau_B^{k_B} \rho_{B|A} \right)} { \Tr\left( \rho_A \tau_B^{k_B} \rho_{B|A} \right) }.
\end{equation}
To verify that this state behaves as required, simply apply the channel $\rho_{BC|A} = \rho_{B|A}\rho_{C|A}$ and trace out the output at $B^\mathrm{in}$. The result is equal to (up to normalisation)
\begin{equation}
\sigma_{C|k_B} = \Tr_{B^\mathrm{out}} \left( (\mathcal{E}_B^{k_B} \otimes \id_C)( \mathcal{E}_{BC|A}( \rho_A ) ) \right)
\end{equation}
which is exactly the marginal state on $C$ after the agent finds the outcome $k_B$ as predicted by standard quantum theory. In other words, rather than taking the joint state on $BC$ and measuring $B$, the agent can find the marginal at $C$ after measurement by inputting $\tilde{\rho}_A$ into the channel $\rho_{BC|A}$. Just as in the classical case, to make this work conditional independence of the channel $\rho_{BC|A}$ was necessary. Also similarly to the classical case, all that is needed to find $\tilde{\rho}_A$ is the original state on $A$, the channel $\rho_{B|A}$ and the outcome $k_B$. The agent can express $\tilde{\rho}_A$ in complete ignorance of $C$ or the channel to it.

Consider, for example, the coherent copy channel discussed in Sec.~\ref{sec:CI:example-copies}. There it was seen that this channel, summarised in Eq.~(\ref{eq:CI:coherent-copy-channel}), does not satisfy conditional independence, which initially seems surprising. However, this updating procedure provides another way of understanding why it should not. 

Suppose the input to the coherent copy [Eq.~(\ref{eq:CI:coherent-copy-channel})] is $|+\rangle_A = (|0\rangle + |1\rangle)/\sqrt{2}$. The output is therefore the Bell state $|\Phi^+\rangle_{BC} = (|00\rangle + |11\rangle)/\sqrt{2}$. If a measurement were made at $B$ and outcome $|+\rangle_B$ obtained then the resulting collapsed marginal at $C$ would become $|+\rangle_C$. However, there is no state which can be input to that channel to produce $|+\rangle_C$ as the marginal output at $C$. Therefore, $\tilde{\rho}_A$ doesn't exist for this channel and measurement. This is in violation of the discussion above, where it was seen that every common cause channel and subsequent measurement can be understood in this way.

This whole procedure of knowledge ``following the arrows'' is reminiscent of Bayesian updating. It is, however, much more limited in scope. This discussion only directly applies to complete common causes. Recall also the disanalogy between quantum causal models and Bayesian conditionals noted in Sec.~\ref{sec:CM:conditionals}.

\wcomment{
\section{d-separation} \label{sec:CM:d-separation}

\red{ [You need to do some proving if you're going to include this section.] }

\red{ [Motivate the importance of d-separation in the standard ways.] }

\red{ [Find a way to link that back to our ontology theme.] }

\red{ [Define d-separation as a graphical criterion here.] }

\red{ [---] }

\red{ [Even if I don't have the proof, I can still put the proposed result in as a conjecture. Motivate my QCI-generalisation as the correct one from the QCI introduced in the last chapter, then use that to say ``look, if this is a good way of generalising causal models then we'd expect $d$-separation to hold, and it should probably look like this.'' This would be a very nice replacement, even if I can't do the proof exactly.] }

\subsection{d-separation in Classical Causal Models}

\red{ [On top of the obvious purpose of this section, see if you can find anything explicit on interventional classical causal models and d-separation.] }

\subsection{d-separation in Quantum Causal Models}

\red{ [Basically just state and prove our result.] }

\red{ [If you don't have a full proof, state the semi-results you have and reason why you think the full result should be true as a generalisation of the common cause case above.] }

\red{ [Think about the mutual-information-over-time way of doing this. There are close links with the quantum states over time people and their way of doing a joint measurement over time.] }

\red{ [Remember that our conditional independences are more like do-conditional independences, so don't take into account input states at all. Consider whether this is important.] }

} 

\section{Summary and Discussion} \label{sec:CM:summary-and-discussion}

This chapter completes the investigation of ontological causal influences in quantum theory in this thesis. Following the classical example, this has been done by generalising quantum Reichenbach of Chap.~\ref{ch:CI} to a full framework of quantum causal models. Since quantum Reichenbach is well motivated, the resulting framework is placed on a solid foundation.

In order to respect the great success of the theory of classical causal models, any framework of quantum causal models should also be a generalisation of that theory. To this end, Sec.~\ref{sec:CM:classical-CM} briefly outlined the theory of classical causal models and demonstrated one justification for them, mirroring the justification for Reichenbach's principle in Sec.~\ref{sec:CI:justifying-classical-Reichenbach}.

Just as with Reichenbach's principle, finding a quantum generalisation of classical causal models becomes easier when the original is bisected. The causal part, causal structures, can be used in quantum theory with little change. All that is needed is a specification of what the nodes correspond to. The real challenge is finding a generalisation of the probabilistic/statistical part---the Markov condition Eq.~(\ref{eq:CM:Markov-condition}).

To this end, the definition of quantum causal models in Sec.~\ref{sec:CM:quantum-definition} sought to most simply generalise quantum Reichenbach to general causal networks. Each node in the causal structure was modelled as a local laboratory, with input system, output system, and the possibility of some agent who intervenes there. This allowed the nodes to be thought of as events, observations, or system states as needed. It also facilitated a quantum Markov condition that applies to a global \emph{model state} [Eq.~(\ref{eq:CM:quantum-Markov})] that is largely disconnected, avoiding the problems with defining true joint quantum states over time \cite{HorsmanHeunen+16}.

To further justify this somewhat speculative definition, Sec.~\ref{sec:CM:predictions} demonstrated how it can be used to obtain predictions for experiments in line with standard quantum theory. Further, Sec.~\ref{sec:CM:classical-limit} explicitly showed how these quantum causal models generalise classical causal models. This was followed by some examples in Sec.~\ref{sec:CM:examples}, which touch on the power of quantum causal models for explaining phenomena and highlight some key differences with classical causal models.

\wcomment{ \red{ [Insert the summary of the d-separation section once it's written.] } }

As noted in Sec.~\ref{sec:CM:previous-attempts}, this is not the first time that a framework of quantum causal models has been defined. The framework given here, however, has two key advantages over previous attempts. First, it is a natively quantum framework. All classical systems described by the network are quantum systems in some classical decohering limit. This is a great strength over frameworks which treat quantum systems as add-ons to classical causal models. Here, quantum systems are the only first-class citizens. Second, it rests on a thoroughly motivated quantum Reichenbach's principle.  So, at least in complete common cause scenarios, the causal models inherit this motivation and have many close analogies to the classical case.

In particular, the important differences between quantum causal models given here and those of Ref.~\cite{CostaShrapnel16} should be noted. Both use a similar set-up, with input and output spaces at each node with possible interventions between them. The important difference is in how these nodes can be linked. In Ref.~\cite{CostaShrapnel16}, the output space from one node must factorise as a tensor product of spaces for each outgoing causal arrow. This is much stronger that the condition here, which requires that output Hilbert spaces must factorise \emph{in linear subspaces} (which may be one-dimensional) \emph{when they are a complete common cause}. Compare, for example, the causal models of Sec.~\ref{sec:CM:confounding-common-cause} where this thesis only requires the combined $AD$ system to factorise in this way, while Ref.~\cite{CostaShrapnel16} would require $A$ and $D$ to individually factorise as simple tensor products.

The definition of quantum Reichenbach provided in Sec.~\ref{sec:CI:quantum-Reichenbach} rests on a concept of quantum conditional independence of outputs $B$ and $C$ given input $A$, which has four equivalent definitions given in Thm.~\ref{thm:CI:QCI-full}. A corresponding set of definitions of conditional independence of $k>2$ outputs given one input are given in Thm.~\ref{thm:CI-QCI-multipartite-full}. Consider again each of these defining conditions.

Condition (1) fits perfectly with the idea that all quantum dynamics is fundamentally unitary, a popular position known as the ``church of the larger Hilbert space''. It defines exactly which unitaries can be considered to act as complete common cause channels (satisfying conditional independence of outputs given input). In such a view, therefore, quantum conditional independence provides a strong characterisation of the structure of ontological transformations. Moreover, this definition of conditional independence can be found just by assuming a larger Hilbert space position, as in Sec.~\ref{sec:CI:justifying-quantum-Reichenbach}.

Condition (2) is perhaps the most pleasingly simple, but is not directly terribly useful. Not only is the Choi-Jamio\l{}kowski state $\rho_{BC|A}$ a somewhat roundabout description of the channel, but the product of two such states has no direct physical meaning in general. The great exception, of course, is in a decohering classical limit, when it becomes the familiar condition for classical conditional independence.

Condition (3) is a very direct generalisation of the corresponding classical condition. To obtain the state $\hat{\rho}_{BC|A}$ as a simple joint quantum state, one can prepare a maximally entangled input $|\Phi\rangle = \sum_i |i\rangle_A |i\rangle_A$ for any basis $\{|i\rangle\}_i$ of $\mathcal{H}_A$ and input half of it into the channel $\rho_{BC|A}$. The joint output state, including the extra $A$ system, will be mathematically equal to $\hat{\rho}_{BC|A}$. This closely mirrors how one can obtain the corresponding classical distribution $\hat{\mathbb{P}}(Y,Z,X)$ as discussed in Sec.~\ref{sec:CI:alternative-QCI}. Condition (3) also provides an interesting link to the study of approximate quantum Markov chains \cite{IbinsonLinden+08,FawziRenner15,Wilde15,SutterTomamichel+16,LiWinter14,JungeRenner+15,BertaLemm+15,SutterFawzi+16}. Typically, (approximate) quantum Markov chains are defined for a tripartite joint quantum state with (approximately) vanishing conditional mutual information. In the approximate case, questions of how characterise quantum states with small-but-non-zero conditional mutual information are of great importance. Such results should also be very important for the study of quantum causal models, especially if one wants to reason about causal structures from imperfect experimental data.

Condition (4) is perhaps the most curious, but is also very persuasive. As discussed in Sec.~\ref{sec:CI:circuits}, it seems to directly capture the idea that such channels are those allowing two agents to independently act on a single input. A general procedure for this is to perform a von Neumann measurement on the input and factorise the resulting state in a manner that can depend on the measurement result. Equivalently, the channel factorises into two, potentially differently in different linear subspaces. The factorisation in linear subspaces is analogous to the classical copy operation, in so far as it allows two agents to act independently on a single input. This is an interesting new analogue for the copy operation. In Sec.~\ref{sec:CI:circuits}, these operations were represented in circuit diagrams with the new symbol \qcopy{}.

The use of \qcopy{} in circuit diagrams immediately suggests extending circuit formulations of quantum theory to natively support such structures. In particular, it would be interesting to see how \qcopy{} could be incorporated into the diagrammatic reformulation of quantum theory of Refs.~\cite{Coecke10,CoeckeKissinger17} based on category-theoretical underpinnings.

Perhaps, if an operational formalisation of the concept of two agents ``independently'' acting on a single input were found, one could derive condition (4) as a consequence. This would provide another strong justification for quantum Reichenbach as defined here, as well as being an interesting result in its own right.

Classical causal networks have become invaluable in the study of statistics. In particular, the field of causal inference \cite{Pearl09,SpirtesGlymour+01,LeeSpekkens15,WolfeSpekkens+16,ChavesLuft+14b}, which aims to deduce properties of the causal structure from uncontrolled statistical data. This has significant applications in many areas of science and beyond. The quantum causal models presented here should provide the appropriate framework for achieving similar utility for quantum experiments. \wcomment{ \red{ [If d-separation has been done, say that this is encouraging for being able to directly import algorithms.] } }

As noted in Sec.~\ref{sec:CM:general-quantum-experiments}, classical causal models have provided techniques for deriving Bell-like inequalities for more general causal structures \cite{LeeSpekkens15,WolfeSpekkens+16,Chaves16,RossetBranciard+16}. The quantum causal models presented here should therefore prove useful in adapting these techniques to finding bounds on observables in quantum experiments with the same causal structures \cite{ChavesLuft+14a,ChavesMajenz+15,ChavesKueng+15}.

There has been much recent interest in the idea of ``indefinite causal structures'' in quantum theory. That is, considering the possibility of a coherent superposition of different causal structures \cite{Chiribella12,OreshkovCosta+12,MacLeanReid+16,FeixBrukner16}. This may be significant for the project of unifying quantum and general relativistic theories \cite{Hardy09}. One might expect the framework of quantum causal models defined here, as well as the understanding of common causes, to be able to significantly contribute to this conversation.

Finally, there is a notable way in which the development of the quantum causal models presented in Sec.~\ref{sec:CM:quantum-definition} could be improved. Section~\ref{sec:CI:justifying-classical-Reichenbach} presented a justification of Reichenbach's principle from fundamental determinism and Sec.~\ref{sec:CM:justifying-classical-CM} did the same for classical causal models. By way of analogy, Sec.~\ref{sec:CI:justifying-quantum-Reichenbach} presented a justification for quantum Reichenbach from fundamental unitarity, giving the first definition of quantum conditional independence. What is missing is the analogous justification of quantum causal models from fundamental unitarity. That is, one might reasonably expect a framework of quantum causal models which generalises quantum Reichenbach to be derivable by assuming that the mappings between nodes are all unitary, supplemented by latent nodes which satisfy the qualitative part of Reichenbach's principle. This has not yet been achieved. The interconnected nature of general causal structures makes this more difficult, since a global unitary evolution normally assumes a sense of global time. However, it is suspected that the task is not insurmountable and the result would place quantum causal models on even firmer footing.

\chapter{Causal Loops: Time Travel to the Past} \label{ch:TT}

\section{Closed Timelike Curves and the Ontology of Time Travel} \label{sec:TT:intro}

This thesis is primarily concerned with the restrictions that quantum theory places on ontology and causality. These topics meet head on when considering the possibility of time travel to the past in a quantum universe. Quantum theory is famously linear but as soon as time travel is included non-linearities creep in. This destroys the equivalence between many ontological interpretations possible in standard quantum theory. Moreover, the lack of consensus towards quantum ontology means a range of possible models for time travel should be considered. Most of the work in this chapter has been published in Ref.~\cite{Allen14}.

\subsection{Why Study Time Travel?} \label{sec:TT:why-TT}

There are few areas of physics in which one confronts the idea of time travel to the past. Indeed, from the causal perspective, true causal loops are normally ruled out by hypothesis (this assumption being built into the use of directed \emph{acyclic} graphs in causal models, Chap.~\ref{ch:CM}). When it is discussed, time travel is often associated with particular ways of thinking about quantum theory and quantum field theory. One sometimes hears the view that in quantum teleportation, for example, the teleported information travels back in time to the point at which entanglement was created, before proceeding forward in time to the recipient \cite{Penrose98,Jozsa98,Jozsa04,Timpson06}. Another example might be remarks about antiparticles in quantum field theory being akin to particles ``travelling backwards in time''. The time travel considered in this chapter is rather different. Here, it will be assumed that time travel into the past is a possible physical process and the question is asked: how could or should quantum theory be modified to account for time travel?

There are a few ways to motivate this. Perhaps the most common is in relation to general relativity and \emph{closed timelike curves} (CTCs), discussed in Sec.~\ref{sec:TT:CTCs-and-GR}. More pertinent to this thesis are the insights it may give to the relationship between causality and ontology in quantum theory. In a classical and ostensibly deterministic universe our familiar ideas of causality and physical ontology are strong and considering time travel to the past is easy enough to be the subject of many of popular stories. For quantum theory, on the other hand, there are a wide variety of possible ontological interpretations, none of which are uncontroversial. As shall be seen, this gives rise to many reasonable models for quantum time travel. Time travel to the past forces us to consider quantum features such as indeterminism and inseparability of states in a new light.

These conversations are valuable even if time travel to the past is not possible in our universe. Treated as a thought experiment, time travel can draw particular attention to issues that may be overlooked in other situations. In the words of Deutsch in the original paper on quantum time travel as approached in this thesis: ``It is curious that the analysis of a physical situation which might well not occur should yield so many insights into quantum theory'' \cite{Deutsch91}.

Taking the opposite perspective, theoretical studies into quantum time travel may yield \emph{a priori} restrictions on whether and how such effects are logically possible. It is well known that classical treatments of time travel are prone to paradoxes, but a common feature of most quantum approaches to time travel is that certain paradoxes can be avoided altogether. This may be taken as evidence that time travel to the past is not as nonsensical as our classically-influenced brains might assume. On the other hand, if a very well-motivated argument is found stating that no sensible account of quantum time travel can be given this may be taken as evidence that time travel to the past is impossible.

There are also computational motivations for studying time travel. One of the main tools in the complexity theorist's box is analysing the abilities of different computational paradigms when given access to powerful additional resources, often specified as ``oracles'' \cite{AroraBarak09}. This allows the derivation of so-called ``relativised'' separation results, which establish relationships between complexity classes given access to certain oracles. The situation with time travel to the past need be no different, allowing comparison of different models of computation when given access to these anachronistic resources. This is of particular interest in quantum computation. Most confirmed computational speed-ups of quantum computation over classical computation are in tasks which involve an oracle's resources \cite{NielsenChuang00}. It is therefore interesting to compare the capabilities of quantum and classical computers with access to different models of time travel.

Finally, by incorporating time travel to the past, one obtains a non-linear extension to quantum theory \cite{Hawking95,Anderson95,FewsterWells95,Hartle94}. Standard quantum theory is characteristically linear and yet the possibility of non-linear extensions to it remain a subject of interest, especially with regard to any hypothetical post-quantum physics \cite{GhirardiRimini+86,Weinberg89,Goldin08,Kent05,Penrose98,AaronsonWatrous09,Aaronson05,CavalcantiMenicucci+12}. In this view, quantum time travel provides an interesting source of reasonable non-linear extensions which can serve as either examples or counter-examples for how non-linear quantum theory can or must behave.

\subsection{Closed Timelike Curves and General Relativity} \label{sec:TT:CTCs-and-GR}

By far the most common motivation given for studying time travel to the past is that general relativity allows for closed timelike curves (CTCs). CTCs are features of exotic spacetimes which allow massive particles to travel to their own past, while heading apparently forwards in time at all points. As such they give a concrete possible mechanism for time travel to the past. It has long been known that solutions containing CTCs can be found to the Einstein field equations \cite{Godel49,Bonnor80,Gott91}. Unsurprisingly, there are doubts that any such solutions could be found in nature \cite{Hawking92,MorrisThorne+88} but it has not been possible to absolutely preclude them. It is therefore prudent to treat them as real possibilities and, since our universe is quantum, the problem of quantum theory along CTCs (that is, quantum time travel to the past) must be addressed.

The most obvious approach to analysing quantum theory with CTCs is to use quantum field theory and general relativity in curved spacetimes. This is far from easy, however. Spacetimes containing CTCs are globally non-hyperbolic and typically contain no Cauchy surfaces. This makes them incompatible with standard relativistic quantum field theory and without a well-defined initial value problem \cite{BirrellDavies82,Wald10}.

Nonetheless, there is a tradition of attempts to model quantum time travel using the path-integral approach to quantum theory \cite{EcheverriaKlinkhammer+91,BirrellDavies82,Hartle94,Politzer94}. The approach taken in this chapter will, however, be somewhat more abstract. Rather than dealing directly with curved spacetimes, one can use the quantum circuit approach \cite{Deutsch91}. The result is an abstraction away from the precise time travel mechanism, retaining only the essential feature that some system is sent into its own past \cite{Allen14}. This move has turned out to be fruitful, giving rise to the two most popular current models of time travel in quantum theory: \emph{D-CTCs} (or ``Deutschian'' CTCs) and \emph{P-CTCs} (or ``post-selection'' CTCs).

This general approach to modelling quantum time travel will be detailed in Sec.~\ref{sec:TT:circuits-time-travel} together with a short discussion of the types of paradoxes that time-travel can cause in the classical case. Section~\ref{sec:TT:review} will then briefly outline the models of D- and P-CTCs and discuss the ontological issues raised, in particular with respect to non-linearity. The main results of the chapter will be in Sec.~\ref{sec:TT:alternative-models}, where alternatives to the D- and P-CTC models are found. One of these alternative models, dubbed T-CTCs, will be fully fleshed out and compared with the established models. Finally, Sec.~\ref{sec:TT:summary-and-discussion} will summarise these results and discuss their significance for the ontology of quantum theory, time travel, and non-linear extensions to quantum theory more generally.

\section{Modelling Time Travel with Quantum Circuits} \label{sec:TT:circuits-time-travel}

\subsection{The Standard Form of Time Travel Circuits} \label{sec:TT:standard-form}

The quantum circuit approach to time travel is based on a particular form of circuit introduced in Ref.~\cite{Deutsch91}. This convenient building block, from which all other circuits involving time travel can be built, will be called the \emph{standard form circuit}. It allows different circuits and models to be easily and concisely specified.

The quantum circuit model is an abstraction from the precise physical mechanisms that neatly separates quantum evolution from spatial motion. Quantum interactions described by unitary gates are assumed to only occur in small, freely falling, non-rotating regions of spacetime so that they obey non-relativistic quantum theory. To include time travel to the past in such a quantum circuit model is therefore equivalent to saying that the classical paths quantum systems take between gates are allowed to go back in time. Within this approach, different models are then defined by their behaviour when sending systems back in time. For clarity, this chapter will only deal with finite-dimensional Hilbert spaces.

The standard form circuit, illustrated in Fig.~\ref{fig:TT:standardform}, contains a single time travel event in a localised spacetime region and a single unitary quantum interaction $U$. To simplify the discussion, it will be assumed that there is some system in the circuit which does not go back in time\footnote{This assumption could be dropped, but at the expense of a more fiddly discussion. It also seems reasonable to expect any model containing non-localised time travel events to be extendible to a model where all are localised.}. The \emph{chronology violating} (CV) system arrives from its own future in the state $\tau_{i}$ and after the interaction is said to be in the state $\tau_{f}$. Meanwhile, the \emph{chronology respecting} (CR) system arrives from the unambiguous past in the state $\rho_{i}$ and emerges into the unambiguous future in the state $\rho_{f}$. Therefore, a standard form circuit is completely specified by these two systems and $U$, while a model for quantum theory with time travel
is a general specification of $\rho_{f}$ given $U$ and $\rho_{i}$.

\begin{figure}
\begin{centering}
\includegraphics[scale=0.35,angle=0]{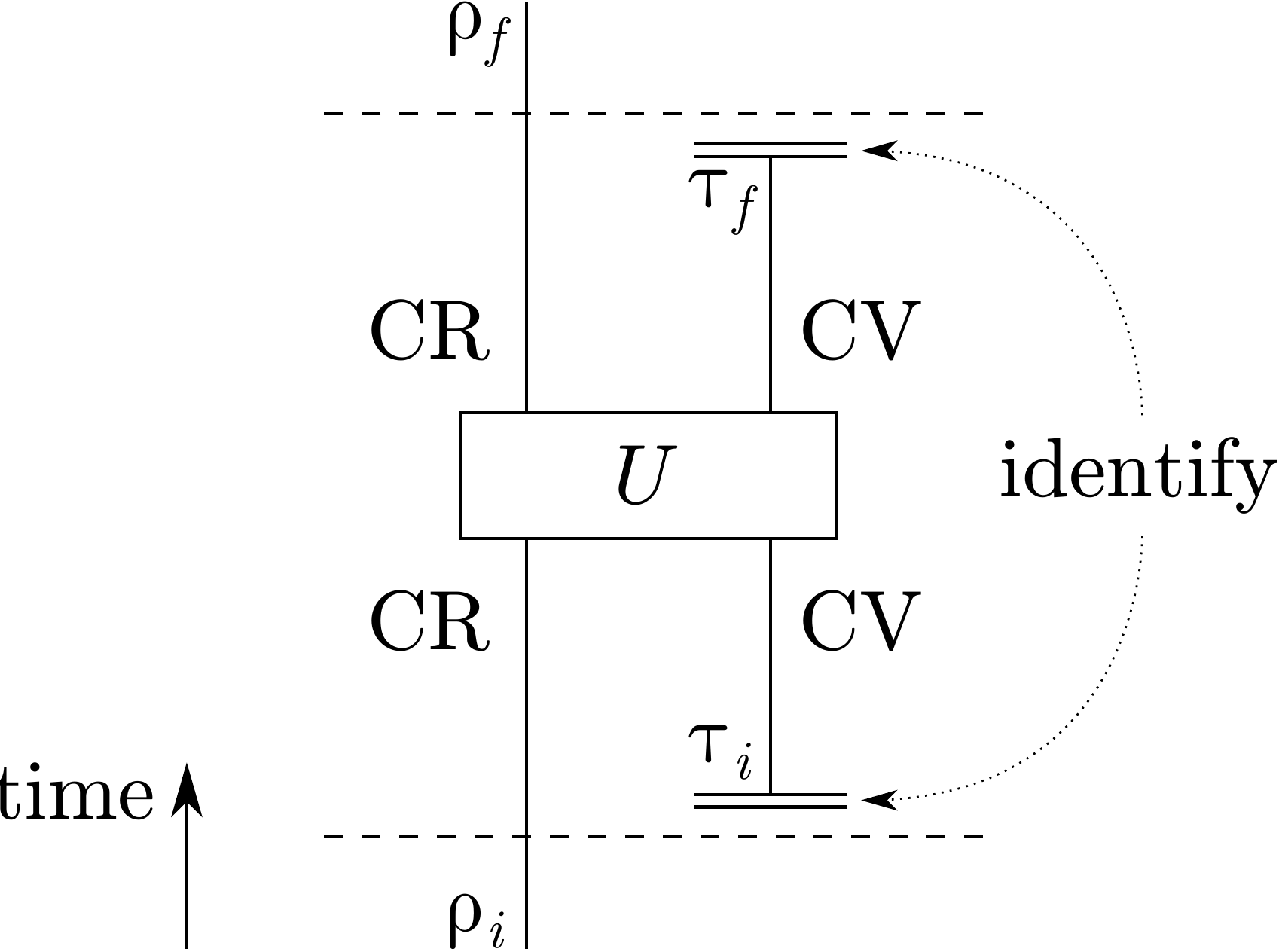}
\par\end{centering}
\protect\caption{Schematic diagram for the standard form circuit described in the text. The chronology respecting (CR) and chronology violating (CV) systems are shown entering the gate labelled by unitary $U$. Time, when it is unambiguous at least, increases up the diagram. The double bars represent the time travel event and may be thought of as two depictions of the same spacelike hypersurface forming a CTC. The dashed lines represent spacelike boundaries of the region in which time travel takes place; CV system is restricted to that region.} \label{fig:TT:standardform}
\end{figure}

In cases where the CV system is initially known to be in a pure state it may be written as $|\phi\rangle$ so that $\tau_{i}=|\phi\rangle\langle\phi|$. Similarly, in cases where the CR system input or output states are known to be pure, they may be written $|\psi_{i,f}\rangle$ so that $\rho_{i,f}=|\psi_{i,f}\rangle\langle\psi_{i,f}|$. In Sec. \ref{sec:TT:non-linearity} it will be seen that $\rho_{i}$ can always be purified to some $|\psi_i\rangle$.

A standard assumption used in all circuit models of quantum time travel is that the CR and CV systems are not initially entangled. That is, before the action of $U$ they are in the product state $\rho_i \otimes \tau_i$. \cite{DeJongheFrey+10}. It has been argued that this is an unreasonable assumption and that prior entanglement should be considered as the CV system contains $\rho_i$ in its past \cite{Politzer94}. That possibility will not be considered here because alternative assignment methods all have undesirable features. The problem is one of finding an assignment procedure that gives a joint state $\omega$ on $\mathcal{H}_{\textsc{CR}}\otimes\mathcal{H}_{\textsc{CV}}$ given only a state $\rho_{i}$ on $\mathcal{H}_{\textsc{CR}}$. The procedure assumed here---where $\omega=\rho_{i}\otimes\tau_{i}$---is the only one for which: $\omega$ is always positive, $\Tr_{\textsc{CV}}\omega=\rho_{i}$, and mixtures are preserved \cite{Alicki95}.

\subsection{Time Travel Paradoxes and the Classical Model} \label{sec:TT:paradoxes}

In order to define the possible types of paradox in time travel it is useful to leave quantum theory to one side briefly and concentrate on classical time travel. Consider a classical version of the standard form circuit, with classical states $\tilde{\rho}_{i,f}$ and $\tilde{\tau}_{i,f}$ and some classical dynamical evolution $\tilde{U}$ replacing their quantum counterparts. Since classically it can be assumed that there is no fundamental stochasticity, assume that $\tilde{\rho}_{i,f}$ and $\tilde{\tau}_{i,f}$ are both ontic.

The standard way of introducing time travel into classical theories is to impose a consistency condition on states that go back in time; that is, $\tilde{\tau}\eqdef\tilde{\tau}_{i}=\tilde{\tau}_{f}$. In other words, the ontic state that emerges into the past is required to be the same one that left from the future. With knowledge of a $\tilde{\rho}_{i}$ and $\tilde{U}$, a consistent $\tilde{\tau}$ can be deduced, from which $\tilde{\rho}_{f}$ may be calculated.

Classical time travel to the past may give rise to both \emph{dynamical consistency} paradoxes and \emph{information} paradoxes. This chapter takes a slightly unconventional view on paradoxes whereby a computation, which algorithmically produces unambiguous output from an input, is never paradoxical. So a theory that predicts ``absurdly'' powerful communication or computational abilities will not be called paradoxical because of them, even if those abilities make the theory appear unreasonable or hard to accept. Some ``absurd'' conclusions should probably be expected when ``absurdity'' of time travel has been assumed.

A dynamical consistency paradox is a situation in which a consistent history of events is not possible \cite{Deutsch91,LloydMaccone+11a,LloydMaccone+11b}. The usual example is the ``grandfather paradox'' where a grandchild travels back to kill their infant grandfather. Dynamical consistency paradoxes occur when a model fails to specify any valid final state from some initial state and evolution. Classically, this is exactly because the consistency condition cannot be satisfied\footnote{It is worth noting that, although these paradoxes appear to be possible in classical models of physics, such situations are quite difficult, if not impossible, to construct in classical models with continuous state spaces. Reference \cite{ArntzeniusMaudlin13} is a useful introduction to such issues.}. The only way to avoid these paradoxes in general is to disallow the interactions that lead to them from the model \cite{Allen14}.

An information paradox is a situation with consistent dynamics but information that appears from nowhere; \emph{viz.} the information has not been computed. The prototypical example is the ``unproved theorem'' paradox: a mathematician reads the proof of a theorem from a book only to travel back in time to author that same book. The proof has no ultimate source.

Under these definitions, information paradoxes arise if and only if a theory contains a \emph{uniqueness ambiguity}: the model specifies more than one final state given some initial state and evolution, but fails to give probabilities for each possibility. Any other dynamically consistent evolution is counted as a computation and information paradoxes have been defined as exactly those where an uncomputed output is produced. For example, in the classical unproved theorem paradox there are many time travelling states $\tilde{\tau}$ compatible with the consistency condition, each producing a different $\tilde{\rho}_{f}$---one $\tilde{\tau}$ produces a theorem answering $\mathsf{P}\eqquest\mathsf{NP}$, another answering $\mathsf{BPP}\eqquest\mathsf{BQP}$, many others where the ``theorem'' is nonsense, \emph{etc}. It should be noted that the equivalence of the uniqueness ambiguity and information paradoxes is not universally used, but it follows from the definitions used here as described above. Readers preferring different definitions may replace each further instance of ``information paradox'' with ``uniqueness ambiguity'' without affecting the meaning.

For example, some authors prefer a wider definition of paradox in which time travel circuits are counted as having information paradoxes when the only consistent evolution reveals a fixed point of some given function \cite{Deutsch91,LloydMaccone+11c}. Such circuits uniquely produce solutions to problems that are hard if $\mathsf{P}\neq\mathsf{NP}$ very rapidly \cite{GareyJohnson79}. As such, they are counted as very powerful computations rather than paradoxes as defined here. This is not to claim that such processes are necessarily reasonable, but rather to reflect that they are qualitatively different from the uniqueness ambiguities equivalent to information paradoxes as defined here. If there is no uniqueness ambiguity then any information that appears as a result of the time travel circuit is uniquely specified by the structure of, and input to, the circuit \cite{LloydMaccone+11c}. It is therefore reasonable to say that this circuit is algorithmically computing as instructed.

\section{Quantum Time Travel Models, Non-linearity, and Ontology} \label{sec:TT:review}

The model of D-CTCs, introduced by Deutsch, was the first to use the quantum circuit model to analyse quantum time travel \cite{Deutsch91}. P-CTCs are a more recent development which make use of post-selection and ideas from quantum teleportation to construct a very different model. In this section each will be briefly introduced and discussed, focussing on what they suggest about ontology and non-linearity in quantum theory.

\subsection{Overview of D-CTCs} \label{sec:TT:D-CTCs}

While not the original line of reasoning, the D-CTC model can be rapidly constructed by assuming that reduced density operators are ontic states, following Ref.~\cite{WallmanBartlett12}. Just as in the classical model, this imposes a consistency condition $\tau\eqdef\tau_i=\tau_f$ on the ontic time travelling states. For a standard form circuit, this implies 
\begin{equation} \label{eq:TT:D-CTC-CC}
\tau=\mathcal{G}(\tau)\eqdef\Tr_{\textsc{CR}}\left( U(\rho_{i}\otimes\tau)U^{\dagger} \right).
\end{equation}
Noting that the CR and CV systems have been separated using a partial trace in Eq.~(\ref{eq:TT:D-CTC-CC}), this suggests that the final state should be given by
\begin{equation} \label{eq:TT:D-CTC-EoM}
\rho_{f}=\Tr_{\textsc{CV}}\left( U(\rho_{i}\otimes\tau)U^{\dagger} \right).
\end{equation}

The key difference between this and ordinary unitary quantum theory is the implied map
\begin{equation} \label{eq:TT:D-CTC-dynamical-change}
U\left(\rho_{i}\otimes\tau\right)U^{\dagger}\rightarrow\rho_{f}\otimes\tau,
\end{equation}
that replaces the quantum state after $U$ with the product of its reduced density operators. This defines the action of the time travel event in D-CTCs.

Between them, Eqs.~(\ref{eq:TT:D-CTC-CC}, \ref{eq:TT:D-CTC-EoM}) define the D-CTC model of time travel. Given $\rho_i$ and $U$, one can solve Eq.~(\ref{eq:TT:D-CTC-CC}) to get $\tau$ and then calculate $\rho_f$. Equation~(\ref{eq:TT:D-CTC-EoM}) takes the role of an equation of motion and is clearly both non-linear and non-unitary in general.

The first thing to note is that a solution $\tau$ for Eq.~(\ref{eq:TT:D-CTC-CC}) always exists. This follows from Schauder's fixed point theorem which guarantees that every trace-preserving quantum channel, such as $\mathcal{G}$, has at least one fixed point \cite{Schauder30,Tychonoff35,Zeidler85}. Therefore dynamical consistency paradoxes cannot arise in the D-CTC model\footnote{A direct proof of this can be found in Ref.~\cite{Deutsch91}}. On the other hand, solutions for $\tau$ are not always unique. Therefore, D-CTCs have the same uniqueness ambiguity present in classical time travel and hence information paradoxes.

To avoid uniqueness ambiguities the \emph{maximum entropy rule} has been suggested \cite{Deutsch91}, stating that one should choose the unique $\tau$ with maximum von Neumann entropy. However, this is not universally accepted as part of the D-CTC model and alternative principles do exist \cite{Politzer94,DeJongheFrey+10}. Interestingly, these ambiguities vanish in the presence of arbitrary non-zero noise. Suppose one incorporates a noise channel $\mathcal{N}$ applied to the CV system so that Eq.~(\ref{eq:TT:D-CTC-CC}) becomes $\tau = \mathcal{N}(\mathcal{G}(\tau))$. As a noisy channel, $\mathcal{N}(\mathcal{G}(\cdot))$ should be \emph{strictly contractive} (that is, trace distance should always decrease under its action) \cite{Raginsky02} and therefore have a unique fixed point\footnote{This argument resolves a conjecture from Ref.~\cite{Allen14}.} \cite{NielsenChuang00}.

Thus, D-CTCs are free from dynamical consistency paradoxes but can have information paradoxes (which vanish under arbitrarily small noise). The non-linearity of the model enables them to have the following abilities, beyond those of ordinary quantum circuits.

They can solve any problem in $\mathsf{PSPACE}$ in polynomial time and are therefore likely to be vastly more powerful even than quantum computers \cite{AaronsonWatrous09}. D-CTCs are also capable of producing discontinuous mappings from $\rho_i$ to $\rho_f$, which means that (for all practical purposes) the model loses predictive power near these discontinuities. 

Given any finite set of pure states (which are not necessarily orthogonal), there is a D-CTC circuit that can render them distinguishable with a single measurement \cite{BrunWilde12}. In this way, D-CTCs can violate the Holevo bound \cite{BrunHarrington+09,Holevo73}. Given this distinguishing ability, it is perhaps unsurprising that D-CTC circuits are also capable of cloning arbitrary quantum states \cite{AhnMyers+12,BrunWilde+13} (though, of course, entanglement is not cloned as this is forbidden by monogamy \cite{CoffmanKundu+00}).

\subsection{Overview of P-CTCs} \label{sec:TT:P-CTCs}

The P-CTC model is due to Svetlichny \cite{Svetlichny09} (inspired by diagrammatic approaches to quantum theory \cite{CoeckeKissinger17}) and Lloyd \emph{et al. }\cite{LloydMaccone+11a,LloydMaccone+11b} (based on the unpublished work of Bennett and Schumacher \cite{BennettSchumacher05} and inspired by Ref.~\cite{HorowitzMaldacena04}). Reference \cite{BrunWilde12} contains an accessible introduction. It is defined by ignoring the precise mechanism behind the time travel and postulating only that the effect is mathematically equivalent to teleportation into the past, achieved by the following unphysical operational protocol schematically illustrated in Fig. \ref{fig:TT:P-CTCs}.

\begin{figure}
\centering{}\includegraphics[scale=0.35,angle=0]{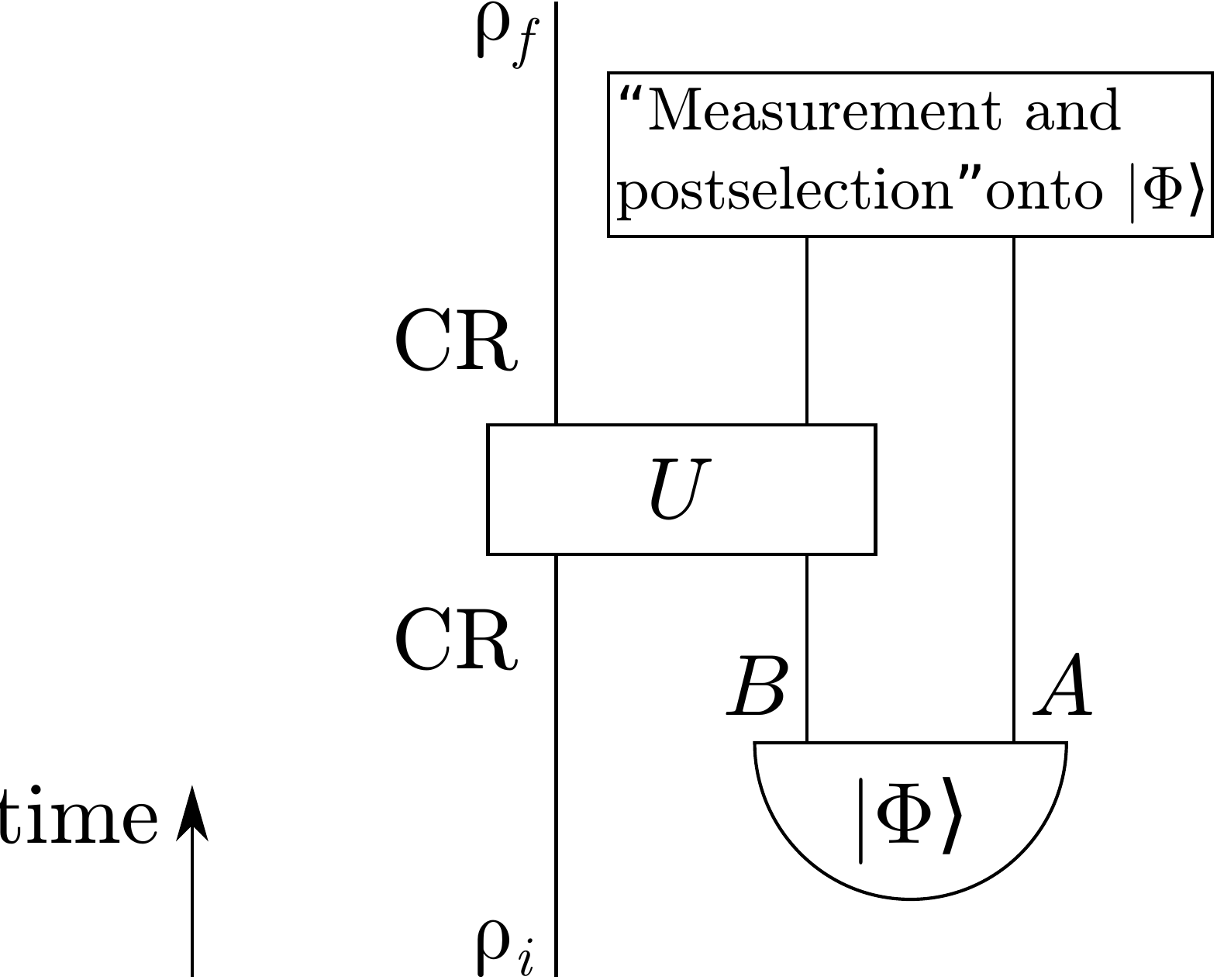}
\protect\caption{Schematic illustration of the protocol defining the action of P-CTCs as described in the text. In standard teleportation, Alice and Bob share entangled systems $A$ and $B$ and Alice can teleport a state to Bob by performing a joint measurement with $A$ and obtaining the outcome $|\Phi\rangle$. This situation differs from that of standard teleportation in two ways. First, the system that Alice teleports is the same as $B$, just at a later time. Second, that Alice can get the outcome $|\Phi\rangle$ by postselection with certainty and so no classical communication to Bob is required to complete the teleportation.}
\label{fig:TT:P-CTCs}
\end{figure}

Prepare two copies of the CV system, $A$ and $B$, in the maximally entangled state $|\Phi\rangle \propto \sum_{i}|i\rangle_{B}|i\rangle_{A}$, where $\{|i\rangle\}_{i}$ is any orthonormal basis of the CV system. Let $B$ interact with the CR system as in the standard form circuit and then perform a joint measurement on $B$ and $A$ which contains outcome  $|\Phi\rangle$. The unphysical step is to postselect this measurement on the outcome $|\Phi\rangle$. This is equivalent to simply projecting the tripartite system of CR, $B$, and $A$ onto $\langle\Phi|$ and then renormalising the resulting state. Comparing this to the standard quantum teleportation protocol, the effect is to ``teleport'' the final state of $B$ back onto the $B$ system in the past. This protocol may be simulated in the laboratory by manual postselection of measurement outcomes \cite{Svetlichny09,LloydMaccone+11a}.

If the CR system is initially in the pure state $|\psi_i\rangle$ then the effect of this protocol is (up to renormalisation)
\begin{eqnarray}
_{BA}\langle\Phi|U_{\textsc{CR},B}|\psi_{i}\rangle|\Phi\rangle_{BA} & \propto & \sum_{i}\,_{B}\langle i|U_{\textsc{CR},B}|i\rangle_{B}|\psi_{i}\rangle\nonumber \\
 & \propto & \Tr_{\textsc{CV}}(U)\,|\psi_{i}\rangle.
\end{eqnarray}
Generalising this result to mixed input $\rho_{i}$ and including the renormalisation, the general action of a P-CTC becomes
\begin{equation} \label{eq:TT:P-CTC-EoM}
\rho_{f}=\frac{P\rho_{i}P^{\dagger}}{\Tr(P\rho_{i}P^{\dagger})},\quad P\eqdef\Tr_{\textsc{CV}}(U).
\end{equation}
The map of Eq.~(\ref{eq:TT:P-CTC-EoM}) completely specifies the action of a P-CTC standard form circuit. So having obtained this result, one can take Eq.~(\ref{eq:TT:P-CTC-EoM}) as the definition of the P-CTC model and take the above unphysical protocol simply as an argument showing that the model is equivalent to teleporting the final state of the CV system into the past. This equation of motion is both non-linear and non-unitary in general.

Since Eq.~(\ref{eq:TT:P-CTC-EoM}) maps each $\rho_{i}$ onto a specific $\rho_{f}$ without ambiguity, P-CTCs cannot suffer uniqueness ambiguities and so are not vulnerable to information paradoxes. On the other hand, the $P$ operators defined in Eq.~(\ref{eq:TT:P-CTC-EoM}) can act as $P\rho_i P^\dagger = 0$ on some input states $\rho_i$. These cases are dynamical consistency paradoxes as no consistent outcome can be obtained from these inputs and interactions. However, such dynamical consistency paradoxes vanish under arbitrary non-zero noise \cite{Allen14}. Note how P-CTCs generically suffer from dynamical consistency paradoxes, while D-CTCs do not, but D-CTCs suffer from information paradoxes, while P-CTCs do not, and in both cases the paradoxes vanish under the action of any finite noise.

As with the D-CTC model, the non-linearity of this model permits various super-quantum abilities. P-CTC circuits can solve any problem in $\mathsf{PP}$ in polynomial time \cite{LloydMaccone+11b}, making them likely much more powerful than quantum computers but likely much less powerful than D-CTCs. This follows because P-CTCs and quantum theory with postselection are computationally equivalent \cite{BrunWilde12,Aaronson05}. Unlike D-CTCs, P-CTCs can never produce discontinuous evolutions because of the more mild form of non-linearity, as shown in Sec.~\ref{sec:TT:non-linearity}.

There are P-CTC circuits that can render any finite set of linearly independent pure states distinguishable with a single measurement \cite{BrunHarrington+09}. Note that, while still able to distinguish non-orthogonal quantum states with certainty, this is much less powerful than distinguishing with D-CTCs, as linear independence is required (in particular, the Holevo bound cannot be violated in this way). P-CTCs are also capable of generically deleting arbitrary quantum states, something impossible in standard linear quantum theory \cite{PatiBraunstein00}.

\subsection{The Problems of Non-linearity} \label{sec:TT:non-linearity}

Even from this brief overview, it should be clear that the non-linearity of D-CTCs and P-CTCs causes some striking departures from standard linear quantum theory. Central features, including no-cloning and indistinguishability of non-orthogonal states, are broken due to the non-linearity and non-unitarity of Eqs.~(\ref{eq:TT:D-CTC-EoM}, \ref{eq:TT:P-CTC-EoM}). However, the abilities of both models are still bounded and it is sensible to consider how and why these two models differ.

There are different types of non-linear maps one can consider applying to quantum states. Generally, non-linearity occurs due to input state appearing at quadratic or higher orders in an equation of motion, as in D-CTCs. A special case of non-linearity is \emph{renormalisation non-linearity}, where the equation of motion is linear except for a scalar factor which simply normalises the final state. P-CTCs are renormalisation non-linear since Eq.~(\ref{eq:TT:P-CTC-EoM}) would be entirely linear were it not for the renormalising denominator. A non-linear equation may be called \emph{polynomial non-linear} if it is not renormalisation non-linear.

This distinction facilitates a more general discussion of different non-linear theories and their features. For example, renormalisation non-linear equations cannot lead to the discontinuous state evolutions possible with D-CTCs. It is simple to verify that if $\rho_f(\rho_i) = \mathsf{L}(\rho_i)/\Tr(\mathsf{L}(\rho_i))$ is a renormalisation non-linear equation of motion (that is, if $\mathsf{L}$ is linear in this equation) then 
\begin{equation}
\lim_{\epsilon\rightarrow0}\rho_{f}(\rho_{i}+\epsilon\sigma) = \rho_{f}(\rho_{i})
\end{equation}
for any $\sigma$ so the mapping is always continuous.

There are also different types of mixed quantum states that can be identified. Typically, one encounters \emph{proper mixtures}---epistemic mixtures due to the observer's ignorance of the actual quantum state---and \emph{improper mixtures}---the way to describe only part of an entangled system. Both of these are described using density operators and in standard quantum theory two mixtures described by the same density operator behave identically regardless of the type.

The operational equivalence of different types of mixture directly facilitates the operational equivalence of many interpretations and ontologies for quantum theory. For example, Everettian ontologies have only a single global quantum state that is hyper-entangled and evolves deterministically, implying an absence of proper mixtures in favour of improper mixtures\footnote{Of course, any Everettian observer could still choose to be ignorant, but this does not give rise to mixed states in the same way. To be clear, consider how proper mixtures arise in objective collapse models. An observer $O$ sets up a measurement $M$ on a system $S$, but chooses (by not looking, or what have you) to remain ignorant of the outcome of $M$. In an objective collapse model, $O$ now knows (assuming sufficient understanding of $S$, $M$, and quantum theory) that the state of the universe has now collapsed into one of multiple possible states with corresponding probabilities; this ensemble forms a proper mixture which $O$ uses to describe $S$ after $M$. In an Everettian model, however, $O$ now knows that the universal state has evolved into some macroscopic superposition, the only uncertainty is about which branch of this $O$ is in. For more on this see, for example, Ref.~\cite[especially \S3.2]{Albert10}.}. A GRW-style collapse theory \cite{GhirardiRimini+86} would, on the other hand, allow both types of mixture and specifically have many fewer improper mixtures.

In any extension to quantum theory with non-linearity, however, it is not valid to describe proper mixtures using density operators. This follows because, for a non-linear evolution $\mathcal{E}$ acting on some ensemble of states and corresponding probabilities $\{(\rho_{j},p_{j})\}_{j}$, applying $\mathcal{E}$ to the initial density operator is not generally the same as the density operator obtained by applying $\mathcal{E}$ to each individual state in the ensemble:
\begin{equation}
\mathcal{E}\left( \sum_{j}p_{j}\rho_{j} \right) \neq \sum_{j}p_{j}\mathcal{E}(\rho_{j}).
\end{equation}
On the other hand, density operators are the correct way to describe improper mixtures under non-linear evolution. By examining the derivation of reduced density operators, as given in \cite{NielsenChuang00} for example, it is easily seen that linearity of operations is not assumed at any point.

A third distinct type of mixed state becomes particularly evident when discussing time travel models \cite{Deutsch91,BrunWilde12}. These are \emph{true mixtures}: mixtures that are not entangled to any reference system
and yet would still be described as mixed by an observer with maximum knowledge. These arise naturally in the model of D-CTCs, where it is possible for a mixed output to be produced from a pure state. Moreover, the CV state $\tau$ in D-CTCs will generally be mixed and cannot be purified, so must also be a true mixture \cite{PatiChakrabarty+10}.

However, both improper and true mixtures may still be validly described by density operators in a non-linear theory (unlike proper mixtures). It therefore follows that non-linear theories do not treat true and improper mixtures
differently. So whilst models of time travel may introduce true mixtures conceptually, they do not affect the way in which calculations are performed. As such the purification theorem still holds for $\rho_{i}$ and one may always assume that $\rho_{i}=|\psi_{i}\rangle\langle\psi_{i}|$ by simply extending $U$ to act on the purification ancilla as the identity.

The different ontologies implied by proper, improper, and true mixtures must be carefully borne in mind when dealing with non-linear models, such as those of time travel. Insufficient clarity on this point caused significant controversy over the capabilities of D-CTCs, the so-called ``linearity trap'' \cite{BennettLeung+09,RalphMyers10,CavalcantiMenicucci10,CavalcantiMenicucci+12,BrunWilde12,Pienaar13}. The solution is to not only specify the density operator of a mixture in a non-linear model, but also specify the ontology of the mixture so that the correct treatment can be used \cite{Allen14}.

Moreover, since all non-linear evolutions treat proper and improper mixtures differently in general then if the difference is observable the result is an \emph{entanglement detector}: a device capable of telling whether or not a system is entangled with another. In any interpretation that involves instantaneous disentanglement by measurement then an entanglement detector necessarily facilitates superluminal signalling \cite[\S3.1.3]{Pienaar13}. D-CTCs and P-CTCs therefore both lead to superluminal signalling with such an interpretation.

Instantaneous disentanglement is not a necessary feature of quantum measurement, however. There are many ontologies that do not require instantaneous disentanglement. Moreover, one can construct alternative models of quantum measurement that prevent superluminal signalling in non-linear theories by construction without subscribing to a particular interpretation \cite{CzachorDoebner02,Kent05}. So whilst non-linearity does not necessarily lead to signalling, it may do depending on the ontology \cite[\S3]{Pienaar13}.

\subsection{The Role of Ontology} \label{sec:TT:role-of-ontology}

Any non-linear extension to quantum theory (and, in particular, models of time travel) raises ontological problems compared to standard quantum theory. The clearest examples are in the differences that non-linearity introduces between proper and improper mixtures. As discussed in Sec.~\ref{sec:TT:non-linearity}, this essentially means that density operators cannot be used to describe proper mixtures in a non-linear model.

Traditionally, the only differences between the types of mixture in quantum theory have been ones of interpretation and preferred ontology. Non-linearity breaks this long-standing equivalence of different interpretations. For example, with non-linearity Everettian quantum theory need not produce the same predictions as Bohmian quantum theory simply because they differ ontologically. Therefore, if a method of time-travel to the past were discovered (a real CTC, for example) then in principle it could be used to experimentally distinguish between certain ontological interpretations.

There is a wrinkle in this argument, however, that points towards the other role of ontology in models of time travel. That is, there are at least two models of time travel in quantum theory each of which is somewhat well-motivated: D-CTCs and P-CTCs. They are only ``somewhat'' well-motivated as neither comes with a solid first-principles argument deriving the model, they are instead developed using plausibility arguments. A clear idea of a quantum ontology will play a key role in selecting the correct model for time travel.

Compare this to the classical case. An uncontroversial model for classical time travel to the past was briefly outlined in Sec.~\ref{sec:TT:paradoxes}. Why was this model used and why is it not controversial? Because, as far as physics is concerned, there is broad consensus on the essential features of classical ontology. Classical physics suggests that systems have a local ontology that describes their entire state at any point, so it is easy to simply transport this local ontology back in time when required.

As an aside, it is interesting to note that in Ref.~\cite{AaronsonWatrous09} a curiously different model of classical time travel was used. In that paper, it was claimed that the computational power of D-CTCs is equivalent to a classical computer equipped with time travel abilities (both being $\mathsf{PSPACE}$). However, this is only true if one takes the unusual position of requiring that classical time travel imposes a consistency condition on \emph{probability distributions} rather than ontic states, as in Sec.~\ref{sec:TT:paradoxes}. In this way, Ref.~\cite{AaronsonWatrous09} avoids dynamical consistency paradoxes normally associated with classical time travel\footnote{Another way to look at this oddity is to note that, since classical theory is a subset of quantum theory, P-CTCs should be able to model classical time travel but P-CTCs only have the computational power of $\mathsf{PP}$. If classical time travel has the power of $\mathsf{PSPACE}$ this would seem to suggest $\mathsf{PP} = \mathsf{PSPACE}$ which would be a hugely surprising result (of course, this is not implied mathematically, but merely illustrates the odd model of classical time travel assumed).}.

To illustrate how ontology influences the plausibility of different time travel models, consider D-CTCs. When first introduced, the suggested interpretation of D-CTCs was in an Everettian ontology where the CV system heading back in time enters a different branch or ``world'' of the quantum state \cite{Deutsch91}. This is not an ontological assumption from which the model is derived, but rather a suggested interpretation that makes a certain amount of sense of the model and was likely influenced by Deutsch's pre-existing preference for that interpretation\footnote{In Ref.~\cite{Deutsch91} he repeatedly refers to Everettian quantum theory as ``unmodified'' quantum theory.}. On the other hand, D-CTCs can be analysed from an epistemic realist perspective on quantum state ontology---see Sec.~\ref{sec:SO:desired-ontologies}---and found to be inconsistent \cite{WallmanBartlett12}. The conclusion is clear: someone, like Deutsch, who prefers Everettian ontology can consistently choose the D-CTC model (though they are not necessarily forced to) while an epistemic realist cannot. Moreover, as in Sec.~\ref{sec:TT:D-CTCs}, one can try to justify D-CTCs by assuming that reduced mixtures are ontic states and therefore suitable candidates for a classical-like consistency condition, but taking this ontological route seriously requires a careful justification of the very unusual quantum ontology that results.

Interestingly, P-CTCs were introduced without any suggested ontological basis at all. In fact, one of the papers that introduced the model preferred to call it ``effective quantum time travel'' that one might simulate in a lab using manual postselection \cite{Svetlichny09}. The primary motivation is that quantum teleportation creates a quantum communication channel and that the (unphysical) protocol described in Sec.~\ref{sec:TT:P-CTCs} modifies this channel to communicate with the past. Despite this, there are still ontological arguments to be made for and against P-CTCs. For example, P-CTCs have been shown to be compatible with one model for quantum time travel that takes the path-integral approach, rather than the circuit approach used here \cite{Politzer94}. On the other hand, the P-CTC model fails to specify a state for the system that travels back in time analogous to $\tau$ in D-CTCs. It is therefore unclear how to answer questions like ``what is the state of the CR and CV system just after applying $U$?'' with P-CTCs, for example. This may not be necessary for some mechanisms of time travel, but it would certainly be bizarre if there were not a well-defined state for a system traversing a CTC.

It is a strength of the abstract quantum circuit-based approach that it can easily isolate these ontological concerns. In particular, it facilitates direct comparisons between the classical and quantum cases which then suggests two core features of quantum theory that contribute to the ontological ambiguity found in quantum time travel. First, the stochasticity found in quantum theory is normally thought to be irreducible, while classical physics is normally taken to be fundamentally deterministic. Classically, this allows one to unproblematically consider the state of the CV system to be in a definite ontic state, even if the precise state is not known. Second, quantum states of the CV and CR systems are generally non-separable, whereas classically there is always a well-defined concept of the state of the CV system separate from the CR system. Classically, this allows the CV system to be extracted from the future and transplanted to the past easily, while quantum mechanically some non-trivial operation is required to separate the two. With D-CTCs this is a partial trace, while with P-CTCs it is a projection. Without these two features, defining an uncontroversial quantum model for time travel would likely be as easy as it is classically.

On the other hand, the circuit approach also introduces some further ambiguities as a result of being somewhat divorced from the precise physical mechanisms. In particular, one prominent ambiguity might be called the \emph{dynamical ambiguity} \cite{Allen14}. In a D-CTC circuit, for example, Eq.~(\ref{eq:TT:D-CTC-dynamical-change}) specifies a non-trivial dynamical change but it is not clear exactly when this should occur. There is a similar ambiguity in P-CTCs: there is a physical change that has no well-defined location. These ambiguities are related to how, in quantum circuits, one can slide gates along wires freely without changing the overall effect of the circuit and is probably therefore a fundamental ambiguity for this circuit approach. While a little inelegant, this ambiguity need not be of particular concern, however, because it is entirely unobservable in both models (the predictions are the same regardless of where Eq.~(\ref{eq:TT:D-CTC-dynamical-change}) is placed, for example).

So ontology can and should inform one's approach to quantum time travel. Likewise, theorising about quantum ontology can highlight certain ontological issues such as non-separability in quantum theory and whether it can make sense to discuss a fundamentally different sort of evolution on one part of a system (CV) from the other (CR). There is also one further link between the studies of quantum time travel and ontology: models of quantum time travel might directly influence the development of ontologies for quantum theory. Suppose, for example, that a physical CTC were discovered. Experiments on it could reveal a particular model of quantum time travel to be correct which should then,  in turn, suggest some quantum ontologies as more plausible than others (\emph{e.g.} D-CTCs may suggest Everettian ontology and rule out many epistemic realist ontologies). It may even possible for this to happen entirely theoretically: if an exceptionally well-motivated and natural model for quantum time travel were developed, then this would probably also provide clues towards natural ontologies for quantum theory.

\section{Alternative Time Travel Models} \label{sec:TT:alternative-models}

The discussion of D-CTC and P-CTC models in the previous section makes three things clear. First, there can be more than one reasonable way to extend quantum theory to include time travel to the past. Second, there are several reasons to dislike either of these models, depending on one's philosophical bent. Third, the ontological and interpretational foundations of each are far from certain but, if one thing is known, they definitely differ.

These points raise some questions. What other reasonable models of quantum theory with time travel might exist? How might they compare to these existing examples? What ontological implications would they suggest? These questions will be tackled in this section. Before developing some new theories, it will be useful to first review some background on integrating over quantum states.

\subsection{Integrating over Quantum States} \label{sec:TT:integrating}

Consider any $d$-dimensional quantum system with Hilbert space $\mathcal{H}$. For any scalar function $\mathcal{J}:\mathcal{H}\rightarrow\mathbb{C}$ one can consider the integral over the pure states 
\begin{equation}
J = \int_{\mathcal{P}(\mathcal{H})} \d[\phi] \,\mathcal{J}(\phi),
\end{equation}
where the integration measure $\d[\phi]$ is yet to be defined. Conveniently, there exists a unique natural measure over $\mathcal{P}(\mathcal{H})$ that is invariant under unitary transformations given by taking a random unitary matrix distributed according to the Haar measure on the group $U(d)$ \cite{ZyczkowskiSommers01}. One way to write this is in the Hurwitz parametrisation \cite{BengtssonZyczkowski06,ZyczkowskiSommers01,Hurwitz97} defined with respect to some orthonormal basis $\{|\alpha\rangle\}_{\alpha=0}^{d-1}$ of $\mathcal{H}$, such that any pure state $|\phi\rangle$ takes the form
\begin{multline} \label{eq:TT:Hurwitz}
|\phi\rangle=\prod_{\beta=d-1}^{1}\sin\theta_{\beta}|0\rangle + \sum_{\alpha=1}^{d-2}e^{i\varphi_{\alpha}}\cos\theta_{\alpha}\prod_{\beta=d-1}^{\alpha+1}\sin\theta_{\beta}|\alpha\rangle + e^{i\varphi_{d-1}}\cos\theta_{d-1}|d-1\rangle,
\end{multline}
for some parameters $\theta_{\alpha}\in[0,\pi/2]$ and $\varphi_{\alpha}\in[0,2\pi)$. In this parametrisation, the integration measure takes the form
\begin{equation} \label{eq:TT:integration-measure}
\d[\phi(\theta_\alpha, \varphi_\alpha)] = \prod_{\alpha=1}^{d-1} \cos\theta_\alpha ( \sin\theta_\alpha )^{2\alpha-1} \d\theta_\alpha \d\varphi_\alpha.
\end{equation}
This is then a natural measure for integrating over pure quantum states that is unique up to a multiplicative constant.

This natural measure is invariant under unitary operations, so that under $|\phi\rangle\rightarrow U|\phi\rangle$ the measure transforms as $\d[\phi] \rightarrow \d[U\phi] = \d[\phi]$. It may be useful to observe that the Hurwitz parametrisation is a generalisation of the Bloch sphere parametrisation often used for qubits $\mathcal{H} = \mathbb{C}^2$ where the measure is, up to a scalar, the rotationally-invariant area measure on a sphere:
\begin{eqnarray}
|\phi\rangle & = & \sin\theta|0\rangle+e^{i\varphi}\cos\theta|1\rangle,\\
\d[\phi] & \propto & \sin(2\theta) \d(2\theta) \d\varphi.
\end{eqnarray}

Integrals over mixed quantum states can be considered similarly. For some scalar function $\mathcal{J}:\mathcal{D}(\mathcal{H})\rightarrow\mathbb{C}$, consider the integral over the density operators on $\mathcal{H}$,
\begin{equation}
J = \int_{\mathcal{D}(\mathcal{H})} \d[\tau] \,\mathcal{J}(\tau),
\end{equation}
for some integration measure $\d[\tau]$. Unlike $\mathcal{P}(\mathcal{H})$, there is no unique natural measure on $\mathcal{D}(\mathcal{H})$ \cite{BengtssonZyczkowski06} and so one must be chosen, along with a useful way to parametrise $\tau$. The result is that there is no unique natural way to define $J$; it will depend on the choice of measure used.

\subsection{Desiderata} \label{sec:TT:desiderata}

When considering how new models of quantum theory with time travel might be developed, it is useful to first consider how it might be desirable for such a theory to behave. A list of possible desirable features follows. Of course, all desiderata are linked to various philosophical prejudices, but there is still utility in considering them.

\begin{enumerate}

\item The model should have sound physical motivation and an ontological interpretation.

\item The model should be compatible with standard quantum theory, at least in so far as current experiments are concerned. In the case of CTCs it should reproduce quantum theory locally along the CTC, as well as in spacetime regions far from the CTC. It is also expected to be locally approximately consistent with special relativity and relativistic causal structure. Specifically, it should not allow superluminal signalling.

\item The model should not allow dynamical inconsistencies. In other words, it should not have disallowed evolutions that lead to dynamical consistency paradoxes for any $U$ or $\rho_i$.
 
\item The model should specify $\rho_{f}$ uniquely given $U$ and $\rho_{i}$. If multiple possible output states are considered, then probabilities for each of these should be specified. In other words, it should not have uniqueness ambiguities that lead to information paradoxes (Sec.~\ref{sec:TT:paradoxes}).

\item The model should specify a state $\tau$ that travels back in time; this should either be uniquely specified or an ensemble of possibilities with corresponding probabilities should be uniquely specified.

\item Given a pure $\rho_{i}$, prejudice might require either or both of $\rho_{f}$ and $\tau$ to also be pure.

\item The model should not be able to distinguish non-orthogonal states in a single measurement, neither should it be able to clone arbitrary quantum states.

\end{enumerate}

Feature (1) is the most subtle and subjective of these but arguably the most important. It was discussed to some extent in Sec.~\ref{sec:TT:role-of-ontology} for D- and P-CTCs.

D-CTCs have feature (2) so long as ontological assumptions regarding collapse are made that rule out superluminal signalling. P-CTCs only have feature (2) in the presence of finite noise, and even then similar assumptions about collapse are required to rule out signalling. However, as noted in Sec. \ref{sec:TT:non-linearity}, adding any non-linear evolution to quantum theory opens up the possibility of signalling in this way.

Feature (3) is not in the D-CTC model but is in the P-CTC model, while feature (4) is definitely in P-CTCs but is only in D-CTCs by adding an extra postulate or in the presence of noise, as noted in Sec.~\ref{sec:TT:D-CTCs}.

Neither the D-CTC nor P-CTC models have feature (5). P-CTCs do not specify any $\tau$, while D-CTCs specify $\tau$ but not necessarily uniquely. Feature (6) is perhaps the least compelling feature listed and is one that neither D-CTCs nor P-CTCs have. Feature (7) is also not one respected by either P-CTCs or D-CTCs.

Of course, it would be ambitious to ask for a model that satisfies all these desiderata. Notably, the standard way of introducing time travel into classical mechanics does not satisfy features (2), (3), (4), or (5). However, they do form a helpful guide for where one might start looking for new time travel models.

\subsection{Some Alternative Models} \label{sec:TT:alternative-model-selection}

The above desiderata can be used to guide the construction of new models of time travel in quantum theory. In this section two overlapping classes of new models will be considered: \emph{weighted D-CTCs} and
\emph{transition probability models}. An example of the latter, dubbed \emph{T-CTCs}, will be fully fleshed out in Sec.~\ref{sec:TT:T-CTCs}.

Weighted D-CTCs represent an extension of the model of D-CTCs. These are described by parametrising the convex subset of density operators $\tau_\alpha$ allowed by the consistency condition Eq.~(\ref{eq:TT:D-CTC-EoM}) with $\alpha$ and then assigning a weight $w_\alpha \geq 0$ to each of these valid CV states. This weighted mixture is then used to calculate $\tau=\int\d\alpha \,w_\alpha\tau_\alpha / \int\d\alpha \,w_\alpha$. The D-CTC protocol can then be used with this uniquely determined choice of $\tau$.

Clearly, this describes a whole class of models based on how the parametrisation is done and which integration measure is chosen. For example, one could weight all possibilities equally $w_{\alpha}=1$, giving a \emph{uniform} weighted D-CTC model. In terms of the desiderata, this theory would gain features (4) and (5) at least, possibly at the expense of feature (1) depending on the details and motivation of the model. Such a model is essentially that of D-CTCs, with an alternative to the maximum entropy rule.

Transition probability models make use of some useful intuition from standard quantum theory. It is common to say that the probability of an initial state $|I\rangle$ to transition into a final state $|F\rangle$ under the unitary transformation $V$ is given by the transition probability $|\langle F|V|I\rangle|^{2}$. More precisely, what is meant is that $|\langle F|V|I\rangle|^{2}$ is the Born rule probability of finding the system in state $|F\rangle$ if one were to measure the system to see if it is in state $|F\rangle$ after the transformation. As an example of this useful way of thinking, consider starting with a bipartite system, initially in state $|\psi_{i}\rangle|\phi\rangle$, and act upon it with the unitary $U$; the ``probability of finding the second system in $|\phi\rangle$'' after the transformation is $p(\phi)=\left\Vert \langle\phi|U|\psi_{i}\rangle|\phi\rangle\right\Vert ^{2}$. 

The transition probability can be generalised to apply to to mixed states $\rho$ and $\tau$. One way to do this is to use the Hilbert-Schmidt inner product $\Tr(\rho\tau)$, which is the probability for $\rho$ to be found in an eigenstate of the proper mixture $\tau$ and for $\tau$ to be a realisation of that same eigenstate, averaged over all eigenstates of $\tau$. Another option is the square of the fidelity, the interpretation of which involves considering $\rho$ as an improper mixture on which a measurement is performed by projective measurement of the larger purified system \cite{Uhlmann11}. Both of these options reduce to the transition probability in the case of both states being pure. The use of any mixed state transition probability must be motivated by its operational meaning in context. For the remainder of this chapter it will be assumed that the appropriate generalisation of transition probability to mixed states is the Hilbert-Schmidt inner product, so that the probability for a bipartite system initially in the state $\rho_{i}\otimes\tau$ to have the second system found in the state $\tau$ after some unitary transformation $U$ is given by\footnote{One curiosity of using this generalisation of transition probability is that the probability for $\rho$ to transition to $\rho$ under unitary $\mathbbm{1}$ is strictly less than unity for mixed $\rho$. This is simply a reflection of the fact that mixed states can be viewed as epistemic states over the pure states.} $p(\tau)=\Tr\left(\tau U(\rho_{i}\otimes\tau)U^{\dagger}\right)$.

Transition probability models are obtained by applying these ideas to time travel. The choices that need to be made to be define a specific model include: whether pure or mixed CV states are used, which $|\phi\rangle$ or $\tau$ are to be considered, and how $\rho_{f}$ should be separated from the CV system.

One example of a transition probability model is also a weighted D-CTC model. This is obtained by choosing the weights of a weighted D-CTC model to be the transition probabilities $w_\alpha = p(\tau_\alpha) = \Tr(\tau_\alpha^2)$, giving the equation of motion
\begin{equation}
\rho_f = \frac{ \int \d\alpha \Tr\left( \tau_\alpha^2 \right) \Tr_{\textsc{CV}} \left( U( \rho_i \otimes \tau_\alpha )U^\dagger \right) }{ \int \d\alpha \Tr\left( \tau_{\alpha}^2 \right) }.
\end{equation}

Another collection of transition probability models is found by integrating over all possible initial CV states, weighted by the corresponding transition probability, and so 
\begin{equation} \label{eq:TT:transition-tau_i}
\tau_i = Z^{-1} \int_{\mathcal{P}(\mathcal{H}_{\textsc{CV}})} \d[\phi] \,p(\phi)|\phi\rangle\langle\phi|
\end{equation}
is used if pure CV states are considered, or similarly with an integral over $\mathcal{D}(\mathcal{H}_{\textsc{CV}})$ if mixed CV states are considered. $Z>0$ normalises the state $\tau_i$ and $p(\phi)$ is the transition probability as described above. There are two options for then isolating $\rho_f$: either take the partial trace as with D-CTCs or use the same partial projection used to calculate the transition probability. The first of these, using pure CV states, gives a model with the equation of motion
\begin{equation} \label{eq:TT:pure-TI-CTCs}
\rho_f = Z^{-1} \int \d[\phi] \,p(\phi) \Tr_{\textsc{CV}}\left( U( |\psi_i\rangle\langle\psi_i| \otimes |\phi\rangle\langle\phi| )U^\dagger \right),
\end{equation}
which is one of various models on this theme.

Another such transition probability model is the \emph{model of T-CTCs}, which is the model found by integrating over all pure CV states, weighted by transition probability, but using the partial projection to find $\rho_f$. T-CTCs will be developed fully in Sec.~\ref{sec:TT:T-CTCs}, complete with a discussion of their physical motivation and ontological implications. Some of the other theories that are variations on this theme will be briefly considered in Sec.~\ref{sec:TT:relation-to-other-alternatives}.

\subsection{The Uniqueness Ambiguity and Epistemic Reasoning} \label{sec:TT:T-CTC-uniqueness-ambiguity}

Before proceeding to detail the model of T-CTCs, some remarks are in order about the uniqueness ambiguity. In Sec.~\ref{sec:TT:paradoxes} this ambiguity was introduced as a necessary and sufficient condition for a model to suffer information paradoxes, as defined in that section. The argument hinges on the idea that if the final state is uniquely determined by the initial state and dynamics, then what has occurred can be regarded as a (possibly very powerful) computation and is therefore not paradoxical according to that definition.

The argument still holds if the unique final state is an epistemic state (that is, a probability distribution over ontic states), so long as the probabilities in the epistemic state are determined by the physics rather than purely epistemic principles. If the probabilities are physically determined then any new information obtained can be viewed as being due to a probabilistic computation. For any particular final state to be likely, the physics must not only establish that final state as a possibility, but also that the corresponding probability is sufficiently high. Models for probabilistic computation are well-established and certainly not paradoxical.

Compare this to D-CTCs without either noise or the maximum entropy rule. In Sec.~\ref{sec:TT:D-CTCs} it was claimed that D-CTCs suffer from the uniqueness ambiguity and therefore information paradoxes. This is different from a probabilistic computation since D-CTCs assign no probabilities to the possible final states; they are merely left as possibilities.

Therefore, as far as the definitions in Sec.~\ref{sec:TT:paradoxes} go, information paradoxes are still avoided if uniqueness ambiguity is avoided; \emph{viz.}, when a unique physically determined epistemic state is specified. It is for this reason that requirements (4) and (5) of Sec.~\ref{sec:TT:desiderata} allow for uniquely specified epistemic states.

For example, suppose a time travel circuit is designed to produce previously unknown theorems. If this circuit produces unique theorems with certainty from the input, then, as discussed in Sec.~\ref{sec:TT:paradoxes}, this is a type of computation: a novel automated theorem prover. Similarly, if the circuit produces one of a selection of possible theorems from the input, each with a given probability, then the circuit is performing a (possibly novel) probabilistic computation. On the other hand, if a model allows for a circuit that could produce one of a range of possible theorems but has no way of giving a probability for each, then that is an information paradox.

\subsection{T-CTCs} \label{sec:TT:T-CTCs}

The model of T-CTCs may be motivated as follows. Consider a CR observer watching a standard form time travel circuit evolve and suppose that the primitive states of quantum theory are pure states. Because of this, and the purification theorem, take $\rho_i = |\psi_i \rangle\langle \psi_i|$ to be pure.

This observer watches a CV system emerge from the future in some unknown pure state $|\phi\rangle$. This is then observed to interact with a CR system, initially in state $|\psi_{i}\rangle$, via unitary $U$. The CV system then proceeds to head back in time. At this point, the observer may judge whether any given $|\phi\rangle$ is a consistent initial CV state. If someone were to have measured the CV system immediately before it travelled back in time, then the probability of their finding that any given $|\phi\rangle$ is a consistent initial state would be $p(\phi) = \left\Vert \langle\phi|U|\psi_i\rangle|\phi\rangle \right\Vert^2$. So for any pair $|\phi_1\rangle$ and $|\phi_2\rangle$, the former would be found to be consistent $p(\phi_1)/p(\phi_2)$ times more often than the latter. It therefore seems reasonable to conclude that $|\phi_1\rangle$ is $p(\phi_1)/p(\phi_2)$ times more likely to have been the initial state than $|\phi_2\rangle$. The observer therefore considers $\tau_i$ as a proper mixture over all $|\phi\rangle \in \mathcal{P}(\mathcal{H}_{\textsc{CV}})$, each weighted by $p(\phi)$, Eq. (\ref{eq:TT:transition-tau_i}). On the other hand, consistency demands that if $|\phi\rangle$ was the initial CV state, then on heading back in time the CV system must be found to be in the same state again. So the observer can describe the final state of the CR system in each case by the partial projection $\langle\phi|U|\psi_i\rangle|\phi\rangle / \left\Vert \langle\phi|U|\psi_i\rangle|\phi\rangle \right\Vert^2$ consistent with this being the case. The resulting final state for the CR system is, therefore,
\begin{eqnarray}
\rho_f & = & Z^{-1} \int \d[\phi] \,U_\phi |\psi_i\rangle\langle\psi_i| U_\phi^\dagger, \label{eq:TT:T-CTC-EoM} \\
U_\phi & \eqdef & \langle\phi|U|\phi\rangle, \\
Z & \eqdef & \int \d[\phi] \,\langle\psi_i|U_\phi^\dagger U_\phi|\psi_i\rangle, \label{eq:TT:T-CTC-normalisation}
\end{eqnarray}
where the operator $U_\phi$ acts only on $\mathcal{H}_{\textsc{CR}}$ and the constant $Z>0$ is defined to normalise $\rho_f$.

Equations~(\ref{eq:TT:T-CTC-EoM}--\ref{eq:TT:T-CTC-normalisation}) define the behaviour of the model of T-CTCs. The CR input state $|\psi_i\rangle$ was assumed to be pure, but these equations can equally be applied to mixed input states by simply replacing instances of $|\psi_i\rangle\langle\psi_i|$ with $\rho_i$.

This is not intended as a derivation, but a motivational explanation for T-CTCs akin to those given for D-CTCs and P-CTCs. As with those models, T-CTCs are defined by the Eqs.~(\ref{eq:TT:T-CTC-EoM}--\ref{eq:TT:T-CTC-normalisation}) rather than by any particular interpretation. One might even use a similar argument to motivate other models, including some of the other transition probability models mentioned in Sec.~\ref{sec:TT:alternative-model-selection}.

Several features of T-CTC immediately follow from Eqs.~(\ref{eq:TT:T-CTC-EoM}--\ref{eq:TT:T-CTC-normalisation}). First, it is a non-unitary and non-linear model. Second, it is only renormalisation non-linear and as such it only gives rise to continuous evolutions, as discussed in Sec.~\ref{sec:TT:non-linearity}. Third, there is no ambiguity in the equation of motion (\ref{eq:TT:T-CTC-EoM}), so there is no uniqueness ambiguity and no information paradoxes. Before proceeding to consider what other features T-CTCs may have, it will first be useful to re-write Eq.~(\ref{eq:TT:T-CTC-EoM}) in a simpler form.

\subsubsection{Simplification of the T-CTC Equation of Motion} \label{sec:TT:T-CTC-simplification}

The equation of motion (\ref{eq:TT:T-CTC-EoM}) for T-CTCs in its current form is rather opaque. In order to more easily calculate with the model it is useful to perform the integration in generality and thereby simplify this equation.

Let $\{|\alpha\rangle\}_{i=0}^{d-1}$ be an orthonormal basis for the $d$-dimensional CV system and expand the unitary $U$ in the Kronecker product form in this basis $U = \sum_{\alpha,\beta} A_{\alpha\beta} \otimes |\alpha\rangle\langle\beta|$, where $A_{\alpha\beta}$ are operators on the CR system. In this form, the equation of motion is
\begin{equation} \label{eq:TT:T-CTC-simplification-1}
\rho_f = Z^{-1} \sum_{\alpha,\beta,\gamma,\delta} I_{\alpha\beta,\gamma\delta} A_{\alpha\beta} |\psi_i\rangle\langle\psi_i| A_{\gamma\delta}^\dagger,
\end{equation}
having defined the integrals
\begin{equation} \label{eq:TT:T-CTC-simplification-2}
I_{\alpha\beta,\gamma\delta} = \int \d[\phi] \, \langle\phi|\alpha\rangle \langle\beta|\phi\rangle \langle\phi|\delta\rangle \langle\gamma|\phi\rangle.
\end{equation}

Now consider expanding both $\d[\phi]$ and $|\phi\rangle$ in the Hurwitz parametrisation [Eqs.~(\ref{eq:TT:Hurwitz}, \ref{eq:TT:integration-measure})] with respect to the same basis. Since, for each $\alpha$, $\d\varphi_{\alpha}$ factorises out of the measure, any integrand in which the only $\varphi_\alpha$-dependence is an integer power of $e^{i\varphi_\alpha}$ will integrate to zero. Considering the integrals in Eq.~(\ref{eq:TT:T-CTC-simplification-2}), every integrand will have such a phase factor unless at least one of the two following conditions is met: $\alpha = \beta$ and $\gamma=\delta$, or $\alpha = \gamma$ and $\beta = \delta$. In these cases, all phase factors will cancel out and the phase integrals will not come to zero. Discarding these zero integrals in Eq.~(\ref{eq:TT:T-CTC-simplification-1}) it is, therefore, found that
\begin{multline}
\rho_f = Z^{-1} \left( \sum_{\alpha\neq\beta} I_{\alpha\beta,\alpha\beta} A_{\alpha\beta} |\psi_i\rangle\langle\psi_i| A_{\alpha\beta}^\dagger + \sum_{\alpha\neq\beta} I_{\alpha\alpha,\beta\beta} A_{\alpha\alpha} |\psi_i\rangle\langle\psi_i| A_{\beta\beta}^\dagger \right. \\
 + \left. \sum_\alpha I_{\alpha\alpha,\alpha\alpha} A_{\alpha\alpha} |\psi_i\rangle\langle\psi_i|A_{\alpha\alpha}^\dagger \right). \label{eq:TT:T-CTC-simplification-3}
\end{multline}

By unitary invariance of the integration measure one may rotate $|\phi\rangle$ in each of these integrals so that only the $|d-1\rangle$ and $|d-2\rangle$ components contribute. Therefore, for $\alpha \neq \beta$, $I_{\alpha\beta,\alpha\beta}$ and $I_{\alpha\alpha,\beta\beta}$ are both equal to
\begin{multline}
\int \d[\phi] \, |\langle\phi|d-1\rangle|^2 |\langle\phi|d-2\rangle|^2 \\
= (2\pi)^{d-1} \left( \int \prod_{\gamma=1}^{d-3} \d\theta_\gamma \cos\theta_\gamma ( \sin\theta_\gamma )^{2\gamma-1} \right) \\
 \times \int \d\theta_{d-1} \d\theta_{d-2} \cos^3\theta_{d-1} \cos^3\theta_{d-2} \sin^{2d-1}\theta_{d-1} \sin^{2d-5}\theta_{d-2}, \label{eq:TT:T-CTC-integral-1}
\end{multline}
where in the final line the integrand has been expanded out in the Hurwitz parametrisation. Similarly 
\begin{multline}
I_{\alpha\alpha,\alpha\alpha} = \int \d[\phi] \,|\langle\phi|d-1\rangle|^4 \\
 = (2\pi)^{d-1} \left( \int\prod_{\gamma=1}^{d-3} \d\theta_\gamma \cos\theta_\gamma ( \sin\theta_\gamma )^{2\gamma-1} \right) \\
 \times \int \d\theta_{d-1} \d\theta_{d-2} \cos^5\theta_{d-1} \cos\theta_{d-2} \sin^{2d-3}\theta_{d-1} \sin^{2d-5}\theta_{d-2}. \label{eq:TT:T-CTC-integral-2}
\end{multline}
By evaluating the final lines of Eqs.~(\ref{eq:TT:T-CTC-integral-1}, \ref{eq:TT:T-CTC-integral-2}) it is seen that, for $\alpha \neq \beta$, the ratio $I_{\alpha\alpha,\alpha\alpha}/I_{\alpha\beta,\alpha\beta} = 2$. From Eq.~(\ref{eq:TT:T-CTC-simplification-3}) one therefore finds
\begin{equation}
\rho_f \propto \sum_{\alpha,\beta} \left( A_{\alpha\beta} |\psi_i\rangle\langle\psi_i| A_{\alpha\beta}^\dagger + A_{\alpha\alpha} |\psi_i\rangle\langle\psi_i| A_{\beta\beta}^\dagger \right).
\end{equation}

Finally, note the following identities, which may readily be verified by expanding the traces: $P \eqdef \Tr_{\textsc{CV}}U = \sum_\alpha A_{\alpha\alpha}$ and $\sum_{\alpha,\beta} A_{\alpha\beta} |\psi_i\rangle\langle\psi_i| A_{\alpha\beta}^\dagger = \Tr_{\textsc{CV}} \left( U\left( |\psi_i\rangle\langle\psi_i| \otimes \mathbbm{1} \right) U^\dagger \right)$. Using these, and introducing a normalising scalar $z>0$ (which is
generally different from $Z$ used before), the final form of the equation of motion becomes
\begin{equation} \label{eq:TT:T-CTC-simple-EoM}
\rho_f = z^{-1} \left( P|\psi_i\rangle\langle\psi_i|P^\dagger + d \Tr_{\textsc{CV}} \left( U\left( |\psi_i\rangle\langle\psi_i| \otimes \frac{\mathbbm{1}}{d} \right)U^\dagger \right) \right).
\end{equation}

Equation~(\ref{eq:TT:T-CTC-simple-EoM}) is in a much more revealing form than Eq.~(\ref{eq:TT:T-CTC-EoM}). It shows that the T-CTC equation of motion is a weighted mixture of the corresponding P-CTC equation of motion (\ref{eq:TT:P-CTC-EoM}) with an ordinary quantum channel. This gives the impression that T-CTCs are akin to noisy P-CTCs. Like Eq.~(\ref{eq:TT:T-CTC-EoM}), this simplified equation of motion can equally be applied to mixed CR input states $\rho_i$.

\subsubsection{The \texorpdfstring{$P$}{P} Operator} \label{sec:TT:P-operator}

The operator $P$ in Eq.~(\ref{eq:TT:T-CTC-simple-EoM}) is the same as used in P-CTCs in Eq.~(\ref{eq:TT:P-CTC-EoM}). It is the partial trace of a unitary operator and is therefore not generally unitary itself. For instance, for a general $P$ there can be states $|\phi\rangle$ in $\mathcal{P}(\mathcal{H})$ for which $P|\phi\rangle = 0$, which would not be possible if $P$ were unitary.

So $P$ does not generally preserve state norms, but the effect of $P$ on state norms is still bounded. Consider $P$ acting on a vector $|\psi\rangle$ and let $\{|\alpha\rangle\}_{\alpha=0}^{d-1}$ be any orthonormal basis on the $d$-dimensional CV system. Using the triangle inequality and unitarity one can bound $\left\Vert P |\psi\rangle \right\Vert$ as follows
\begin{equation} \label{eq:TT:P-bound}
\left\Vert P|\psi\rangle \right\Vert \leq \sum_\alpha \left\Vert \langle\alpha|U|\psi\rangle|\alpha\rangle \right\Vert \leq \sum_\alpha \left\Vert U|\psi\rangle|\alpha\rangle\right\Vert = d \left\Vert \psi\right\Vert. 
\end{equation}
Moreover, this bound can be achieved as demonstrated in Ref.~\cite{BrunWilde17} and is therefore the tightest possible general bound.

\subsubsection{Paradoxes} \label{sec:TT:T-CTC-paradoxes}

It has already been noted that T-CTCs have no uniqueness ambiguities and therefore no information paradoxes. With the simplified equation of motion (\ref{eq:TT:T-CTC-simple-EoM}) to hand it is now easy to see that T-CTCs are also always dynamically consistent. Even though it is possible for $P|\psi_i\rangle = 0$, the second term of Eq.~(\ref{eq:TT:T-CTC-simple-EoM}) will always give a non-zero density operator. Therefore, T-CTCs contain neither type of paradox identified in Sec.~\ref{sec:TT:paradoxes}. Unlike P-CTCs and D-CTCs, no noise or extra rule is required to avoid these paradoxes.

An example is instructive. Consider the following toy model of an unproved theorem paradox as a standard form time travel circuit, introduced in Ref.~\cite{LloydMaccone+11a} and illustrated in Fig.~\ref{fig:TT:unproved-theorem}. The CR system is a pair of qubits, $M$ and $B$ representing the mathematician and book respectively, initially in state $|0\rangle_B |0\rangle_M$. The CV system is a single qubit representing the mathematician heading back in time. The unitary $U$ consists of a pair of $\textsc{CNOT}$ gates representing the writing and reading of the book with a swap for when the mathematician swaps places with their time-travelling self, as illustrated. Clearly, by extending this toy model to use $N$ qubits for each of $M$, $B$, and CV it would allow a theorem to be encoded in an $N$-bit string.

\begin{figure}
\begin{centering}
\includegraphics[scale=0.35,angle=0]{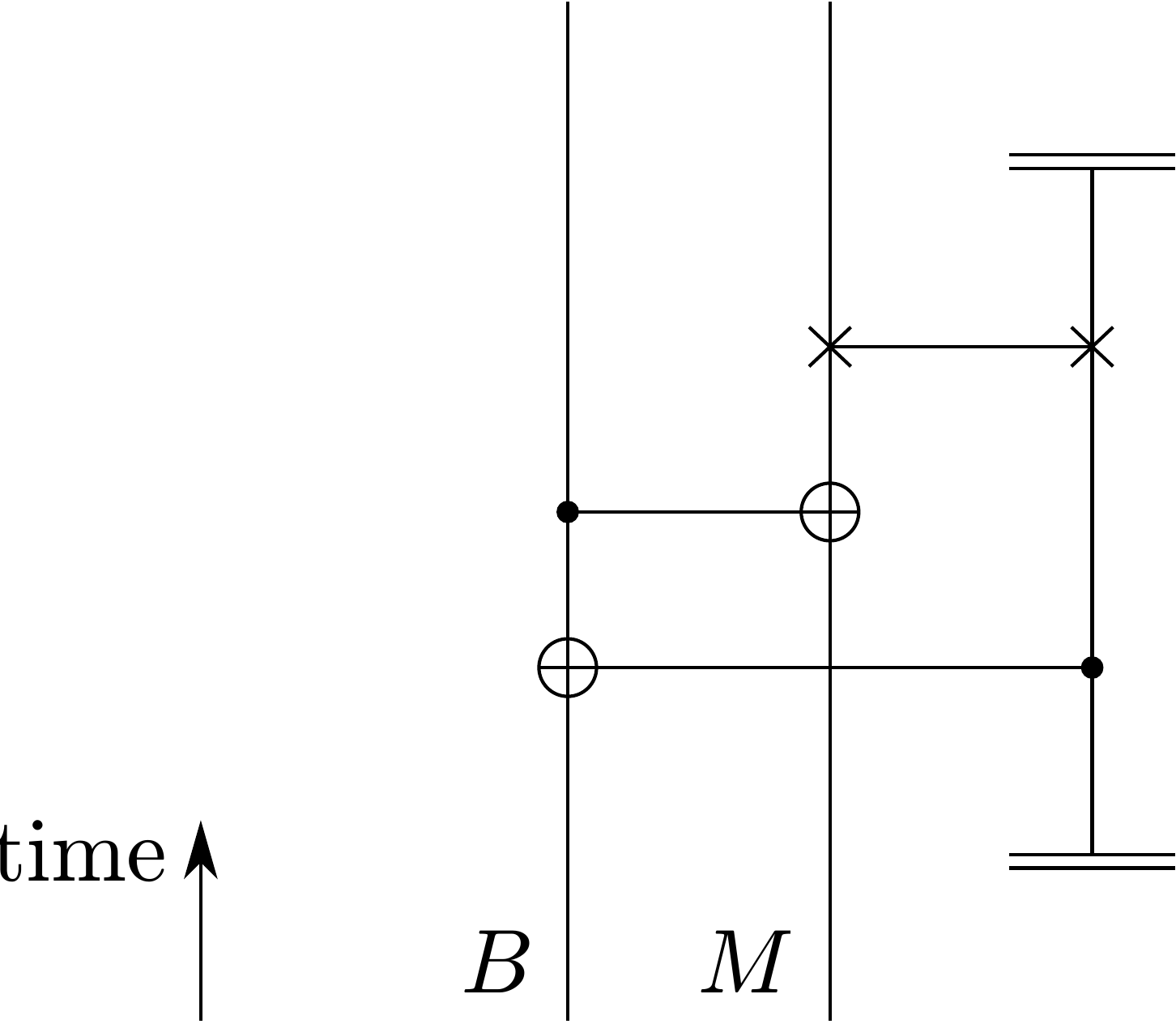}
\par\end{centering}
\protect\caption{Toy model of an unproved theorem paradox as a standard form circuit from Ref.~\cite{LloydMaccone+11a}. The two CR qubits are labelled $B$ for the book and $M$ for the mathematician. The unitary gates illustrated represent, from bottom to top, the book being written with a $\textsc{CNOT}$, the book being read with a $\textsc{CNOT}$, and the mathematician swapping places with their future self.} \label{fig:TT:unproved-theorem}
\end{figure}

For this circuit the $P$ operator is
\begin{equation} \label{eq:TT:unproved-theorem-P}
P = |0\rangle_B\langle0| \otimes |0\rangle_M\langle0| + |0\rangle_B\langle1| \otimes |1\rangle_M\langle1| 
+ |1\rangle_B\langle1| \otimes |0\rangle_M\langle1| + |1\rangle_B\langle0| \otimes |1\rangle_M\langle0|
\end{equation}
and therefore $P|00\rangle_{BM} = |00\rangle_{BM} + |11\rangle_{BM}$. To get the second term of Eq.~(\ref{eq:TT:T-CTC-simple-EoM}), note that $\Tr_{\textsc{CV}} \left( U( |00\rangle_{BM}\langle00| \otimes \mathbbm{1} )U^\dagger \right) = |00\rangle_{BM}\langle00| + |11\rangle_{BM}\langle11|$. The output the unproved theorem T-CTC circuit is therefore
\begin{equation}
\rho_f = \frac{1}{2} |00\rangle_{BM}\langle00| + \frac{1}{4} |00\rangle_{BM}\langle11| + \frac{1}{4} |11\rangle_{BM}\langle00| + \frac{1}{2} |11\rangle_{BM}\langle11|.
\end{equation}
This result may also be verified, with rather more effort, directly from the integral expression Eq.~(\ref{eq:TT:T-CTC-EoM}).

Compare this to the cases of D-CTCs and P-CTCs. For D-CTCs, any state diagonal in the computational basis will satisfy Eq.~(\ref{eq:TT:D-CTC-CC}) for this circuit. The result is a one-parameter continuous family of consistent evolutions representing any probabilistic mixture of possible ``theorems'' along the CV system. Notably, since the D-CTC description is diagonal it is entirely classical, except the consistency condition requires consistency on probability distributions over the theorems rather than the theorems themselves\footnote{\emph{C.f.} the unusual model for classical time travel used in Ref.~\cite{AaronsonWatrous09}, discussed in Sec.~\ref{sec:TT:role-of-ontology}}. For P-CTCs, Eq.~(\ref{eq:TT:unproved-theorem-P}) gives $|\psi_f\rangle = \frac{1}{\sqrt{2}}\left( |0\rangle_B |0\rangle_M + |1\rangle_B |1\rangle_M \right)$ as the circuit's output, so one obtains an equal superposition over all possible ``theorems''.

\subsubsection{Computation} \label{sec:TT:T-CTC-computation}

Some basic facts about the computational abilities of T-CTCs can be read straight from Eq.~(\ref{eq:TT:T-CTC-simple-EoM}). The second term could be realised in ordinary quantum theory and so is limited to the power of $\mathsf{BQP}$, while the first term is the P-CTC equation of motion (\ref{eq:TT:P-CTC-EoM}). It follows immediately that P-CTCs can therefore trivially simulate T-CTCs, so clearly T-CTCs cannot efficiently solve any problems that are not contained within $\mathsf{PP}$. This shows that P-CTCs are at least as computationally powerful as T-CTCs.

Moreover, the form of Eq.~(\ref{eq:TT:T-CTC-simple-EoM}) suggests that T-CTCs may be less powerful than P-CTCs. This is because, for a T-CTC, $\rho_{f}$ is only a pure state if either $P|\psi_{i}\rangle=0$ or if the two terms in Eq. (\ref{eq:TT:T-CTC-simple-EoM}) are equal. So every T-CTC algorithm that outputs a pure state is achievable on an
ordinary quantum computer in exactly the same way. Any potential algorithm for a T-CTC-equipped computer that makes computational use of the extra power of first term in Eq. (\ref{eq:TT:T-CTC-simple-EoM}) would therefore have to output mixed states. This observation also prevents T-CTCs from being able to perform an arbitrary postselected quantum measurement, since many postselected measurement outcomes are pure states. Therefore, one could not prove that T-CTCs have the power of $\mathsf{PP}$ using the same method used for P-CTCs \cite{LloydMaccone+11b}.

Despite this, it has recently been shown that T-CTCs can efficiently simulate any P-CTC circuit to arbitrary precision \cite{BrunWilde17} (that is, the undesired ``error term'' can be made exponentially small with a linear number of CV qubits). This does not invalidate the above comments, which demonstrate that T-CTCs cannot \emph{perfectly} simulate P-CTCs. This approximate simulation is, however, enough to show that T-CTCs do indeed have the computational power of $\mathsf{PP}$.

\subsubsection{Mixed States and Non-linearity} \label{sec:TT:T-CTC-non-linearity}

In Sec.~\ref{sec:TT:non-linearity} it was noted that both improper and true mixtures are validly described with density operators in non-linear extensions of quantum theory and that the purification still holds so $\rho_i = |\psi_i\rangle\langle\psi_i|$ may always be assumed. This remains true in the model of T-CTCs. It also remains true that proper mixtures are not validly described by density operators and that non-linearity of T-CTCs allows for the possibility of creating an entanglement detector and thereby signalling, exactly as with D-CTCs and P-CTCs.

Another consequence of non-linearity is that both D-CTCs and P-CTCs are capable of distinguishing non-orthogonal states in a single measurement. However, this is not the case with T-CTCs as shall now be shown. 

Consider the problem of distinguishing between two states $\rho$ and $\sigma$. The probability of success when using a single optimal measurement is given by $\frac{1}{2}\left( 1+D(\rho,\sigma) \right)$, where $D(\rho,\sigma) \eqdef \frac{1}{2}\Tr|\rho-\sigma|$ is the trace distance between the states \cite{FuchsvandeGraaf99}. Therefore, $\rho$ and $\sigma$ are perfectly distinguishable in a single measurement if and only if $D(\rho,\sigma)=1$.

Another measure of distinguishability of quantum states is the the fidelity between $\rho$ and $\sigma$, defined as \cite{NielsenChuang00}
\begin{equation} \label{eq:TT:fidelity}
F(\rho,\sigma) \eqdef \Tr \sqrt{ \rho^{1/2} \sigma \rho^{1/2} }.
\end{equation}
In the case of pure states $|a\rangle$ and $|b\rangle$, the Fidelity takes on the particularly simple form $F(|a\rangle, |b\rangle) = |\langle a|b\rangle|$. 

Suppose one wishes to distinguish quantum states using a T-CTC. Only pure state inputs need be considered, so what is required is a bound on the distinguishability of the output states of some T-CTC circuit, $\rho_f^a$ and $\rho_f^b$, for which the input states were $|a\rangle$ and $|b\rangle$, respectively. Note first that trace distance is bounded by fidelity \cite{NielsenChuang00}
\begin{equation} \label{eq:TT:disitinguishability-1}
D(\rho_f^a, \rho_f^b) \leq \sqrt{ 1 - F(\rho_f^a, \rho_f^b)^2 }.
\end{equation}
It is then useful to separate the terms of $\rho_f^{a,b}$ seen in Eq. (\ref{eq:TT:T-CTC-simple-EoM}). Therefore, write $\rho_f^\psi = ( 1 - \lambda^\psi )\sigma^\psi + \lambda^\psi \tau^\psi$, where $\tau^\psi \eqdef \Tr\left( U( |\psi\rangle\langle\psi| \otimes \frac{\mathbbm{1}}{d} )U^\dagger \right)$ and 
\begin{equation} \label{eq:TT:disitinguishability-2}
\lambda^\psi \eqdef \frac{ d }{ d + \left\Vert P|\psi\rangle \right\Vert^2 } \geq \frac{1}{d+1}
\end{equation}
where this inequality is a result of Eq. (\ref{eq:TT:P-bound}). Using strong concavity and monotonicity of fidelity under quantum operations \cite{NielsenChuang00} it is seen that 
\begin{eqnarray} \label{eq:TT:disitinguishability-3}
F(\rho_f^a, \rho_f^b) & \geq & \sqrt{ \lambda^a \lambda^b } F(\tau^a, \tau^b ) \nonumber \\
 & \geq & \sqrt{ \lambda^a \lambda^b } F(|a\rangle, |b\rangle) \geq \frac{1}{d+1} |\langle a|b\rangle|.
\end{eqnarray}
Finally, observe that by Eqs.~(\ref{eq:TT:disitinguishability-1}--\ref{eq:TT:disitinguishability-3})
\begin{equation}
D(\rho_f^a, \rho_f^b) \leq \sqrt{ 1 - \frac{|\langle a|b\rangle|^2 }{ (d + 1)^2 } }\leq 1,
\end{equation}
with equality to unity only possible if $\langle a|b\rangle = 0$.

This proves that the output states of a T-CTC circuit are only perfectly distinguishable from one another in a single measurement if the input states were. However, as proved in Ref.~\cite{BrunWilde17}, T-CTCs can approximately simulate P-CTCs to arbitrary precision efficiently. It therefore follows that while T-CTCs cannot \emph{perfectly} distinguish non-orthogonal quantum states, they can distinguish any set of linearly independent non-orthogonal states to arbitrary precision by simulating the corresponding P-CTC circuit.

It is comparatively very simple to observe that T-CTCs are incapable of cloning pure states. A pure state cloning machine always outputs a pure state. Since any T-CTC outputting a pure state can be simulated exactly by an ordinary quantum operation, then the no-cloning theorem for T-CTCs is simply a result of the no-cloning theorem in ordinary quantum theory. In exactly the same way, it also immediately follows that T-CTCs are incapable of deleting arbitrary pure states. However, the question as to whether mixed states can be broadcast \cite{BarnumCaves+96} is left open.

\subsubsection{Relation to Other Alternatives} \label{sec:TT:relation-to-other-alternatives}

Before discussing these results, it should be noted that some of the results proved for T-CTCs above can be easily modified to apply to some of the closely related models introduced in Sec.~\ref{sec:TT:alternative-model-selection}.

First, consider the modification to T-CTCs where, instead of integrating over pure CV states, mixed CV states are integrated over. This represents a class of models since there is no unique natural choice for the integration measure, but it is still possible to deduce some general properties. These models are only renormalisation non-linear and so continuity follows immediately. It is also possible to show that these models always define a unique non-zero $\rho_f$ for every $\rho_i$ and $U$, so that the models suffer neither dynamical consistency nor information paradoxes\footnote{The proof that $\rho_f$ is non-zero follows by showing that the integrand is positive semi-definite and that there exist some $\tau$ for which the integrand is non-zero. Importantly, an appropriate assumption would be needed about the positivity of the chosen measure.}.

Now consider the modification to T-CTCs expressed in Eq.~(\ref{eq:TT:pure-TI-CTCs}), where, instead of separating CV and CR systems by a projection, they are separated by a partial trace. This model always defines a unique non-zero\footnote{The proof that $\rho_f$ is always non-zero follows similarly to the previous case, by proving that the integrand is always positive semi-definite and there always exist CV states $|\phi\rangle$ for which both the integrand and $\d[\phi]$ are non-zero.} $\rho_{f}$, thus avoiding both dynamical consistency and information paradoxes. It is not, however, renormalisation non-linear but polynomial non-linear.

There are, of course, further variations which could be considered. For example, one could use an alternative generalisation of transition probability for mixed states, as noted in Sec.~\ref{sec:TT:alternative-model-selection}. The purpose of this discussion is to show that whilst T-CTCs were focussed on above, the other theories mentioned in Sec.~\ref{sec:TT:alternative-model-selection} also have reasonable properties and may be worthy of further development.

\section{Summary and Discussion} \label{sec:TT:summary-and-discussion}

This chapter approached ontology and causality in quantum theory from an unusual angle. By assuming that time travel to the past is possible, one is confronted with an interesting cross-section of problems from ontology and causality. In particular, one is drawn towards non-linear extensions of quantum theory which cast light on certain ontological issues in quantum theory.

Having discussed the various motivations for studying quantum time travel in Sec.~\ref{sec:TT:intro}, this chapter followed the tradition of circuit-based approaches to time travel to the past. This is a convenient abstract approach that can be summarised using standard form circuits introduced in Sec.~\ref{sec:TT:standard-form}. Of course, time travel to the past always raises the possibility of paradoxes, which were systematically introduced in Sec.~\ref{sec:TT:paradoxes} in the context of classical time travel.

Following this, previous work on quantum time travel in the circuit approach was outlined, focussing on the model of D-CTCs in Sec.~\ref{sec:TT:D-CTCs} and the model of P-CTCs in Sec.~\ref{sec:TT:P-CTCs}. The roles of non-linearity and ontology in these models were discussed in Secs.~\ref{sec:TT:non-linearity}, \ref{sec:TT:role-of-ontology} where it was noted that not only do one's ontological preferences influence time travel models, but that the types of reasonable non-linear extensions to quantum theory may hint at approaches to quantum ontology.

The main results of this chapter were in Sec.~\ref{sec:TT:alternative-models}. The strengths and shortcomings of D- and P-CTCs were used to identify two classes of new quantum time travel models. In particular, the model of T-CTCs was developed in full and its properties were derived.

Non-linear extensions of quantum theory are subtle since the long-standing plurality of co-existing interpretations is broken. When considering time travel this manifests itself in at least two ways. The first is in the development and motivation of various possible models: ontological bias will affect decisions made. The second is in using those models: mixed states with ontological differences but the same density operator may behave differently, as discussed in Sec.~\ref{sec:TT:non-linearity}. Neither of these issues arise when considering time travel classically, since ontology is generally clear and non-linear evolutions are commonplace.

This uniquely quantum issue has both positive and negative effects on the resulting models. The way in which quantum theory works allows models of time travel that do not suffer from the paradoxes that are present classically, but which generally break some of the central structure of quantum theory. Distinguishability of non-orthogonal states, state cloning/deleting, and the spectre of superluminal signalling all present themselves. It also appears that computational power is greatly increased even beyond that of quantum computers.

The existing models of quantum time travel are not without their shortcomings. Most troubling is that both D-CTCs and P-CTCs suffer from paradoxes (of the information and dynamical consistency types respectively). While both can be eliminated by arbitrarily small noise, the models themselves remain paradoxical. The two classes of new models presented in Sec.~\ref{sec:TT:alternative-model-selection} were designed to avoid these paradoxes and hopefully also satisfy many of the other desiderata laid out in Sec.~\ref{sec:TT:desiderata}.

Of the new models, that of T-CTCs was selected primarily due to its physical motivation. To illustrate the strength of the physical story told in Sec.~\ref{sec:TT:T-CTCs}, consider applying the same reasoning in a classical context.

A CR observer watching a classical time travel circuit sees a CV system in an unknown ontic state $\tilde{\tau}_{i}$ emerge from the future, interact with a CR system in a known state, and then disappear back to the past in the ontic state $\tilde{\tau}_{f}$. For each $\tilde{\tau}_{i}$, the observer knows that when it heads back in time it must be found to be in the same state. Whilst in quantum theory the probability of finding a system in a given state is given by the transition probability, the corresponding probability classically is either unity or zero: Either $\tilde{\tau}_{i}=\tilde{\tau}_{f}$ or $\tilde{\tau}_{i}\neq\tilde{\tau}_{f}$. So when the CR observer takes a probability distribution over all possible CV states $\tilde{\tau}_{i}$, the only states to which non-zero probabilities are assigned are precisely those states for which $\tilde{\tau}_{i}=\tilde{\tau}_{f}$ after the interaction. What is missing from this account is a way of specifying the relative probabilities assigned to each CV state. One might choose to use the principle of indifference and weight each possibility equally, but this is not necessary.

A very similar argument could be used to argue in favour of D-CTCs, for example, by demanding exact equality of reduced density operators rather than consistency via the transition probability. In this case, using the principle of indifference would lead to a result that is equivalent to the uniform weighted D-CTCs mentioned in Sec.~\ref{sec:TT:alternative-model-selection}. However, by accepting the interpretation of transition probabilities and supposing that only pure states are primitive, T-CTCs do have a clear physical motivation.

Of course, this was never meant to be a cast-iron argument for T-CTCs and there is a certain amount of vagueness in the description given in Sec.~\ref{sec:TT:T-CTCs}. Similar interpretational vagueness is found with both D-CTCs and P-CTCs and should be expected when attempting to extend quantum theory (which lacks consensus on interpretation) to a non-linear regime that is so alien to it.

These arguments for the classical model and uniform weighted D-CTCs differ from the argument for T-CTCs in an important respect. The probabilities assigned to the different possible histories with T-CTCs are physically determined: they are proportional to the transition probabilities. On the other hand, in the above discussion of the classical model and D-CTCs using the same narrative there is no physical assignment of probabilities. The use of principle of indifference is an epistemic move, not a physical one.

Consider the desirable features listed in Sec.~\ref{sec:TT:desiderata} in the light of the model of T-CTCs. The above argument aims to satisfy feature (1) at least as far as with D- or P-CTCs. As with D- and P-CTCs, feature (2) is satisfied so long as an appropriate ontology of measurement is chosen [Sec.~\ref{sec:TT:non-linearity}]. Features (3), (4), and (5) are satisfied without condition, as discussed in Sec.~\ref{sec:TT:T-CTC-paradoxes}. Feature (6) is potentially partially satisfied, in that both $\tau$ and $\rho_{f}$ might be considered ontologically pure, but since a proper mixture is taken over so many possibilities, the mathematical form of either is very rarely pure. In Sec.~\ref{sec:TT:T-CTC-non-linearity} it is shown that feature (7) is satisfied strictly speaking, although if one allows for arbitrarily small error then it is violated.

Does this mean that T-CTCs are a better model for quantum theory with time travel or, more specifically, for quantum theory in the presence of CTCs? Not necessarily. The physical motivation for T-CTCs is far from a ``first principles'' argument and there is still the question as to how, and when, the state projection occurs. However, both D-CTCs and P-CTCs have incomplete physical motivations and both leave questions as to exactly how, and when, a proposed physical change occurs---all three models share the unobservable dynamical ambiguity.

What has been comprehensively shown is that there is a whole landscape of other theories out there. The quantum circuit approach to quantum theory with time travel may be very attractive in that it abstracts away from knotty problems with spacetime geometry or any other exact mechanism for time travel, but it is perhaps too general for the problem at hand. In order to identify a more robustly physical solution to quantum theory with time travel it may be necessary to use a different approach, such as path integral or field theoretic ideas. Alternatively, by very carefully committing to a specific ontology for quantum theory it may be possible to identify the corresponding theory of time travel. When non-linearity is present vagueness on this point is problematic.

\chapter{Conclusions and Further Work} \label{ch:CO}

This thesis has considered the possibilities and impossibilities for ontologies of quantum systems in three broad ways. First, by taking a rather traditional approach to the ontology of quantum states. Second, by considering how to properly analyse causality in quantum systems. Third, by considering how one might account for time travel to the past in a quantum universe and how this interacts with possible ontologies. The conclusions for each will now be discussed in turn, together with what they suggest in terms of further work.

\section{State Ontology and Macro-realism}

The limitations that quantum theory places on the ontology of states and the possibility of macro-realist ontologies were discussed in Chaps.~\ref{ch:SO} and \ref{ch:MR} respectively. In the first instance, these chapters were concerned with which ontological features can or cannot be compatible with quantum theory, with extensions concerning compatibility with experiments being developed later.

The primary result of Chap.~\ref{ch:SO} was Thm.~\ref{thm:SO:superpositions-are-ontic} which proved that almost every quantum superposition state must be ontic in $d>3$ dimensions (equivalently, they cannot be epistemic). The conclusion for any prospective epistemic realist must be that even though a superposition inherits all of its properties and dynamics from underlying basis states, one cannot use this fact to construct a simpler ontology where the superposition also inherits its ontic states from underlying basis states. Note that the theorem gives almost the strongest version of this statement possible, holding for almost every superposition with respect to any given orthonormal basis. 

To illustrate, consider Bohmian mechanics of a single particle and its implied ontology where ontic states are of the form $\lambda = (\vec{r},|\psi\rangle)$, as described in Sec.~\ref{sec:MR:loopholes}. Here, the position basis of the particle plays a special role, with its ``true value'' being represented by the position vector $\vec{r}$. A less extravagant variation on this ontology might aim to make use of this special role and have other quantum states inherit their ontology from the ontic states of the position basis. However, Thm.~\ref{thm:SO:superpositions-are-ontic} shows that this can never be achieved---superpositions will always require novel ontic states not accessible to the underlying basis in order to be compatible with quantum theory.

One avenue for developing on this result is to consider \emph{ontic independence} more generally. By analogy with linear independence in geometry, one might define a quantum state $|\psi\rangle$ to have ontic independence with respect to some set of other quantum states if it can access ontic states not covered by preparations of that set. In this language, Thm.~\ref{thm:SO:superpositions-are-ontic} requires almost all quantum states to have ontic independence with respect to any given basis to be compatible with quantum theory in $d>3$ dimensions. It would be interesting to see if more powerful ontic independence results could be proved, showing that certain quantum states necessarily have ontic independence with respect to other sets that aren't necessarily bases.

The techniques of Thm.~\ref{thm:SO:superpositions-are-ontic} then formed the basis of the rest of the results in Chaps.~\ref{ch:SO}, \ref{ch:MR}. One of the primary motivations for epistemic realist ontologies with ontic overlaps is to explain the indistinguishability of non-orthogonal quantum states in terms of these overlaps. Much like how Ref.~\cite{BarrettCavalcanti+14} proved that such overlaps cannot fully explain \emph{all} indistinguishabilities (by a failure of being maximally $\psi$-epistemic), Thm.~\ref{thm:SO:no-maximally-epistemic} proved that the indistinguishability between \emph{any pair} of quantum states cannot be fully explained in this way (for $d>3$ dimensions). Moreover, it was proved that no individual quantum state can be maximally $\psi$-epistemic at all. By considering the effect of higher dimensions, this was adapted in Thm.~\ref{thm:SO:no-overlap-large-d} to prove that many pairs of identifiable quantum states have necessarily small ontic overlap in large dimensions. In particular, as dimension $d\rightarrow\infty$ the ontic overlap between many such pairs must approach zero, while maintaining finite Born rule overlap.

These conclusions make it very difficult for the epistemic realist to use ontic overlaps to explain indistinguishability in particular, but also any other phenomena that might seem to naturally gel with ontic overlaps, such as no-cloning for example. The situation becomes very difficult in large-dimensional systems, where the ontic overlaps between any pair of quantum states satisfying $|\langle\phi\rangle|\psi\rangle|^2 < \frac{1}{4}$ must be very small (by Thm.~\ref{thm:SO:no-overlap-large-d}). Importantly, these results apply to many \emph{identifiable} pairs of quantum states, \emph{viz.} you give me a pair of quantum states and I can tell you whether the results apply to that pair in particular. This closes an important loophole common to previous similar theorems \cite{Maroney12a,BarrettCavalcanti+14,Branciard14,Leifer14a,Ballentine14}, which only proved the existence of pairs of quantum states with bounded overlaps. Theorem~\ref{thm:SO:no-overlap-large-d} also addressed another shortcoming of those previous theorems, in that its limiting case applies to states with finite Born rule overlap, not states that approach orthogonality (as discussed in Sec.~\ref{sec:SO:previous-theorems}).

These results improve on the current state of no-go theorems for epistemic realist ontologies in several ways, but there are plenty of opportunities for further improvement. Perhaps the ultimate aim for such results would be a proof that any ontological model for a quantum system must be \emph{sometimes $\psi$-ontic} \cite{Leifer14b}. An ontological model is sometimes $\psi$-ontic if for each quantum state there is some finite-measure set of ontic states that can only be obtained by preparing that particular quantum state---\emph{i.e.} each quantum state keeps a region of the ontic state space to itself. This is very similar to $\psi$-ontic, but is not ruled out by the existence of $\psi$-epistemic ontological models \cite{LewisJennings+12,AaronsonBouland+13}. Perhaps the greatest advantage to such a proof would be that it implies several other important foundational results, including Bell's theorem, and therefore could act to unify those results \cite{Leifer14b}.

Chapter~\ref{ch:MR} shifted focus to macro-realism---another possible property of quantum state ontology. The main result was Thm.~\ref{thm:MR:no-ESMR-EMMR}, which used Thm.~\ref{thm:SO:superpositions-are-ontic} to prove that two of three types of macro-realism, ESMR and EMMR, are incompatible with quantum theory in $d>3$ dimensions. This is more powerful than the Leggett-Garg argument, which is only able to rule out EMMR, while both leave the possibility of a third type of macro-realism, SSMR. Since Bohmian mechanics reproduces quantum predictions and satisfies SSMR, then no theorem can prove that quantum theory is incompatible with every SSMR ontology. However, Bohmian mechanics is $\psi$-ontic (a very restrictive condition) and it therefore may still be possible for further work to prove that quantum theory is incompatible with all $\psi$-epistemic (opposite of $\psi$-ontic) SSMR ontologies. More generally, it would be good for further work to clarify whether any significant subset of SSMR ontologies is also incompatible with quantum theory.

These main results are all unfortunately intolerant to error. That is, the proofs do not directly generalise to the case where quantum probabilities are only assumed to be approximately correct. The core reason for error-intolerance is that they are all based on the asymmetric overlap, a quantity that loses meaning in the presence of finite error. To address this, Thm.~\ref{thm:SO:small-symmetric-overlap-large-d} adapted the proof Thm.~\ref{thm:SO:no-overlap-large-d} to use the alternative symmetric overlap, resulting in a somewhat weaker yet crucially error-tolerant result. The conclusion is that an appropriate set of experiments that reproduce quantum predictions to within some small $\pm\epsilon$ could in principle be used to experimentally bound the ontic overlap between an appropriate specific pair of quantum states.

However, Thm.~\ref{thm:SO:small-symmetric-overlap-large-d} is something of a proof-of-concept result, as the proof is not very well matched to the symmetric overlap and the resulting error term does not scale very well. An alternative approach was taken in Sec.~\ref{sec:MR:error-tolerant-argument} where the $\epsilon$-asymmetric overlap was introduced as an error-tolerant generalisation of the asymmetric overlap. This was used to prove Thm.~\ref{thm:MR:no-large-overlap-for-basis}, which can be used to state error-tolerant variations on Thms.~\ref{thm:SO:superpositions-are-ontic}, \ref{thm:MR:no-ESMR-EMMR}. Similarly, the conclusion is that appropriate sets of experiments could be used to rule out epistemic superpositions and ESMR/EMMR macro-realism to within given precisions.

Both of these error-tolerant results show the way for further experimental work in this area. However in order to achieve that, a more thorough analysis of the exact measurements needed for each would be required. The introduction of the $\epsilon$-asymmetric overlap also suggests that other results using the asymmetric overlap could be made error-tolerant and amenable to experimental investigation in this way, including overlap results already present in the literature \cite{Ballentine14,LeiferMaroney13,Maroney12a}. An experimental refutation of ESMR macro-realist ontologies would be particularly useful, as approaches based on the Leggett-Garg argument are incapable of achieving this.

Section~\ref{sec:SO:communication} applied these methods to information theory and classical simulations of quantum channels in particular. The result was Thm.~\ref{thm:SO:communication-bound}, proving that to perfectly simulate an $n_q$-qubit noiseless quantum channel with one-way classical communication a noiseless classical channel of at least $2^{n_q + \mathcal{O}(1)} -1$ bits is needed, even when the sender and receiver have access to arbitrarily large shared random data. This may be seen as a kind of ``anti-Holevo'' bound [Sec.~\ref{sec:SO:communication-bound}], since while the Holevo bound requires $n_c$ qubits to store $n_c$ classical bits this requires $\mathcal{O}(2^{n_q})$ classical bits to store $n_q$ qubits. The bound proved in Thm.~\ref{thm:SO:communication-bound} asymptotically matches the best known bounds for the same simulation but, as noted in Sec.~\ref{sec:SO:further-communication-bounds}, has three key advantages. First, it is much simpler to prove. Second, because of its simplicity it may be easily seen as a result of a certain property of quantum states. Most importantly, third, the proof method can be reused to produce potentially better bounds, given a certain class of classical error-correction codes.

Clearly then, a good avenue for further work is finding examples of classical error-correcting codes to give better and more explicit bounds via Thm.~\ref{thm:SO:communication-bound}. In particular, an explicit family of such codes would produce a bound of the form $c 2^{n_q} - 1$ for some specific $c$. If $c>0.293$, this bound would exceed the best known specific bound \cite{Montina11b}.

More generally, Thm.~\ref{thm:SO:communication-bound} demonstrates the power of a fact that, to my knowledge, has not been utilised in quantum information before. That is, the existence of sets of quantum states exponentially large in dimension where every triple from that set is anti-distinguishable. Given the three-way incompatibility displayed by anti-distinguishable triples, it seems likely that this fact may become useful in the study of quantum cryptography.

Recall from Sec.~\ref{sec:SO:previous-theorems} that the current generation of ontology theorems concentrate on single systems to avoid the issues raised by Bell's theorem and the PBR theorem in multipartite environments. In particular, the preparation independence postulate (PIP) is assumed in the PBR theorem, which uses it to prove that all compatible ontological models are $\psi$-ontic. It may be interesting to revisit some of these multipartite issues in light of the quantum causal models of Chap.~\ref{ch:CM}. Bell locality and the PIP are both conditions on the way that the ontologies of multiple quantum systems can combine into a global ontology, while causal models naturally describe how local independent systems can interact in general, so it seems likely that some properties of the causal models framework might influence how one treats multiple systems in ontological models.

\section{Quantum Causality}

The appropriate way to describe and analyse causality in quantum theory was the subject of Chaps.~\ref{ch:CI}, \ref{ch:CM}. In Chap.~\ref{ch:CI}, a definition for quantum conditional independence was given and motivated, enabling a definition for a quantum Reichenbach's principle and a characterisation of quantum common causes. These were used in Chap.~\ref{ch:CM} to motivate and give a corresponding definition for full quantum causal models, capable of describing any acyclic causal scenario.

Quantum conditional independence has four equivalent definitions, enumerated in Thm.~\ref{thm:CI:QCI-full}. The particular strength of this quantum conditional independence comes from the fact that each of these definitions is a natural generalisation of a corresponding definition for classical conditional independence. While the first of these (conditions (1) and (2)) were obtained in Sec.~\ref{sec:CI:justifying-quantum-Reichenbach} by assuming fundamentally unitary dynamics, there are many other ways of obtaining the same quantum conditional independence from different classical starting points. For example, an information-theoretic approach may obtain condition (3) as the first definition of quantum conditional independence. Regardless of one's philosophical predisposition, however, the fact that so many natural definitions are equivalent gives them all strength.

Condition (4) of Thm.~\ref{thm:CI:QCI-full} is of particular interest. In Sec.~\ref{sec:CI:circuits}, it was informally argued that channels satisfying this condition can represent two agents acting independently on a single input system. An important contribution would be to formalise this concept. That is, from a principled (perhaps operational) definition of what it means to ``act independently'' on a common input, prove that all channels that do so are of the form required by condition (4). This would provide another good way of justifying quantum conditional independence via condition (4). Assuming this is possible, it would establish channels of this form as very important generalisations of factorised quantum channels, potentially having important uses in quantum information, quantum theory in relativistic causal settings, and beyond.

In Sec.~\ref{sec:CI:circuits}, a new symbol, \qcopy{}, was added to quantum circuit diagrams to concisely depict channels satisfying condition (4). There is much potential for formalising the role of \qcopy{} in quantum circuits, as well as other diagrammatic formulations of quantum theory \cite{CoeckeKissinger17,ChiribellaDAriano+10,ChiribellaDAriano+11,ChiribellaDAriano+16,Hardy01,Hardy11}. This would be especially useful given a full characterisation of these channels as discussed above. Moreover, simply deriving rules for how diagrams involving \qcopy{} can be manipulated could greatly simplify many calculations in quantum causal models.

Quantum conditional independence characterises channels where the input can act as a complete common cause for the outputs. This enables a definition of quantum Reichenbach's principle. Classically, the framework of causal models may be seen as a generalisation of Reichenbach's principle. By way of analogy, Sec.~\ref{sec:CM:quantum-cm-generalise-reichenbach} proposed a definition of quantum causal models that most simply generalises quantum Reichenbach and then illustrated its utility with examples.

However, as noted in Sec.~\ref{sec:CM:quantum-definition}, it would be better to justify these causal models by an argument from fundamental unitarity, mirroring the justification for quantum conditional independence given in Sec.~\ref{sec:CI:justifying-quantum-Reichenbach}. Filling this gap is an important piece of further work, one that would no doubt benefit from a more thorough understanding of the behaviour of \qcopy{}. Not only that, it would be of particular relevance to those who believe that quantum dynamics \emph{is} fundamentally unitary just as the argument in Sec.~\ref{sec:CI:justifying-quantum-Reichenbach} is.

In quantum causal models, the process of ``linking out'' nodes was introduced in Sec.~\ref{sec:CM:predictions} as the 	quantum causal model analogue of marginalising over nodes in a classical causal model---that is, the mathematical procedure corresponding to ignoring that node. As a first step to being able to derive more properties of quantum causal models, including in applications to specific experiments, it would be very useful to investigate the properties of this linking out operation. In particular, a general specification for how linking out affects the model state, the causal structure, and the relationship between them should be found. One simple case is for linking out a leaf node (that is, a node with no children), where it is easy to check that the resulting model state is simply obtained by removing the factor corresponding to that node from Eq.~(\ref{eq:CM:quantum-Markov}). A corresponding specification for linking out a general node would be more difficult to find, but potentially much more useful.

Beyond these opportunities for further work on the formalism of quantum causal models itself, there are also many ways in which they might be used and extended.

An important use case for quantum causal models is in the analysis of experiments. Since the output of experiments is, generally speaking, statistical data, a first step in this direction would be a thorough analysis of what characterises statistics obtained from quantum experiments with certain causal structures. Given a set of experimental data, it is known how to easily check whether a certain variable could represent a complete common cause of another pair of variables, for example. Developing similar techniques based on quantum causal models would go a long way towards bringing them to experimental relevance.

Another use case for quantum causal models is in bringing a causal understanding to certain ``paradoxes'' or other puzzling thought experiments in quantum foundations. For example, the \emph{extended Wigner's friend} ``paradox'' \cite{Wigner67,Deutsch85,Brukner17,FrauchigerRenner16} has a branched causal structure and a foundational analysis of it may benefit from a causal understanding facilitated by quantum causal models, since the exact disanalogy between the classical and quantum cases remains somewhat unclear \cite{FrauchigerRenner16,BaumannHansen+16}.

One powerful way to extend the quantum causal models framework would be to find a d-separation theorem, or something similar \cite{Pearl09,SpirtesGlymour+01,HensonLal+14}. Classically, d-separation is a graphical criterion that applies to causal structures and the d-separation theorem establishes it as both sound and complete for conditional independence in a Markov causal model. This is one of the foundational results in the study of causal models and the basis of many causal discovery algorithms \cite{Pearl09,SpirtesGlymour+01}. If a similar enough theorem were found for quantum causal models, then it may be possible to use some of these algorithms already derived in the classical setting with minimal changes. Techniques used in Ref.~\cite{HensonLal+14} to prove a d-separation theorem for their version of quantum causal models could be useful for developing such an extension.

Finally, it is interesting to consider whether the techniques of Chaps.~\ref{ch:CI}, \ref{ch:CM} could be used as a template for developing formulations of causal models in frameworks for beyond-quantum physical theories. One example is the \emph{general probabilistic theory} (GPT) framework \cite{Hardy01,Hardy11,ChiribellaDAriano+10,ChiribellaDAriano+11}, which is a general way to discuss a wide landscape of theories including classical and quantum. Such a GPT causal model formalism would not only be useful for the study of GPTs themselves, but also could potentially unify quantum and classical causal models. As noted in Sec.~\ref{sec:CM:previous-attempts}, Refs.~\cite{HensonLal+14,Fritz16} present alternative formulations of quantum causal models (inequivalent to that presented here) that apply to broader frameworks in this way and may be useful resources for this effort.

\section{Time Travel and CTCs}

The final approach to ontology in quantum theory taken in this thesis was in Chap.~\ref{ch:TT}, which considered time travel to the past in a quantum universe. This was done by taking the quantum circuit approach, which has yielded the models of D-CTCs and P-CTCs. After reviewing this previous work in Sec.~\ref{sec:TT:review}, the significance of ontology and non-linearity in these models was discussed in Secs.~\ref{sec:TT:non-linearity}, \ref{sec:TT:role-of-ontology}. This led to the identification of two classes of new models for quantum time travel in Sec.~\ref{sec:TT:alternative-model-selection}: weighted D-CTCs and transition probability CTCs. From these, a specific transition probability model called T-CTCs was selected, thoroughly investigated, and compared to the previous models.

Both P-CTCs and T-CTCs make use of an operator $P \eqdef \Tr_{\textsc{CV}}\left( U \right)$ which is the partial trace of a unitary operator [Eq.~(\ref{eq:TT:P-CTC-EoM})]. Any further work that develops a fuller understanding of the properties of such operators would be in a good place to quickly identify further properties of P-CTCs and T-CTCs. The only general property of such operators so far proved is a tight upper bound on the action of $P$ on vector norms [Sec.~\ref{sec:TT:P-operator}].

The original paper on D-CTCs discussed how the second law of thermodynamics remains respected in that model (except within the chronology violating region) \cite{Deutsch91}. No such analysis has yet appeared for either the P-CTC or T-CTC models. This analysis would be interesting for further comparing the properties of these models and may have a bearing on their relative plausibilities.

Of the opportunities for further work based into the properties of D-, P-, and T-CTCs, perhaps the most intriguing is using them to clarify the relationships between non-linearity, computation, and distinguishability of quantum states. In particular, a well-known result in Ref.~\cite{AbramsLloyd98} claims that ``virtually any'' non-linear extension to quantum theory is able to solve $\mathsf{NP}$-complete and $\mathsf{\#P}$ problems in polynomial time; however, the method used in the proof of that result does not apply to the non-linear evolutions provided by T-CTCs. While T-CTCs can solve any $\mathsf{PP}\supseteq\mathsf{NP}$ problem in polynomial time and does therefore not provide a counter-example, it does call for a more thorough analysis of this claim. Reference \cite[Chap. 9]{Aaronson13a} suggests, in a conversational manner, that deterministic distinguishability of states is a necessary consequence of non-linearity. However, it has already been established that T-CTCs, while generally non-linear, cannot distinguish non-orthogonal quantum states in a single shot with certainty. In general, results such as these demonstrate a close relationship between non-linearity, computational power, and distinguishability of non-orthogonal states, but one that is not yet fully understood. It is clear that non-linearity is necessary for the latter two, but further study of time travel models may help to clarify the degree to which it is also sufficient.

Beyond these open questions for the study of D-, P-, and T-CTCs, this also suggests fruitful further work for the study of time travel beyond these models.

Perhaps the most obvious avenue for work beyond these models is to identify other particular weighted D-CTCs or transition probability CTCs that may warrant further investigation. A few of these are noted in Secs.~\ref{sec:TT:alternative-model-selection}, \ref{sec:TT:relation-to-other-alternatives} but none other than T-CTCs have received a full treatment. This would be especially interesting if a convincing ontological argument could be given for a particular model.

In Sec.~\ref{sec:TT:role-of-ontology}, intrinsic stochasticity and non-separability were identified as features of quantum theory that make it particularly difficult to find an ontologically sound model of time travel, compared to the classical context. The precise roles of these features could be clarified by starting with a generic classical model with time travel and incrementally adding stochastic and non-separable features. By starting from the ontologically uncontroversial classical setting, this could reveal concrete suggestions of how to consistently model quantum time travel. A complementary task would be to develop a model of time travel for Spekkens' toy model \cite{Spekkens07} (since that is a toy model for quantum theory that has an explicit ontology) and use it as an inspiration for time travel models in quantum theory.

The analysis of computation in the presence of time travel has potentially interesting implications for computation in generalised probabilistic theories (GPTs) \cite{LeeBarrett15,BarrettBeaudrap+17}. Since the GPT framework is designed to characterise an exceedingly large class of theories that might be applicable to physics, it seems likely that D-, P-, and T-CTC models should be expressible as GPTs. Expressing them in this way should help with comparisons between them, to classical theories, and to standard quantum theory. In particular, a general framework of computation in GPTs is developed in Ref.~\cite{LeeBarrett15} using several plausible assumptions, including those of ``tomographic locality'' and a ``uniformity condition''. There it was proved that no GPT computation in that framework is capable of solving problems outside $\mathsf{AWPP}\subseteq\mathsf{PP}$ in polynomial time. This is puzzling in light of the known results about D-CTCs, which are capable of solving all problems in $\mathsf{PSPACE}$ in polynomial time, since it is believed that $\mathsf{PSPACE}$ is strictly larger than $\mathsf{PP}$ \cite{AroraBarak09} (and therefore also $\mathsf{AWPP}$). This puzzle reveals a tension between the model of GPT computation of Ref.~\cite{LeeBarrett15} and the model of computation used in Ref.~\cite{AaronsonWatrous09} to prove that D-CTCs have the computational power of $\mathsf{PSPACE}$.

By explicitly expressing D-, P-, and T-CTCs as GPTs it should be possible to compare them to the assumptions of Ref.~\cite{LeeBarrett15}. Doing this would certainty clarify the discrepancy between the results of Ref.~\cite{AaronsonWatrous09} and Ref.~\cite{LeeBarrett15} and may even suggest ways in which they might be improved. \emph{A priori}, D-, P-, and T-CTCs seem like they should all define plausible enough GPTs which have similarly plausible notions of computation. Therefore, it would be interesting to see, especially in the case of D-CTCs, whether the assumptions used in Ref.~\cite{LeeBarrett15} hold for these theories, and whether any subtleties are brought to light. In particular, it is not obvious whether the GPTs corresponding to D-, P-, or T-CTCs would be either tomographically local or satisfy the uniformity condition. In general, the study of GPTs could benefit from some more fully worked-out examples and time travel models could benefit from some new settings in which to examine them.

Finally, the framework of quantum causal models developed in Chap.~\ref{ch:CM} could perhaps also be of use for in the study of quantum time travel. The stipulation of time travel to the past is, after all, a causal one. This will not be as simple as writing a causal model for a CTC, since causal models rule out such causal loops by fiat. However, certain intuitions and understandings from causal models, including the characterisation of channels acting independently on a single system, may be useful in trying to understand smaller parts of the time travel process.

\section{To Conclude}

It is undoubtedly ambitious to attempt to describe what ``fundamental reality'' can and cannot be like. The way to even partially achieve this, while not losing your head or embarrassing yourself too much, is to be careful and precise with your statements and arguments. This makes concluding a bothersome business, as substantial statements tend to need so many caveats and conditions you might as well just re-state the theorems. Perhaps it is best then to simply say this. If you want real states of affairs respecting probability, then you probably need to accept a certain similarity to quantum states (including superpositions). If you want to correctly describe causal influences, then you need to understand quantum conditional independence and all that it implies. If you want to make sense of time travel to the past, then there are many ontological choices to be made and a wide variety of models to explore. If you want to know what quantum theory means for ontology, then there is plenty more work to be done.

\makeatletter
\renewcommand{\@makechapterhead}[1]{\chapterheadstartvskip%
  {\size@chapter{\sectfont\raggedleft
      {\chapnumfont
        \ifnum \c@secnumdepth >\m@ne%
        \if@mainmatter {\LARGE Appendix\quad}\thechapter%
        \fi\fi
        \par\nobreak}%
      {\raggedleft\advance\leftmargin10em\interlinepenalty\@M #1\par}}
    \nobreak\chapterheadendvskip}}
\makeatother

\begin{appendices}
\include{appendix-proof-thm-no-maximally-epistemic}
\include{appendix-proof-thm-small-symmetric-overlap-large-d}
\end{appendices}

\addcontentsline{toc}{chapter}{Bibliography}

\bibliography{references}

\end{document}